\documentclass[lettersize,journal,onecolumn,comsoc,10pt]{IEEEtran}

\IEEEoverridecommandlockouts                              

\usepackage{graphicx} 
\usepackage[utf8]{inputenc}
\usepackage{scalerel}
\usepackage{cite}
\usepackage{amsmath} 
 
\usepackage{amsthm}
\usepackage{amssymb}
\usepackage{amsfonts}
\usepackage{float} 
\usepackage{graphicx}
\usepackage{subcaption}
\newcommand{\headcav}{{{\scaleto{\mathrm{H}}{3.5pt}}}}
\newcommand{
\tailcav}{{{\scaleto{\mathrm{T}}{3.5pt}}}}
\newcommand{
\HVn}{{{\scaleto{\mathrm{N}}{3.5pt}}}}
\newcommand{\platoon}{\mathrm{p}}

\newcommand{\hhv}{\mathrm{d}}
\newcommand{\diff}{\,\mathrm{d}}
\newcommand{\imaginaryj}{\mathrm{j}}
\newtheorem{theorem}{Theorem}
\newtheorem{lemma}{Lemma}
\newtheorem{definition}{Definition}

\theoremstyle{definition}
\newtheorem{remark}{\normalfont\bfseries Remark}
\usepackage{xcolor}
\usepackage{arydshln}
\usepackage{authblk}
\usepackage{multirow}
\allowdisplaybreaks
\usepackage[section]{placeins}

\begin{document}

\title{\LARGE \bf 
Leveraging Cooperative Connected Automated Vehicles \\ for Mixed Traffic Safety}

\author[1]{Chenguang Zhao}
\author[2]{Tamas G. Molnar}
\author[1,*]{Huan Yu\thanks{* corresponding author: Huan Yu (huanyu@ust.hk)}}

\affil[1]{Thrust of Intelligent Transportation, The Hong Kong University of Science and Technology (Guangzhou), Nansha, Guangzhou, 511400, Guangdong, China.}
\affil[2]{Department of Mechanical Engineering,
Wichita State University, Wichita, KS 67260, USA.}

\renewcommand\Authands{ and }

\maketitle


\begin{abstract}
    The introduction of connected and automated vehicles (CAV) is believed to reduce congestion, enhance safety, and improve traffic efficiency. Numerous research studies have focused on controlling pure CAV platoons in fully connected automated traffic, as well as single or multiple CAVs in mixed traffic with human-driven vehicles (HVs). CAV cruising control designs have been proposed to stabilize the car-following traffic dynamics, but few studies has considered their safety impact, particularly the trade-offs between stability and safety.  In this paper, we study how cooperative control strategies for CAVs can be designed to enhance the safety and smoothness of mixed traffic under varying penetrations of connectivity and automation. Considering mixed traffic where a pair of CAVs travels amongst HVs, we design cooperative feedback controllers for the pair CAVs to stabilize traffic via cooperation and, possibly, by also leveraging connectivity with HVs. The real-time safety impact of the CAV controllers is investigated using control barrier functions (CBF). We construct CBF safety constraints, based on which we propose safety-critical control designs to guarantee CAV safety, HV safety and platoon safety. Both theoretical and numerical analyses have been conducted to explore the effect of CAV cooperation and HV connectivity on stability and safety. Our results show that the cooperation of CAVs helps to stabilize the mixed traffic while safety can be guaranteed with the safety filters.
    Moreover, connectivity between CAVs and HVs offers additional benefits: if an HV connects to an upstream CAV (i.e., the CAV looks ahead), it helps the CAV to stabilize the upstream traffic, while if an HV connects to a downstream CAV (i.e., the CAV looks behind), the safety of this connected HV can be enhanced. 
\end{abstract}


\section{Introduction}
The integration of automation and connectivity in intelligent vehicles has been envisioned to improve road transportation efficiency, fuel consumption, and driving safety. Many studies have focused on the control of connected and automated vehicles (CAVs) to explore their potential for improving traffic under different penetration rates ranging from a single automated vehicle to fully connected automated traffic systems~\cite{mousavi2021investigating,ding2020penetration}.
Before fully automated and connected traffic becomes reality, there will still be an inevitable transition period of mixed traffic systems, that are characterized by frequent interactions between CAVs and conventional human-driven vehicles (HVs). 
In mixed traffic, the cooperative control of CAVs remains a significant challenge.
Especially, the safety impact of cooperative CAV control strategies on surrounding vehicles needs to be addressed. In this paper, we investigate this problem for a mixed vehicle platoon that includes a pair of cooperative CAVs traveling amongst (possibly connected) HVs as shown in Fig.~\ref{fig:framework}. We explore in detail how coordination between CAVs and feedback from surrounding connected HVs influence the effectiveness of control strategies on safety and smoothness of mixed traffic, under varying penetrations of connectivity and automation.

\subsection{Stability and safety by controlling a single CAV}

Automated vehicles have the potential to stabilize traffic by introducing smooth driving motions through car-following behaviors of vehicle platoon. Speed perturbations are attenuated in the process of propagating from the leader vehicle to the follower.~\cite{giammarino2020traffic,stern2018dissipation,yu2018stabilization,zhou2023data}. To achieve smooth motion, controllers have been designed for automated vehicles under different penetration rate of CAVs in traffic,  accessibility of surrounding traffic information.
In adaptive cruise control (ACC), the automated vehicle is controlled based on data measured via onboard sensors, that include its speed, the preceding vehicle's speed, and the gap (distance) ahead~\cite{gunter2020commercially}. The performance of automated vehicles can be further improved by using additional information from vehicle-to-vehicle (V2V) connectivity. CAVs equipped with communication devices may obtain data from downstream or upstream connected HVs that are also equipped with communication devices. This is leveraged by connected cruise control, which enables CAVs to respond not only to the immediate preceding vehicle but also to other connected HVs downstream. Field experiments in~\cite{jin2018experimental} have shown that this additional information from the downstream traffic helps the CAV to drive smoother compared to ACC-only vehicle, which improves traffic stability.
Connectivity between the CAV and upstream HVs can also improve traffic stability  by the notion of leading cruise control~\cite{wang2021leading}. As the CAV responds to connected HVs upstream, the stability of the upstream traffic is further enhanced.

In order to implement controllers for CAVs in practice, safety must be guaranteed~\cite{ye2019evaluating}. In this paper, we focus on longitudinal control, and thus safety refers to eliminating the risk of rear-end collisions. Longitudinal controller design exhibits a trade-off between stability and safety~\cite{li2022trade}. For example, when the vehicle ahead of the CAV decelerates, a smaller deceleration of the CAV leads to smoother traffic but also a higher risk of collisions. To ensure safety for the controlled dynamical systems, representative control techniques include model predictive control (MPC)~\cite{bai2022robust}, reachability analysis~\cite{althoff2014online,zhao2022formal}, and control barrier function (CBF)~\cite{ames2014control,xiao2021bridging}. MPC typically minimizes stability-related indices such as speed perturbation or energy consumption while safety is usually considered as constraint over the prediction horizon.  CBF, on the other hand, directly gives a constraint on the controller to meet safety of the current system state. Therefore, CBFs avoid the computation burden from the prediction of future dynamics which could possibly be inaccurate. Furthermore, CBFs can be integrated with pre-designed nominal controllers. In particular, safety-critical controllers have been developed that minimize the deviation from the nominal controller while satisfying CBF safety constraint~\cite{ames2019control,ames2014control}. In this approach, CBF acts as a safety filter that only alters the nominal controller when the system is in danger of violating safety. 

Many recent results have employed CBFs to enhance traffic safety.
CBFs have been integrated with adaptive cruise control in~\cite{ames2014control} to ensure the safety of automated vehicles. CBF-based safety filters have also been developed for connected cruise control in~\cite{molnar2023safetyCCC} to avoid collisions for a CAV that is connected to HVs downstream. Furthermore, CBF constraints have been designed in~\cite{zhao2023safety} to ensure safety for both the CAV and following HVs when the CAV is connected to HVs upstream. In this paper, we develop a novel CBF framework to evaluate the stability and safety impact of cooperative CAV control strategies on the upstream and downstream traffic. More importantly, we analyze the trade-offs between stability and safety, in particular, how connectivity of HVs, ACC and cooperation between CAVs play a role in it.  The safety-critical control designs proposed in this paper achieve stability and safety simultaneously.

\subsection{Cooperative control of multiple CAVs}

Besides analyzing the effect of a single automated vehicle on traffic, research has also shown that communication and coordination between multiple CAVs may further improve the overall performance of the traffic system~\cite{garg2023can,do2019simulation,papadoulis2019evaluating,talebpour2016influence}. The cooperative controllers have been considered for various traffic scenarios, such as cooperative cruise control and platooning~\cite{wang2024robust,johansson2023platoon}, cooperative merging at highway on-ramps~\cite{rios2016automated,wang2018distributed}, cooperative lane changing~\cite{zheng2019cooperative,zhang2022hybrid}, and cooperative eco-driving at signalized and unsignalized intersections~\cite{han2020energy,bai2022hybrid,chen2023nearly}. In longitudinal control that this paper focuses on, an important case is when multiple adjacent CAVs follow each other and form a CAV platoon.  Related research has revealed the positive improvement brought by such CAV platoons on traffic throughput~\cite{jin2020analysis,zhou2021analytical}, energy~\cite{kim2021compact}, stability~\cite{zhou2023autonomous}, mobility~\cite{liu2020mobility}, and safety~\cite{yi2020using}. 

To unleash the potential of CAV platoons, however, it is crucial to design controllers that effectively coordinate CAVs control strategies. A prevailing approach is cooperative adaptive cruise control (CACC), which is an extension of ACC~\cite{dey2015review,dai2022exploring}. Compared with ACC, CACC allows the CAVs within the platoon to have a smaller gap, and improves traffic stability, efficiency, and throughput~\cite{brunner2022comparing}. 
As for safety, CACC with well-designed controller gains also achieves collision-free driving for CAVs~\cite{bekiaris2023robust}.

Despite their benefits, CAV platoons may be difficult to implement in practice, since full penetration of automation and connectivity is required within the platoon.
As opposed, the penetration rate of CAVs is typically low in current traffic conditions. With low penetration, a more common scenario is mixed vehicle platoons consisting of both CAVs and HVs as shown 
in Fig.~\ref{fig:framework}. To coordinate CAVs that are separated by HVs, controllers have been designed via multiple tools, such as MPC~\cite{gong2018cooperative,qiu2023cooperative}, feedback control~\cite{guo2023connected}, and reinforcement learning~\cite{shi2021connected}. Analyses in these studies have shown that, compared with controlling the two CAVs separately, cooperative strategies further improves mixed traffic efficiency and stability. For safety, MPC-related controllers~\cite{gong2018cooperative,qiu2023cooperative} include CAV safety conditions as constraints in the optimization problem, while the other controllers fail to provide safety guarantees.

The above research considers the case where HVs are not connected, and thus, feedback about their motion may not be available to the CAVs.  Yet, as it has been discussed above for the control of a single CAV, the traffic efficiency~\cite{orosz2016connected,wang2021leading,shladover2012impacts,di2019cooperative} and safety~\cite{zhao2023safety} can be further improved via connecting HVs to CAVs. For the coordination of multiple CAVs, it remains unexplored how additional connectivity to HVs affects traffic stability and safety.  On the other hand, scarce studies have focused on designing CAV control strategies to guarantee the safety of surrounding HVs in mixed traffic, which will be discussed in our paper.  We design safety-critical controllers for a cooperating pair of CAVs that stabilizes the mixed traffic while considering both CAV and HV safety, and we analyze how the connection of HVs can be leveraged to improve the stability and safety of the system.

\begin{figure}[t]
    \centering
    \includegraphics[width=1\linewidth]{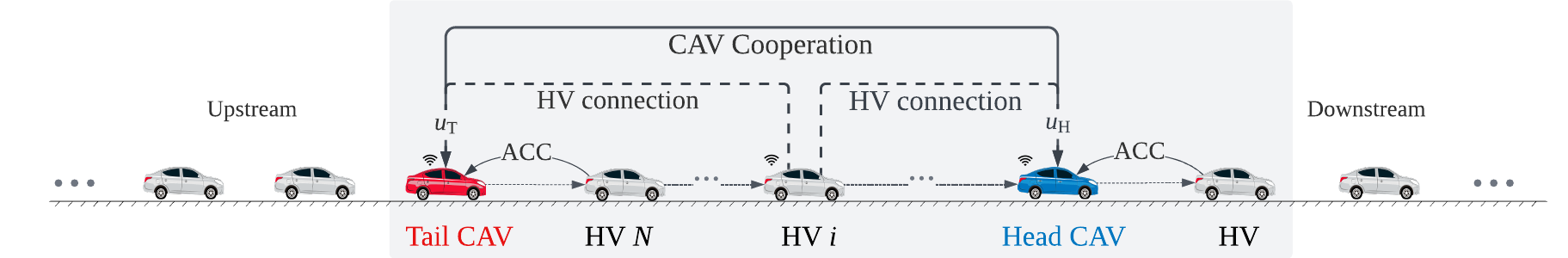}
    \caption{Safety-critical stabilization of a pair of CAVs traveling in mixed traffic. The head CAV follows a downstream head HV  and leads $N$ following HVs, while the tail CAV follows the last HV in this vehicle platoon.  We design controllers for the CAVs to alleviate congestion and also maintain formal safety guarantees considering the mixed vehicle platoon shaded in grey.
    }
    \label{fig:framework}
\end{figure}

\subsection{Contribution}

In this paper, we consider a mixed vehicle platoon shown in Fig.~\ref{fig:framework} that contains HVs enclosed by two CAVs. The CAVs are equipped with communication devices and thus can obtain information about each other's state, while some of the middle HVs may also have the ability for V2V communication. By leveraging the information from connectivity, we design controllers for the pair CAVs to stabilize the motion of the mixed vehicle platoon while maintaining safety of the CAVs and the middle HVs. 
We first design a nominal stabilizing controller consisting of three parts: adaptive cruise control, feedback of middle HV states (if there are connected HVs), and CAV coordination. 
Then, we discuss safety against rear-end collisions and propose safety notion in terms of the CAVs, the middle HVs, and the overall platoon by CBF. 
We perform safety analysis for the nominal controller, and then design CBF safety filters to guarantee CAVs, HVs, and the overall platoon respectively.
Finally, we conduct numerous simulations to validate our safety-critical cooperative CAV controllers
 and analyze their impact on stabilty and safety, sensitivity to penetration rate of CAVs, and robustness to uncertain human driver behaviors. 

The main contribution of this paper lies in proposing a methodological cooperative control framework for CAVs to guarantee safety of the mixed traffic, including CAVs safety, HVs safety and platoon safety. 
In particular, we analyze how the connection of HVs affects stability and safety, and show that if an HV is equipped with V2V communication devices, the CAVs can simultaneously enhance the safety of the HV and the smoothness of the following traffic.
In our previous works~\cite{zhao2023safety, guo2023connected, zhao2024safetyACC}, we obtained preliminary results on using CBFs to enhance traffic safety. The impact of cooperative strategies on trades-off between stability and safety has not been explored yet, especially for mixed traffic. More importantly, we leverage connection of HVs to improve the safety of mixed system by establishing CBFs safety constraints for the pair CAVs, HVs, and platoon. The practical impact of our proposed control strategies has been extensively studied across various safety-critical scenarios, considering different penetration rates of CAVs and accounting for uncertain human driver behaviors.

The remainder of this paper is organized as follows. We formulate the model of the mixed vehicle platoon in Section~\ref{sec:formulation}. We establish a nominal controller and address stability in Section~\ref{sec:stability}.
In Section~\ref{sec:safety}, we conduct safety impact analysis and propose cooperative control designs that guarantee the safety of the CAVs, HVs, and platoon based on CBFs. We conduct numerical simulations in Section~\ref{sec:simulation} to validate the designed controller.
We analyze controller performance from various aspects in Section~\ref{sec:performance analysis}, including controller parameter selection and robust stability and safety against unmodeled HV dynamics.

\section{Problem Formulation}\label{sec:formulation}

We consider a mixed vehicle platoon, shown in Fig.~\ref{fig:framework}, that includes two CAVs and $N$ HVs in-between.
The head CAV (H-CAV) follows an HV that leads the vehicle platoon and may cause velocity disturbance, while the tail CAV (T-CAV)  follows the last HV (HV-$N$) in the platoon. The two CAVs can measure their own gap, their own speed, and the speed of their preceding vehicle by on-board sensors (such as radar, lidar, or camera).
Furthermore, the head and tail CAVs share information of each other's real-time position and speed via V2V communication.  The states of the middle HVs are available to the CAVs if they are equipped with communication devices (i.e., they are connected HVs). We will discuss how to design safe cooperative CAV control strategies from null connectivity to full connectivity of middle HVs. 

The state variable of the mixed traffic system $x(t)\in \mathbb{R}^n$ with $n=2N+4$ is defined as:
\begin{align}
    x = [ \underbrace{s_{\headcav},v_{\headcav}}_{\text{head CAV}}, \underbrace{  s_{1},v_{1},\cdots, s_{\HVn},v_{\HVn}}_{\text{HV}}, \underbrace{s_{\tailcav},v_{\tailcav}}_{\text{tail CAV}}]^{\top} \in \mathbb{R}^{n},
\end{align}
with $s_{\headcav}(t)\in \mathbb{R}$ being the gap between the head CAV and its leader HV, $v_{\headcav}(t)\in \mathbb{R}$ being the speed of the head CAV, $s_{\tailcav}(t)\in \mathbb{R}$ and $v_{\tailcav}(t)\in \mathbb{R}$ being the gap and speed of the tail CAV, $s_{i}(t)\in \mathbb{R}$ and $v_{i}(t)\in \mathbb{R}$ being the gap and speed of HV-$i$ with $i \in \{1,\ldots,N\}$. Meanwhile, the leader HV's speed $v_{\hhv}(t)\in \mathbb{R}$ is an external disturbance for the vehicle platoon.
We design a control input $u_{\headcav}(t)\in \mathbb{R}$ for the head CAV and a control input $u_{\tailcav}(t)\in \mathbb{R}$ for the tail CAV. The controller is represented by:
\begin{align}
    u=\begin{bmatrix}u_{\headcav} \\ u_{\tailcav}
    \end{bmatrix} \in \mathbb{R}^2.
\end{align}
The dynamics of the mixed-traffic vehicle platoon are governed by: 
\begin{align}\label{eq:system}
    \dot x = f(x,v_{\hhv}) + g(x) u + d,
\end{align}
with $ f : \mathbb{R}^n \times \mathbb{R} \to \mathbb{R}^{n}$, $g : \mathbb{R}^{n} \to \mathbb{R}^{n\times 2}$, and $d(t)\in \mathbb{R}^n$ given below according to the models used for capturing the car-following motions of CAVs and HVs.
We describe the longitudinal motion of HVs and two CAVs as follows. 

\textbf{HV dynamics}:
the car-following model of HV-$i$ is given by
\begin{align}
    \dot s_i &= v_{i-1}-v_i, \label{eq:system HV s} \\
    \dot v_i &= F_i(s_i,v_i,\dot s_i) +  d_i.  \label{eq:system HV v}
\end{align}
The acceleration function  $F_i: \mathbb{R}^3 \to \mathbb{R}$ models the longitudinal driving behaviors of HV-$i$ based on its gap $s_i$, speed $v_i$, and speed difference $\dot s_i$.
Specifically, function $F_i$ is a generic acceleration function that allows the adoption of many commonly-used car-following models such as the optimal velocity model (OVM)~\cite{Bando1998} or the intelligent driver model (IDM)~\cite{treiber2000congested}.
In practice, human drivers often present more complex car-following behaviors that are difficult to describe accurately by simple models, hence we use $d_i(t)\in   \mathbb{R}$ to represent the unmodeled HV dynamics as an additive disturbance. It is noted that HV-1 follows the head CAV; thus, we use the notation $v_0=v_\headcav$.

\textbf{CAV dynamics}:
the head CAV follows an HV, and we will design the driving strategy of H-CAV as a control input, given by
\begin{align}
    \dot s_\headcav &= v_{\hhv} - v_\headcav, \label{eq:system head CAV s}\\
    \dot v_\headcav &= u_\headcav. \label{eq:system head CAV v}
\end{align}
In analogy, the gap $s_{\tailcav} $ and speed $v_{\tailcav} $ of the tail CAV are governed by:
\begin{align}
    \dot s_{\tailcav} &= v_{N} - v_{\tailcav}, \label{eq:system tail CAV s}\\
    \dot v_{\tailcav} &= u_{\tailcav}, \label{eq:system tail CAV v}
\end{align}
where $u_{\tailcav} \in \mathbb{R}$ controls the acceleration of the T-CAV.

For the mixed vehicle platoon controlled by the two cooperative CAVs, we thus have the system model as~\eqref{eq:system} with:
\begin{equation}
\begin{split}
& f(x,v_{\hhv}) = 
\begin{bmatrix}
    f_{\headcav}(x,v_{\hhv}) \\   
    f_{1}(x) \\
    \vdots \\
    f_{\HVn}(x) \\
    f_{\tailcav}(x) 
\end{bmatrix} \in \mathbb{R}^n,   \quad
f_{\headcav}(x,v_{\hhv}) = 
\begin{bmatrix}
    v_{\hhv} - v_{\headcav}\\
    0
\end{bmatrix},    \quad 
f_i(x) = 
\begin{bmatrix}
    v_{i-1} - v_i \\ F_i(s_i,v_i,v_{i-1}-v_i)
\end{bmatrix},  \quad
f_{\tailcav}(x) = 
\begin{bmatrix}
    v_{N} - v{\tailcav} \\ 0
\end{bmatrix}, \\
& g(x) = 
\begin{bmatrix}
    g_{\headcav} & g_{\tailcav}
\end{bmatrix}  \in \mathbb{R}^{n\times 2},  \quad
g_{\headcav} = 
\begin{bmatrix}
    0  \\1  \\0  \\\vdots  \\0 
\end{bmatrix}  \in \mathbb{R}^n, \quad
g_{\tailcav} = 
\begin{bmatrix}
    0 \\ \vdots \\0 \\0 \\1
\end{bmatrix}  \in \mathbb{R}^n, \quad
d =  \begin{bmatrix}
        0 & 0 & 0 & d_1 & \cdots & 0 & d_{N} & 0 & 0
\end{bmatrix}^{\top} \in \mathbb{R}^n. 
\end{split}
\label{eq:system_expressions}
\end{equation}

In the following parts, we first consider the case in which the human driver's model is fully known, i.e., $d_i=0$ in \eqref{eq:system HV v}. In Section~\ref{sec:stability}, we design cooperative CAV controllers to stabilize the system~\eqref{eq:system}, and further analyze the safety impact of the proposed cooperative controllers and then design safety filters to enhance the safety of the mixed traffic system in Section~\ref{sec:safety}. We run simulations to validate the safety filters in Section~\ref{sec:simulation} and conduct performance analysis in Section~\ref{sec:performance analysis}. In Section~\ref{sec:performance analysis}, we also analyze the robustness of the controller with an inaccurate driver's model, i.e., $d_i\ne 0$.

\section{Cooperative Stabilizing CAV Controllers}\label{sec:stability}

In this section, we design stabilizing nominal controllers for the two CAVs so that stable motion of the vehicle platoon is achieved for two different notions of stability, i.e. \textit{plant stability} and 
 \textit{head-to-tail string stability}.

\subsection{Nonlinear cooperative control design for the head and tail CAVs }

The nominal CAV controllers are designed to respond to the state $x$ of the traffic and the speed $v_{\hhv}$ of the leader HV:
\begin{align}
\begin{split}
    u_{\headcav} & = k_{\headcav,{\rm n}}(x,v_{\hhv}), \\
    u_{\tailcav} & = k_{\tailcav,{\rm n}}(x),
\end{split}
\label{eq:nominal_controller}
\end{align}
where subscript ${\rm n}$ stands for nominal, and the expressions of the nominal head and tail CAV controllers, ${k_{\headcav,{\rm n}}: \mathbb{R}^n \times \mathbb{R} \to \mathbb{R}}$ and ${k_{\tailcav,{\rm n}}: \mathbb{R}^n \to \mathbb{R}}$, are chosen as follows.
The controller of each CAV is designed to consist of three parts:
(i) adaptive cruise control (ACC) based on the preceding vehicle;
(ii) state feedback of the middle HVs that are connected to the CAV (if there are any);
and (iii) cooperative response to the other CAV. If there is no connectivity between the  CAV pair and HVs, the control strategy of the two CAVs will fall back to the simple ACC. We denote the set of middle HVs that are connected to the head CAV as $\mathcal{N}_{\headcav}\subseteq \{1,2,\ldots,N\}$.
The controller of the H-CAV is given by:
\begin{align}
    k_{\headcav,{\rm n}}(x,v_{\hhv}) = \underbrace{\alpha_{\headcav}(V_\headcav(s_\headcav) - v_{\headcav}) + \beta_{\headcav,{\hhv}}(W(v_{\hhv}) - v_{\headcav}) }_{\text{Adaptive cruise control}} +\underbrace{\textstyle{
    \sum_{i\in \mathcal{N}_{\headcav}}} \beta_{\headcav,i} (W(v_i) - v_{\headcav}) }_{\text{HV feedback}}   + \underbrace{\beta_{\headcav,\tailcav} (W(v_{\tailcav}) - v_{\headcav})}_{\text{CAV cooperation}} . \label{eq:nominal controller head CAV}
\end{align}
The first two terms are the response to the preceding vehicle, the third term is the response to the middle HVs (that is omitted if no HVs are connected to the head CAV, i.e., ${\mathcal{N}_{\headcav} = \emptyset}$), and the fourth term is the response to the tail CAV. We propose this controller by extending the design in ~\cite{guo2023connected} with response to the middle HVs. Parameters $\beta
_{\headcav,\hhv}$, $\beta
_{\headcav,i}$, and $\beta
_{\headcav,{\tailcav}}$ are the control gains with respect to the speeds of the head HV, middle HVs, and the tail CAV, respectively. The function $W:\mathbb{R}\to \mathbb{R}$ is defined as:
\begin{align}
    W(v) = \min\{v,v_{\max}\}, \label{eq:W}
\end{align}
with $v_{\max}$ being the maximum speed.
The control gain $\alpha_{\headcav}$ adjusts the head CAV's acceleration with respect to the desired speed $V_\headcav(s_\headcav)$ based on the gap $s_\headcav$.
Function $V_\headcav:\mathbb{R}\to \mathbb{R}$ is given below.

The tail CAV's acceleration is controlled as:
\begin{align}
    k_{\tailcav,{\rm n}}(x) = \underbrace{\alpha_{\tailcav} (V_{\tailcav}(s_\tailcav) - v_\tailcav) + \beta_{\tailcav,N} (W(v_N) - v_{\tailcav})}_{\text{Adaptive cruise control}}
    + \underbrace{\textstyle{\sum_{i\in \mathcal{N}_{\tailcav}} \beta_{\tailcav,i} (W(v_i) - v_{\tailcav})}}_{\text{HV feedback}} + \underbrace{\beta_{\tailcav,\headcav}(W(v_\headcav) - v_{\tailcav})}_{\text{CAV cooperation}}, \label{eq:nominal controller tail CAV}
\end{align}
where the meaning of each term in~\eqref{eq:nominal controller tail CAV} is analogous to that in~\eqref{eq:nominal controller head CAV}.
In controller~\eqref{eq:nominal controller tail CAV} $\mathcal{N}_{\tailcav}\subseteq \{1,2,\ldots,N-1\}$ denotes the set of middle HVs that are connected to the tail CAV excluding the preceding HV-$N$ (and the third term is omitted if no HVs are connected to the tail CAV, i.e., $\mathcal{N}_{\tailcav} = \emptyset$).
Parameters $\alpha_{\tailcav}$, $\beta_{\tailcav,N}$, $\beta_{\tailcav,i}$, and $\beta_{\tailcav,\headcav}$ are the corresponding control gains, while $V_\tailcav(s_\tailcav)$ is a desired speed based on the gap $s_\tailcav$.

The gap-dependent desired speeds $V_\headcav$ and $V_\tailcav$ can be designed for each CAV respectively. For instance, they are chosen differently by vehicle manufactures.
The upcoming analysis is conducted for general $V_\headcav$ and $V_\tailcav$ functions.
For numerical examples, we will use the same range policy:
\begin{equation}\label{eq:Vs}
    V_\headcav(s) = V_\tailcav(s) =
    \begin{cases}
    0, & s \leq s_{\mathrm{st}}, \\
    \kappa (s-s_{\mathrm{st}}), & s_{\mathrm{st}}<s<s_{\mathrm{go}}, \\
    v_{\max}, & s \geq s_{\mathrm{go}},
    \end{cases}
\end{equation}
where $s_{\mathrm{st}}$ and $s_{\mathrm{go}}$ are the standstill gap and free-driving gap, respectively, and ${\kappa = v_{\max}/(s_{\mathrm{go}}-s_{\mathrm{st}})}$.

With the controller~\eqref{eq:nominal_controller}, the dynamics~\eqref{eq:system} of the mixed vehicle platoon becomes:
\begin{equation}
    \dot{x} = 
    F(x,v_{\hhv}) = f(x,v_{\hhv}) + g_{\headcav} k_{\headcav,{\rm n}}(x,v_{\hhv}) + g_{\tailcav} k_{\tailcav,{\rm n}}(x),
\label{eq:closed_loop_system}
\end{equation}
where ${F(x,v_{\hhv})}$
represents the car-following dynamics of the HV model and the proposed CAV controllers.
Next, we analyze these dynamics, and we design the control gains denoted by $\alpha$ and $\beta$ for the two CAVs such that their controllers achieve stable motion for the vehicle platoon.
The linear stability is analyzed. We provide stability charts that identify the control gains for stability guarantees.

\subsection{Stability analysis}\label{sec:subsec stability chart}

We analyze the stability by considering the behavior of the mixed-autonomy traffic system~\eqref{eq:closed_loop_system} around its equilibrium.
At the equilibrium, each vehicle in the vehicle platoon has the same constant speed $v^*$ and keeps a constant gap. The equilibrium gap for each HV, $s_i^*$, is given by $F_i(s_i^*,v^*,0) = 0$. 
For the H-CAV and T-CAV, their equilibrium gaps $s_{\headcav}^*$ and $s_{\tailcav}^*$ are given by $V_{\headcav} (s_{\headcav}^*) = v^*$ and $V_{\tailcav} (s_{\tailcav}^*) = v^*$, respectively.
For the system~\eqref{eq:closed_loop_system}, the equilibrium state is
\begin{align}
    x^* = [s_{\headcav}^*,v^*,s_1^*,v^*,\cdots,s_{\HVn}^*,v^*,s_{\tailcav}^*,v^*]^\top,
\end{align}
and it satisfies ${F(x^*,v^*)=0}$.

We analyze stability by considering that speeds and gaps fluctuate around their equilibrium value.
The perturbations are described by: 
${\tilde{v}_{\hhv} = v_{\hhv} - v^*}$, $\tilde{s}_{\headcav} = s_{\headcav} - s^*$, $\tilde{v}_{\headcav} =v_{\headcav} -v^*$, ${\tilde{s}_{\tailcav} = s_\tailcav - s_\tailcav^*}$, ${\tilde{v}_{\tailcav} = {v}_{\tailcav} - v^*}$, ${\tilde s_i = s_i  - s_i^*}$, and ${\tilde v_i = v_i - v^*}$, which are written compactly as:
\begin{equation}
    \tilde{x}(t) = x(t) - x^*, \quad
    \tilde{v}_{\hhv} = v_{\hhv} - v^*.
    \label{eq:perturbation}
\end{equation}
By linearizing the mixed traffic system~\eqref{eq:closed_loop_system}, the evolution of these perturbations is given in the form:
\begin{equation}
    \dot{\tilde{x}} = A \tilde{x} + B \tilde{v}_{\hhv}.
    \label{eq:linearized_system}
\end{equation}
This linearized system, with the expressions of $A \in \mathbb{R}^{n \times n}$ and $B \in \mathbb{R}^{n}$, is derived in Appendix~\ref{appdx:stability} as~\eqref{eq:linear AB}. We consider two types of stability for the linearized system: plant stability and head-to-tail string stability.

\begin{definition}[Internal stability]
    System~\eqref{eq:linearized_system} is plant stable if it is asymptotically stable for ${\tilde{v}_{\hhv}(t)=0}$.
\end{definition}

\begin{definition}[Input-output stability]\label{definition:head to tail string stability}
    System~\eqref{eq:linearized_system}
    head-to-tail string stable if
    ${\sqrt{\int_0^{\infty} \tilde{v}_{\tailcav}(t)^2 \diff t} < \sqrt{\int_0^{\infty} \tilde{v}_{\hhv}(t)^2 \diff t}}$
    for any square integrable $\tilde{v}_{\hhv}$ and ${\tilde{x}(0)=0}$.
\end{definition}

\begin{remark}[Plant and string stability]
   The plant stability refers to the \textit{internal} dynamics of the system when the external disturbance $\tilde{v}_{\hhv}=0$ , while head-to-tail string stability describes the system's response to the \textit{external} disturbance $\tilde{v}_{\hhv}$ when initial states ${\tilde{x}(0)=0}$. With respect to their implications for the mixed traffic system, plant stability means that, when the downstream traffic travels at the equilibrium speed $v^*$, the platoon will approach the equilibrium state $x^*$ associated with the same speed $v^*$ and constant spacing for each vehicle. Head-to-tail string stability, on the other hand, requires that when the leader HV has some speed perturbation $\tilde{v}_{\hhv}$ the tail CAV will experience a smaller perturbation $\tilde{v}_{\tailcav}$ in its speed, i.e., the speed perturbations in downstream traffic are attenuated by the vehicle platoon and the traffic becomes smoother. For traffic system, the head-to-tail string stability is a more restrictive performance requirement than the plant stability. 
\end{remark}

We analyze plant and head-to-tail string stability using the {\em head-to-tail transfer function}  defined as:
\begin{align} \label{eq:G_def}
    G(s) = \frac{\widetilde{V}_{\tailcav}(s)}{\widetilde{V}_{\hhv}(s)},
\end{align}
with $\widetilde{V}_{\tailcav}(s)$ and $\widetilde{V}_{\hhv}(s)$ being the Laplace transforms of the speed perturbations of the leader HV, $\tilde v_\hhv$, and the tail CAV, $\tilde v_{\tailcav}$, respectively.
The head-to-tail transfer function $G(s)$ of the linearized system~\eqref{eq:linearized_system} is derived in Appendix~\ref{appdx:stability}.
Its final expression is given as follows.

\begin{lemma}\label{theorem:G}
    The head-to-tail transfer function of system~\eqref{eq:linearized_system} is:
    \begin{align} \label{eq:G}
        G(s) = \frac{N(s)}{D(s)},
    \end{align}
    where the numerator is given by:
    \begin{align}
        N(s) = (\beta_{\headcav,\hhv} s+ \xi_{\headcav} ) \left(\beta_{\tailcav,\headcav}s P_0  + (\beta_{\tailcav,\HVn}s + \xi_{\tailcav}) P_{\HVn} + \sum_{i\in \mathcal{N}_{\tailcav}} \beta_{\tailcav,i} s P_i \right),
        \label{eq:G N}
    \end{align}
    with ${\xi_{\headcav} = \alpha_{\headcav} V'_{\headcav}(s_{\headcav}^*)}$, ${\xi_\tailcav = \alpha_\tailcav V_{\tailcav}'(s_{\tailcav}^*)}$,
    while the denominator is:
    \begin{align}
        D(s) = \left( (s^2 +
        \eta_{\headcav} s  + \xi_{\headcav}) P_0
        - \sum_{i\in \mathcal{N}_{\headcav}}\beta_{\headcav,i} s P_i \right) (s^2 + \eta_{\tailcav} s + \xi_{\tailcav})  -\beta_{\headcav,\tailcav}s \left(\beta_{\tailcav,\headcav}s P_0  + (\beta_{\tailcav,\HVn}s + \xi_{\tailcav}) P_{\HVn} + \sum_{i\in \mathcal{N}_{\tailcav}} \beta_{\tailcav,i} s P_i\right),
    \label{eq:G D}
    \end{align}
    with
    ${\eta_{\headcav} = \alpha_{\headcav} + \beta_{\headcav,{\hhv}} + \sum_{i\in \mathcal{N}_{\headcav}} \beta_{\headcav,i}  + \beta_{\headcav,\tailcav}}$,  ${\eta_\tailcav = \alpha_\tailcav + \beta_{\tailcav,\HVn} + \sum_{i\in \mathcal{N}_{\tailcav}} \beta_{\tailcav,i} + \beta_{\tailcav,\headcav}}$, and:
    \begin{align}
        P_0 = \prod_{j=1}^{N} (s^2 + a_{j2} s + a_{j1}), \quad
        P_i = \prod_{j=1}^{i} (a_{j3}s+a_{j1}) \prod_{j=i+1}^{N} (s^2 + a_{j2} s + a_{j1}), \quad
        P_N = \prod_{j=1}^{N} (a_{j3}s+a_{j1}),
        \label{eq:G Pi}
    \end{align}
    where ${a_{j1}=\frac{\partial F_j}{\partial s_j}(s_j^*,v^*,0)}$, ${a_{j2}=\frac{\partial F_j}{\partial \dot{s}_j}(s_j^*,v^*,0)-\frac{\partial F_j}{\partial v_j}(s_j^*,v^*,0)}$, ${a_{j3}=\frac{\partial F_j}{\partial \dot{s}_j}(s_j^*,v^*,0)}$.
\end{lemma}
Note that the expressions in~\eqref{eq:G Pi} represent the dynamics~\eqref{eq:system HV s}-\eqref{eq:system HV v} of HVs.
Meanwhile, the formulas in~\eqref{eq:G N} and~\eqref{eq:G D} are determined by the CAV controllers in~\eqref{eq:nominal controller head CAV} and~\eqref{eq:nominal controller tail CAV}, hence they depend on the $\alpha$ and $\beta$ controller gains.
Using the head-to-tail transfer function $G(s)$, we provide the conditions on these controller gains to stabilize the linearized system~\eqref{eq:linearized_system}.
This implies local stability for the nonlinear system~\eqref{eq:closed_loop_system}, i.e., when the system state and disturbance are within a small region around the equilibrium.
The stability conditions are summarized in Theorem~\ref{theorem:stability}.

\begin{theorem}\label{theorem:stability}
    System~\eqref{eq:linearized_system} is plant stable if the controller gains $\alpha_{\headcav}$, $\beta_{\headcav,\hhv}$, $\beta_{\headcav,i}$, $\beta_{\headcav,\tailcav}$, $\alpha_{\tailcav}$, $\beta_{\tailcav,\HVn}$, $\beta_{\tailcav,i}$, $\beta_{\tailcav,\headcav}$  are chosen such that all solutions of $D(s) = 0$ have negative real parts, where $D$ is given in~\eqref{eq:G D}.     The system is head-to-tail string stable if $ |G(\imaginaryj \omega) | < 1 $  holds for all $\omega > 0$, where $\imaginaryj^2 = -1$ and $G$ is given in~\eqref{eq:G}-\eqref{eq:G D}.
\end{theorem}

\subsection{Stabilizing control gains}
\label{sec:subsec:stability analysis}

Here we use Theorem~\ref{theorem:stability} to determine the range of controller gains that stabilize the system, and we plot this range as \textit{stability charts} in the $(\beta_{\headcav,\tailcav},\beta_{\tailcav,\headcav})$ plane of controller gains. The corresponding stability boundaries (that bound the region of stabilizing controller gains)  are given as follows.

\textbf{Plant stability}: The mixed traffic system is at the plant-stability boundary when $D(s)=0$ has a real root at the origin, i.e.,  $s=0$,  or has a complex conjugate pair of roots $s=\pm \imaginaryj \omega$ with $\omega>0$. 
In the first case, the stability boundary is:
\begin{align}\label{eq:stability boundary p0}
    D(0) = 0.
\end{align}
In the second case, the stability boundary is:
\begin{align}
\begin{split}
    \mathrm{Re}(D(\imaginaryj \omega)) = & 0, \\
    \mathrm{Im}(D(\imaginaryj \omega)) = & 0,
\end{split}
\label{eq:stability boundary p}
\end{align}
with $\mathrm{Re}(\cdot)$ and $\mathrm{Im}(\cdot)$ being the real and imaginary parts of a complex number.

\textbf{Head-to-tail string stability}:
Based on Theorem~\ref{theorem:stability}, the head-to-tail string stability condition is
$ |G(\imaginaryj \omega)|<1$ for all $\omega>0$. We note that $|G(0)| = 1$. Therefore, there are two cases for the string stability boundaries~\cite{guo2023connected}. In the first case, $\vert G(\imaginaryj \omega)\vert$ gets its maximum value at $\omega = 0$, and the stability boundary is given as:
\begin{align}\label{eq:stability boundary s0}
   \lim_{\omega \to 0^+} \frac{1}{\omega^2} \big( \vert D(\imaginaryj \omega) \vert^2 - \vert N(\imaginaryj \omega) \vert^2 \big) = 0,
\end{align}
In the second case, $|G(\imaginaryj \omega)|= 1$ for some positive $\omega>0$. The stability boundaries form a family of curves parameterized by the wave number $\theta\in [0,2\pi)$, obtained from:
\begin{align}
	G(\imaginaryj \omega) = e^{-\imaginaryj \theta},
\end{align}
which is equivalent to:
\begin{align}
\begin{split}
	\mathrm{Re} (D(\imaginaryj \omega)) \!-\! \mathrm{Re} (N(\imaginaryj \omega)) \cos \theta \!+\! \mathrm{Im} (N(\imaginaryj \omega)) \sin \theta & = 0, \\
	\mathrm{Im} (D(\imaginaryj \omega)) \!-\! \mathrm{Re} (N(\imaginaryj \omega)) \sin \theta \!-\! \mathrm{Im} (N(\imaginaryj \omega)) \cos \theta & = 0.
\end{split}
\label{eq:stability boundary s}
\end{align}

Equations~\eqref{eq:stability boundary p0}-\eqref{eq:stability boundary p} and~\eqref{eq:stability boundary s0}-\eqref{eq:stability boundary s} define the plant and string stability boundaries, respectively.
Note that the left-hand sides of these equations depend on the control gains like $\beta_{\headcav,\tailcav}$ and $\beta_{\tailcav,\headcav}$.
Thus, these equations define curves in the space of controller gains such as in the $(\beta_{\headcav,\tailcav},\beta_{\tailcav,\headcav})$ plane.
By plotting these curves, we create stability charts that identify stabilizing gains.

\begin{figure}[t]
    \centering    \includegraphics[width=0.55\linewidth]{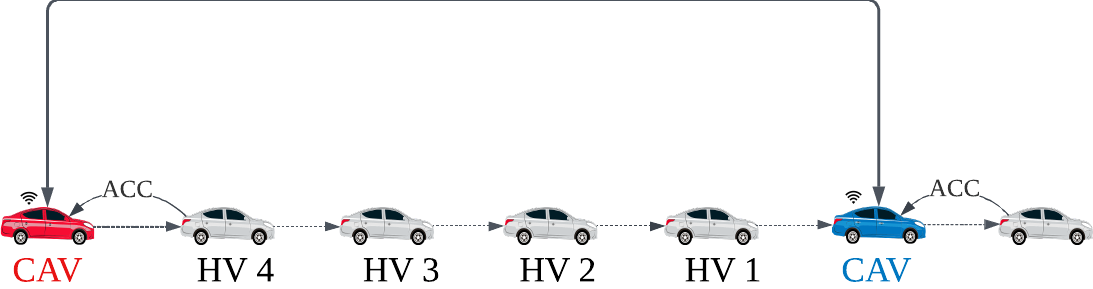}\hspace{2em}
    \includegraphics[width=0.3\linewidth]{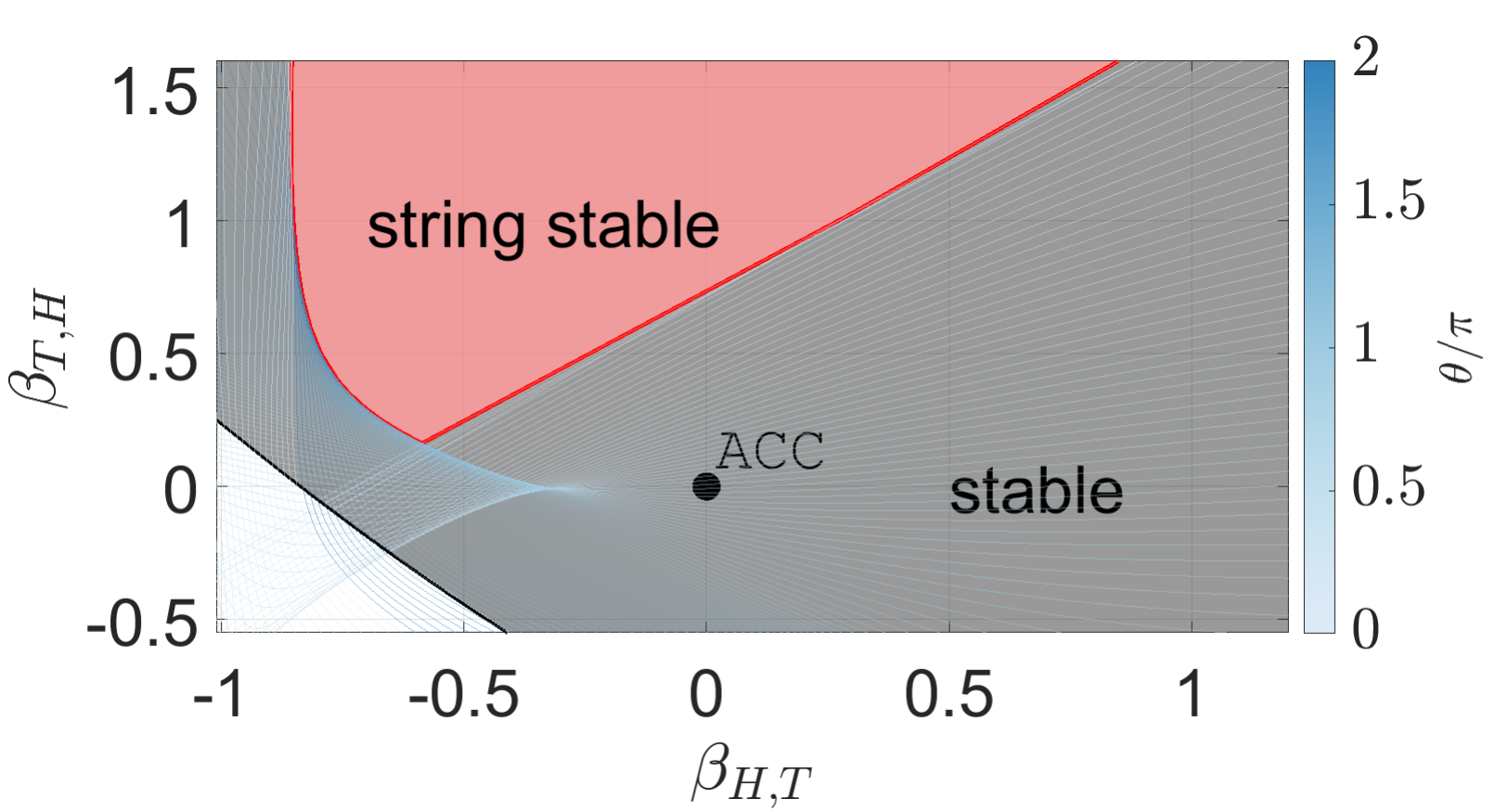}  \\   
    (a)   No HV connection\\
    \includegraphics[width=0.55\linewidth]{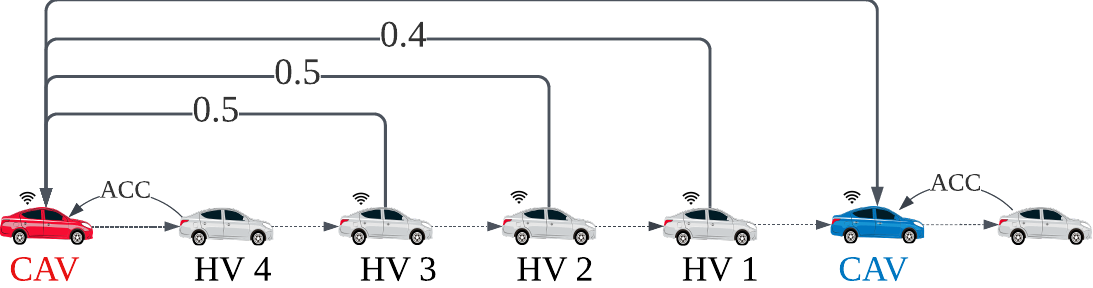}\hspace{2em}
    {\includegraphics[width=0.3\linewidth]{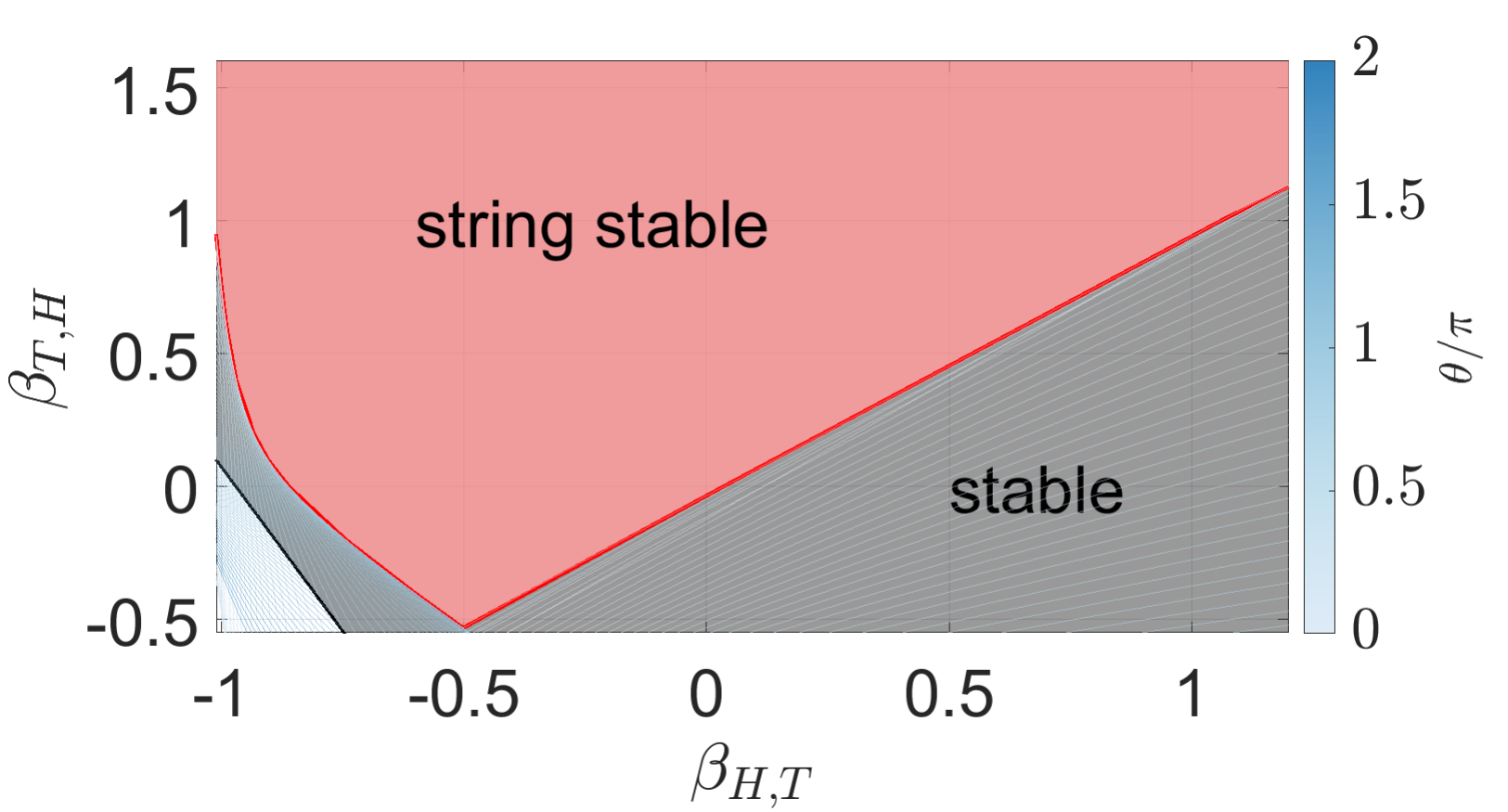} \\  \vspace{0.2cm}
    (b)  Tail CAV looks ahead, i.e.,  HVs are connected to T-CAV  \\
    \includegraphics[width=0.55\linewidth]{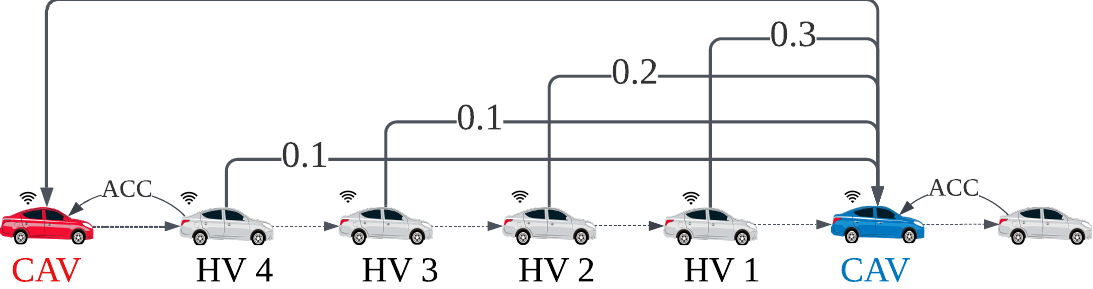} \hspace{2em}
    \includegraphics[width=0.3\linewidth]{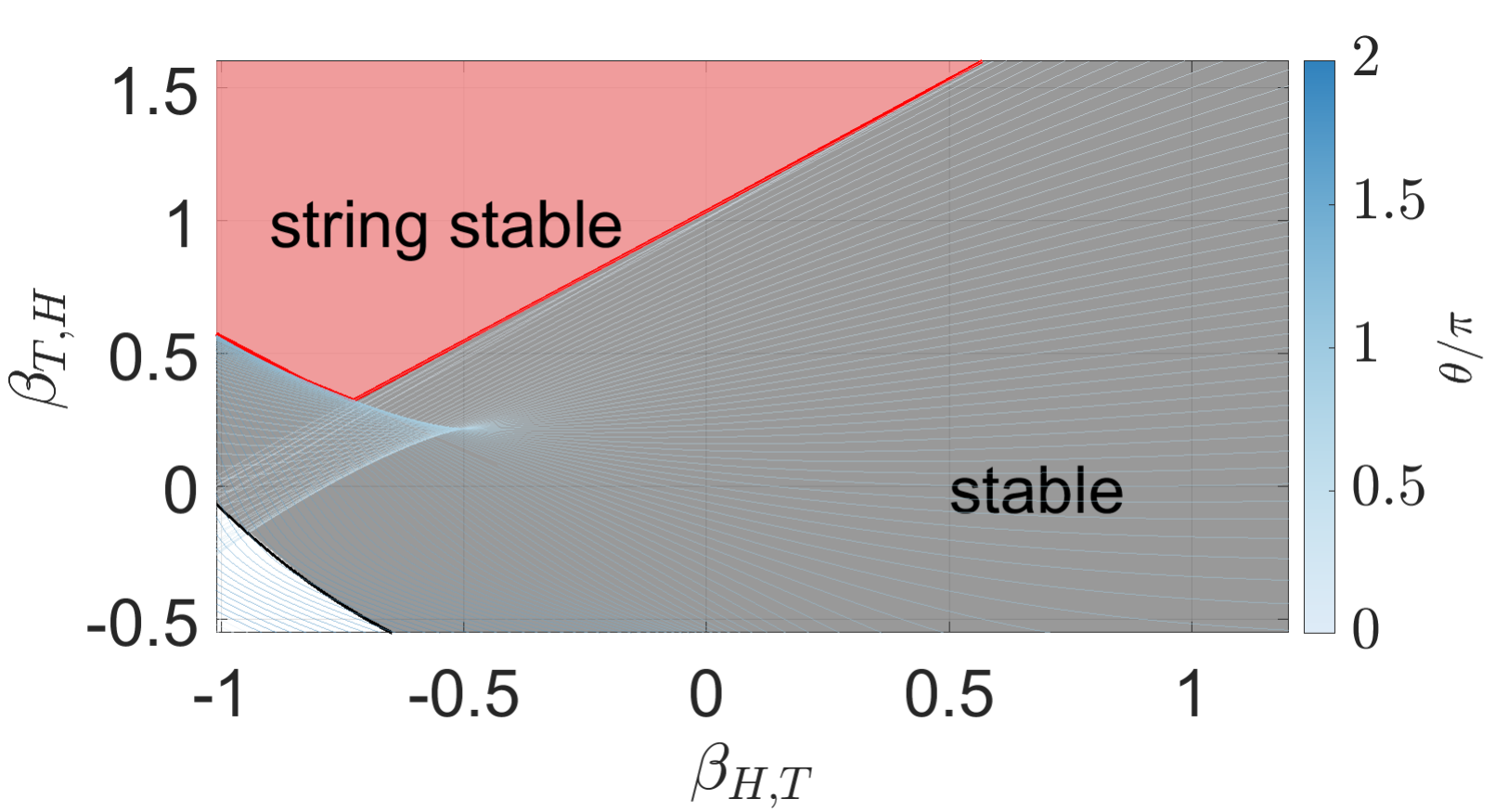}} \\
         \vspace{0.2cm}
    (c) Head CAV looks behind, i.e., HVs are connected to H-CAV 
     \\
     \vspace{0.4cm}
    \includegraphics[width=0.4\linewidth]{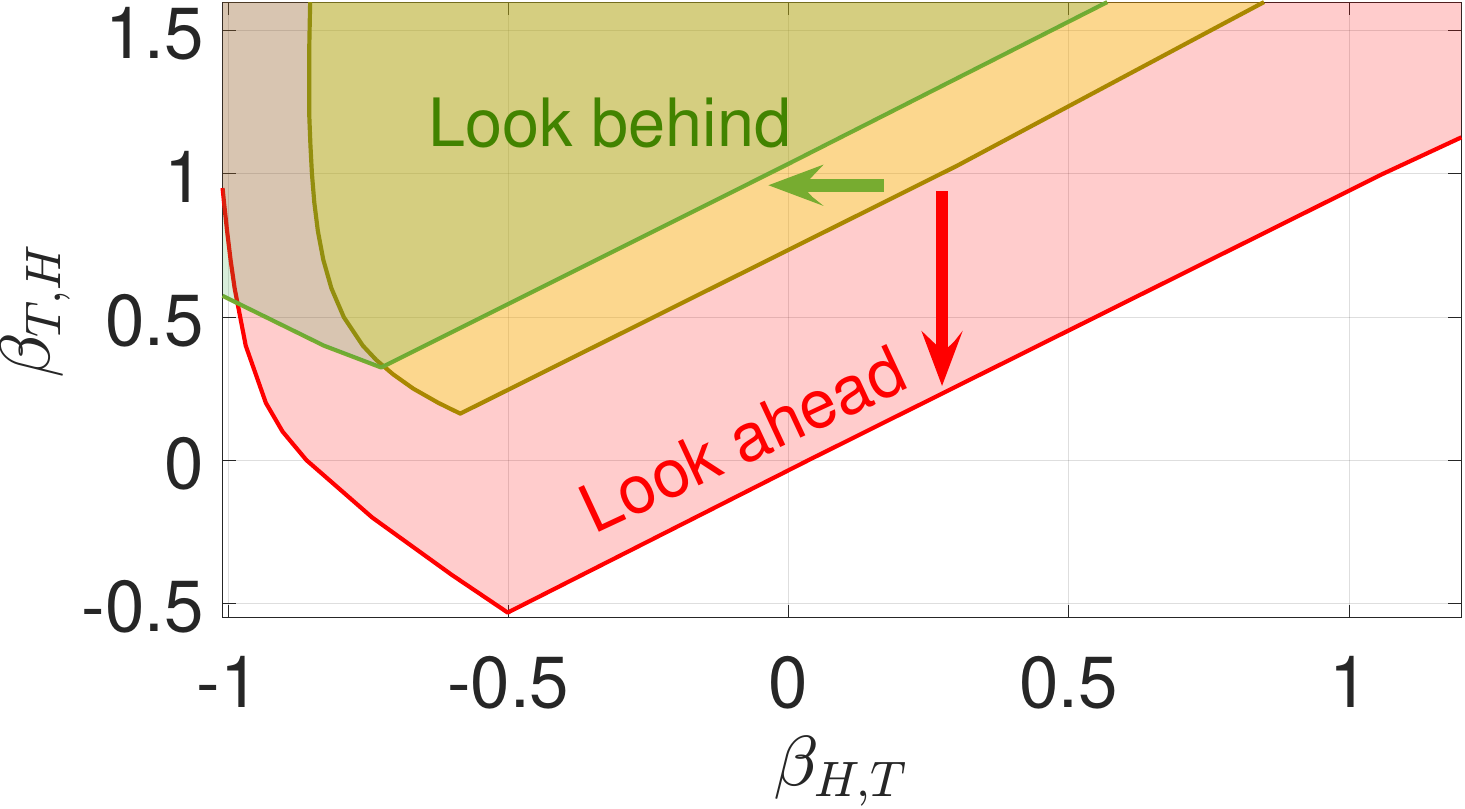}\\
    (d) Comparison of stability region under different HV communication topology: (a) yellow, (b) red, and (c) green
    \caption{Stability charts in the $(\beta_{\headcav,\tailcav},\beta_{\tailcav,\headcav})$ space of control gains. Grey areas and red areas represent plant stability and head-to-tail string stability, respectively. The white one represents unstable region. In panel (a), the middle HVs are not connected and the CAVs do not respond to them. 
    In panel (b), HV~1, HV~2, and HV~3 are connected to the tail CAV who responds to them with gains $\beta_{\tailcav,1} = 0.4$, $\beta_{\tailcav,2} = 0.5$, $\beta_{\tailcav,3} = 0.5$ (while also responding to HV-4 based on range sensors).
    In panel (c), HVs are connected to the head CAV who responds to them with controller gains $\beta_{\headcav,1} = 0.3$, $\beta_{\headcav,2} = 0.2$. $\beta_{\headcav,3} = 0.1$, $\beta_{\headcav,4} = 0.1$.
    In all the three cases, the remaining controller gains are $\alpha_{\headcav} = 0.4$, $\beta_{\headcav,\hhv}=0.6$, $\alpha_{\tailcav}= 0.4$, and $\beta_{\tailcav,4}= 0.6$.
    In panel (d), the stability charts from panels (a)-(c) are compared.}
    \label{fig:stability chart}
\end{figure}

As a numerical example, we consider a vehicle platoon of $N=4$ middle HVs. For the CAV spacing policy in~\eqref{eq:Vs}, we use $s_{\mathrm{st}} = 2$ m, $s_{\mathrm{go}} = 40$ m, and $v_{\max} = 40$ m/s. We set the HV dynamics as the optimal velocity model:
\begin{align}
    F_i(s_i, v_i, \dot{s}_i) = a(V_i(s_i) - v_i) + b\dot s_i, \label{eq:OVM}
\end{align}
where $a>0$ and $b>0$ reflect the human driver's reaction to match its speed $v_i$ to the desired speed $V_i(s_i)$ and the preceding vehicle's speed $v_{i-1}$, respectively. For the human drivers, we take the desired speed $V_i(s)$ also in the form of~\eqref{eq:Vs}, but with $s_{\mathrm{st}}$ and $s_{\mathrm{go}}$ calibrated from trajectories in the NGSIM dataset~\cite{NGSIM}. 
We calibrate the model parameters for HVs as  $a=0.16$ s$^{-1}$, $b=0.16$ s$^{-1}$, $s_{\mathrm{st}} = 1.9 $ m, and $s_{\mathrm{go}} = 46.3$ m.    We set the equilibrium speed of the vehicle platoon as $v^* = 20$ m/s, which gives the equilibrium gap for the HVs as $s_i^* = 24 $ m and for the two CAVs as $s_{\headcav}^* = s_{\tailcav}^* = 21$ m. 

Fig.~\ref{fig:stability chart} presents {\em stability charts} that indicate the stability boundaries and the range of cooperative CAV controller gains in the $(\beta_{\headcav,\tailcav},\beta_{\tailcav,\headcav})$ domain that leads to plant and head-to-tail string stability.

\begin{remark}[Stability impact of ACC mode only]
Fig.~\ref{fig:stability chart}(a) considers the case where the middle HVs are not connected. In this case, if the two CAV are also not connected (i.e., $\beta_{\headcav,\tailcav} = 0$, $\beta_{\tailcav,\headcav} = 0$), then the two CAVs execute ACC, and the vehicle platoon is string unstable. This means that when the head HV has some speed perturbations, the tail CAV will experience a larger speed perturbation through propagation along the vehicle platoon.
\end{remark}

\begin{remark}[Stability impact of CAV cooperation]\label{remark:stability cav}

In Fig.~\ref{fig:stability chart}(a) when the two CAVs only use ACC and CAV coordination, by connecting the two CAVs and using feedback gains ($\beta_{\headcav,\tailcav} $,  $\beta_{\tailcav,\headcav} $) within the string stable region shaded in red, the mixed vehicle platoon becomes head-to-tail string stable. 
This is possible only if the tail CAV's controller responds to the head CAV (i.e., string stability requires $\beta_{\tailcav,\headcav}\ne 0$), while the response of the head CAV to the tail CAV is not necessary (i.e., there exist string stable gains even when $\beta_{\headcav,\tailcav}=0$).
\end{remark}

Now consider the case where the middle HVs are connected to the CAVs. If the tail CAV includes feedback from middle HVs, as in Fig.~\ref{fig:stability chart}(b), then the system can be rendered string stable even if the two CAVs are not connected (i.e., $\beta_{\headcav,\tailcav}=0$ and $\beta_{\tailcav,\headcav}=0$). This shows that for string stability, the tail CAV should include feedback from its downstream traffic, either from the head CAV or the middle HVs. Fig.~\ref{fig:stability chart}(c) shows how the stability region is affected when the head CAV includes middle HV feedback in its controller. Similar to Fig.~\ref{fig:stability chart}(a), the system can be string stable with $\beta_{\headcav,\tailcav}=0$ and proper choice of $\beta_{\tailcav,\headcav}$, i.e., the head CAV can ignore tail CAV's feedback but not vice versa.  

\begin{remark}[Stability impact of connecting HVs]\label{remark:stability hv}
    Fig.~\ref{fig:stability chart}(d) compares the string stability boundaries from Fig.~\ref{fig:stability chart}(a)-(c) that correspond to different communication topologies of HV connection.
    Yellow color indicates when the CAVs are not connected to HVs (cf.~panel (a)), red shows when the tail CAV looks ahead and connects to HVs (cf.~panel (b)), and green corresponds to when the head CAV looks behind and connects to HVs (cf.~panel (c)).
    By connecting HVs to CAVs, i.e. HV feedback considered by CAV control strategies, the string stability boundaries drift as the arrows show. In the numerical example, ``look ahead" HVs behaviors makes stabilization an easier job for cooperation between the two CAVs while ``look behind" HVs behaviors not necessarily implicate that.  For various communication topologies, there exists a large overlap of the string stability regions,  from which we can select the controller gains $\beta_{\headcav,\tailcav}$ and $\beta_{\tailcav,\headcav}$ to stabilize traffic. 
\end{remark}

\section{Safety-critical Control}\label{sec:safety}
In this section, we first define the safety impact of the proposed cooperative CAV controllers~\eqref{eq:nominal controller head CAV} and~\eqref{eq:nominal controller tail CAV}, and then analyze how the choice of controller gains affects the safety.
Secondly, we utilize control barrier functions to design safety filters that constrains control inputs in real-time for safety guarantees. This is realized by formulating and solving an optimization problem that modifies the operation of potentially unsafe nominal controller designs.

\subsection{Longitudinal safety of  CAVs and HVs}

In longitudinal car-following control, safety refers to eliminating the risk of rear-end collision. We adopt surrogate safety measures~\cite{wang2021review} to evaluate the risk of collision for CAVs and HVs in various traffic scenarios. 
To this end, we adopt the constant time headway (CTH) safe spacing policy~\cite{zhao2023safety} for the safety definition of the CAVs and middle HVs.
For a vehicle with speed $v$ and gap $s$,  the CTH policy requires that the gap exceeds a minimum safe value which is the product of a safe time headway $\tau>0$ and the speed:
\begin{align}
    s\ge \tau v.
\end{align}
For the head CAV, we set its safe time headway as $\tau_{\headcav}>0$, and the CTH policy constraints its spacing $s_{\headcav}$ and speed $v_{\headcav}$ as:
\begin{align}
    s_{\headcav} &\ge \tau_{\headcav} v_{\headcav}.
\end{align}
This yields a safe set, i.e., a set of states where the head CAV is considered to be safe: 
\begin{align}
    \mathcal{C}_{\headcav} = \{x \in \mathbb{R}^{n}: h_{\headcav} (x)\ge 0\},
    \label{eq:safe_set_headcav}
\end{align}
with the safety function $h_{\headcav}$ being:
\begin{align}
    h_{\headcav}(x) = s_{\headcav} - \tau_{\headcav} v_{\headcav}.
    \label{eq:safety_function_headcav}
\end{align}
To guarantee safety, the controller should be designed so that if the head CAV is safe initially then it stays safe for all future time.
That is, if $x(0) \in \mathcal{C}_{\headcav}$, then $x(t) \in \mathcal{C}_{\headcav}$ holds for all $t\ge 0$,  which means that the safe set $\mathcal{C}_{\headcav}$ is {\em forward invariant}.
For the tail CAV, similarly, we take a safe time headway as $\tau_{\tailcav}>0$, and define its safe set as:
\begin{align}
    \mathcal{C}_{\tailcav} = \{x \in \mathbb{R}^{n}: h_{\tailcav} (x)\ge 0\},
    \label{eq:safe_set_tailcav}
\end{align}
with the safety function:
\begin{align}
    h_{\tailcav}(x) = s_{\tailcav} - \tau_{\tailcav} v_{\tailcav}.
    \label{eq:safety_function_tailcav}
\end{align}
We aim to ensure forward invariance  of $\mathcal{C}_{\tailcav}$, i.e., if $x(0) \in \mathcal{C}_{\tailcav}$, then $x(t) \in \mathcal{C}_{\tailcav}$ for all ${t \ge 0}$. For HV $i$, we take its safe time headway as $\tau_i$, and the CTH spacing policy 
yields the safety function:
\begin{align}\label{eq:safety_function_hv}
    h_i(x) = s_i - \tau_i v_i.
\end{align}

\begin{remark}[Safety evaluation]
    The safety function $h$ acts as a safety surrogate measurement, and a negative $h$ implies an unsafe car-following scenario that has a higher risk of collision. When $h<0$, the gap $s$ may still be positive. A negative gap $s$ means that a severe collision has happened. In the following analysis, we refer to the CAVs or HVs as ``safe'' if $h \ge 0$.
\end{remark}

\subsection{Safety impact of the nominal controller}

The forward invariance of the safe sets $\mathcal{C}_{\headcav}$ and $\mathcal{C}_{\tailcav}$ for the mixed traffic system~\eqref{eq:system} is established using Nagumo's theorem.
\begin{lemma}[Nagumo's theorem \cite{nagumo1942lage}] \label{theorem:nagumo safety}
Consider the system:
\begin{align}
    \dot x = F(x),
\end{align}
with state $x\in \mathbb{R}^{n}$, safe set $\mathcal{C}$, and safety function $h:\mathbb{R}^{n} \to \mathbb{R}$ such that ${\nabla h(x) \neq 0}$ if ${h(x) = 0}$. 
The system is safe w.r.t.~$\mathcal{C}$ (that is, $\mathcal{C}$ is forward invariant) if and only if:
\begin{equation}
    \dot h(x) = \nabla h(x) \cdot F(x) \geq 0
    \label{eq:Nagumo}
\end{equation}
holds for all ${x \in \mathbb{R}^{n}}$ satisfying ${h(x) = 0}$.
\end{lemma}

Using condition~\eqref{eq:Nagumo} in Lemma~\ref{theorem:nagumo safety}, we determine the gains of the nominal controller~\eqref{eq:nominal controller head CAV} that ensure safety for the head CAV.
This is summarized in Theorem~\ref{theorem:safety nominal head}, whose proof is given in Appendix~\ref{appdx:safety}.

\begin{theorem}[Safety of the nominal head CAV controller]\label{theorem:safety nominal head}
System~\eqref{eq:closed_loop_system} with the nominal controller~\eqref{eq:nominal controller head CAV} of the head CAV and the range policy~\eqref{eq:Vs} is safe w.r.t.~$\mathcal{C}_{\headcav}$ defined in~\eqref{eq:safe_set_headcav}-\eqref{eq:safety_function_headcav},  if $v_{\hhv}, v_{\headcav}, v_{i}, v_{\tailcav} \in [0,v_{\max}]$, $s_{\headcav} \in [s_{\mathrm{st}},s_{\mathrm{go}}]$, and if the controller parameters satisfy ${\kappa \leq 1/\tau_{\headcav}}$ and:
\begin{align}
    \alpha_{\headcav} \ge \bigg( |  1-\tau_{\headcav} \beta_{\headcav,{\hhv}}| + \tau_{\headcav} \sum_{i \in \mathcal{N}_{\headcav}} |\beta_{\headcav,i}| + \tau_{\headcav} |\beta_{\headcav,\tailcav}| \bigg) \frac{v_{\max}}{s_{\mathrm{st}}}.
    \label{eq:head_CAV_safe_gains}
\end{align}
\end{theorem}

\begin{remark}[Safe controller gains] \label{remark:nominal safe gain}
    The safety criterion~\eqref{eq:head_CAV_safe_gains} in Theorem~\ref{theorem:safety nominal head} can be interpreted as follows.
    It provides safe choices of controller gains $\alpha_{\headcav}$ and $\beta_{\headcav,\hhv}$ that are related to the adaptive cruise control term in~\eqref{eq:nominal controller head CAV}:
    \begin{align}
    \frac{s_{\mathrm{st}}}{v_{\max}} \alpha_{\headcav} - |  1-\tau_{\headcav} \beta_{\headcav,{\hhv}}|  \ge  \tau_{\headcav} \sum_{i \in \mathcal{N}_{\headcav}} |\beta_{\headcav,i}| + \tau_{\headcav} |\beta_{\headcav,\tailcav}| .
    \end{align}
    It also gives a maximum safe CAV coordination gain $\beta_{\headcav,\tailcav}$ as:
    \begin{align}
        |\beta_{\headcav,\tailcav}|  \le \frac{s_{\mathrm{st}}}{v_{\max} \tau_{\headcav}} \alpha_{\headcav} - \bigg|  \frac{1}{\tau_{\headcav}}-\beta_{\headcav,{\hhv}} \bigg| -  \sum_{i \in \mathcal{N}_{\headcav}} |\beta_{\headcav,i}|.
    \end{align}
    Such $\beta_{\headcav,\tailcav}$ gain exists only if the right-hand side is non-negative. When $\beta_{\headcav,i} = 0$ and $\beta_{\headcav,\hhv} = 1/\tau_{\headcav}$, the right-hand side has the maximum as $\alpha_{\headcav} s_{\mathrm{st}} / (v_{\max} \tau_{\headcav})$. 
    Similarly, condition~\eqref{eq:head_CAV_safe_gains} also provides safe gains considering the feedback of HV states:
    \begin{align}
        \sum_{i \in \mathcal{N}_{\headcav}} |\beta_{\headcav,i}|  &\le \frac{s_{\mathrm{st}}}{v_{\max} \tau_{\headcav}} \alpha_{\headcav} - \bigg|  \frac{1}{\tau_{\headcav} }-\beta_{\headcav,{\hhv}} \bigg| - |\beta_{\headcav,\tailcav}| , 
    \end{align}
    where the right-hand side should again be non-negative to ensure the existence of such $\beta_{\headcav,i}$. When $\beta_{\headcav,\hhv} = 1/\tau_{\headcav}$ and $\beta_{\headcav,\tailcav} = 0$, the right-hand side has the maximum as $\alpha_{\headcav} s_{\mathrm{st}} / (v_{\max} \tau_{\headcav})$.
\end{remark}

To ensure safety for the tail CAV, the nominal controller~\eqref{eq:nominal controller tail CAV} must satisfy similar criteria given in Theorem~\ref{theorem:safety nominal tail}.
The proof of this theorem follows the same steps as for Theorem~\ref{theorem:safety nominal head}, hence it is omitted.
\begin{theorem}[Safety of the nominal tail CAV controller]\label{theorem:safety nominal tail}
System~\eqref{eq:system} with the nominal controller~\eqref{eq:nominal controller tail CAV} of the tail CAV and the range policy~\eqref{eq:Vs} is safe w.r.t. $\mathcal{C}_{\tailcav}$ defined in~\eqref{eq:safe_set_tailcav}-\eqref{eq:safety_function_tailcav}, if $v_{\hhv},v_{\headcav},v_i,v_{\tailcav} \in [0,v_{\max}]$, $s_{\tailcav} \in [s_{\mathrm{st}},s_{\mathrm{go}}]$, and if the controller parameters satisfy $\kappa \le 1/\tau_{\tailcav}$ and: 
\begin{align}
    \alpha_{\tailcav} \ge \bigg(| 1- \tau_{\tailcav}   \beta_{\tailcav,{\HVn}}| + \tau_{\tailcav} \sum_{i \in \mathcal{N}_{\tailcav}} |\beta_{\tailcav,i}| + \tau_{\tailcav} |\beta_{\tailcav,\headcav}| \bigg) \frac{v_{\max}}{s_{\mathrm{st}}}.
    \label{eq:tail_CAV_safe_gains}
\end{align}
\end{theorem}

\begin{figure}[t]
    \centering
    \subfloat[Safe ACC gain for head CAV]{\includegraphics[width=0.24\linewidth]{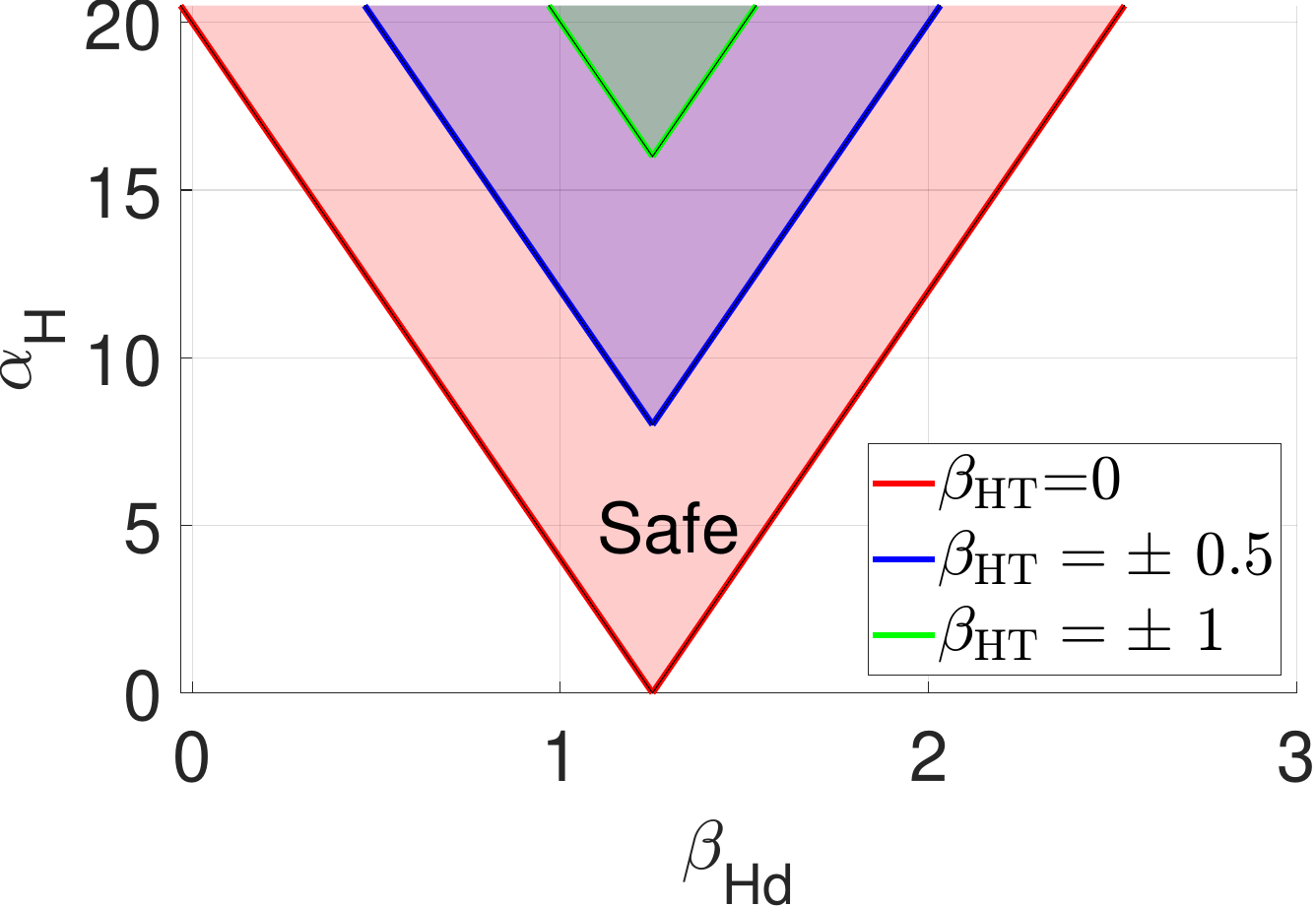}}
    \hspace{0.2em} \subfloat[Safe ACC gain for tail CAV]{\includegraphics[width=0.24\linewidth]{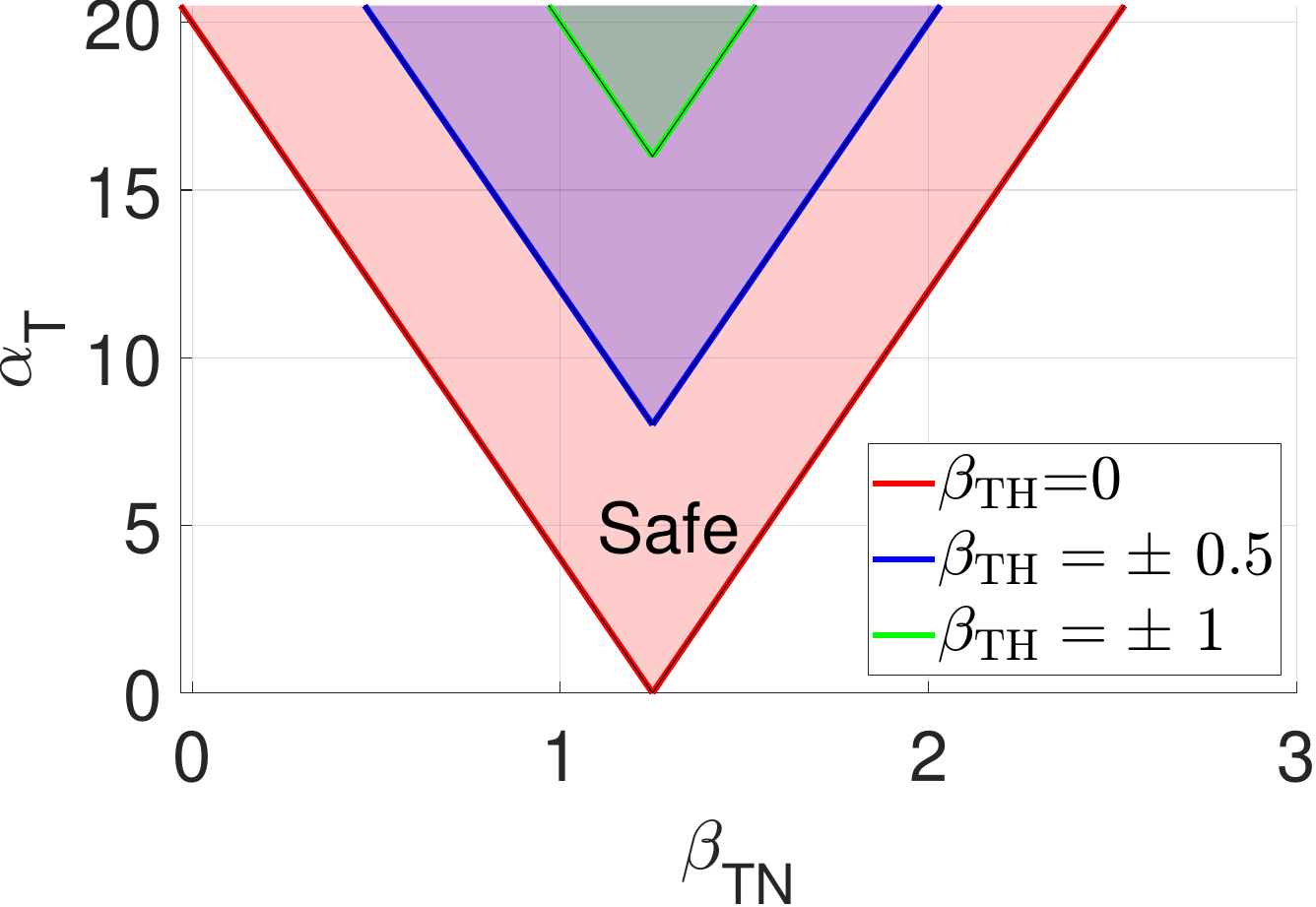}}
    \hspace{0.2em}  \subfloat[Safe ACC gain for head CAV]{\includegraphics[width=0.24\linewidth]{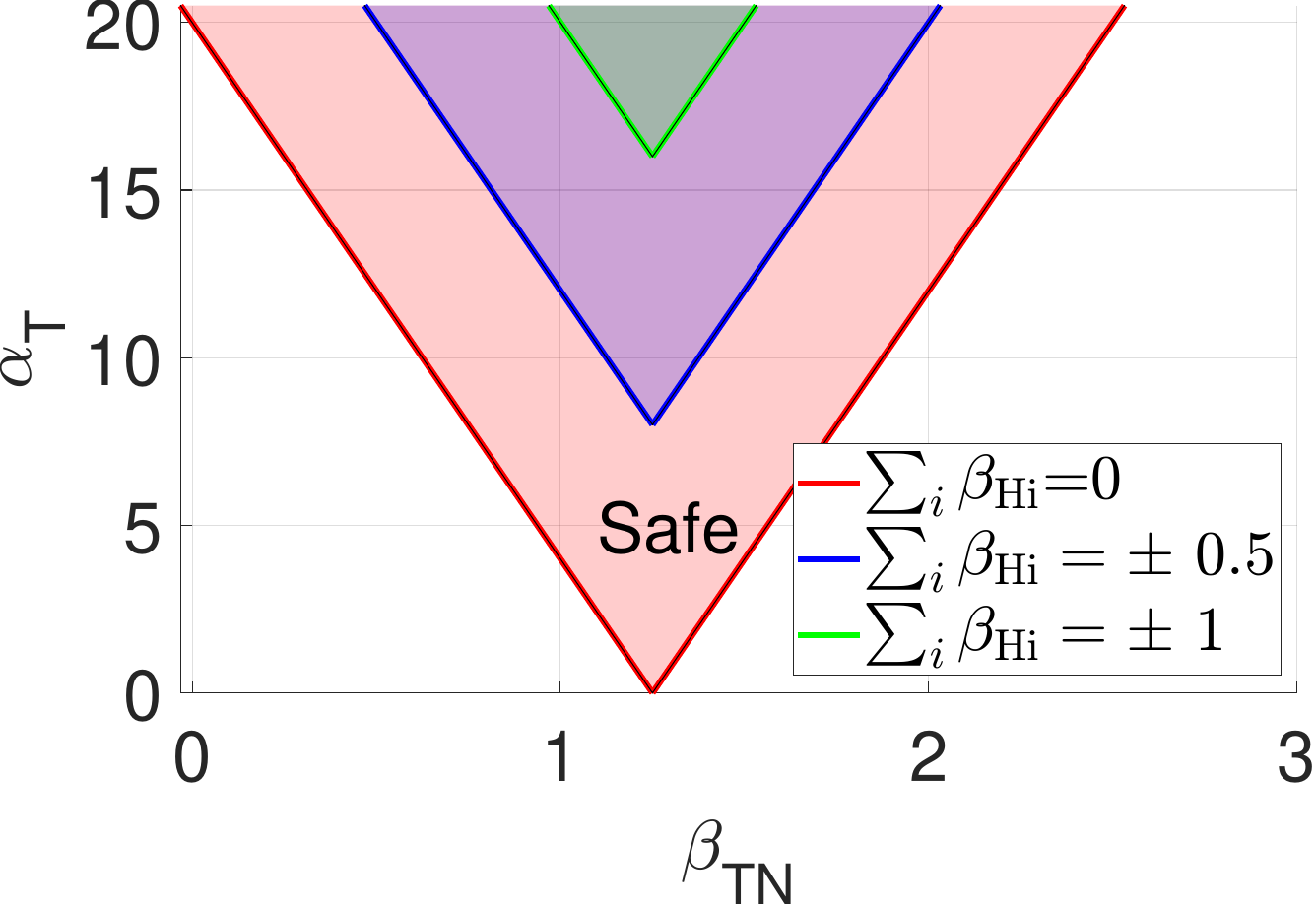}}
    \hspace{0.2em} \subfloat[Safety and stability charts with CAV cooperation  gain]{\includegraphics[width=0.24\linewidth]{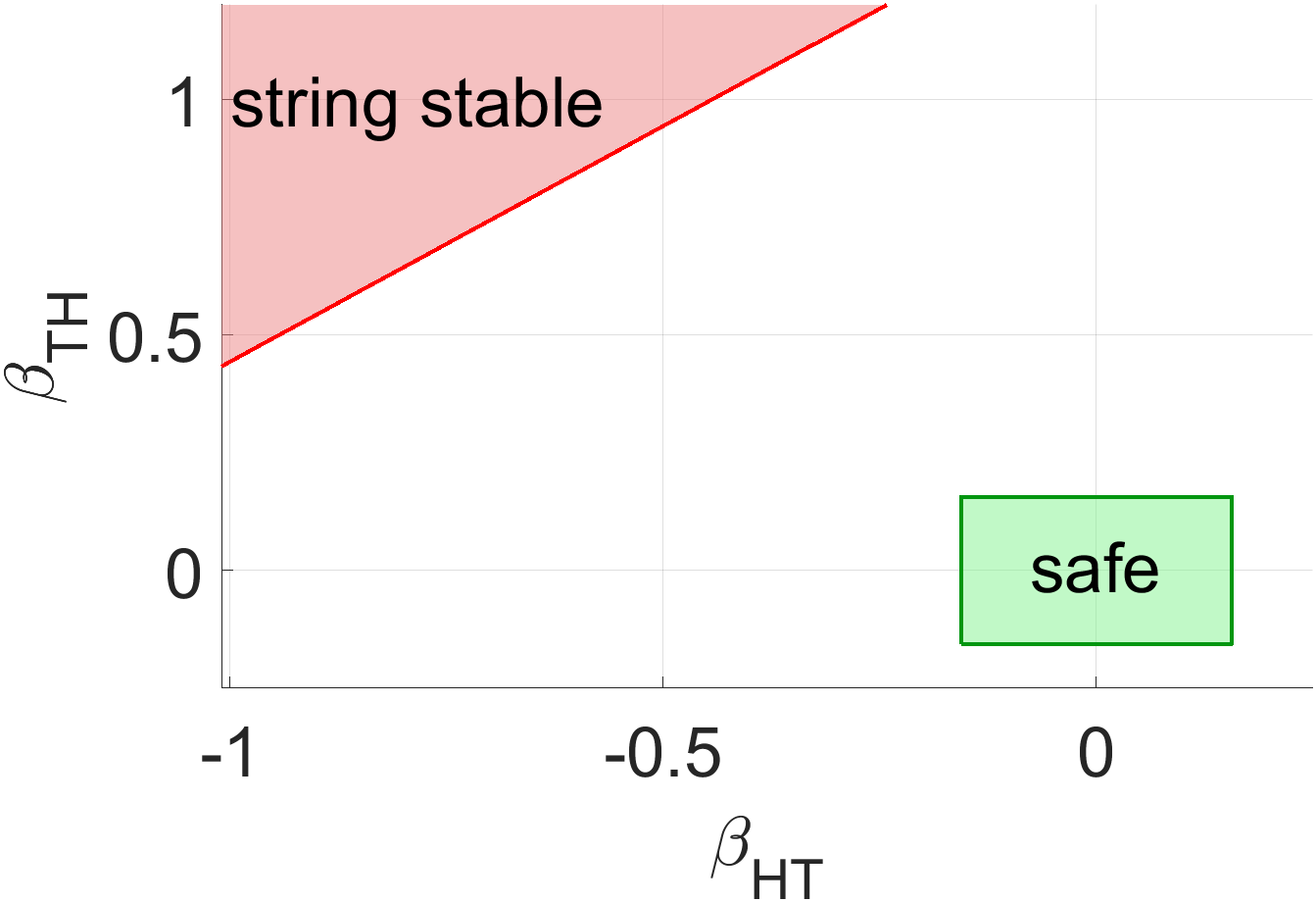}}
    \caption{
    Safety charts in the space of control gains for the the nominal controller~\eqref{eq:nominal controller head CAV}-\eqref{eq:nominal controller tail CAV}.
    (a) Safe $(\beta_{\headcav,\hhv},\alpha_{\headcav})$ gains considering the head CAV's safety (We take the spacing policy $V(s)$  the same as in Fig.~\ref{fig:stability chart}. The CAV coordination gains are set as $\beta_{\headcav,\tailcav} = 0$. We choose safe time headway as $\tau_{\headcav} = 0.8$ s.);
    (b) safe $(\beta_{\tailcav,N},\alpha_{\tailcav})$ gains associated with the tail CAV's safety (We take the spacing policy $V(s)$  the same as in Fig.~\ref{fig:stability chart}. The CAV coordination gains are set as $\beta_{\tailcav,\headcav} = 0$. We choose safe time headway as $\tau_{\tailcav} = 0.8$ s.); (c) safe ACC gain for $(\beta_{\headcav,\hhv},\alpha_{\headcav})$ for the head CAV with different  HV feedback gains (We set $\beta_{\headcav,\tailcav}=0$); and
    (d) safe $(\beta_{\headcav,\tailcav},\beta_{\tailcav,\headcav})$ gains for the safety of both CAVs. 
    The shaded region indicates the range of gains that ensure safety for the respective CAVs based on Theorems~\ref{theorem:safety nominal head} and~\ref{theorem:safety nominal tail}.
    Notice that the gains of the nominal controller are restricted if one intends to achieve provably safe behavior (i.e., $\alpha_{\headcav}$ and $\alpha_{\tailcav}$ must be very high or $\beta_{\headcav,\tailcav}$ and $\beta_{\tailcav,\headcav}$ must be very small).
    This motivates the introduction of safety filters to enforce safe behaviors by deviating from the nominal controller to prevent safety violation.}
    \label{fig:safety chart nominal}
\end{figure}

The range of safe controller gains provided by Theorems~\ref{theorem:safety nominal head} and~\ref{theorem:safety nominal tail} are depicted as {\em safety charts} in Fig.~\ref{fig:safety chart nominal}.

\begin{remark}[Safety impact of CAV cooperation]\label{remark:safety cav}

The shaded domain in Fig.~\ref{fig:safety chart nominal}(a) shows the safe controller gains for the head CAV in the $(\beta_{\headcav,\hhv},\alpha_{\headcav})$ space for various values of $\beta_{\headcav,\tailcav}$ based on~\eqref{eq:head_CAV_safe_gains}.  
It can be observed that including the cooperation with the tail CAV (i.e., taking $\beta_{\headcav,\tailcav} \ne 0$) makes the safety region shift towards higher $\alpha_{\headcav}$ gains. Similarly, Fig.~\ref{fig:safety chart nominal}(b) shows the safe gains for the tail CAV in the $(\beta_{\tailcav,N},\alpha_{\tailcav})$ space for various $\beta_{\tailcav,\headcav}$ values based on~\eqref{eq:tail_CAV_safe_gains}.
The same trend is showcased: the safety region shifts up by including the cooperation with the head CAV ($\beta_{\tailcav,\headcav} \ne 0$).
\end{remark}

\begin{remark}[Safety impact of connecting HVs]\label{remark:safety hv}
  We plot the safe ACC gain in Fig.~\ref{fig:safety chart nominal}(c) for the head CAV. When the CAV includes feedbacks from middle HVs, its safety region shifts up and requires larger $\alpha_{\headcav}$. The effect of HV  feedback on the tail CAV safety is also similar, and we ignore the figure.

\end{remark}

\begin{remark}[Trade-off between stability and safety]

The cooperation between CAVs enhances stability as in Remark~\ref{remark:stability cav} but makes it harder to guarantee safety as in Remark~\ref{remark:safety cav}. This trade-off also exists in the connectivity of HVs: From Remark~\ref{remark:stability hv} regarding Fig.~\ref{fig:stability chart}, the stability region grows by setting proper controller gains for the feedback of HV states. Meanwhile, Remark~\ref{remark:safety hv} shows that the selection of safe gains becomes more limited by connecting to HVs (i.e., the right-hand sides of~\eqref{eq:head_CAV_safe_gains} and~\eqref{eq:tail_CAV_safe_gains} increase for $\beta_{\headcav,i} \ne 0$ and $\beta_{\tailcav,i} \ne 0$).

Fig.~\ref{fig:safety chart nominal}(d) indicates the safe CAV cooperation gains in green in the $(\beta_{\headcav,\tailcav},\beta_{\tailcav,\headcav})$ space considering the safety of both CAVs. We set $\alpha_{\headcav} =1$, $\alpha_{\tailcav} =1$. We consider no HV connection case, i.e., $\beta_{\headcav,i} = 0$, and $\beta_{\tailcav,i} = 0$, and we set   $\beta_{\headcav,\hhv} = 1/\tau_{\headcav}$, $\beta_{\tailcav,\HVn} = 1/\tau_{\tailcav}$, which gives the maximum range of safe CAV coordination gain $\beta_{\headcav,\tailcav}$ and $\beta_{\tailcav,\headcav}$  as in Remark~\ref{remark:nominal safe gain}. We also plot the string-stability region as the red region. As Fig.~\ref{fig:safety chart nominal}(d) shows, the string stability region and safety region do not intersect.
This highlights that guaranteeing both stability and safety by the nominal controller~\eqref{eq:nominal controller head CAV}-\eqref{eq:nominal controller tail CAV} may not be possible. 
Furthermore, Fig.~\ref{fig:safety chart nominal}(a,b,c) also demonstrate that cooperating CAVs require a very high $\alpha$ gain for provable safety (for example, the blue curve in panel (a) shows that the head CAV is safe only for $\alpha_{\headcav}>5$ when the CAV cooperation gain is set to $\beta_{\headcav,\tailcav} = 0.5$). 
Such high $\alpha$ gain is undesired and infeasible for practical CAV control applications. Experiments in~\cite{alan2024integrating} have shown that the gap-related controller gain $\alpha$ is recommended to be smaller than 1.
Therefore, these restrictions on the safe control gains motivate the modification of the nominal controller~\eqref{eq:nominal controller head CAV}-\eqref{eq:nominal controller tail CAV} to actively enforce safety.
This is discussed in the next subsections through the introduction of CBF-based safety filters.    
\end{remark}

\subsection{Preliminaries on CBFs}

Control barrier functions provide a general tool to constrain the control input to ensure safety. We present the basics of CBF as follows. Consider an affine control system with state $x \in \mathbb{R}^{n}$ and control input $u\in \mathbb{R}^{m}$:
\begin{equation}\label{eq:controlafine}
    \dot{x} = f(x) + g(x)u,
\end{equation}
cf.~\eqref{eq:system},
with controller ${u=k(x)}$ and initial condition ${x(0) = x_{0} \in \mathbb{R}^{n}}$.
If $f$, $g$, and $k$ are locally Lipschitz, the system has a unique solution $x(t)$, which we assume to exist for all ${t \ge 0}$.
The system is safe if the solution stays in a safe set $\mathcal{C}$, i.e, $x(t)\in \mathcal{C}$ holds for all ${t \ge 0}$ if ${x_{0} \in \mathcal{C}}$. Let $\mathcal{C}$ be given by a continuously differentiable function $h:\mathbb{R}^{n} \to \mathbb{R}$,
cf.~\eqref{eq:safe_set_headcav} and~\eqref{eq:safe_set_tailcav}.

\begin{definition}[Control Barrier Function~\cite{ames2019control}]
    Function $h$ is called a control barrier function for the system \eqref{eq:controlafine}  on $\mathcal{C}$  if there exists an extended class-$\mathcal{K}_{\infty}$ function $\gamma$ such that:
	\begin{equation}
		\sup_{u\in \mathbb{R}^{m}} L_{f} h(x)+L_{g} h(x) u > -\gamma(h(x)), \quad \forall x\in \mathcal{C}, 
	\end{equation}
	with $L_fh(x) = \nabla h(x) \cdot f(x)$ and $L_gh(x)  = \nabla h(x) \cdot g(x)$.
\end{definition}
The CBF is used to guarantee safety of the closed-loop system with a feedback controller ${k: \mathbb{R}^{n} \to \mathbb{R}^{m}}$, ${u = k(x)}$ for \eqref{eq:controlafine}.
\begin{theorem}[Safety guarantee by CBF~\cite{ames2019control}] \label{theorem:safety}
    If function $h$ is a control barrier function for~\eqref{eq:controlafine}  on $\mathcal{C}$, then any locally Lipschitz continuous controller $u=k(x)$ satisfying:
	\begin{align}
		L_f h(x) + L_gh(x) k(x) \ge -\gamma(h(x)),
        \quad \forall x\in \mathcal{C},
        \label{eq:CBF_constraint}
	\end{align}
	renders the set $\mathcal{C}$ forward invariant (safe), i.e, ${x(t)\in \mathcal{C}}$, ${\forall t \geq 0}$ holds for the closed-loop system for all $x_{0} \in \mathcal{C}$.
\end{theorem}

To control a system with formal safety guarantees, CBFs can be integrated with a pre-designed nominal controller such as~\eqref{eq:nominal controller head CAV} or~\eqref{eq:nominal controller tail CAV}. In particular, a nominal controller $u = k_{\rm n}(x)$ can be modified in a minimal way to synthesize a safety-critical control input $u=k(x)$, by solving the quadratic program (QP): 
\begin{equation}\label{eq:QP CBFintro}
\begin{split}
    k(x) = \underset{u \in \mathbb{R}^{m}}{\operatorname{argmin}} \;  & \Vert u-k_{\rm n}(x)  \Vert^{2}, \\
    \text{s.t.} \ 
    &  L_fh(x) + L_gh(x)u +\gamma(h(x))\ge 0,
\end{split}
\end{equation}
that is also called as {\em safety filter}. When the control input is a scalar (i.e., $u\in \mathbb{R}$) and $L_g h(x)<0$ for all $x \in \mathbb{R}^{n}$, this safety filter simplifies to the form:
\begin{align}
    k(x) = \min \left\{k_{\rm n}(x), - \frac{L_fh (x)  + \gamma(h(x))}{L_gh(x)} \right\}.
\label{eq:safety_filter}
\end{align}

\subsection{Safety filter CBF design for CAV, HV and platoon}

While Theorems~\ref{theorem:safety nominal head} and~\ref{theorem:safety nominal tail} provide conditions on the nominal CAV controllers~\eqref{eq:nominal controller head CAV} and~\eqref{eq:nominal controller tail CAV} to ensure safety, now we utilize CBFs to minimally modify these nominal controllers to obtain safety-critical controllers:
\begin{align}
\begin{split}
    u_{\headcav} & = k_{\headcav}(x,v_{\hhv}), \\
    u_{\tailcav} & = k_{\tailcav}(x),
\end{split}
\end{align}
based on real-time traffic states.
We consider three types of safety: CAV safety, HV safety, and platoon safety.

\textbf{CAV safety}: CAV safety refers to that the two CAVs keep a safe gap behind their preceding vehicles by enforcing the CTH policy. 
For the head CAV, the CBF $h_{\headcav}$ in~\eqref{eq:safety_function_headcav} gives constraints on the controller $u_{\headcav}$ as:
\begin{align}\label{eq:CBF head CAV}
    L_fh_{\headcav}(x,v_{\hhv}) + L_{g_{\headcav}} h_{\headcav}(x) u_{\headcav} \ge -\gamma_\headcav h_{\headcav}(x),
\end{align}
cf.~\eqref{eq:CBF_constraint},
with $L_fh_{\headcav}(x,v_{\hhv}) = v_{\hhv} - v_{\headcav}$, $L_{g_{\headcav}} h_{\headcav}(x) = -\tau_{\headcav}$, and $\gamma_{\headcav} > 0$.
This is equivalent to:
\begin{align}\label{eq:CBF explicit head}
    u_{\headcav} \le \frac{1}{\tau_{\headcav}} (v_{\hhv} - v_{\headcav}) + \gamma_\headcav \Big( \frac{1}{\tau_{\headcav}} s_{\headcav} - v_{\headcav} \Big).
\end{align}
The right-hand side resembles the adaptive cruise control terms in the nominal controller~\eqref{eq:nominal controller head CAV}.
Based on~\eqref{eq:safety_filter}, the safety filter enforcing the head CAV's safety is:
\begin{align}
    k_{\headcav}(x,v_{\hhv}) = \min \left\{k_{\headcav,{\rm n}}(x,v_{\hhv}), \frac{1}{\tau_{\headcav}} (v_{\hhv} - v_{\headcav}) + \gamma_\headcav \Big( \frac{1}{\tau_{\headcav}} s_{\headcav} - v_{\headcav} \Big) \right\},
    \label{eq:min_controller_head}
\end{align}
where the nominal controller $k_{\headcav,{\rm n}}$ is given in~\eqref{eq:nominal controller head CAV}.
Similarly, the safety function $h_{\tailcav}$ of the tail CAV  in~\eqref{eq:safety_function_tailcav} gives the constraint:
\begin{align}\label{eq:CBF tail CAV}
    L_fh_{\tailcav}(x) + L_{g_{\tailcav}} h_{\tailcav}(x) u_{\tailcav} \ge -\gamma_{\tailcav}  h_{\tailcav}(x),
\end{align}
with $L_fh_{\tailcav}(x) = v_{\HVn} - v_{\tailcav}$, $L_{g_{\tailcav}}h_{\tailcav}(x) = -\tau_{\tailcav}$, and $\gamma_{\tailcav} > 0$, which is equivalent to:
\begin{align}\label{eq:CBF explicit tail}
    u_{\tailcav} \le \frac{1}{\tau_{\tailcav}} (v_{\HVn} - v_{\tailcav}) + \gamma_\tailcav \Big( \frac{1}{\tau_{\tailcav}} s_{\tailcav} - v_{\tailcav} \Big);
\end{align}
cf.~the ACC terms in~\eqref{eq:nominal controller tail CAV}.
The corresponding safety filter that ensures the tail CAV's safety is:
\begin{align}
    k_{\tailcav}(x) = \min \left\{k_{\tailcav,{\rm n}}(x), \frac{1}{\tau_{\tailcav}} (v_{\HVn} - v_{\tailcav}) + \gamma_\tailcav \Big( \frac{1}{\tau_{\tailcav}} s_{\tailcav} - v_{\tailcav} \Big) \right\},
    \label{eq:min_controller_tail}
\end{align}
with the nominal controller $k_{\tailcav,{\rm n}}$ in~\eqref{eq:nominal controller tail CAV}.

\textbf{HV safety}:  When a middle HV-$i$ is connected to the head CAV, the CBF also enables the head CAV to improve the HV's safety. It is noted that here safe motions are enforced by car-following behaviors and therefore the head-CAV control input is constrained to this purpose. We design constraints on the head CAV controller to ensure HV safety as follows. For the HV safety function $h_i$ in~\eqref{eq:safety_function_hv}, we have that $L_g h_i(x) = \begin{bmatrix}
    0 & 0
\end{bmatrix}$, which means HV's safety measure $h_i$ does not directly constrain on any of the control inputs $u_{\headcav}$ or $u_{\tailcav}$ due to the relative degree of the system. Thus, we introduce a CBF for HV-$i$,:
\begin{align}
    \bar{h}_i(x) = h_i(x)  - \eta_i h_{\headcav}(x),
\label{eq:safety_function_hv_valid}
\end{align}
where $\eta_i>0$ is a parameter. Then we have $L_g \bar{h}_i(x) = \begin{bmatrix}
    \eta_i \tau_{\headcav} & 0
\end{bmatrix}$, which provides a way to constrain the control input $u_{\headcav}$ for safety. Once ensuring both $h_{\headcav}(x)\ge 0 $ and $\bar{h}_i(x)\ge 0$, we have $h_i(x) \ge 0 $. The safety constraint from $\bar{h}_i$ is then:
\begin{align}\label{eq:CBF HV}
    L_{f} \bar{h}_{i}(x,v_{\hhv}) + L_{g_{\headcav}} \bar{h}_{i}(x) u_{\headcav} \ge -\gamma_{i} \bar{h}_{i}(x),
\end{align}
with $L_f \bar{h}_i(x,v_{\hhv}) = v_{i-1} - v_{i} - \tau_i F_i(s_i,v_i,v_{i-1}-v_i) - \eta_i (v_{\hhv} - v_{\headcav})$, $L_{g_{\headcav}} \bar{h}_i(x) = \eta_i \tau_{\headcav}$, and ${\gamma_i>0}$. Thus $\bar{h}_i$ gives a constraint  on the head CAV's controller as:
\begin{align}\label{eq:CBF explicit HV}
    u_{\headcav} \ge
    \frac{1}{\tau_{\headcav}} (v_{\hhv} - v_{\headcav})
    + \gamma_i \bigg( \frac{1}{\tau_{\headcav}} s_{\headcav} - v_{\headcav} \bigg)
    + \frac{\tau_i}{\eta_i \tau_{\headcav}} \bigg( F_i(s_i,v_i,v_{i-1}-v_i)
    -\frac{1}{\tau_i}(v_{i-1}-v_{i}) - \gamma_i \bigg( \frac{1}{\tau_i} s_i - v_i \bigg) \bigg).
\end{align}

While the head CAV safety constraint~\eqref{eq:CBF explicit head} gives an upper bound on $u_{\headcav}$, the HV safety condition~\eqref{eq:CBF explicit HV} sets a lower bound on $u_{\headcav}$. These two bounds may conflict with each other, thus there may be no available controller to enforce the safety of both the head CAV and the HVs. Therefore, we set the head CAV safety as a hard constraint and HV safety as soft constraints by adding a relaxation term to the HV safety constraints.
This leads to the following safety filter that ensures the safety of head CAV while facilitating the safety of HVs:
\begin{align}
\begin{split}
    k_{\headcav}(x,v_{\hhv}) = \underset{u_{\headcav} \in \mathbb{R}, \sigma_i \geq 0}{\operatorname{argmin}} \; & \Vert u_{\headcav} -k_{\headcav,{\rm n}}(x,v_{\hhv}) \Vert^{2} + \sum_{i=1}^N p_i\sigma_i^2 \\
    \text{s.t.} \;
    &\mathrm{head \; CAV \; safety:} \quad
    L_fh_{\headcav}(x,v_{\hhv}) + L_{g_{\headcav}} h_{\headcav}(x) u_{\headcav} \ge - \gamma_{\headcav} h_{\headcav}(x),
    \\
    & \mathrm{HV \; safety:} \quad\quad\
    \left\{
    \begin{array}{l}
    L_{f} \bar{h}_{1}(x,v_{\hhv}) + L_{g_{\headcav}} \bar{h}_{1}(x) u_{\headcav} \ge - \gamma_{1} \bar{h}_{1}(x) -\sigma_1, \\
     \qquad \vdots \\
     L_{f} \bar{h}_{N}(x,v_{\hhv}) + L_{g_{\headcav}} \bar{h}_{N}(x) u_{\headcav} \ge - \gamma_{N} \bar{h}_{N}(x)    - \sigma_{N},
    \end{array}
    \right.
\end{split}
\label{eq:QP_head}
\end{align}
where $\sigma_i$ represent the relaxation term and $p_i>0$ are penalty parameters for the relaxation.
This type of controller was first proposed for a single CAV in~\cite{zhao2023safety} which was called safety-critical traffic controller.
If a HV is not connected to the head CAV, the corresponding HV safety constraint shall be omitted from~\eqref{eq:QP_head}.
If no HVs are connected to the head CAV, then all HV safety constraints are dropped, and~\eqref{eq:QP_head} reduces to~\eqref{eq:min_controller_head}.

\textbf{Platoon safety}: In the above design, to ensure the safety of each individual HV, the HV is required to connect to the head CAV. If some HVs are non-connected, we propose to enforce {\em platoon safety} that constrains the overall length of the vehicle platoon, i.e., the total gap between the two CAVs must exceed a minimum value. 
Since the two CAVs are connected, the gap $s_{\headcav\tailcav}$ between them is available.
The gap $s_{\headcav\tailcav}$ can be expressed by $s_{\headcav\tailcav} = s_{\tailcav} + \sum_{i=1}^{N} s_i + l_{\tailcav} +\sum_{i=1}^N l_i$ with $l_{\tailcav}$ and $l_i$ being the length of the tail CAV and HV-$i$, thus $s_{\headcav\tailcav}$ is a function of the system state $x$.
Based on $s_{\headcav\tailcav}$, we define the platoon safety function  as:
\begin{align}
    h_{\platoon}(x) = s_{\headcav\tailcav} - l_0 - \tau_{\platoon} (v_{\tailcav}-v_{\headcav}),
\end{align}
where $l_0>0$ is the base length of the platoon and $\tau_{\platoon}>0$ is a parameter.
The safety constraint constructed by $h_{\platoon}$ is:
\begin{align}\label{eq:CBF platoon}
    L_f h_{\platoon}(x) + L_g h_{\platoon}(x)u \ge -\gamma_{\platoon} h_{\platoon}(x),
\end{align}
with $L_f h_{\platoon}(x) = v_{\headcav} - v_{\tailcav}$,
$L_gh(x) = \begin{bmatrix}
    \tau_{\platoon} & -\tau_{\platoon}
\end{bmatrix}$, and $\gamma_{\platoon}>0$.
This is equivalent to:
\begin{align}\label{eq:CBF explicit platoon}
    u_{\tailcav}-u_{\headcav} \le \frac{1}{\tau_{\platoon}} (v_{\headcav} - v_{\tailcav}) + \gamma_{\platoon} \bigg( \frac{1}{\tau_{\platoon}} (s_{\headcav\tailcav} - l_0)   - (v_{\tailcav}-v_{\headcav}) \bigg).
\end{align}

The platoon safety constraint is of the form $u_{\tailcav} - u_{\headcav} \le \bar{u}$
with $\bar{u}$ being the right-hand side of~\eqref{eq:CBF explicit platoon}.
Meanwhile, the safety of the head CAV can be guaranteed by enforcing an upper bound on its control input $u_{\headcav}$ in the form ${u_{\headcav} \le \bar{u}_{\headcav}}$ where $\bar{u}_{\headcav}$ is the right-hand side of~\eqref{eq:CBF explicit head}.
The tail CAV's safety can be guaranteed by a similar upper bound on its input $u_{\tailcav}$ as ${u_{\tailcav} \le \bar{u}_{\tailcav}}$
where $\bar{u}_{\tailcav}$ is the right-hand side of~\eqref{eq:CBF explicit tail}.
Fig.~\ref{fig:feasibility} illustrates these input constraints. Depending on the values of $\bar{u}_{\headcav}$, $\bar{u}_{\tailcav}$, and $\bar{u}$, there are three possible cases, as shown by the three panels. In each case, there exists a feasibility region for the control inputs $u_{\headcav}$ and $u_{\tailcav}$ where the safety of both the head CAV, tail CAV, and platoon can be enforced simultaneously without relaxation (see shaded domain).
This is summarized by the following Lemma.
\begin{lemma}
    There always exist $u_{\headcav} \in \mathbb{R}$ and $u_{\tailcav} \in \mathbb{R}$ that satisfy three constraints: head CAV safety~\eqref{eq:CBF explicit head}, tail CAV safety~\eqref{eq:CBF explicit tail}, and platoon safety~\eqref{eq:CBF explicit platoon}, for all $x \in \mathbb{R}^{n}$ and $v_{\hhv} \in \mathbb{R}$.
\end{lemma}

\begin{figure}[t]
    \centering
    \includegraphics[width=0.8\linewidth]{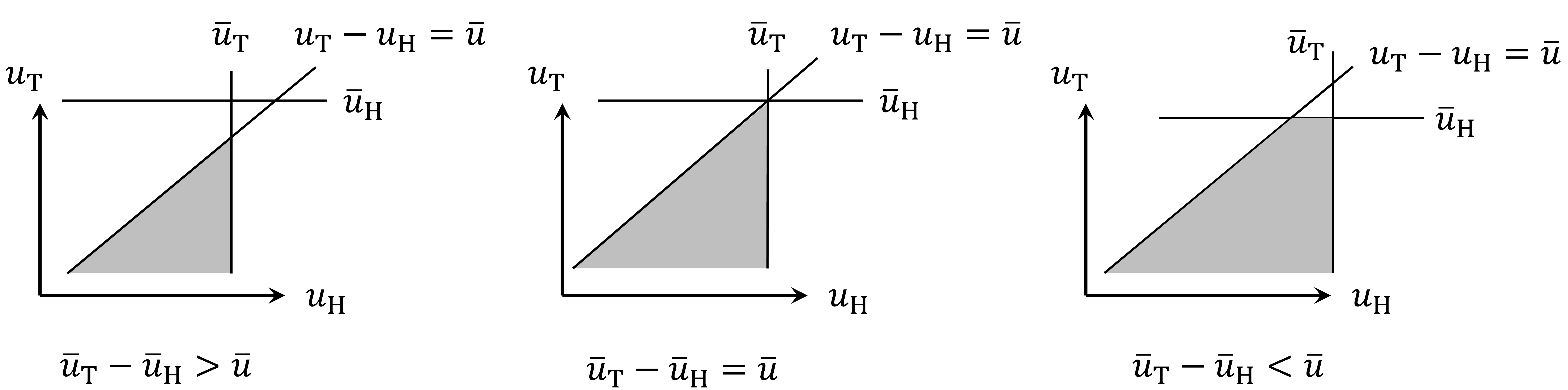}
    \caption{Region of $(u_{\headcav},u_{\tailcav})$ control inputs that ensure the safety of both the H-CAV using~\eqref{eq:CBF explicit head}, the T-CAV using~\eqref{eq:CBF explicit tail}, and the platoon using~\eqref{eq:CBF explicit platoon}.}
    \label{fig:feasibility}
\end{figure}

This finally leads to the following safety constraints design that ensures both CAV safety, HV safety (with relaxation), and platoon safety:
\begin{align}
\begin{split}
    k(x,v_{\hhv}) = \underset{u \in \mathbb{R}^2, \sigma_i \geq 0}{\operatorname{argmin}} \; & \Vert u -k_{\rm n}(x,v_{\hhv}) \Vert^{2} + \sum_{i=1}^N p_i\sigma_i^2 \\
    \text{s.t.} \;
    &\mathrm{CAV \; safety:} \quad \left\{
    \begin{array}{l} 
    L_fh_{\headcav}(x,v_{\hhv}) + L_{g_{\headcav}} h_{\headcav}(x) u_{\headcav} \ge - \gamma_{\headcav} h_{\headcav}(x), \\
    L_fh_{\tailcav}(x) + L_{g_{\tailcav}} h_{\tailcav}(x) u_{\tailcav} \ge - \gamma_{\tailcav} h_{\tailcav}(x),
    \end{array}
    \right.\\
    & \mathrm{HV \; safety:} \quad\
    \left\{
    \begin{array}{l}
    L_{f} \bar{h}_{1}(x,v_{\hhv}) + L_{g_{\headcav}} \bar{h}_{1}(x) u_{\headcav} \ge - \gamma_{1} \bar{h}_{1}(x) -\sigma_1, \\
     \qquad \vdots \\
     L_{f} \bar{h}_{N}(x,v_{\hhv}) + L_{g_{\headcav}} \bar{h}_{N}(x) u_{\headcav} \ge - \gamma_{N} \bar{h}_{N}(x)    - \sigma_{N},
    \end{array}
    \right.\\
    &\mathrm{Platoon \; safety:} \quad  L_{f} h_{\platoon}(x) + L_{g} h_{\platoon}(x) u \ge - \gamma_{\platoon} h_{\platoon}(x).
\end{split}
\label{eq:QP}
\end{align}
It is noted that the platoon safety constraint depends on the inputs of both CAVs.
Therefore, when this constraint is enforced, the two control inputs $u_{\headcav}$ and $u_{\tailcav}$ need to be computed together.
If it is infeasible in practice to compute the control inputs of two CAVs jointly, the platoon safety constraint shall be omitted from~\eqref{eq:QP}.
This leads back to the controller~\eqref{eq:min_controller_tail} for the tail CAV and the controller~\eqref{eq:QP_head} for the head CAV which no longer depend on each other.

\begin{remark}[Look-ahead for stability and look-behind for safety]

As discussed in Remark~\ref{remark:stability cav}-\ref{remark:stability hv} regarding the stability chart Fig.~\ref{fig:stability chart}, for string stability, the tail CAV must look ahead and include feedback from either the middle HVs or the head CAV. As for HV safety, the head CAV should look behind and alter its controller using a safety filter based on the states of the HVs.
    
\end{remark}

\section{Numerical simulation}\label{sec:simulation}

In this section, we conduct a number of simulations to validate the safety and performance of the proposed safety-critical controller~\eqref{eq:QP}. We first consider the case where no HVs are connected to the CAVs, and we discuss CAV safety in subsection~\ref{sec:subsec:simulation CBF}.
Then we consider connected HVs and HV safety in subsection~\ref{sec:subsec:simulation HV}. Finally, we demonstrate the behavior of the controller that enforces platoon safety in subsection~\ref{sec:subsec:simulation platoon}.

\begin{figure}[t]
    \centering
    Nominal control \vspace{0.5ex}\\
    \subfloat[Safety measure $h$]{\includegraphics[width=0.25\linewidth]{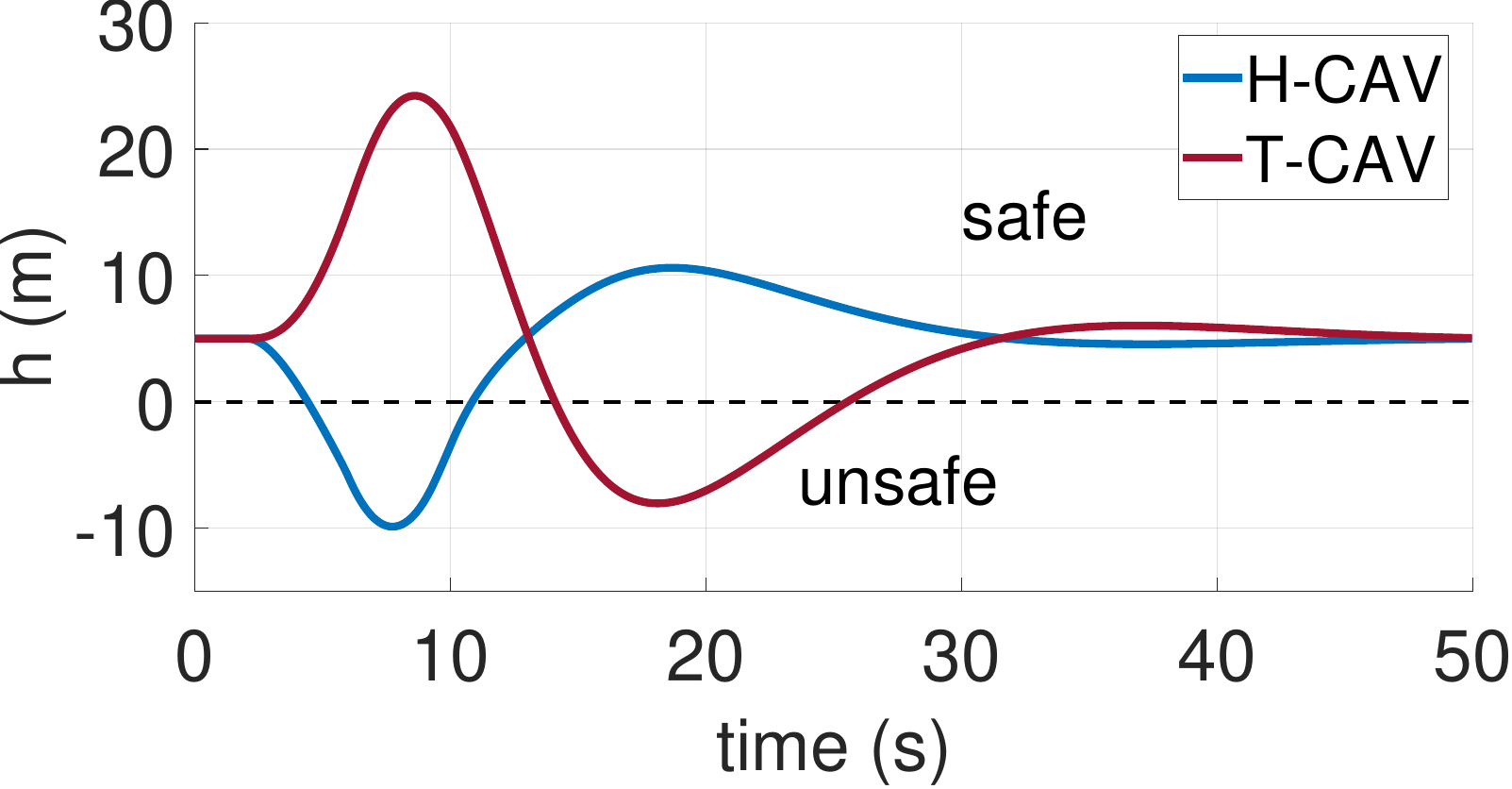}}
    \subfloat[Gap $s$]{\includegraphics[width=0.25\linewidth]{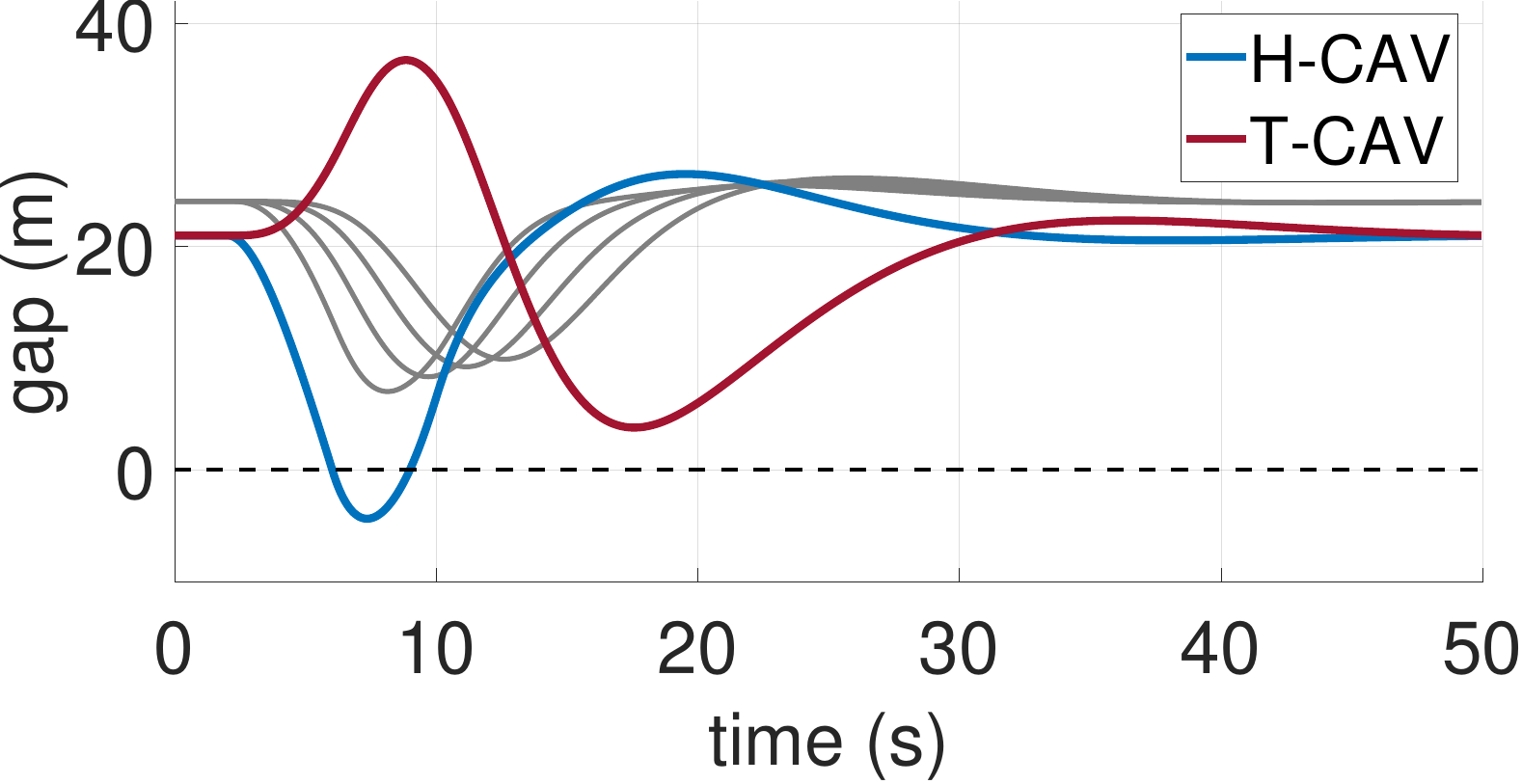}}
    \subfloat[Speed $v$]{\includegraphics[width=0.25\linewidth]{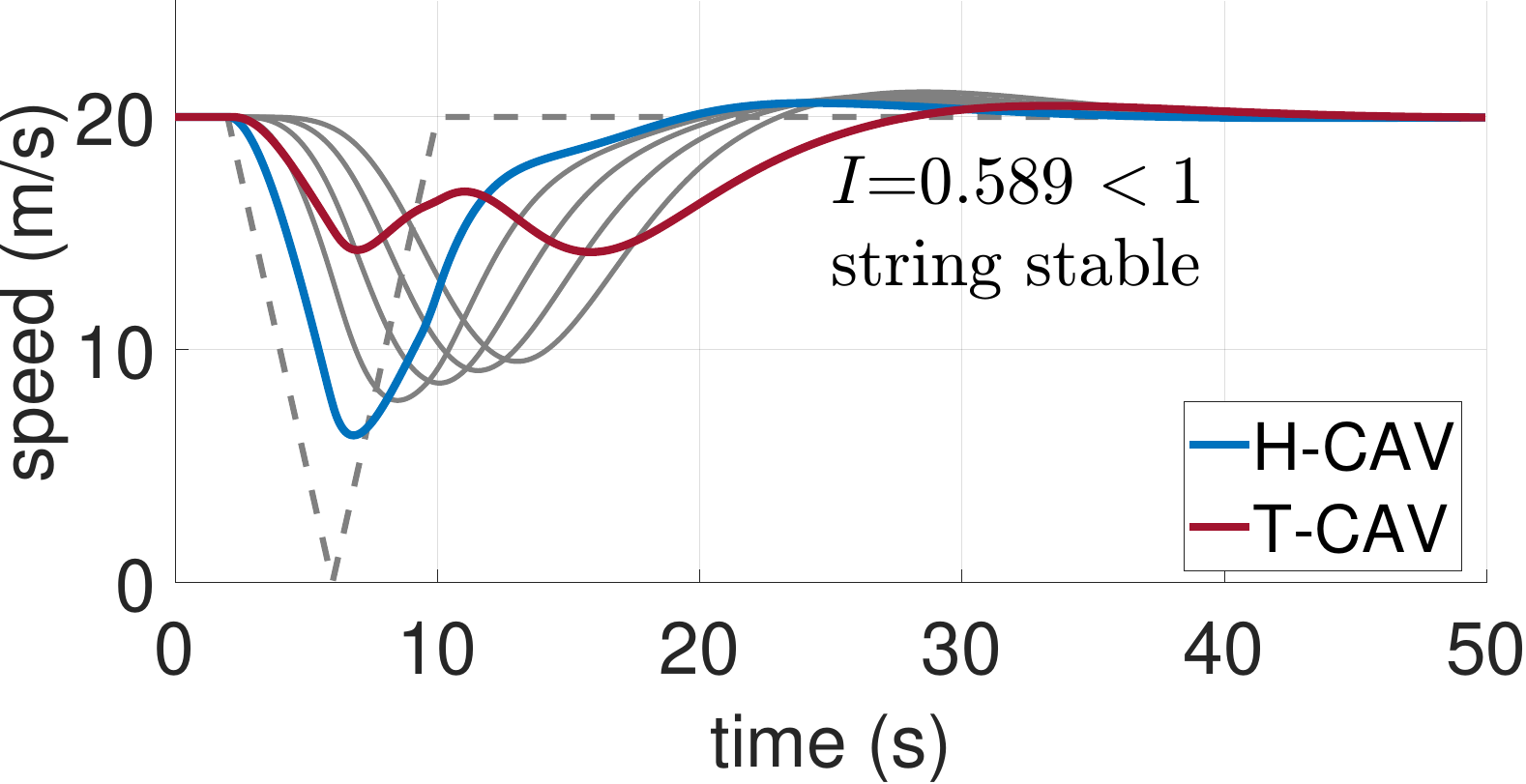}}
    \subfloat[Acceleration $a$]{\includegraphics[width=0.25\linewidth]{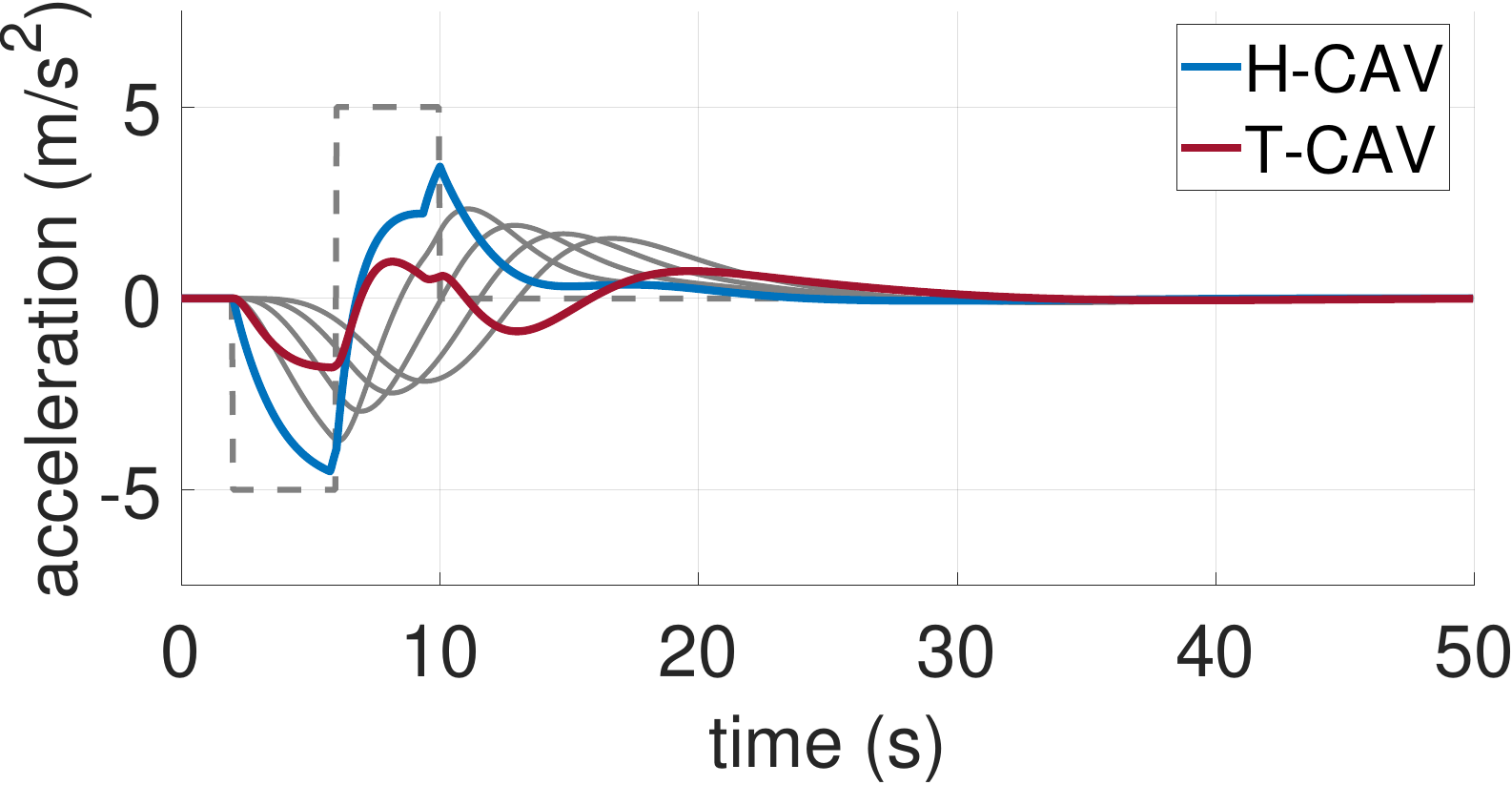}}
    \vspace{2ex} Safety-critical control \vspace{0.5ex}\\
    \subfloat[Safety measure $h$]{\includegraphics[width=0.25\linewidth]{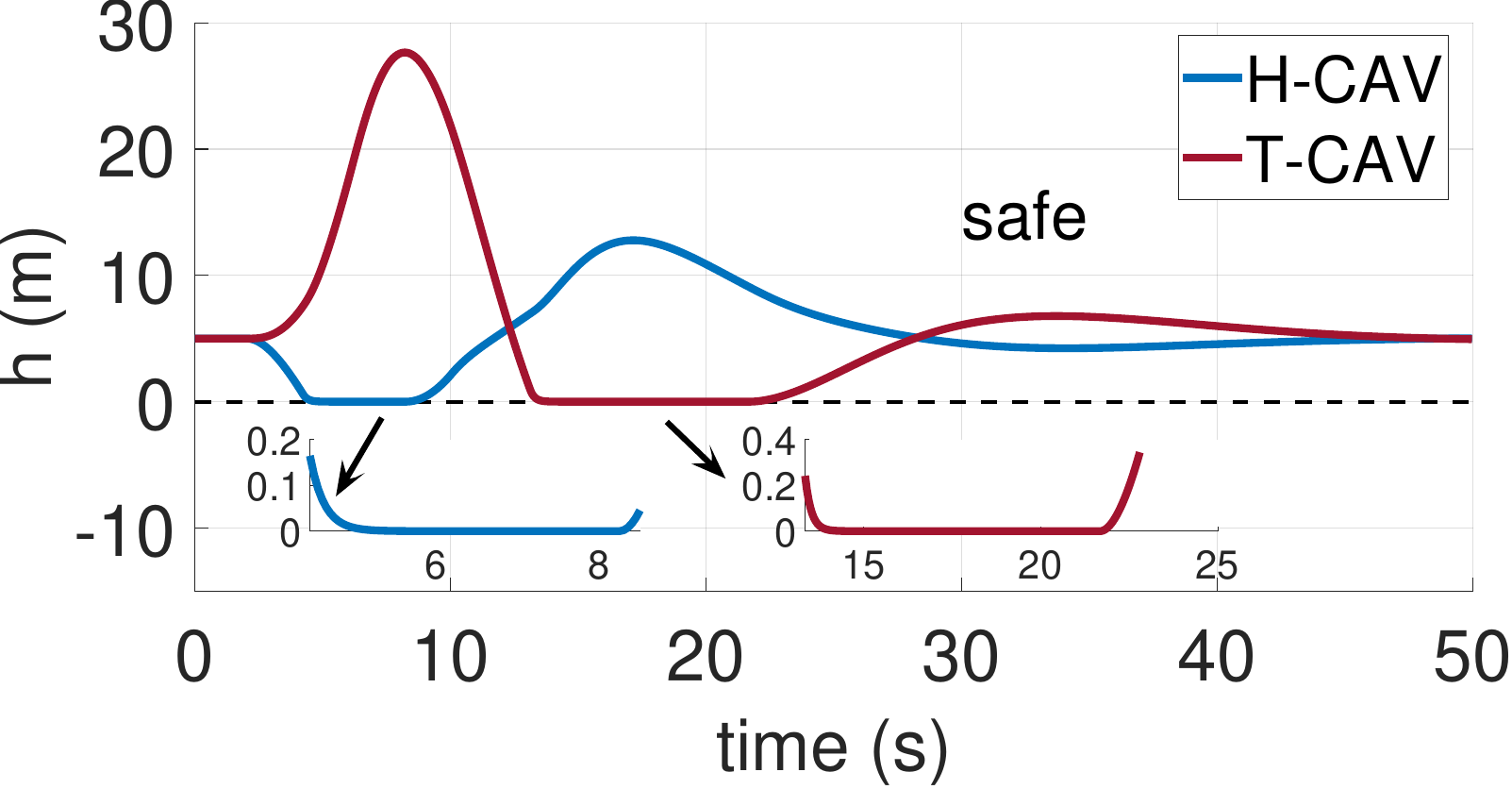}}
    \subfloat[Gap $s$]{\includegraphics[width=0.25\linewidth]{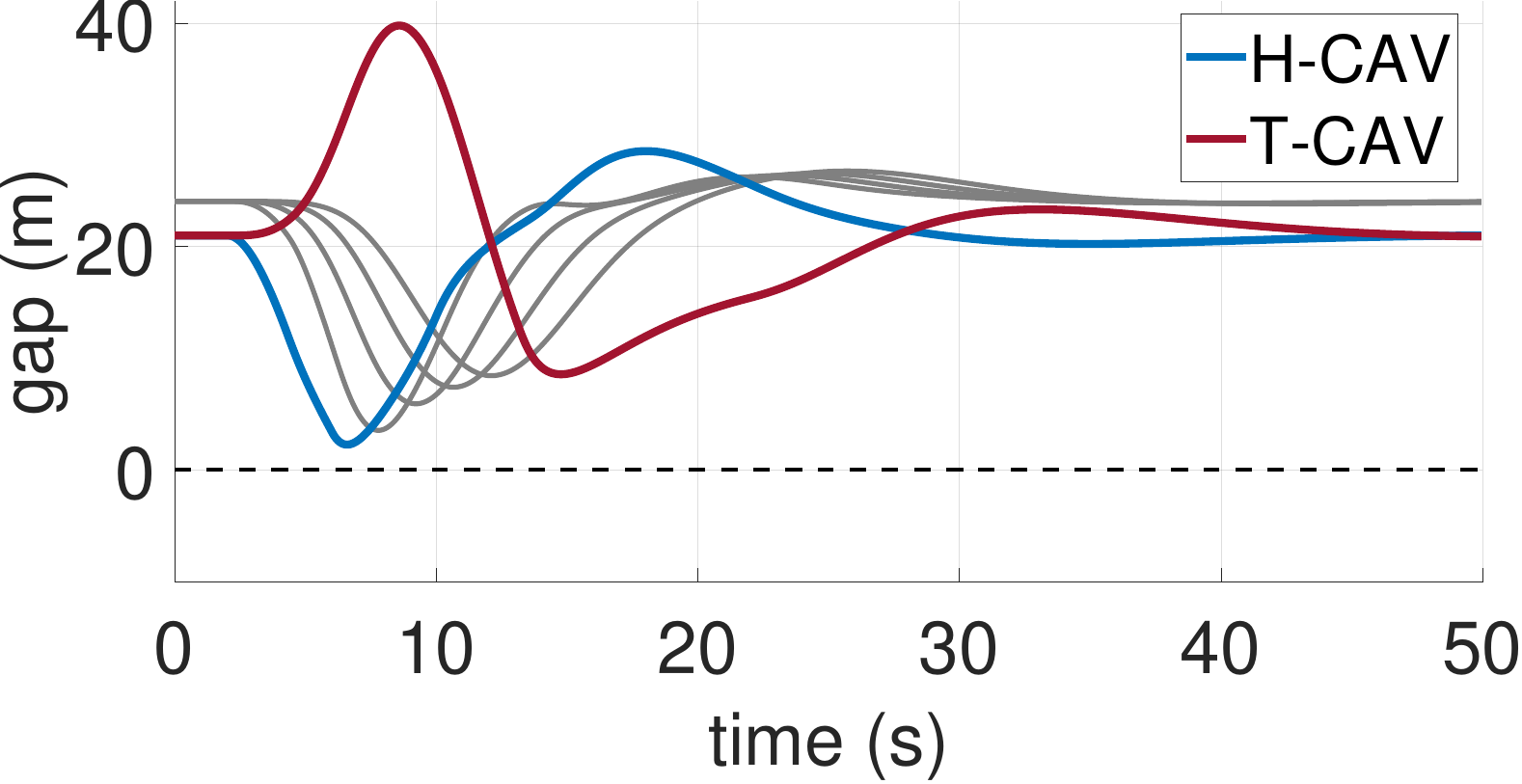}}
    \subfloat[Speed $v$]{    \includegraphics[width=0.25\linewidth]{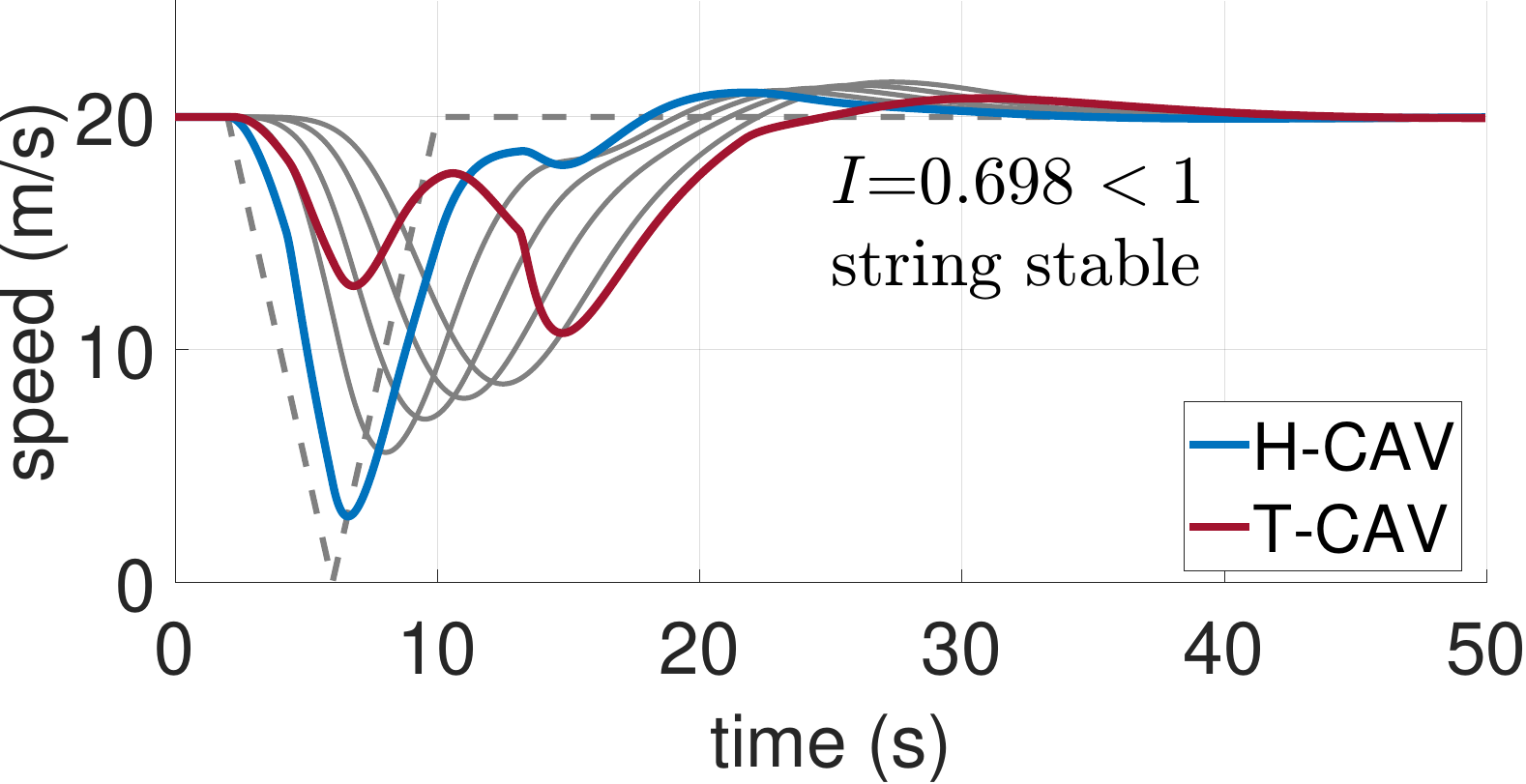}}
    \subfloat[Acceleration $a$]{\includegraphics[width=0.25\linewidth]{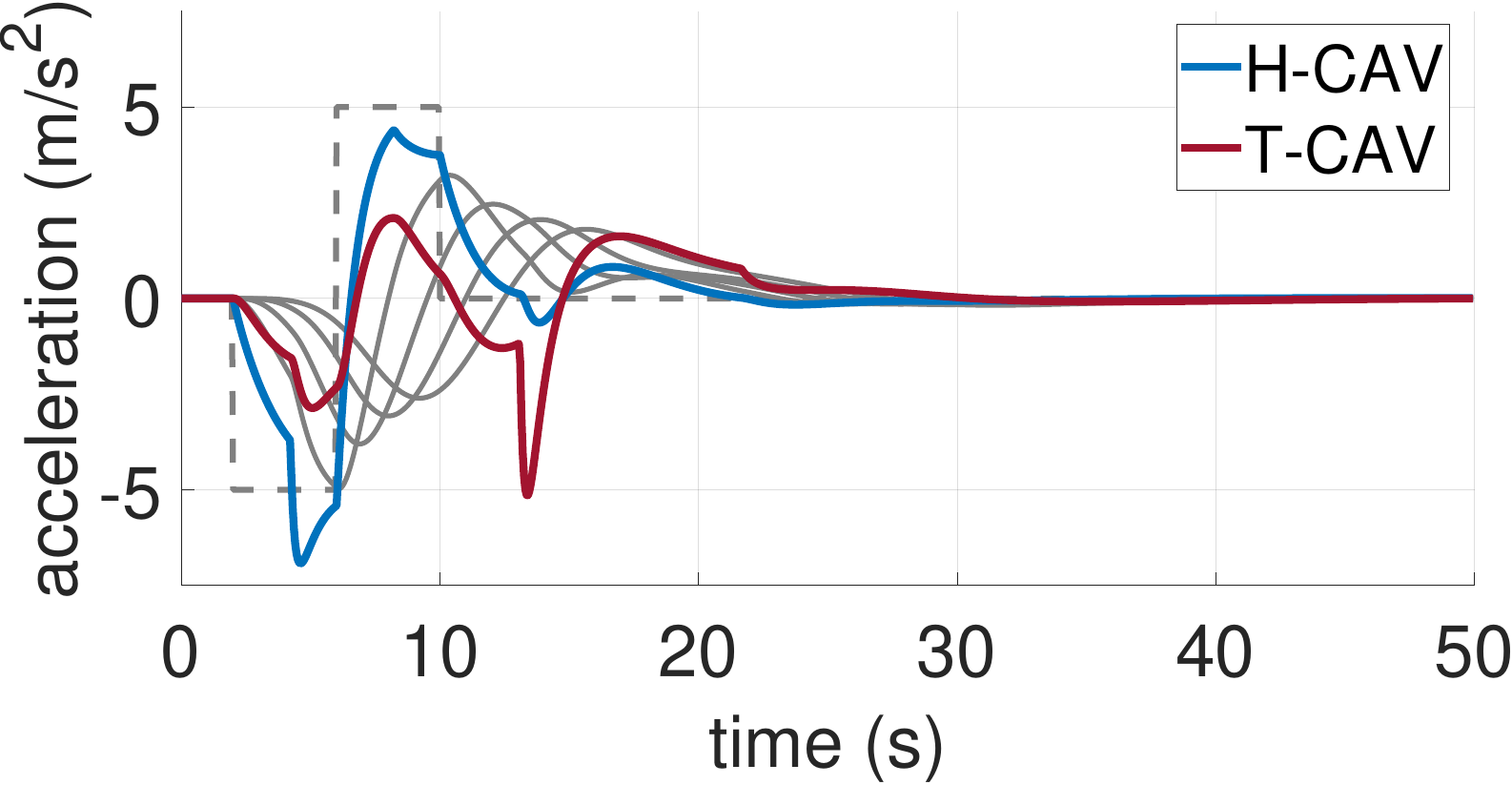}}
    \caption{{\it The head HV suddenly decelerates:} simulated trajectory of mixed vehicle platoon . The nominal controllers cause unsafe driving conditions for the two CAVs, i.e., $h_{\headcav}<0$ and $h_{\tailcav}<0$ occur. By enforcing CBF, the two CAVs become safe with positive $h$.  Besides, the CBF also maintains string stability, i.e., $I<1$.}
    \label{fig:trajectory head HV dec}
\end{figure}
 
\begin{figure}[t]
    \centering
    Nominal control \vspace{0.5ex}\\
    \subfloat[Safety measure $h$]{\includegraphics[width=0.25\linewidth]{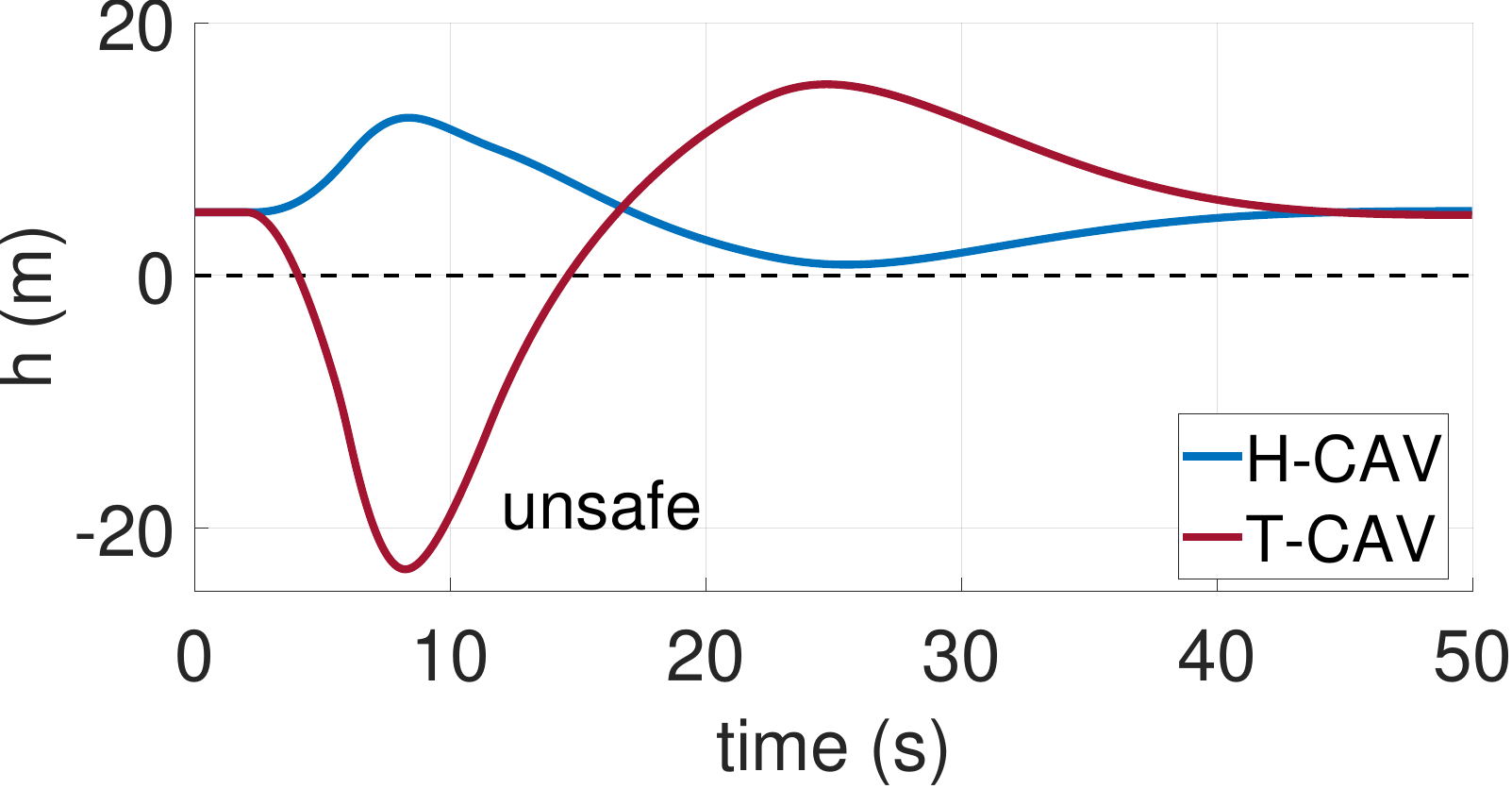}}
    \subfloat[Gap $s$]{\includegraphics[width=0.25\linewidth]{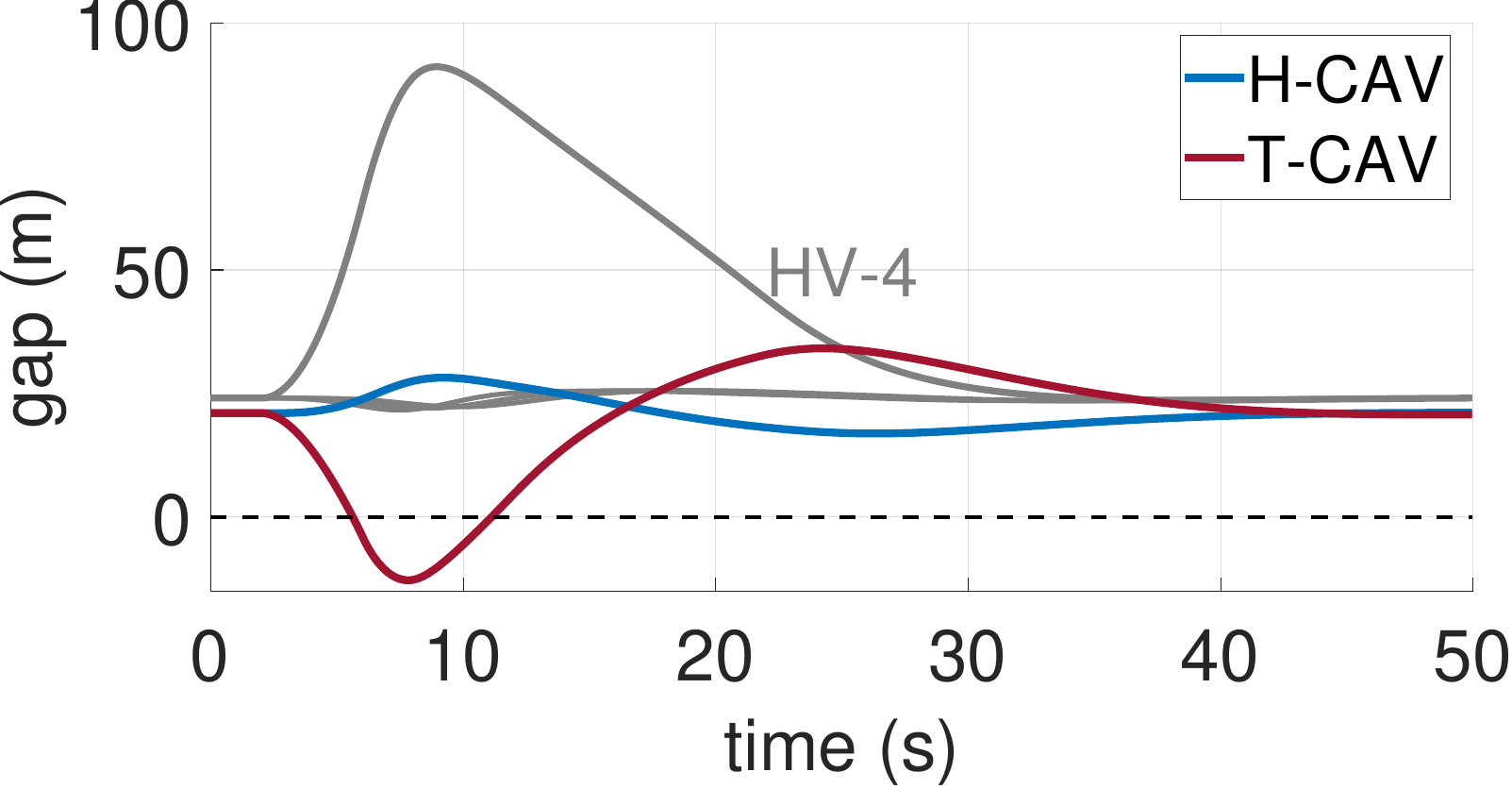}}
    \subfloat[Speed $v$]{\includegraphics[width=0.25\linewidth]{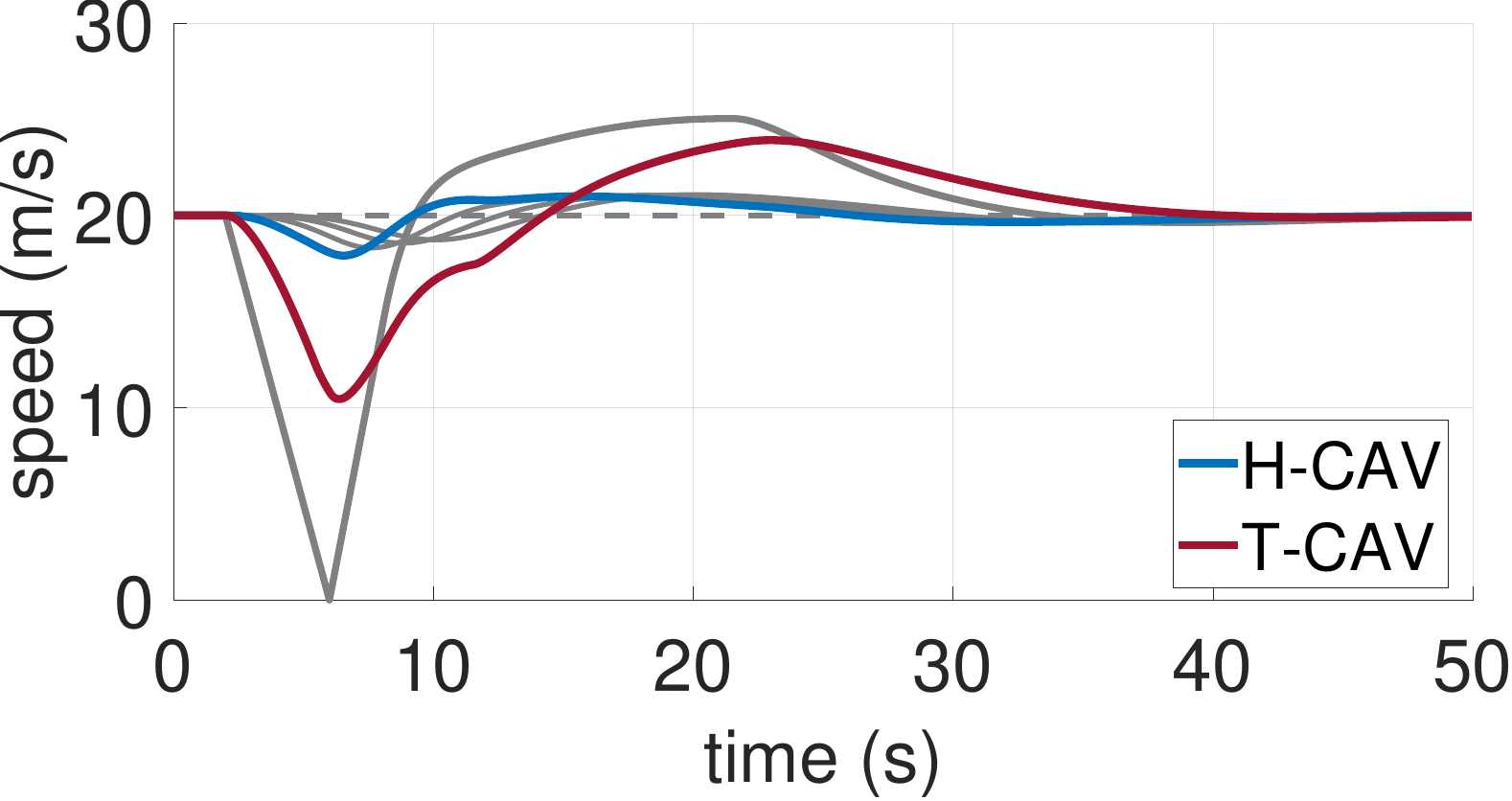}}
    \subfloat[Acceleration $a$]{\includegraphics[width=0.25\linewidth]{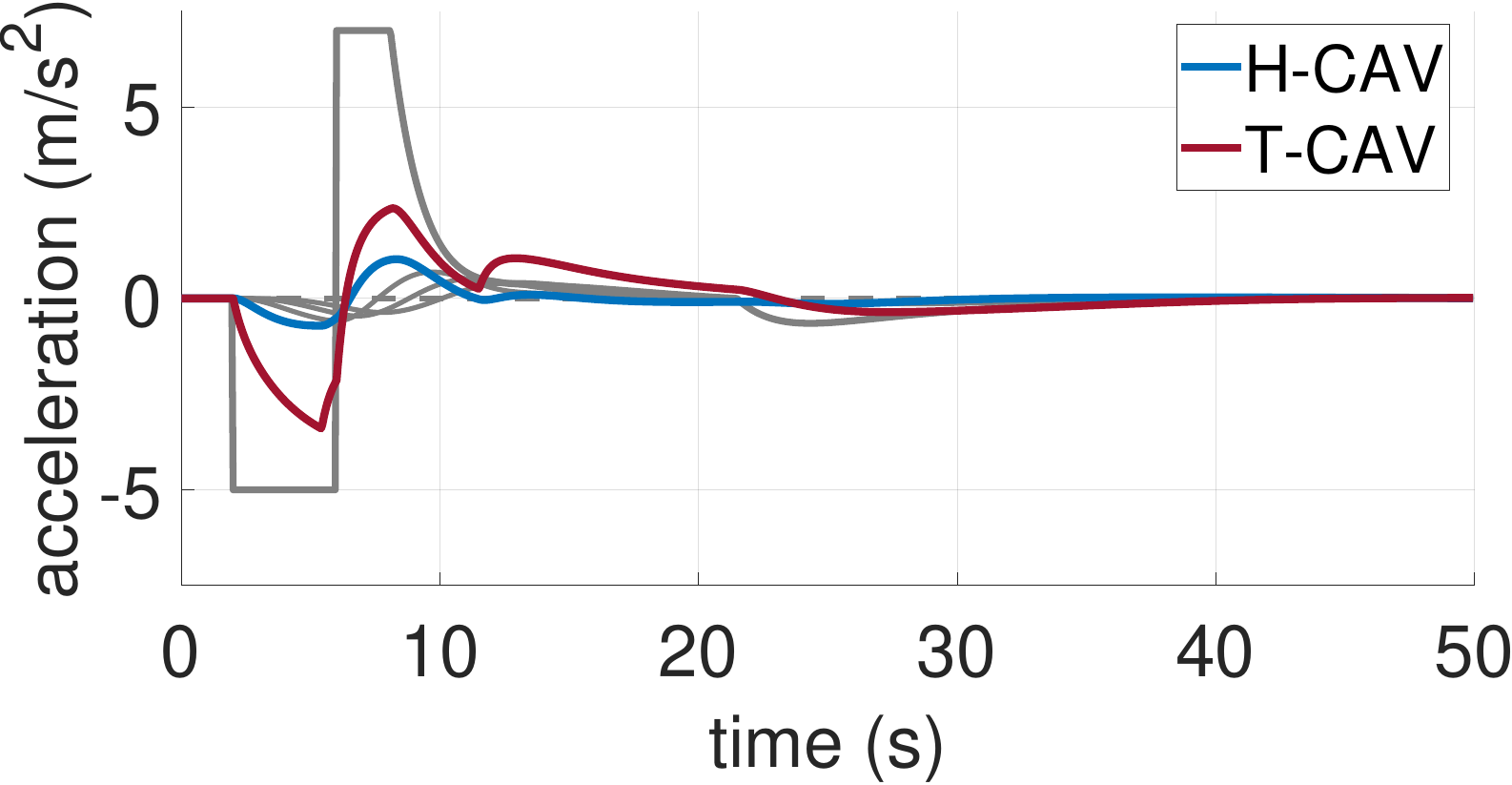}}
    \vspace{2ex} Safety-critical control \vspace{0.5ex}\\
    \subfloat[Safety measure $h$]{\includegraphics[width=0.25\linewidth]{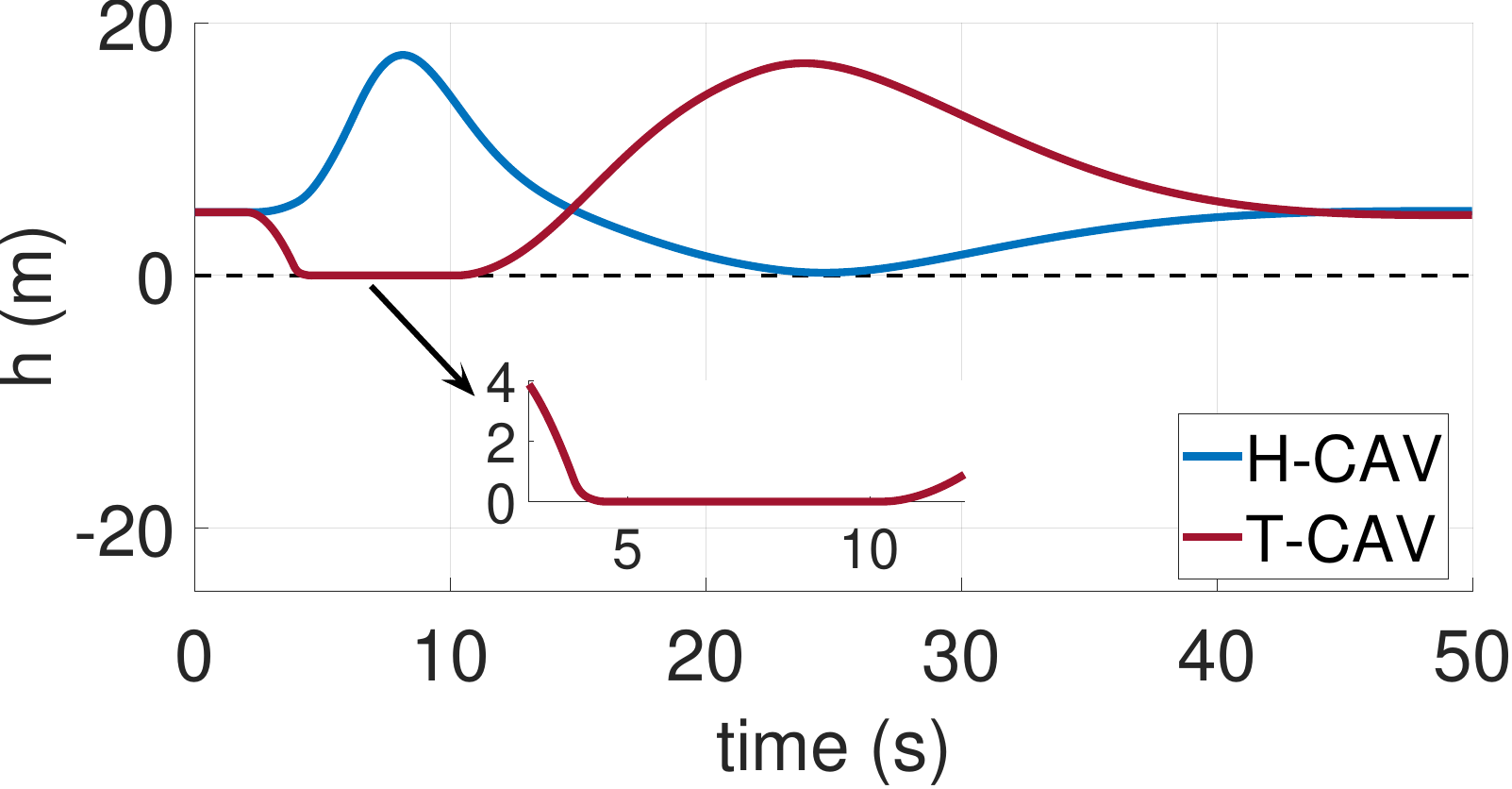}}
    \subfloat[Gap $s$]{\includegraphics[width=0.25\linewidth]{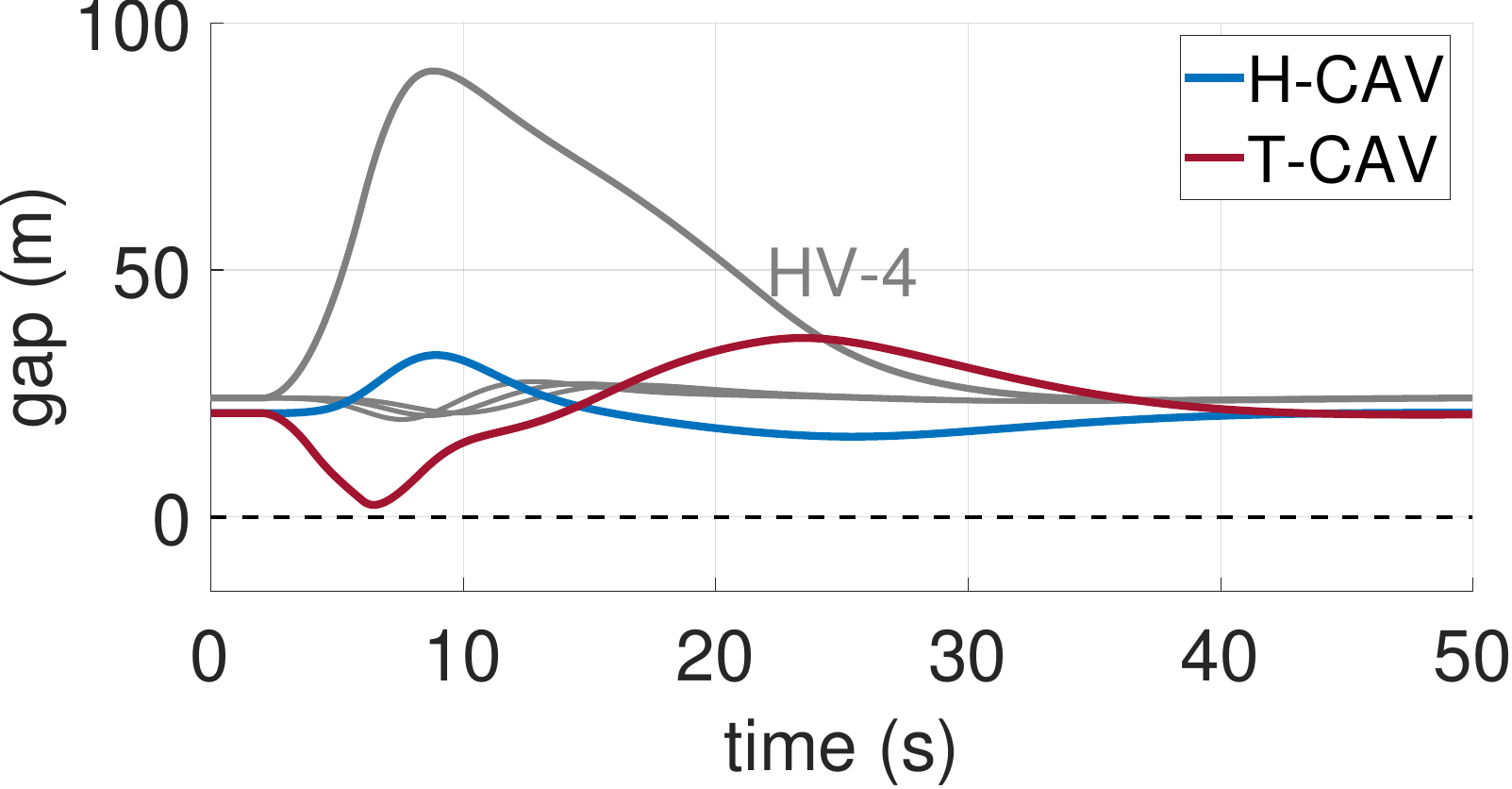}}
   \subfloat[Speed $v$]{\includegraphics[width=0.25\linewidth]{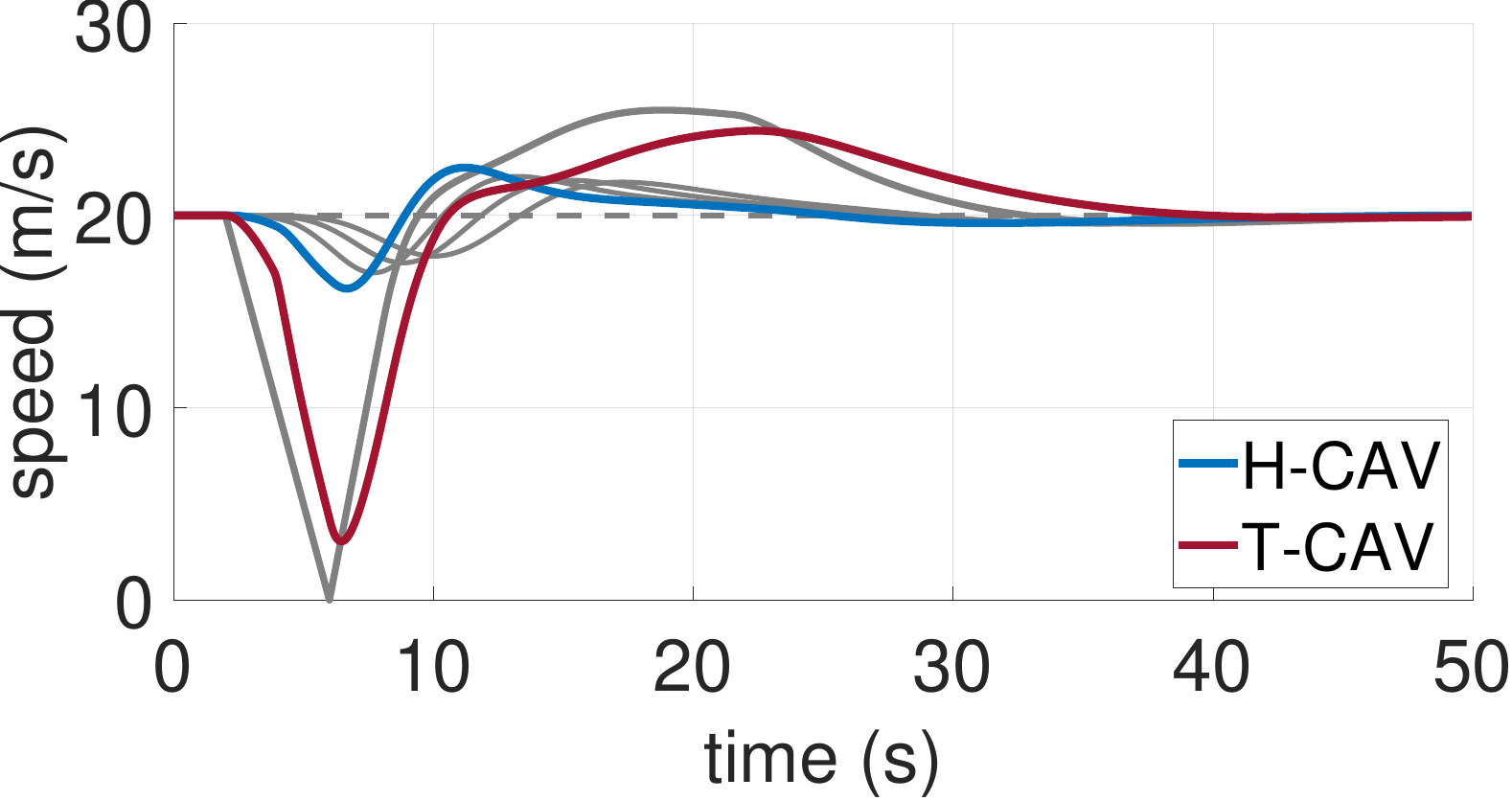}}
    \subfloat[Acceleration $a$]{\includegraphics[width=0.25\linewidth]{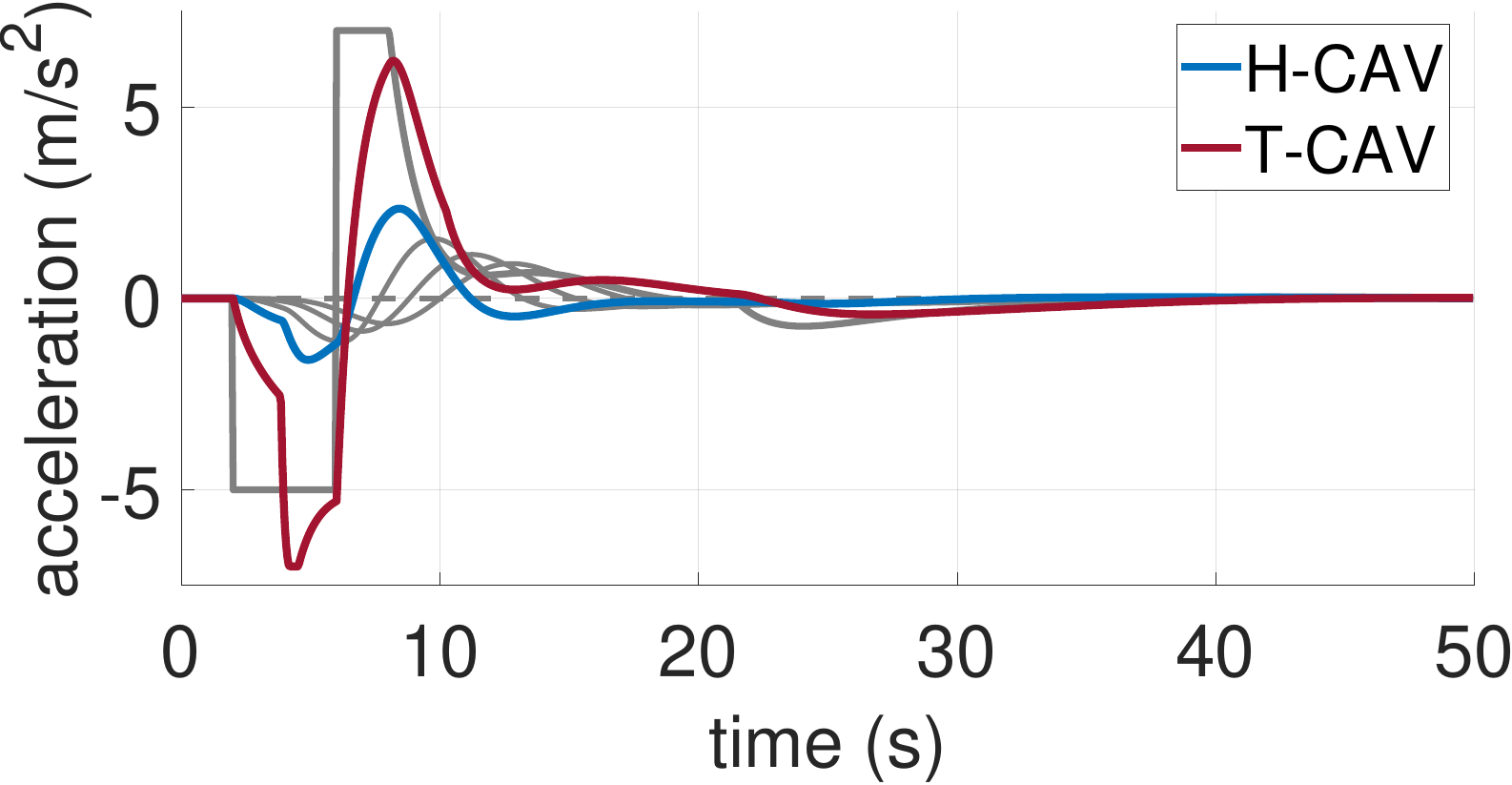}}
    \caption{{\it A middle HV suddenly decelerates:} simulated trajectory of mixed vehicle platoon. Due to the deceleration of HV-4, the tail CAV becomes unsafe when using the nominal controller, while the CBF guarantees safety.}
    \label{fig:trajectory HV dec}
\end{figure}

\subsection{CAV safety guaranteed via CBF} \label{sec:subsec:simulation CBF}

We consider two safety-critical scenarios that may happen in mixed traffic and pose rear-end collision risks:
\begin{itemize}
    \item {\it The head HV suddenly decelerates.} This will pose safety risks for the H-CAV, and possibly also for the T-CAV. This may be caused in real traffic by an aggressive cut-in, pedestrian, or obstacles on the road. In the simulation, we set the head HV's acceleration profile as:
\begin{equation}
    \dot v_\hhv = \left\{
    \begin{array}{ll}
        -a_\hhv, & t \in [2,2+\Delta v_\hhv /a_\hhv], \\
        a_\hhv, & t \in (2+\Delta v_\hhv /a_\hhv,2+2\Delta v_\hhv /a_\hhv], \\
       0, & \mathrm{otherwise},
    \end{array}
       \right.
    \end{equation}
    where $a_\hhv$ is a constant deceleration, $\Delta v_\hhv$ is its speed perturbation, and $\Delta v_\hhv /a_\hhv$ is the duration of deceleration. 
    \item {\it One middle HV suddenly decelerates.} This will pose safety risks for the T-CAV. We make the HV-$i$ suddenly reduce its speed by $\Delta v_i$ with a constant deceleration $a_i$. The acceleration profile of HV-$i$ is, cf.~\eqref{eq:system HV v},:
    \begin{equation}
        \dot{v}_i = \left\{
    \begin{array}{ll}
        -a_i, & t \in [2,2+\Delta v_i /a_i], \\
       F_i(s_i,v_{i},\dot{s}_i), & \mathrm{otherwise},
    \end{array}
       \right.
    \end{equation}

\end{itemize}

Considering the physical limits of vehicle braking systems, we use saturated acceleration in the simulation, i.e., $\dot{v}_i = \mathrm{sat}(F_i(s_i,v_i,\dot{s}_i))$, $\dot{v}_{\headcav}=\mathrm{sat}(u_{\headcav})$, $\dot{v}_{\tailcav}=\mathrm{sat}(u_{\tailcav})$, where the saturation function is
\begin{align}
    \mathrm{sat} (u) =  \begin{cases}
    u_{\min}, & u< u_{\min}, \\
    u, & u_{\min} \le u \le u_{\max}, \\
    u_{\max}, & u> u_{\max},
    \end{cases}
\end{align}
with $u_{\max}$ and $u_{\min}$ being the maximum and minimum acceleration. In the simulation, we use $u_{\max} = 7\; \mathrm{m/s}^2$ and $u_{\min} = -7\; \mathrm{m/s}^2$.  For the nominal stabilizing controller, we consider CAV coordination gain as $\beta_{\headcav,\tailcav} = 0.5$ s$^{-1}$  and $\beta_{\tailcav,\headcav} = 1.2$  s$^{-1}$. We set the safe time headway as $\tau_{\headcav} =0.8 $ s and $\tau_{\tailcav} = 0.8$ s.  For the CBF parameters, we use $\gamma_{\headcav} = 5$ s$^{-1}$ and $\gamma_{\tailcav} = 5$ s$^{-1}$. The human-driver model $F_i$ and the rest of the parameters are the same as in Section~\ref{sec:subsec stability chart}.

Fig.~\ref{fig:trajectory head HV dec} shows the simulated trajectories by using the nominal controller~\eqref{eq:nominal controller head CAV}, \eqref{eq:nominal controller tail CAV} and the safety-critical controller~\eqref{eq:min_controller_head},~\eqref{eq:min_controller_tail} for the first scenario when the head HV suddenly decelerates from ${v^*=20 \ \mathrm{m/s}}$ to a stop with $a_{\hhv} = 5 \ \mathrm{m/s^2}$ and $\Delta v_{\hhv} = 20$ m/s. The behavior of the nominal controller is shown in the top row. The CAVs stabilize the traffic by having a small deceleration (see panels (c) and (d)), but this causes a collision between the head CAV and the head HV (see the negative gap along the blue curve in panel (b)). Besides, the tail CAV also becomes unsafe in the sense that its safety function becomes negative (see the red curve in panel (a)). The behavior of the safety-critical controller that includes CBF-based safety filter is shown in the bottom row.
Initially, when the head HV drives at the constant speed $v^*$, the safety-critical controller matches the nominal controller since it satisfies the CBF constraints~\eqref{eq:CBF head CAV},~\eqref{eq:CBF tail CAV}. When the head HV begins to decelerate, the CBF is activated around 5 sec, and the head CAV has a larger deceleration than with the nominal controller (see panels (g) and (h)).
This way the head CAV successfully avoids collision (see the positive gap in panel (f)).
The safety functions are kept positive throughout the motion (see panel (e)), which indicates that both CAVs are safe.

Besides avoiding unsafe scenarios caused by the nominal controller,  the safety-critical controller also maintains both plant stability and head-to-tail string stability of the system. 
From the profiles of gap $s$ in the second column and speed $v$ in the third column of Fig.~\ref{fig:trajectory head HV dec}, we observe plant stability as $s$ and $v$ converge to the equilibrium values $s^*$ and $v^*$ around 45 sec. To evaluate head-to-tail string stability, we quantify the speed perturbations of the tail CAV relative to the head HV by:
\begin{align}\label{eq:stability index tail}
    I = \frac{\sqrt{\int_0^T (v_{\tailcav} - v^*)^2 \diff t}}{\sqrt{\int_0^T (v_{\hhv} - v^*)^2 \diff t}},
\end{align}
where $T=50$ sec is the total simulation length. By Definition~\ref{definition:head to tail string stability}, if $I<1$, then the system is head-to-tail string stable.  After incorporating the CBF, we have $I = 0.698 <1$.  
Therefore, by implementing CBF, the system is both string stable and safe.
Note that $I = 0.589$ for the nominal controller, i.e., the value of $I$ is higher for the safety-critical controller.
This means that safety is achieved at the price of higher speed fluctuations (but without losing head-to-tail string stability).

In Fig.~\ref{fig:trajectory HV dec}, we plot the profiles of the safety function, gap, speed, and acceleration when one middle HV suddenly decelerates. We consider HV-4 to decelerate with $a_4 = 5 \ \mathrm{m/s^2}$ and $\Delta v_4 = 20$ m/s.  As the top row of Fig.~\ref{fig:trajectory HV dec} shows, considering the nominal controller, the deceleration of HV-4 poses a safety risk for the tail CAV. Meanwhile, based on the bottom row, the safety-critical controller with the CBF has a larger deceleration for the tail CAV to avoid collisions. In this scenario, the safety-filtered controller still stabilizes, since the system converges to the equilibrium state after the perturbation of HV-4 (around 40 sec). We note that the head-to-tail string stability index in~\eqref{eq:stability index tail} is not applicable, since it concerns cases where the head HV has a speed perturbation while in this scenario the head HV drives at a constant speed.

\subsection{Connected HV safety and effect of HV connection}\label{sec:subsec:simulation HV}

If a middle HV is connected to the head CAV, then the head CAV may guarantee the safety of this HV via a CBF constraint as~\eqref{eq:CBF HV}. We validate this HV safety constraint. We consider a scenario where {\it one middle HV suddenly accelerates}. This increases the collision risk between the accelerating HV and its leader. We set HV-$i$ to suddenly accelerate with the acceleration profile:
\begin{equation}
    \dot{v}_i = \left\{
    \begin{array}{ll}
        a_i, & t \in [2,2+ \Delta v_{i}/a_i], \\
       F_i(s_i,v_{i},\dot{s}_i), & \mathrm{otherwise},
    \end{array}
   \right.
\end{equation}
with $a_i$ being a constant acceleration, and $\Delta v_{i}$ being the speed perturbation.
We assume that the accelerating HV is HV-1, which is connected to the head CAV, and the nominal controller of the head CAV~\eqref{eq:nominal controller head CAV} includes HV-1's feedback with the controller gain $\beta_{\headcav,1} = 0.1$. The corresponding safety-critical controller of the head CAV is given by~\eqref{eq:QP_head}, where the HV-1 safety constraint with $\bar{h}_{1}$ is included while the other HV safety constraints are omitted. We set the HV's safe time headway as $\tau_1 = 1$ s, and the CBF parameters as $\gamma_1 = 5$ s$^{-1}$, $\eta_1 = 0.5$, $p_1 = 100$. Other parameters, including the tail CAV's nominal controller~\eqref{eq:nominal controller tail CAV} and safety-critical controller~\eqref{eq:min_controller_tail}, the saturation function, and the CBF parameters for the head and tail CAVs, remain the same. 

Fig.~\ref{fig:trajectory HV acc} plots the trajectories by the nominal controller~\eqref{eq:nominal controller head CAV},~\eqref{eq:nominal controller tail CAV} and the safety-critical controller~\eqref{eq:min_controller_tail},~\eqref{eq:QP_head} with $a_1 =  5 \ \mathrm{m/s^2}$ and $\Delta v_1 = 13$ m/s. The results with the nominal controller are shown  in the top row. When HV-1 suddenly accelerates, the head CAV also accelerates, but the gap between HV-1 and the head CAV still becomes too small, hence HV-1 is unsafe (see negative $h_1$). By using the CBF, as shown in the bottom row, the head CAV generates a larger acceleration.  This enlarges the gap between HV-1 and the head CAV, and thus ensures safety (i.e., $h_1$ is positive for all time).   In this scenario, the safety-filtered controller still stabilizes, since the system converges to the equilibrium state around 25 sec.

\textit{Connectivity of HVs}: Next, we investigate how the connectivity of HVs to CAVs affects the stability and safety of the nominal controller.
These simulations correspond to the stability results presented earlier in Fig.~\ref{fig:stability chart}, which considered three communication topologies: no HV connection, tail CAV connects to HVs (look ahead), and head CAV connects to HVs (look behind). 
We investigate the behavior of the vehicle platoon for each communication topology in three simulation scenarios: head HV decelerates, middle decelerates, and middle HV accelerates; corresponding to Figs.~\ref{fig:trajectory head HV dec},~\ref{fig:trajectory HV dec} and~\ref{fig:trajectory HV acc}.

\begin{figure}[t]
    \centering
    Nominal controller \vspace{0.5ex}\\
    \subfloat[Safety measure $h$]{\includegraphics[width=0.25\linewidth]{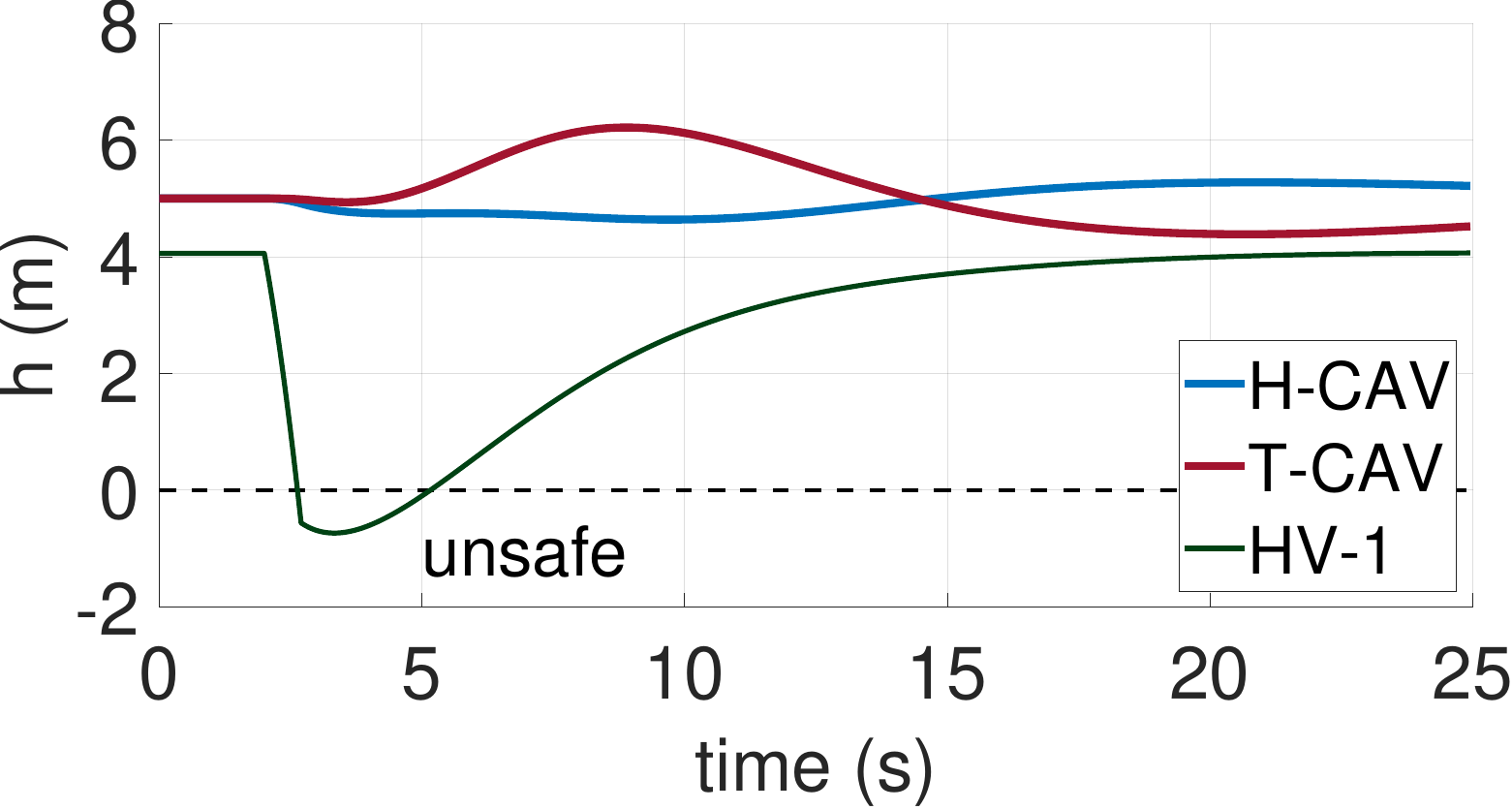}}
    \subfloat[Gap $s$]{\includegraphics[width=0.25\linewidth]{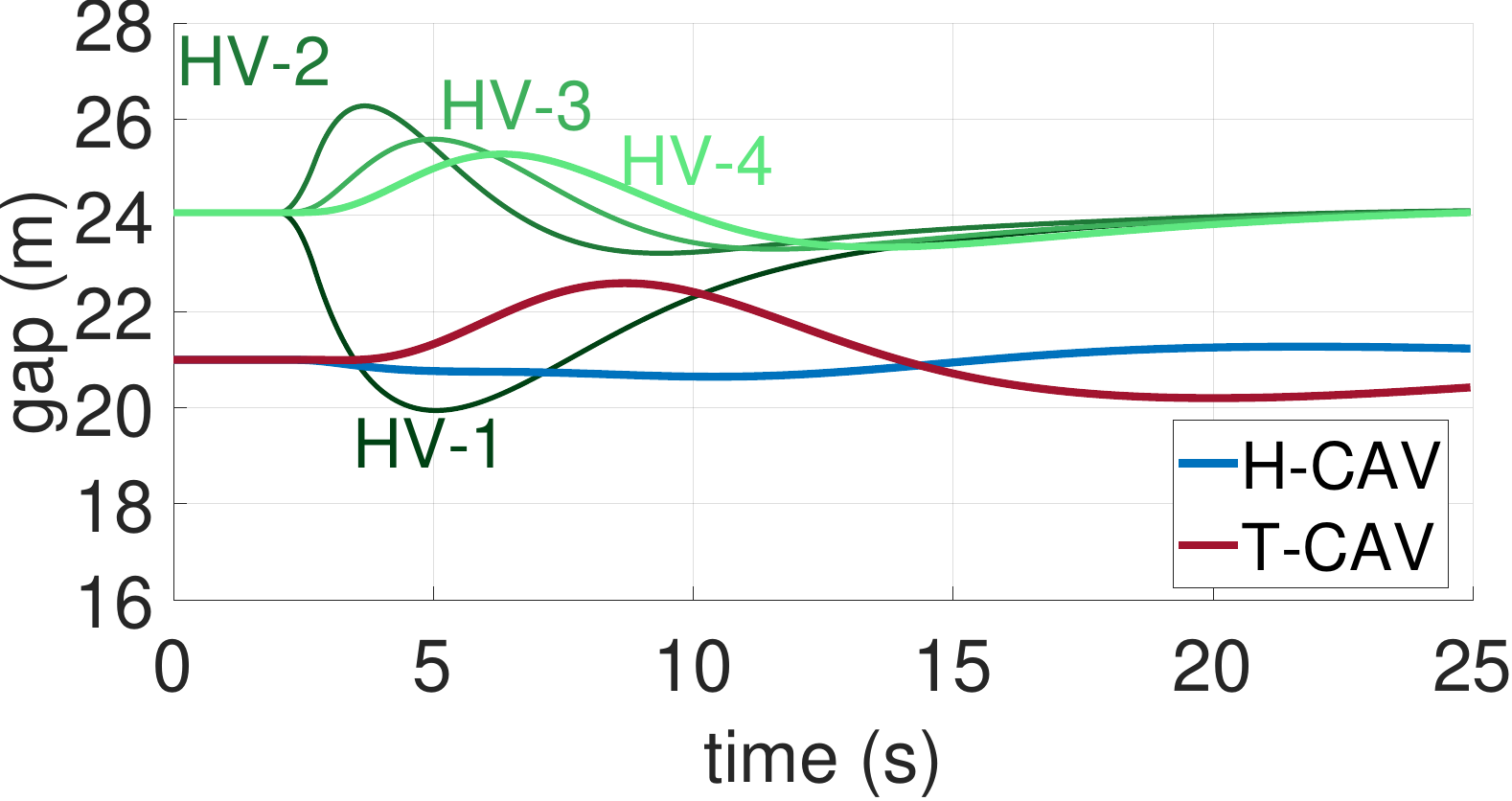}}
    \subfloat[Speed $v$]{\includegraphics[width=0.25\linewidth]{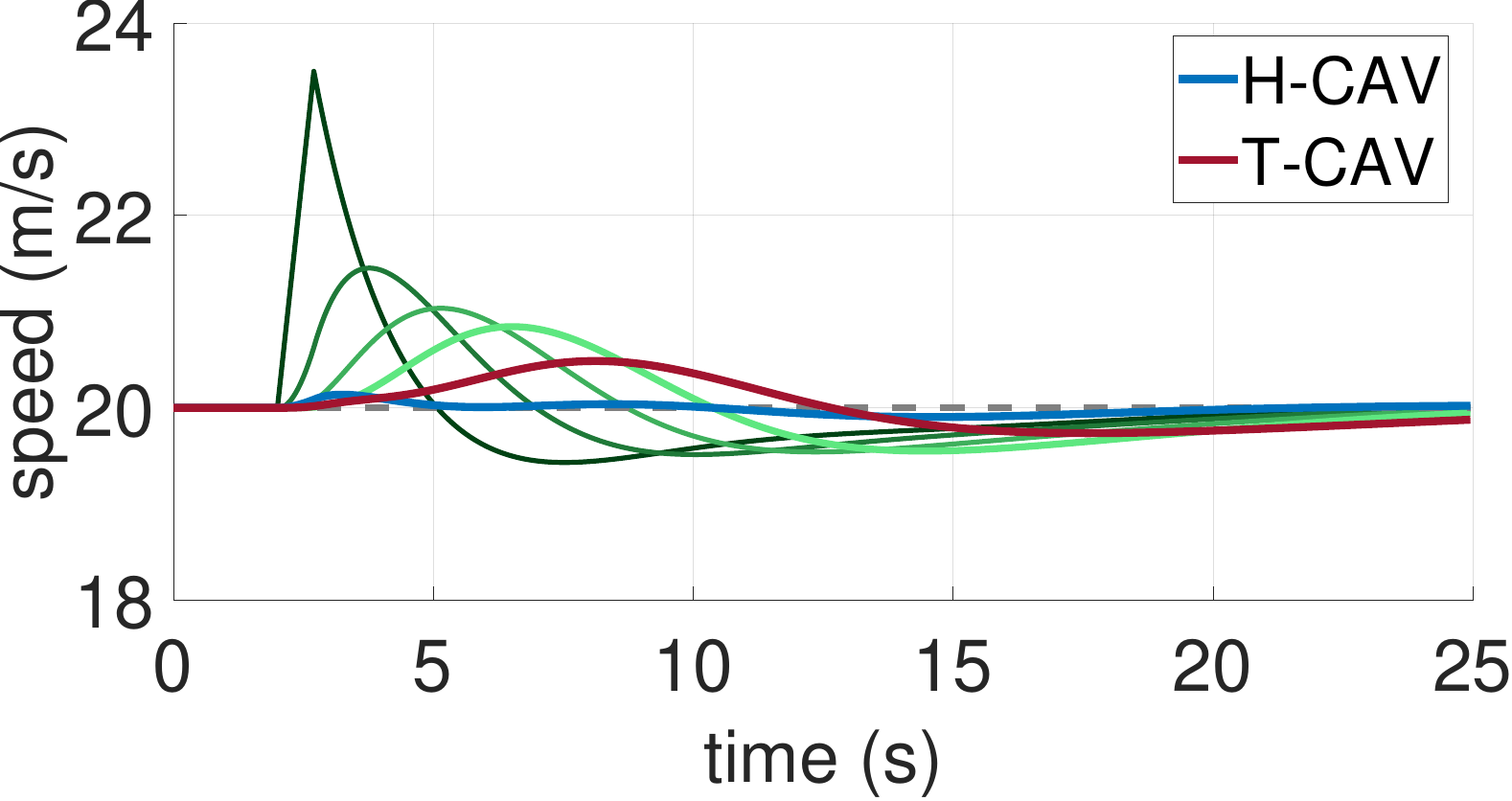}}
    \subfloat[Acceleration $a$]{\includegraphics[width=0.25\linewidth]{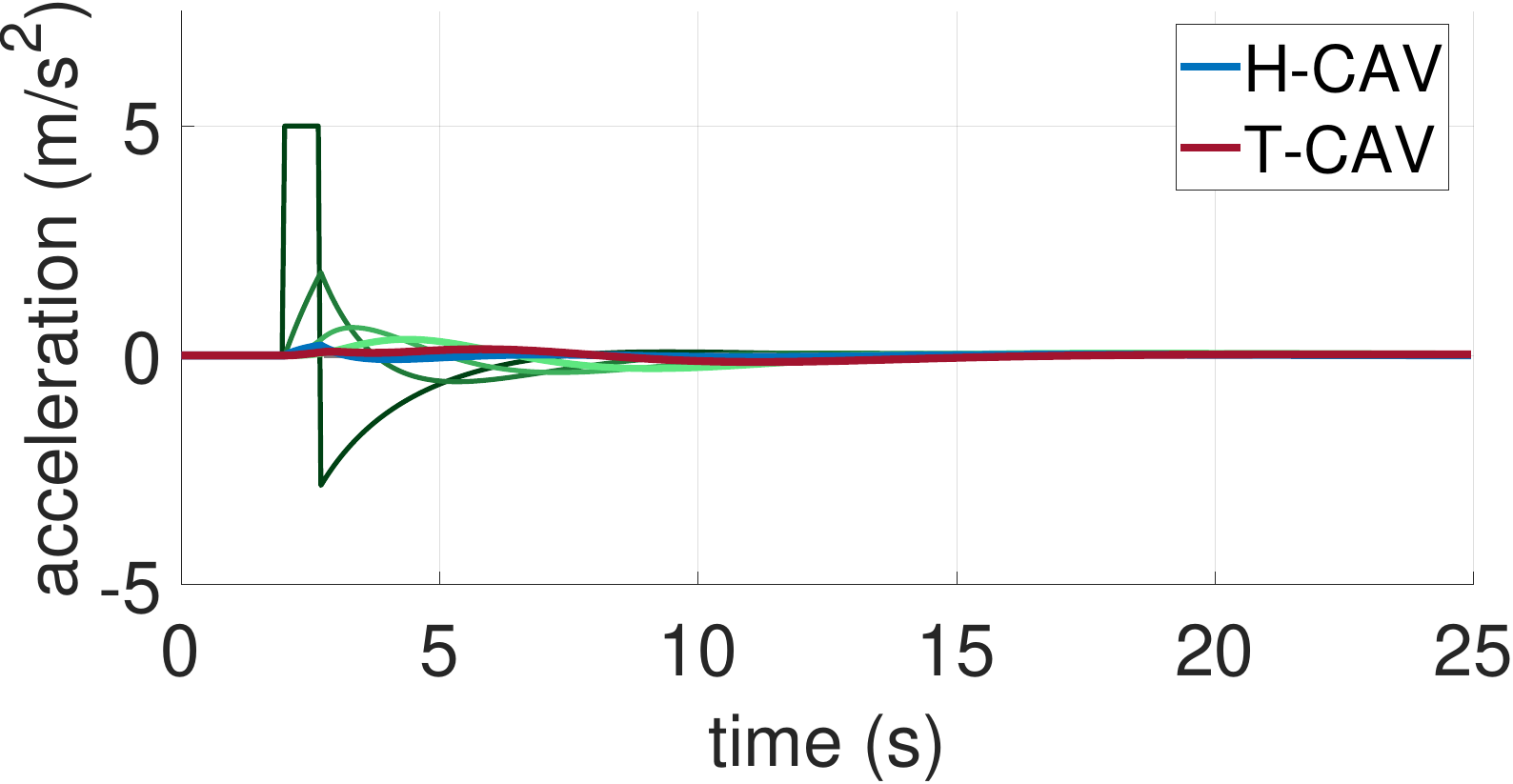}} \\ 
    \vspace{2ex} CBF \vspace{0.5ex}\\
    \subfloat[Safety measure $h$]{\includegraphics[width=0.25\linewidth]{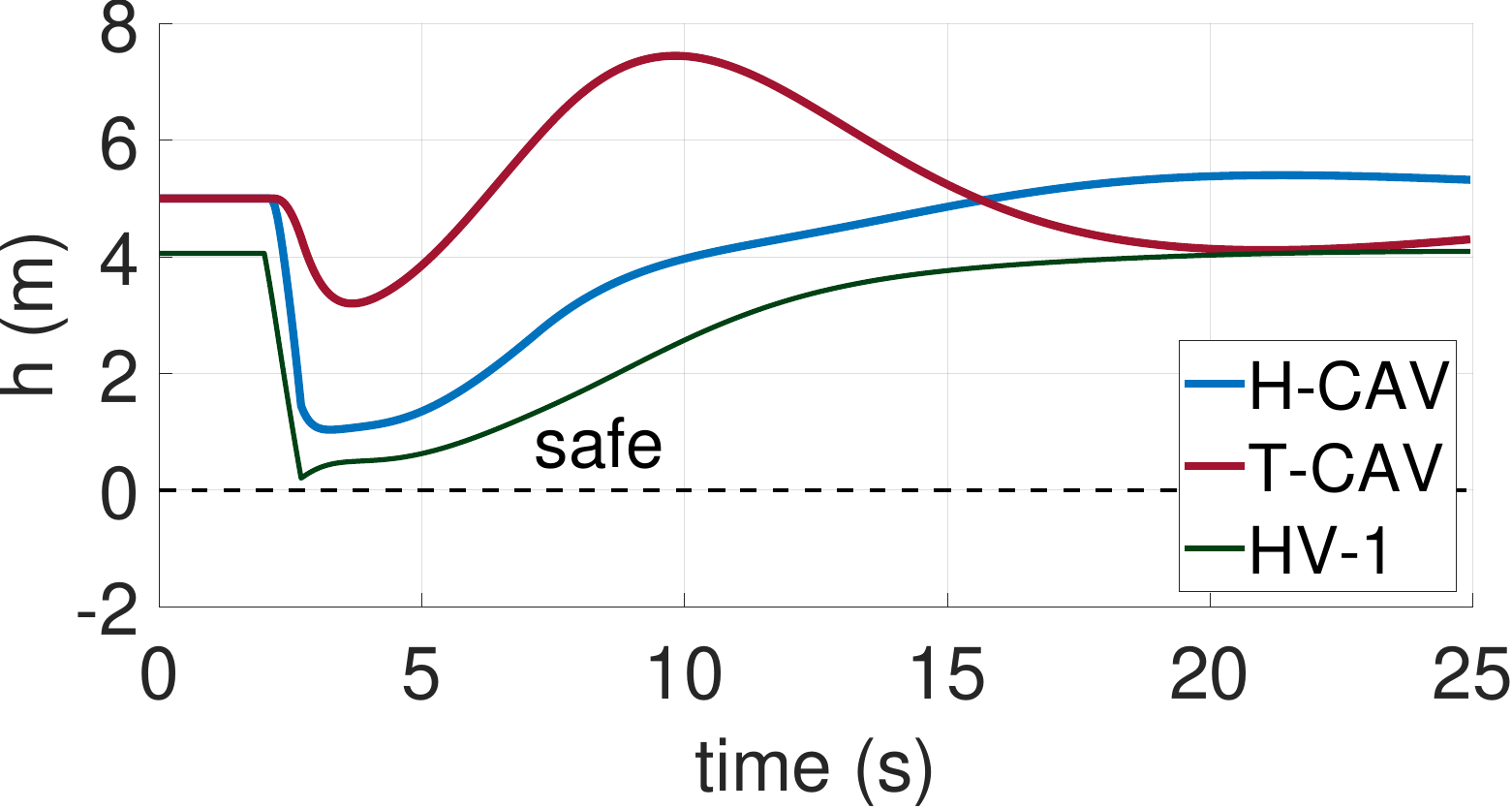}}
    \subfloat[Gap]{\includegraphics[width=0.25\linewidth]{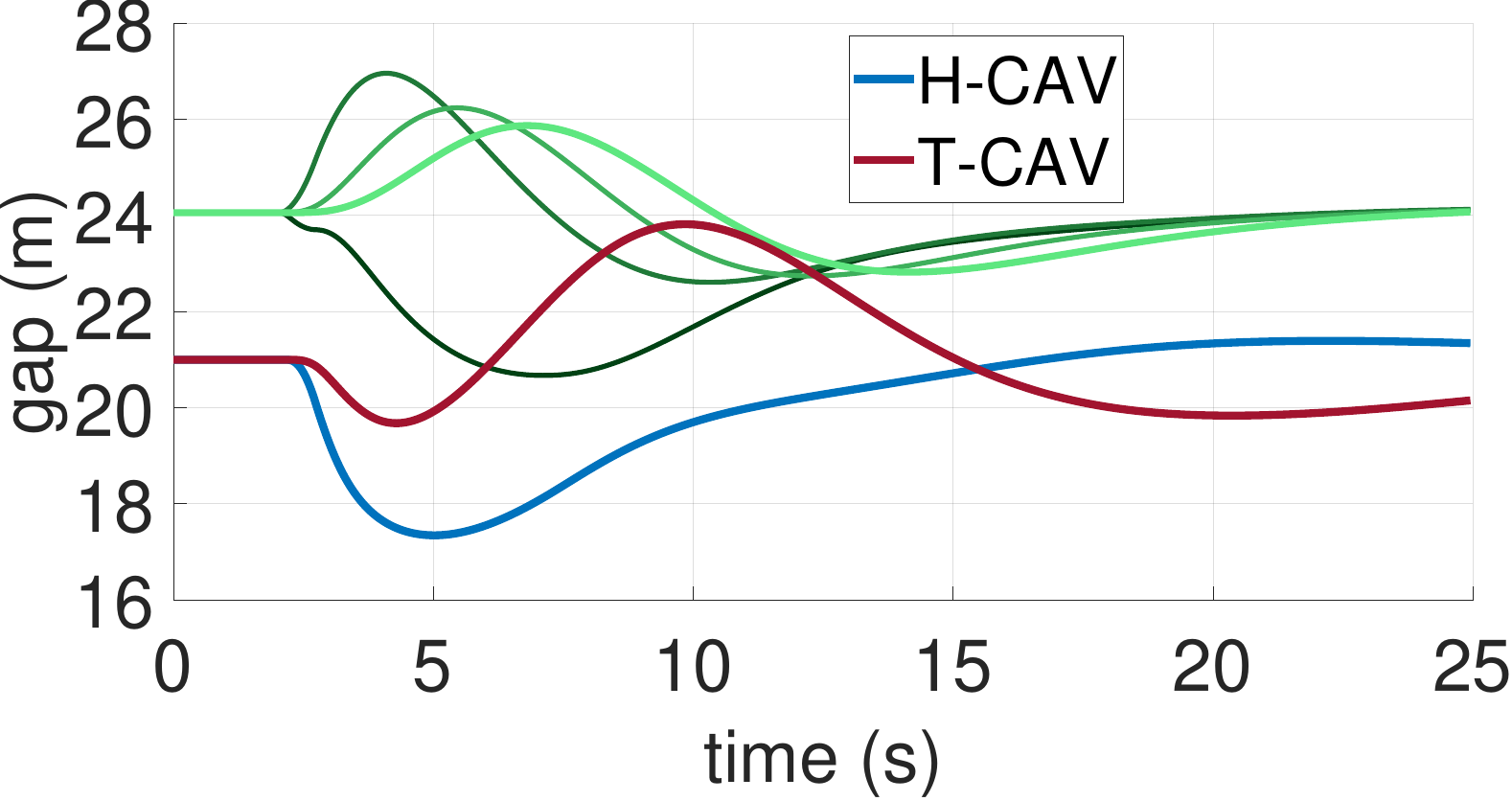}}
    \subfloat[Speed]{\includegraphics[width=0.25\linewidth]{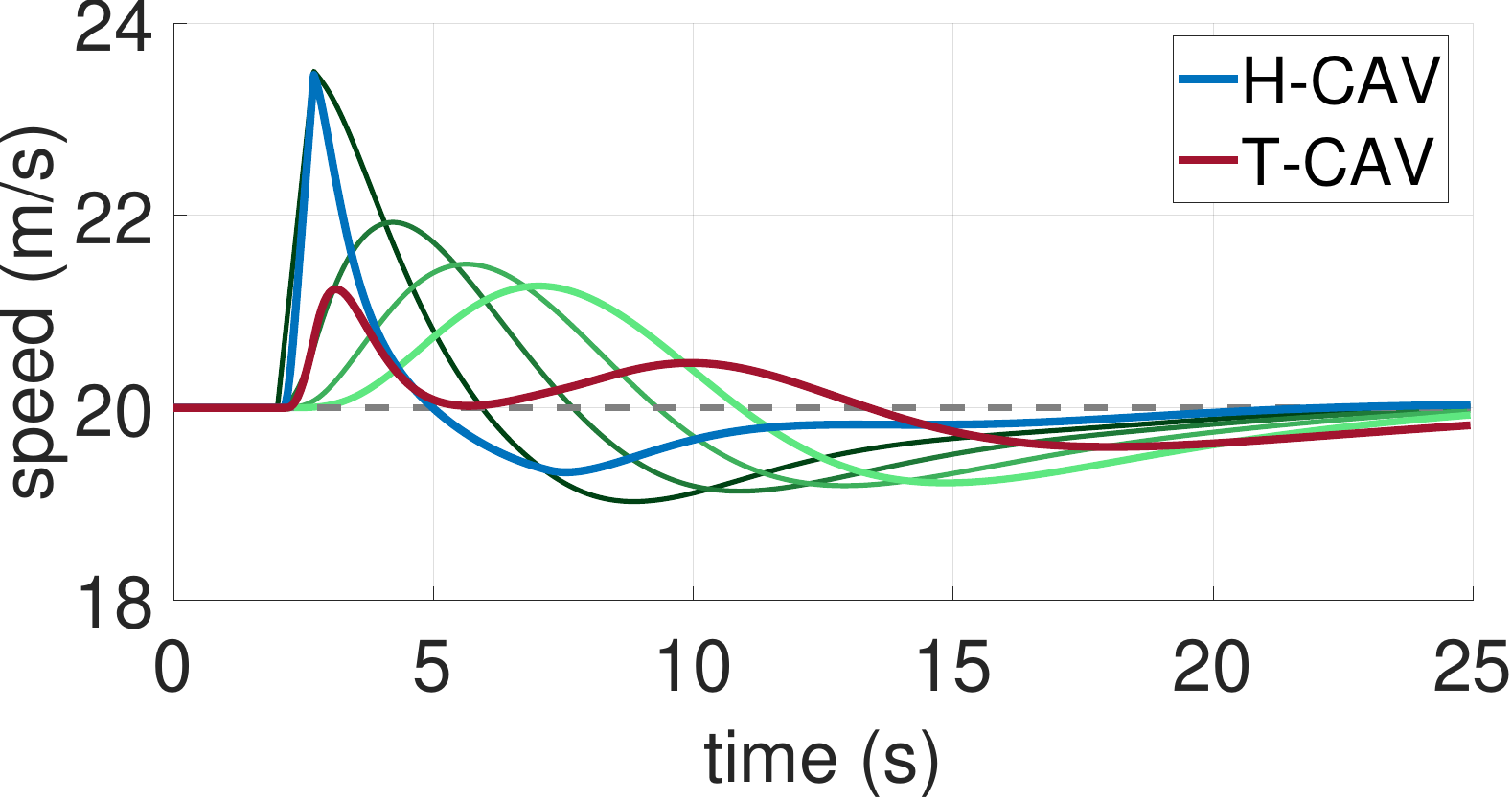}}
    \subfloat[Acceleration]{\includegraphics[width=0.25\linewidth]{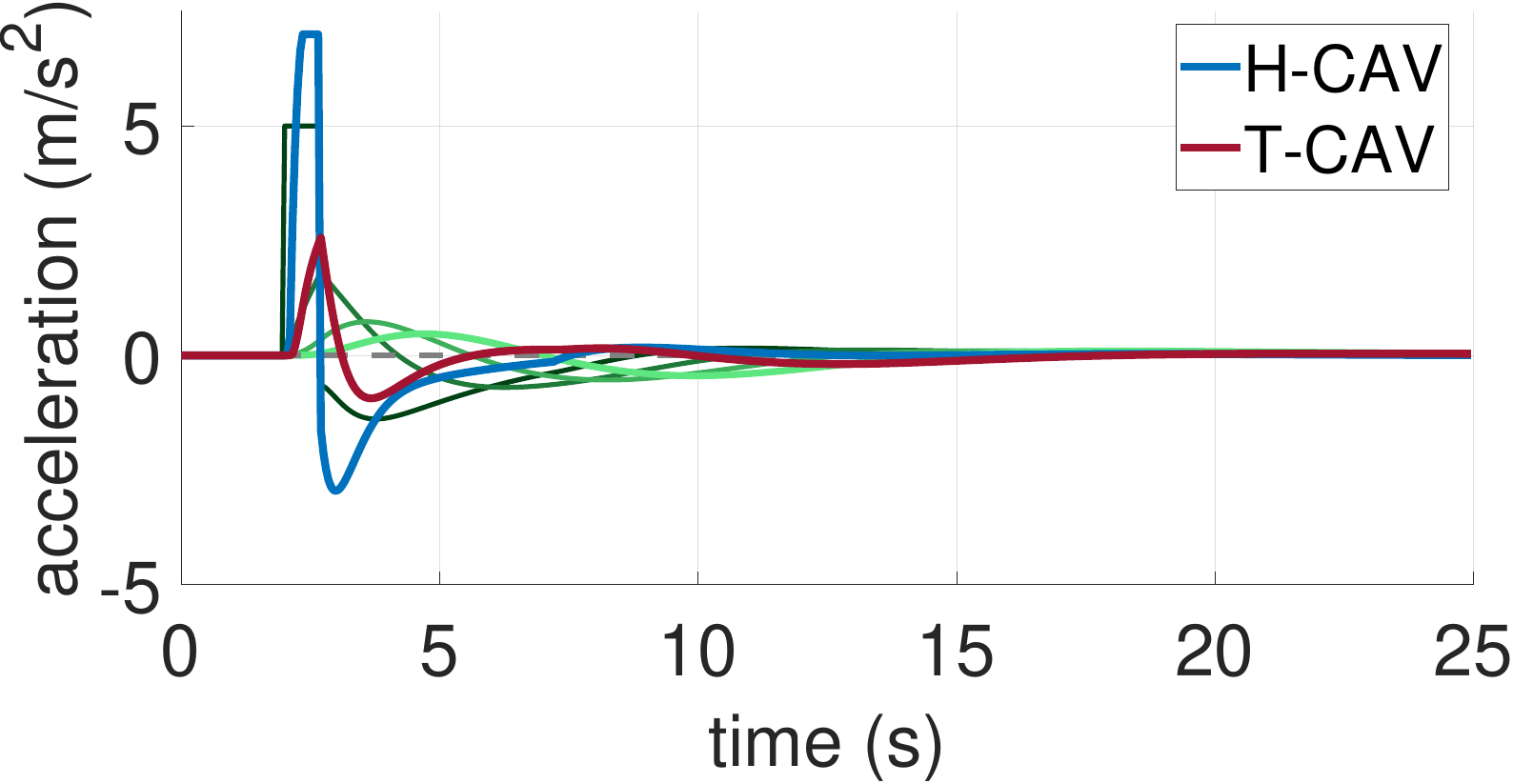}}
    \caption{{\it A middle HV-1 suddenly accelerates:} simulated trajectory of mixed vehicle platoon. The CBF guides the head CAV to accelerate so that collision between HV-1 and the head CAV is avoided.}
    \label{fig:trajectory HV acc}
\end{figure}
The simulation results are shown in Fig.~\ref{fig:analysis HV connection}, where the top row presents the tail CAV's speed, the bottom row shows the tail CAV's safety function, while the three columns correspond to the three simulation scenarios.
Note that the safety functions of the head CAV and middle HVs yield similar curves for each case, thus these plots are omitted.
We see that for all three scenarios looking behind has marginal effect on the results. At the same time, the tail CAV has a smoother speed by looking ahead, which reduces perturbations for the upstream traffic. However, the tail CAV also has a higher risk for unsafe behavior by looking ahead (i.e., the minimum of $h_{\tailcav}$ becomes smaller) compared to the case of no HV connection.
This can be remedied by using CBF-based safety filters, as it was demonstrated in Fig.~\ref{fig:trajectory head HV dec}.

\begin{figure}[t]
    \centering
    \subfloat[head HV decelerates, $v_{\tailcav}$]{\includegraphics[width=0.27\linewidth]{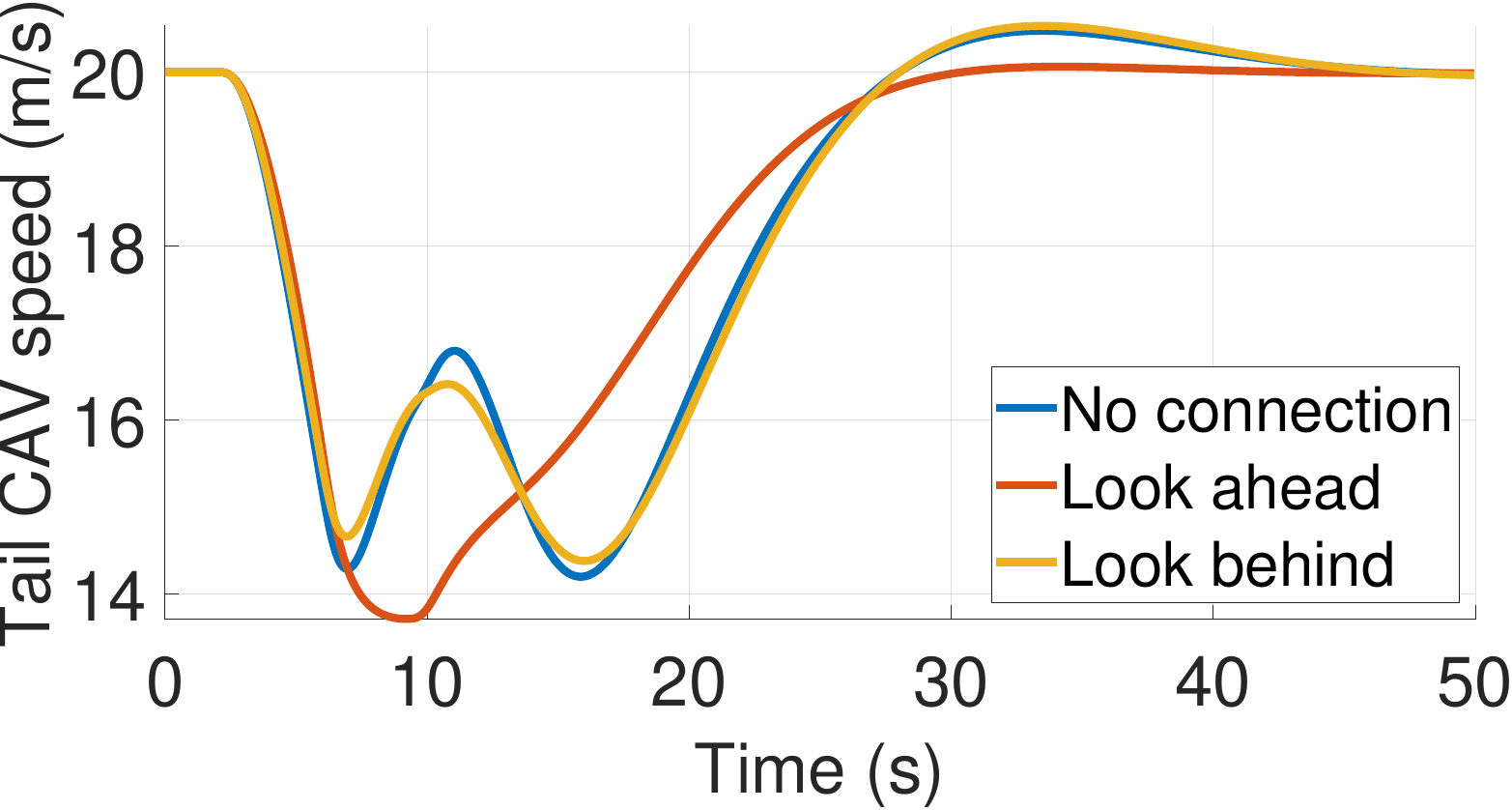}}
    \hspace{1.5 em}\subfloat[middle HV decelerates, $v_{\tailcav}$]{\includegraphics[width=0.27\linewidth]{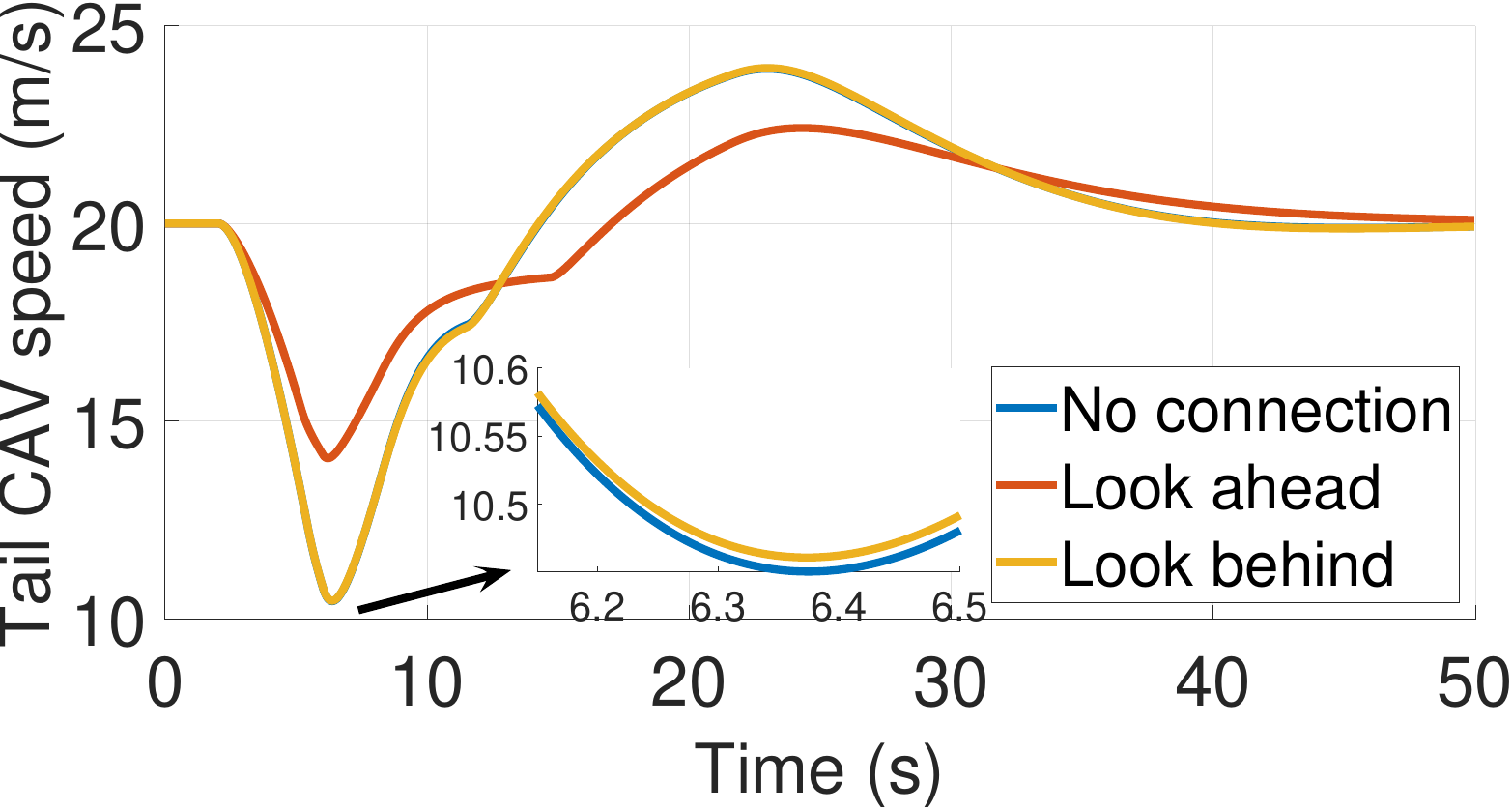}}
    \hspace{1.5 em}\subfloat[middle HV accelerates, $v_{\tailcav}$]{\includegraphics[width=0.27\linewidth]{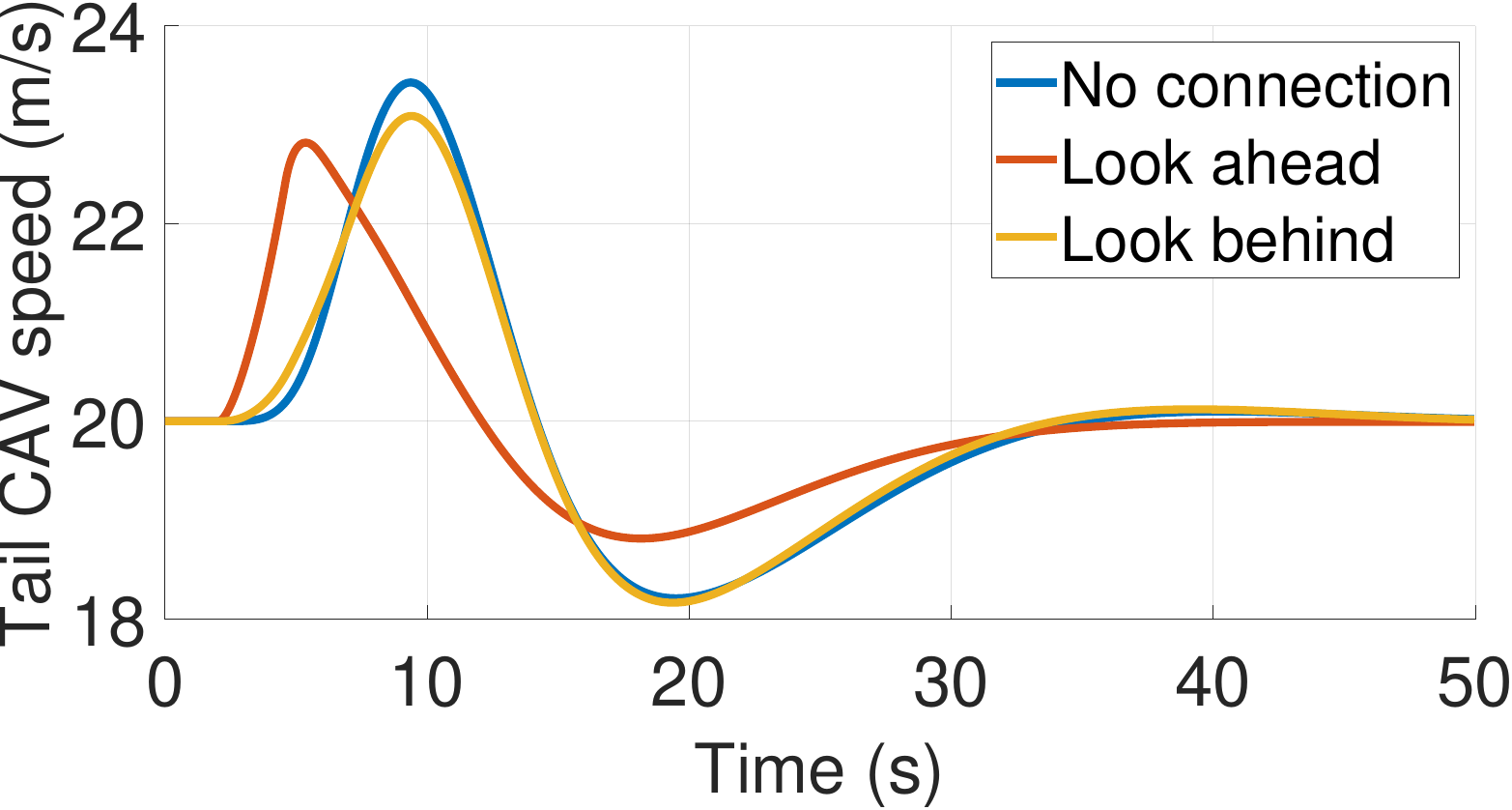}} 
    \\ \vspace{2ex}
    \subfloat[head HV decelerates, $h_\tailcav$]{\includegraphics[width=0.27\linewidth]{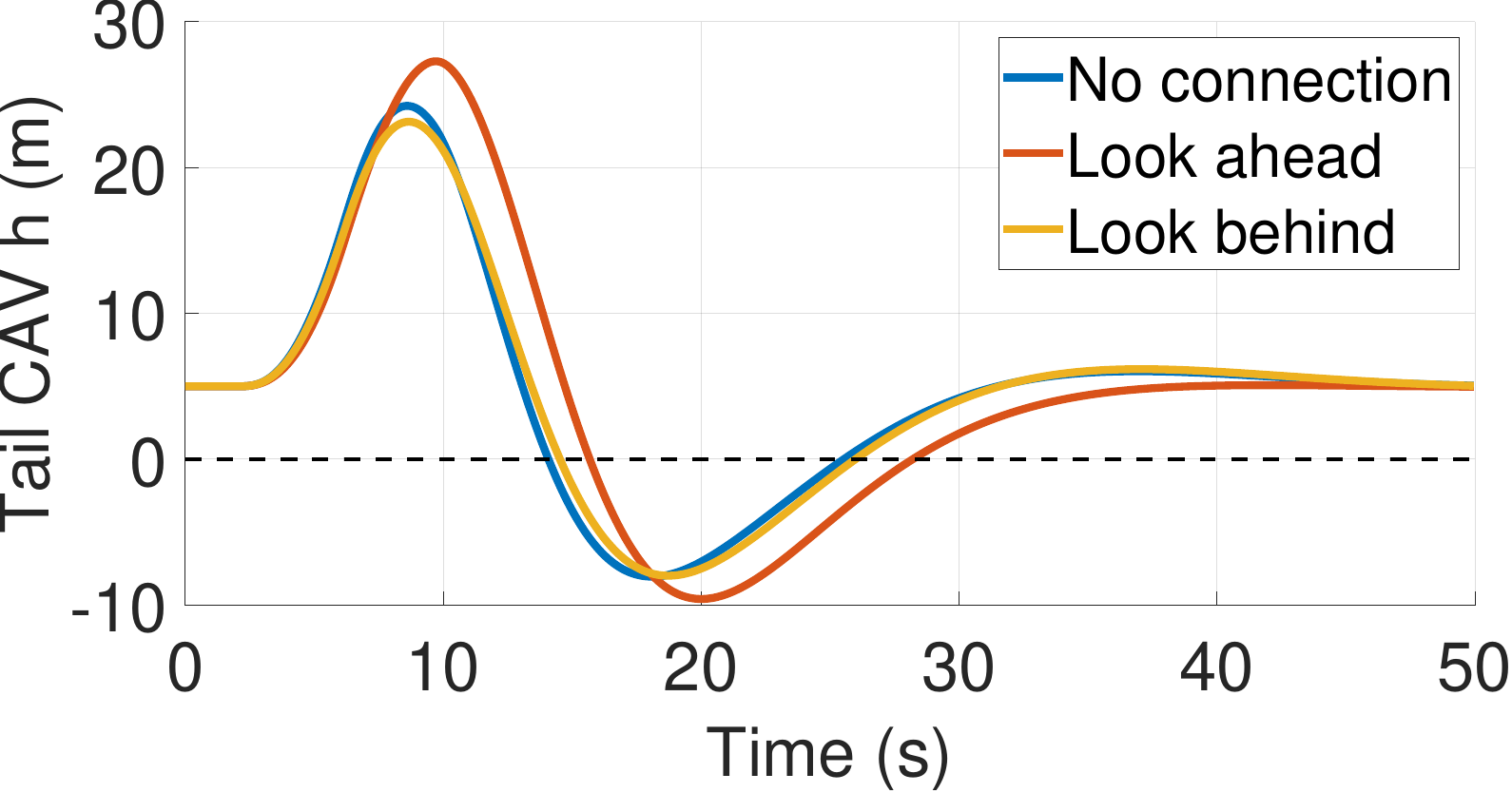}}
    \hspace{1.5 em}\subfloat[middle HV decelerates, $h_{\tailcav}$]{\includegraphics[width=0.27\linewidth]{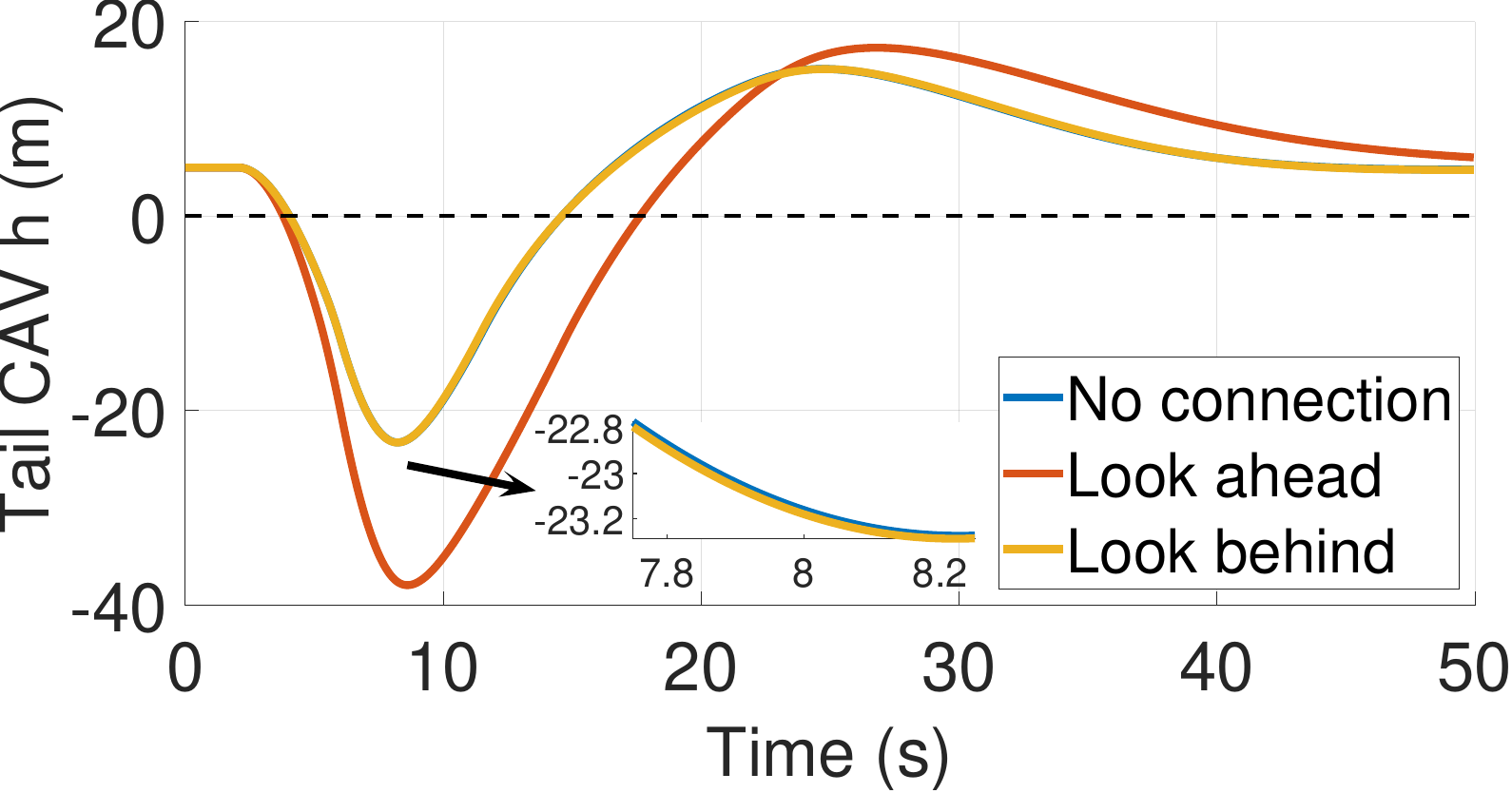}}
    \hspace{1.5 em}\subfloat[middle HV accelerates, $h_1$]{\includegraphics[width=0.27\linewidth]{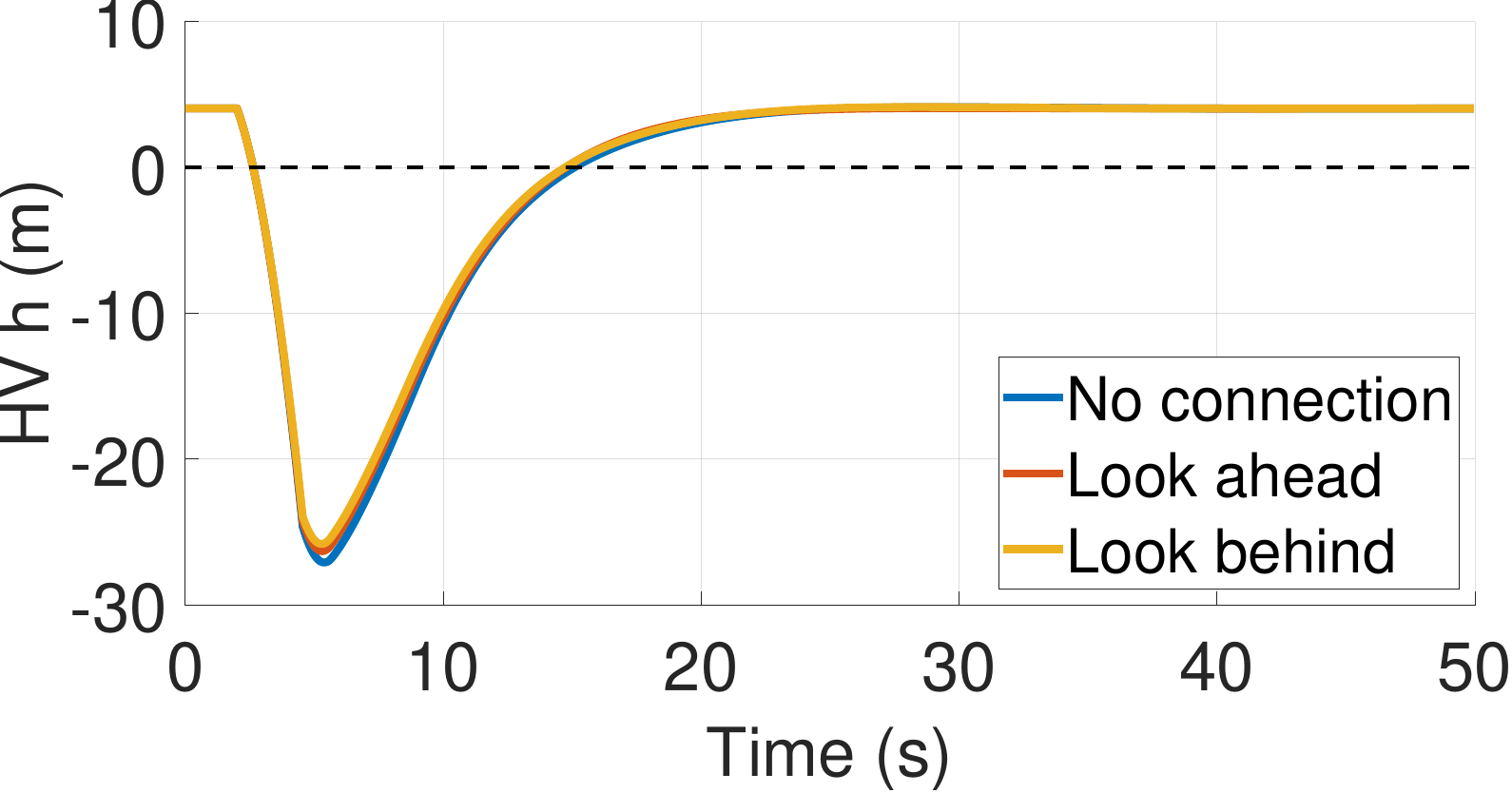}}
    \caption{Simulations of the nominal controller demonstrating the effect of HV connection on stability (first row) and safety (second row).
    When the tail CAV connects to HVs (look ahead), it reduces its speed fluctuations but also hinders its safety compared to the case of no connection.
    When the head CAV connects to HVs (look behind), it has marginal effect.}
    \label{fig:analysis HV connection}
\end{figure}

\begin{figure}[t]
    \centering
    \subfloat[Safety function $h$]{\includegraphics[width=0.24\linewidth]{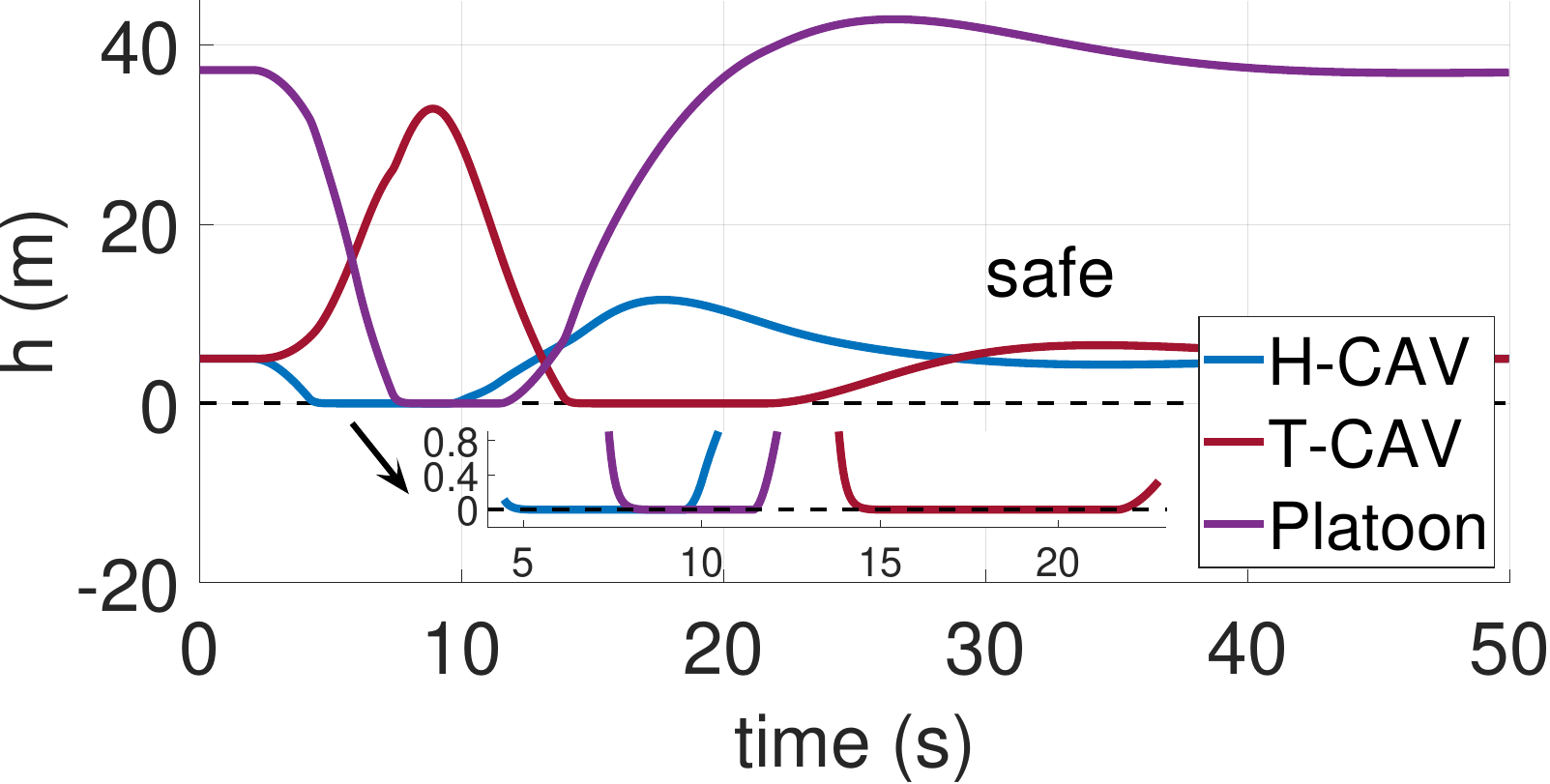}}
    \subfloat[Gap]{\includegraphics[width=0.24\linewidth]{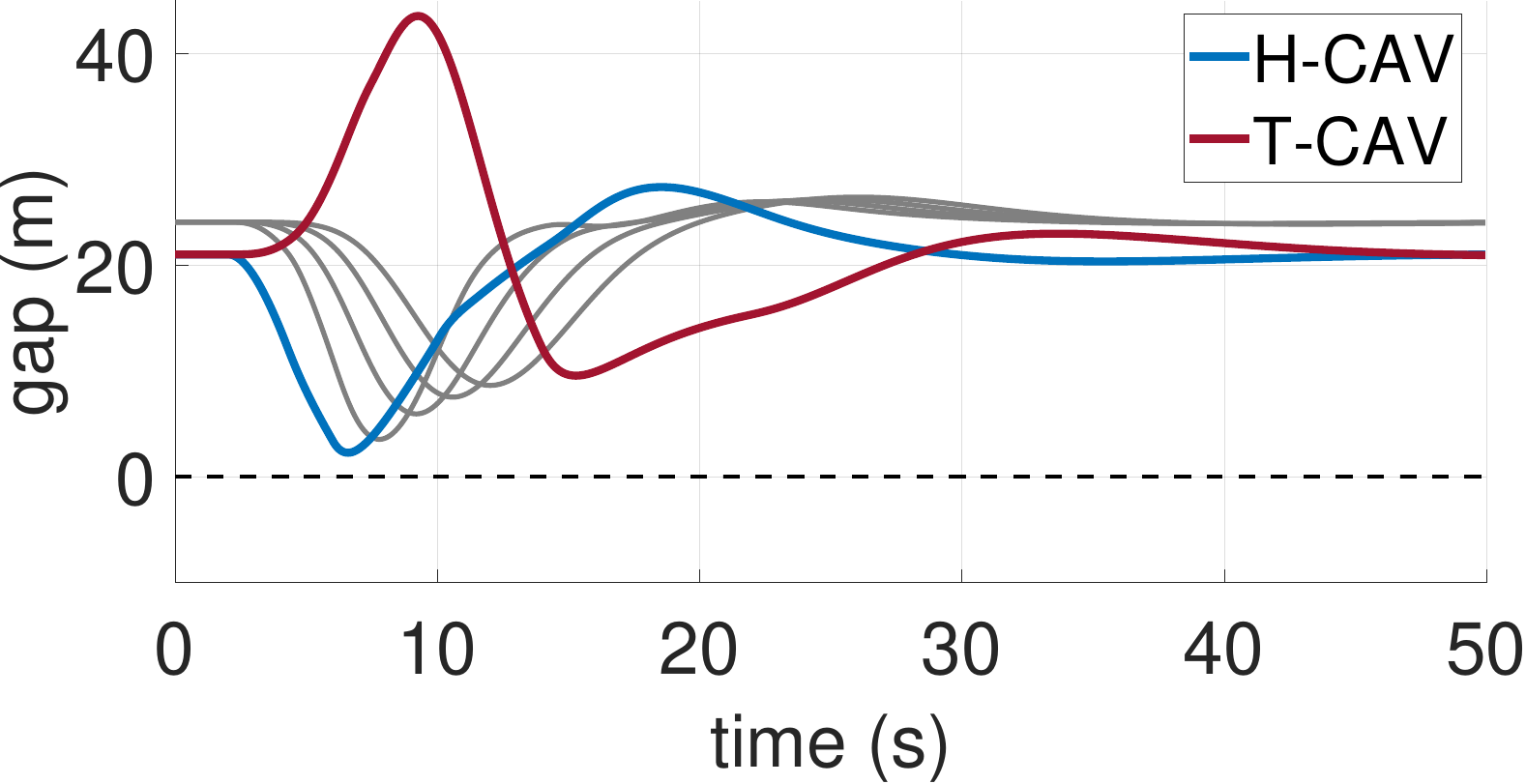}}
    \subfloat[Speed]{\includegraphics[width=0.24\linewidth]{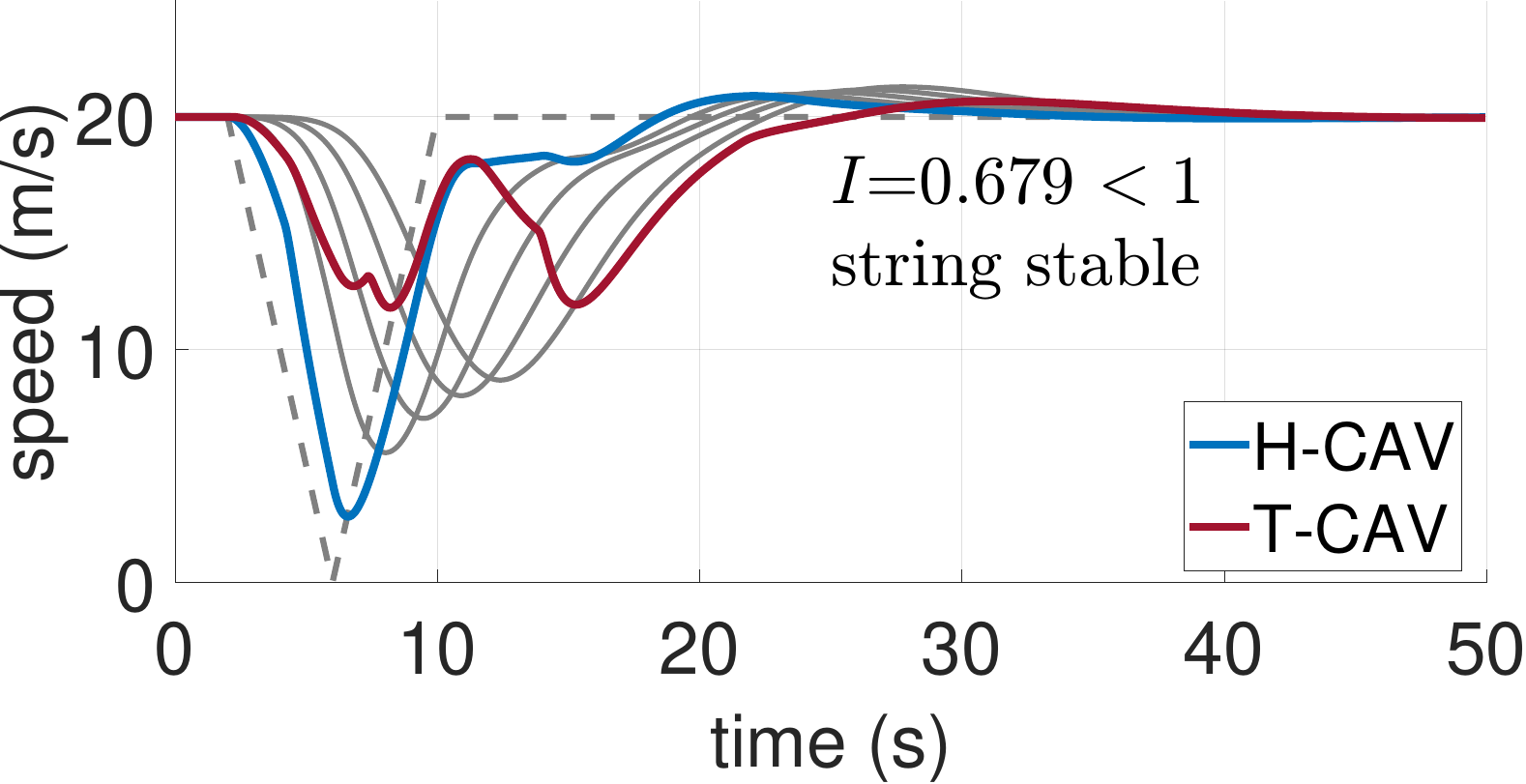}}
    \subfloat[Acceleration]{\includegraphics[width=0.24\linewidth]{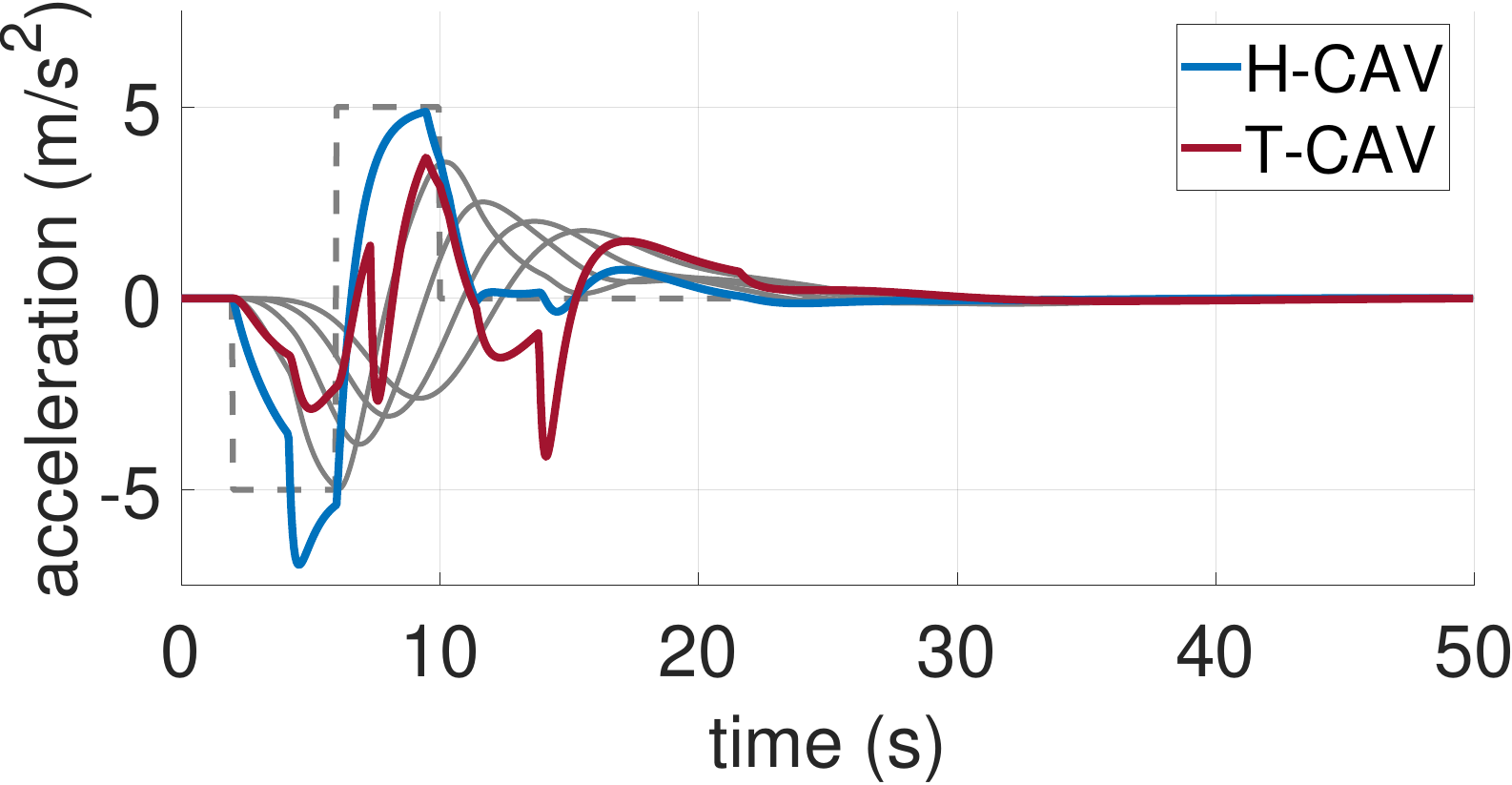}}
    \caption{Simulated trajectories when the head HV suddenly decelerates and platoon safety is enforced by the two CAVs.
    Enforcing platoon safety helps reduce the deceleration of the tail CAV compared to the case without platoon safety (Fig.~\ref{fig:trajectory head HV dec}).}
    \label{fig:trajectory platoon}
\end{figure}

\subsection{Platoon safety}\label{sec:subsec:simulation platoon}

Finally, we investigate platoon safety, which requires that the gap between the two CAVs is greater than a safe minimum value. While platoon safety does not refer to collision-based safety for a single vehicle, it prevents the gap between the two CAVs from becoming too small and increases the safety of the overall platoon.  On the other side, this also enhances the stability of traffic flow. To demonstrate this, we run simulations with the safety-critical controller~\eqref{eq:QP} that enforces platoon safety (while we omit the HV safety constraints).
We consider the scenario of Fig.~\ref{fig:trajectory head HV dec} where the head HV decelerates.
We set the vehicle length as $l_i = 5$ m, and the regular length of the vehicle platoon as $l_0 = 100$ m. We use $\tau_{\platoon} = 1$ s and $\gamma_{\platoon} = 5$ s$^{-1}$ to enforce platoon safety. All other parameters remain the same as in Fig.~\ref{fig:trajectory head HV dec}.

Fig.~\ref{fig:trajectory platoon} plots the simulated trajectories. Compared with the trajectories in Fig.~\ref{fig:trajectory head HV dec}, we see that the string stability index in Fig.~\ref{fig:trajectory platoon} reduces from 0.698 to 0.679, indicating that the tail CAV has smaller speed perturbations. Furthermore, the maximum deceleration of the tail CAV is $-5\;\mathrm{m/s^2}$ when platoon safety is not enforced (see Fig.~\ref{fig:trajectory head HV dec}(h)), while it is only $-4\;\mathrm{m/s^2}$ by enforcing platoon safety (see Fig.~\ref{fig:trajectory platoon}(d)). This is because the tail CAV begins to decelerate earlier when it intends to maintain platoon safety, which results in milder deceleration and smoother motion.

\section{Performance analysis}\label{sec:performance analysis}

The main simulation results have validated that the designed cooperative CAV controllers achieve both stability and safety and, therefore, provides an improvement over the nominal control design. We further analyze the performance of the controllers for a wide range of parameters in this section and discuss the sensitivity and robustness of our design with respect to uncertain human driver behaviors.

\begin{figure}[t]
    \centering
    Head CAV safety \vspace{0.5ex}\\
    \subfloat[Nominal Controller]{\includegraphics[width=0.24\linewidth]{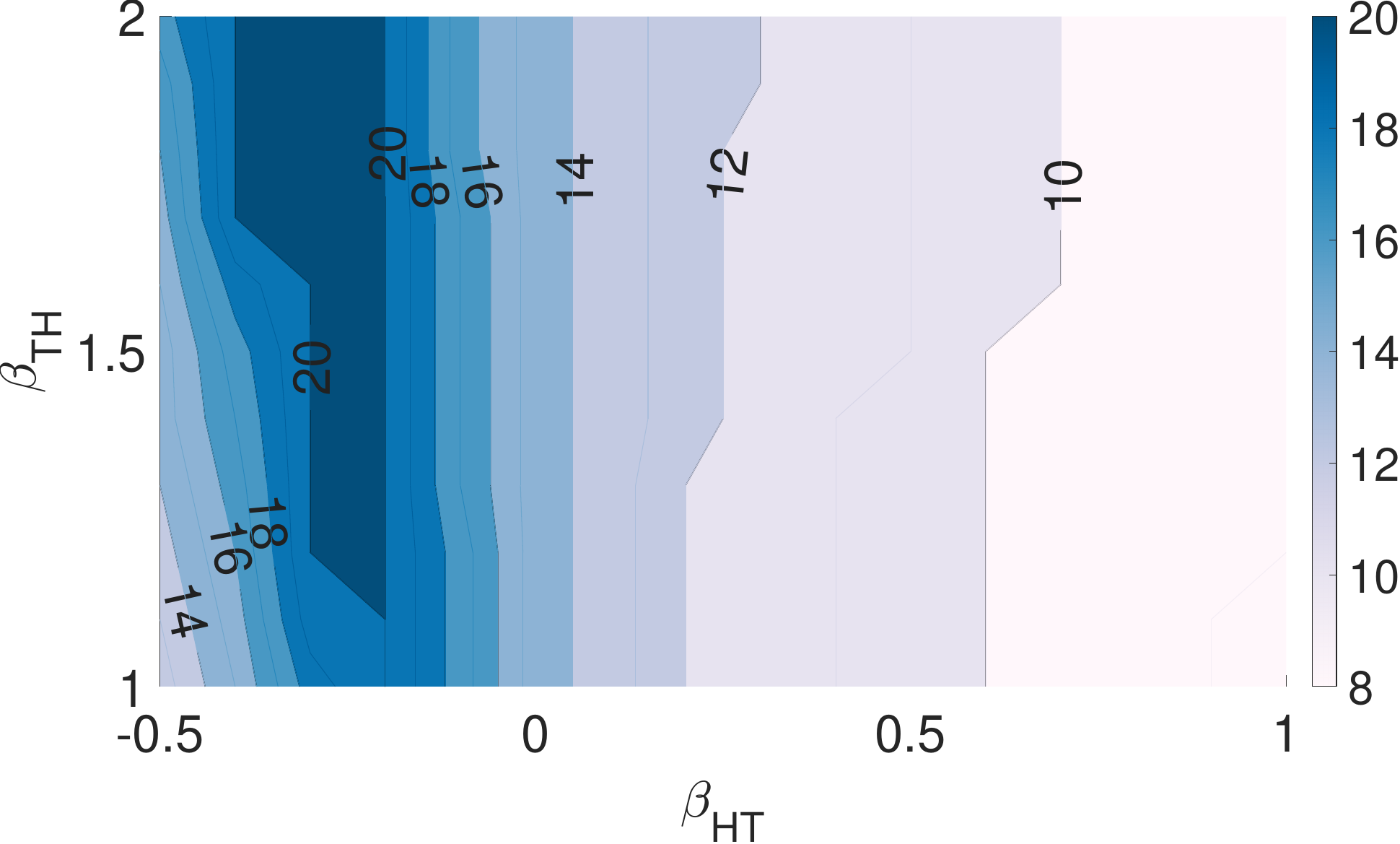}}
    \hspace{0.3em}
    \subfloat[CBF]{\includegraphics[width=0.24\linewidth]{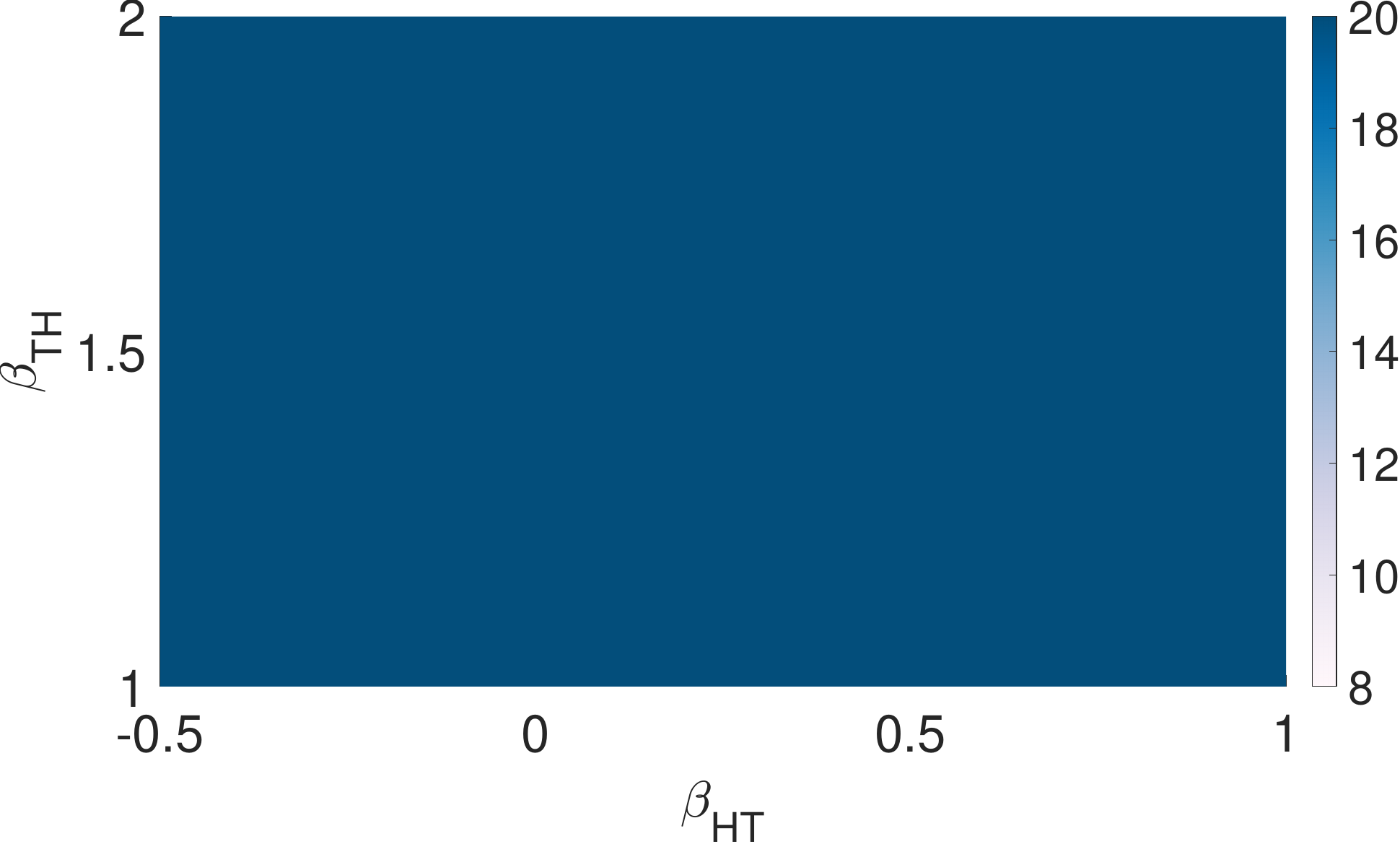}}
    \hspace{0.3em}
    \subfloat[Compare at $\Delta v_{\hhv} = 12$ m/s]{\includegraphics[width=0.24\linewidth]{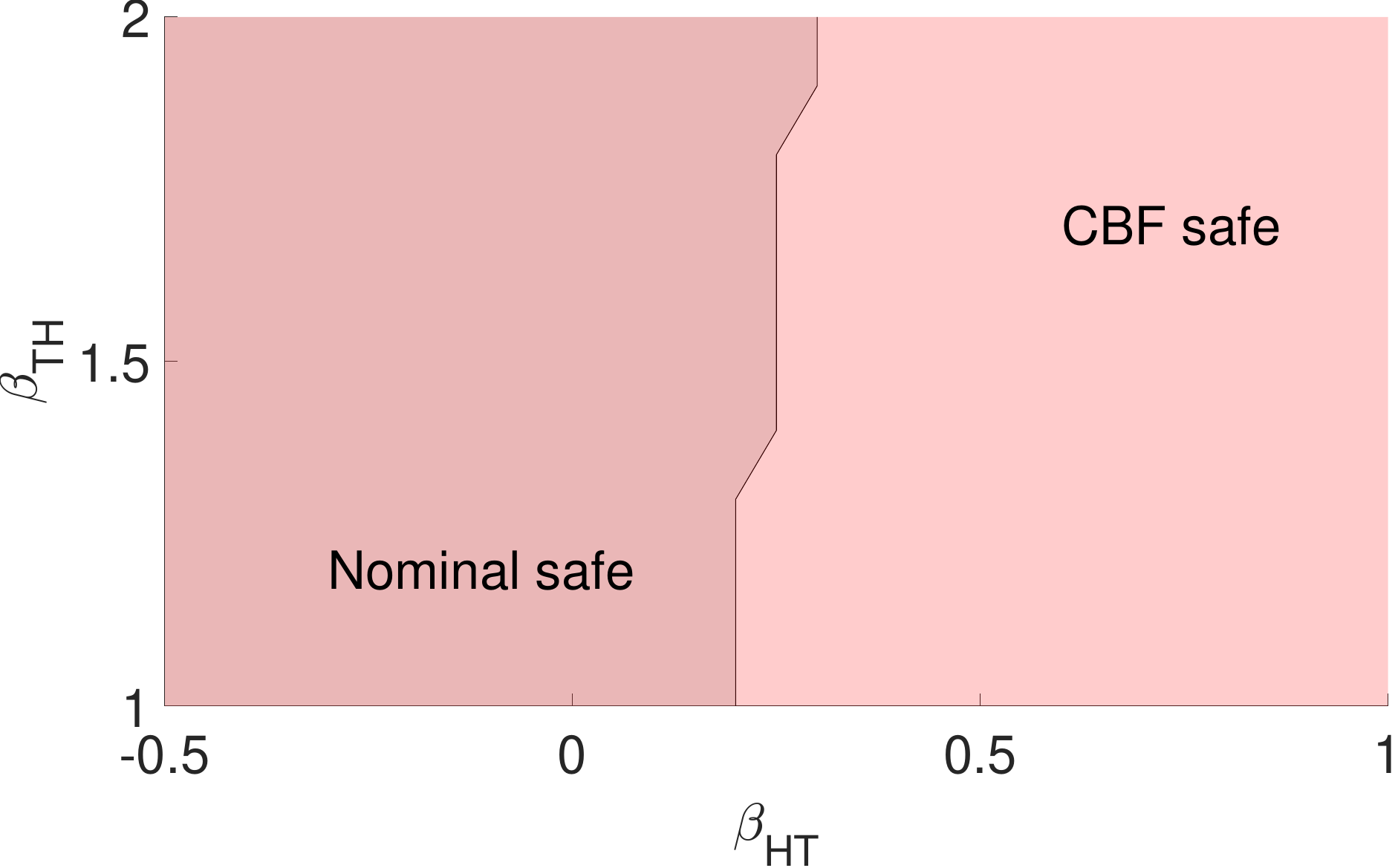}}
    \hspace{0.3em}
    \subfloat[Compare at $\Delta v_{\hhv} = 20$ m/s]{\includegraphics[width=0.24\linewidth]{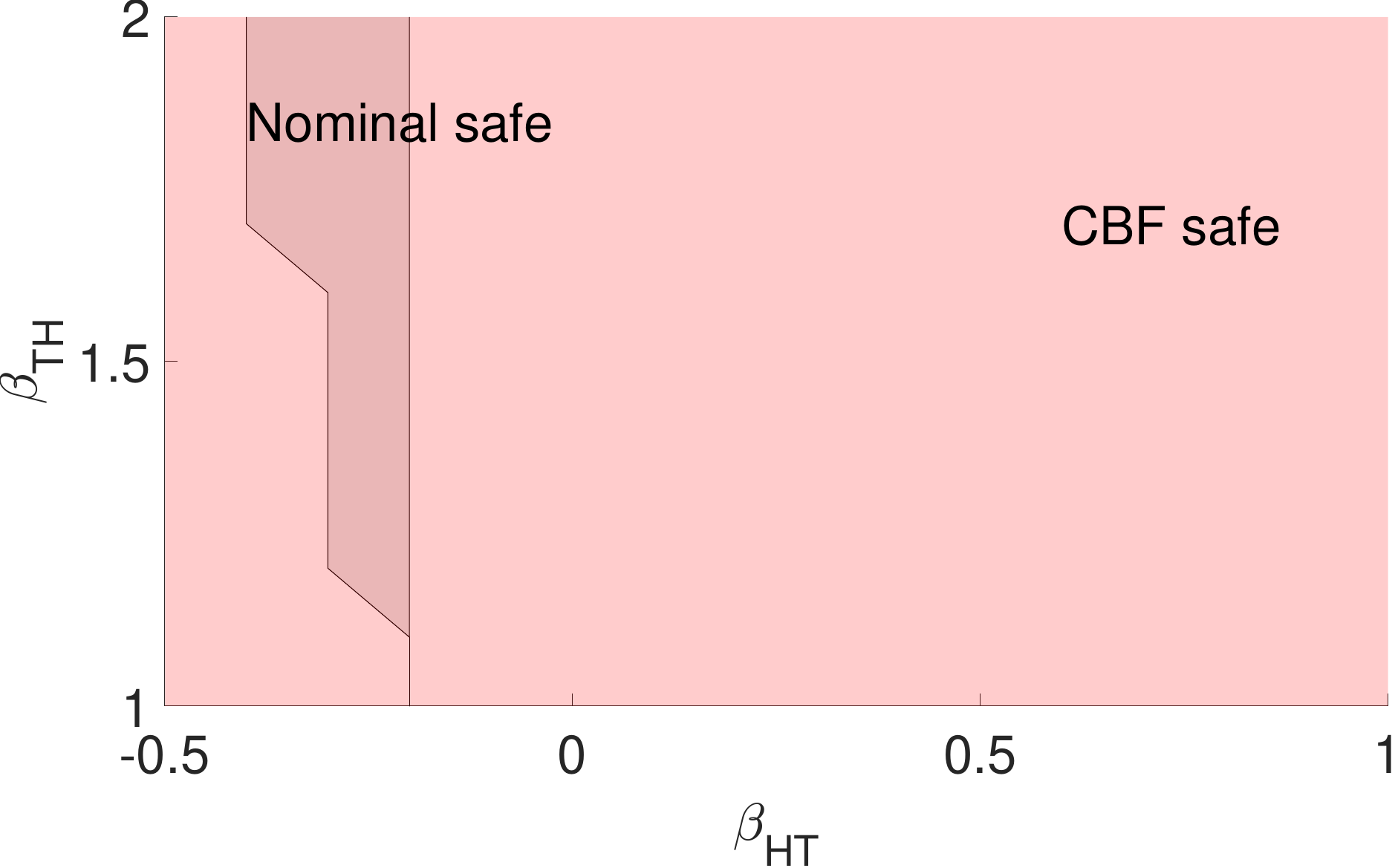}}
    \\
    \vspace{2ex} Tail CAV safety \vspace{0.5ex}\\
    \subfloat[Nominal Controller]{\includegraphics[width=0.24\linewidth]{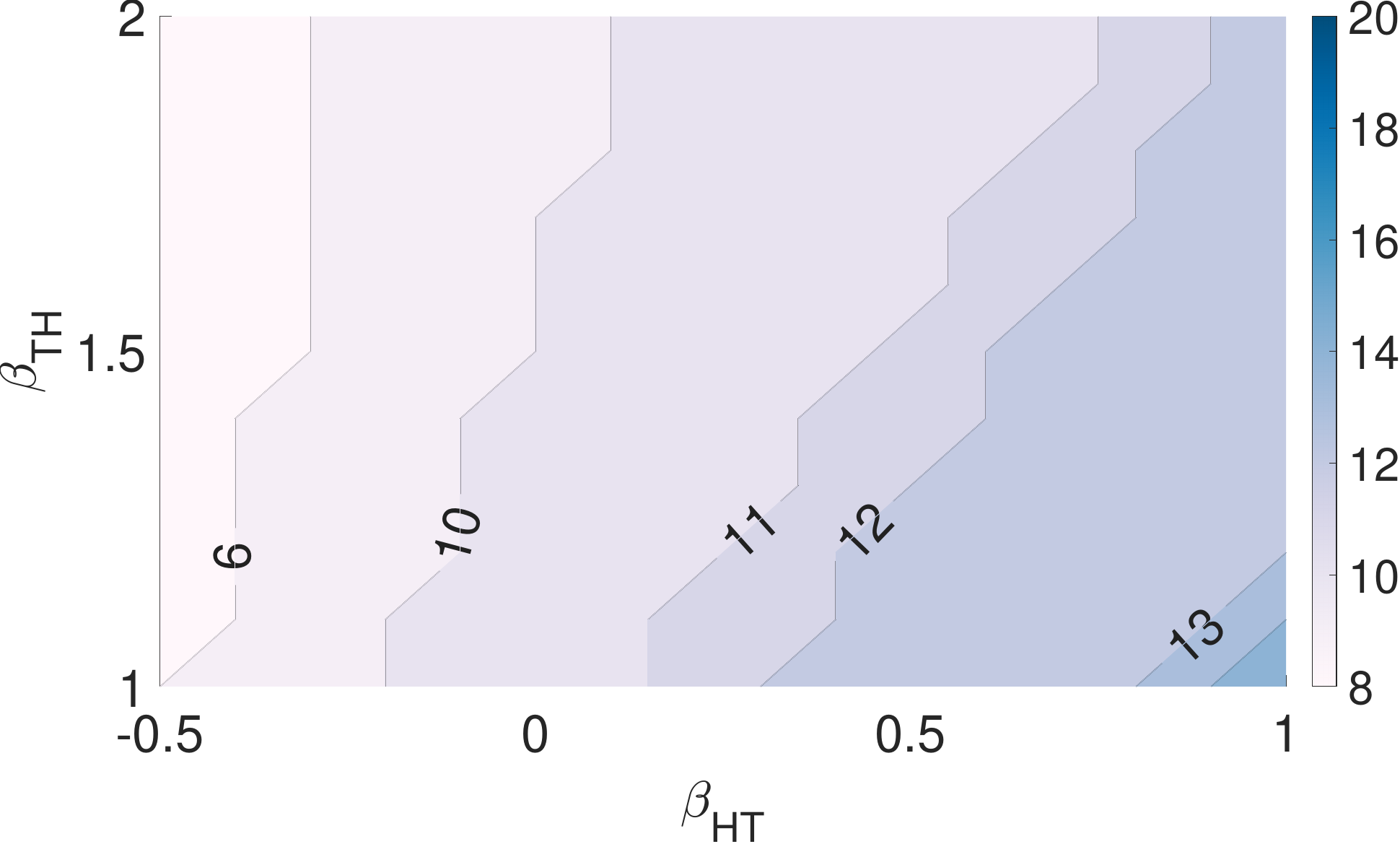}}
    \hspace{0.3em}
    \subfloat[CBF]{\includegraphics[width=0.24\linewidth]{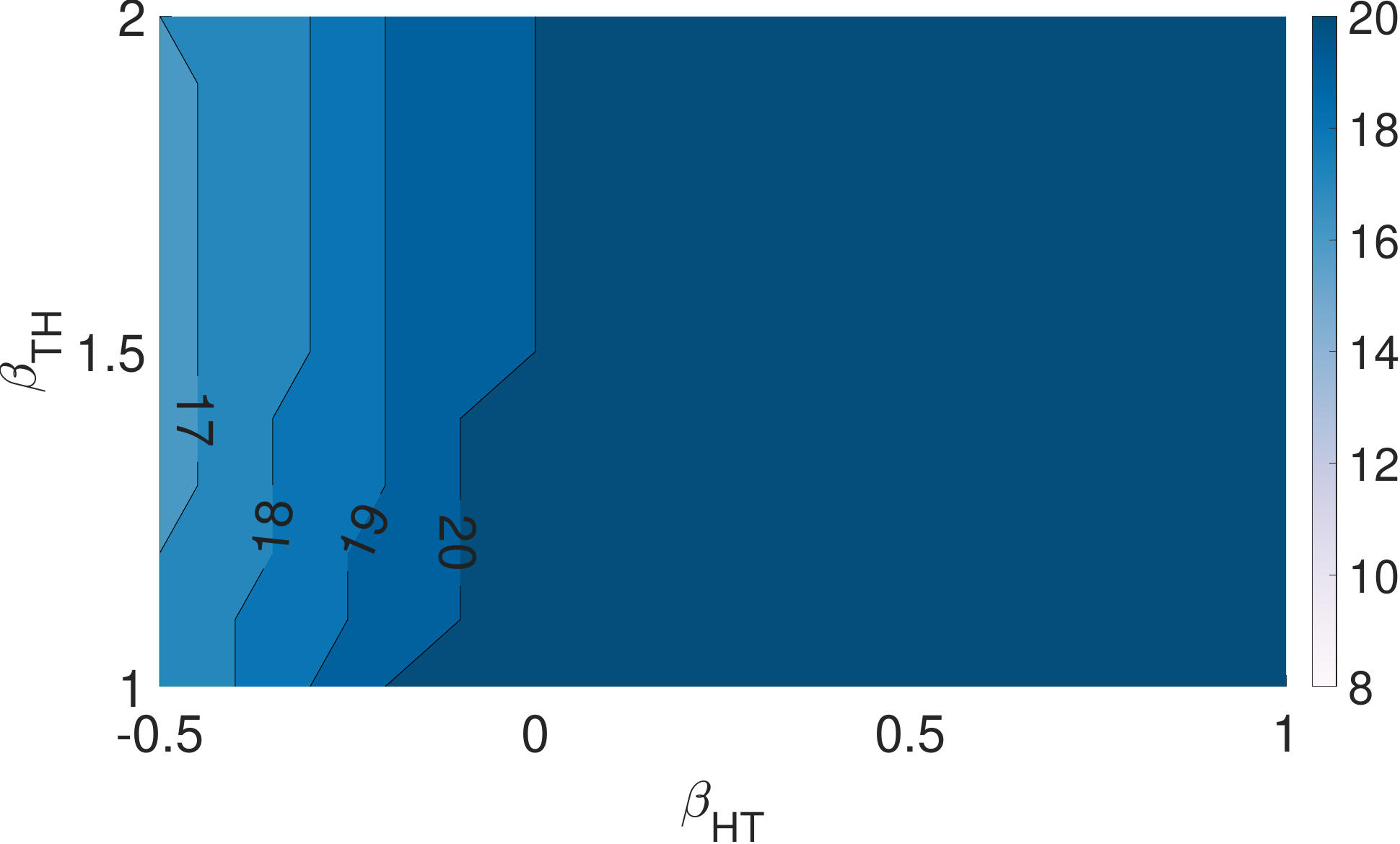}}
    \hspace{0.3em}
    \subfloat[Compare at $\Delta v_{\hhv} = 12$ m/s]{\includegraphics[width=0.24\linewidth]{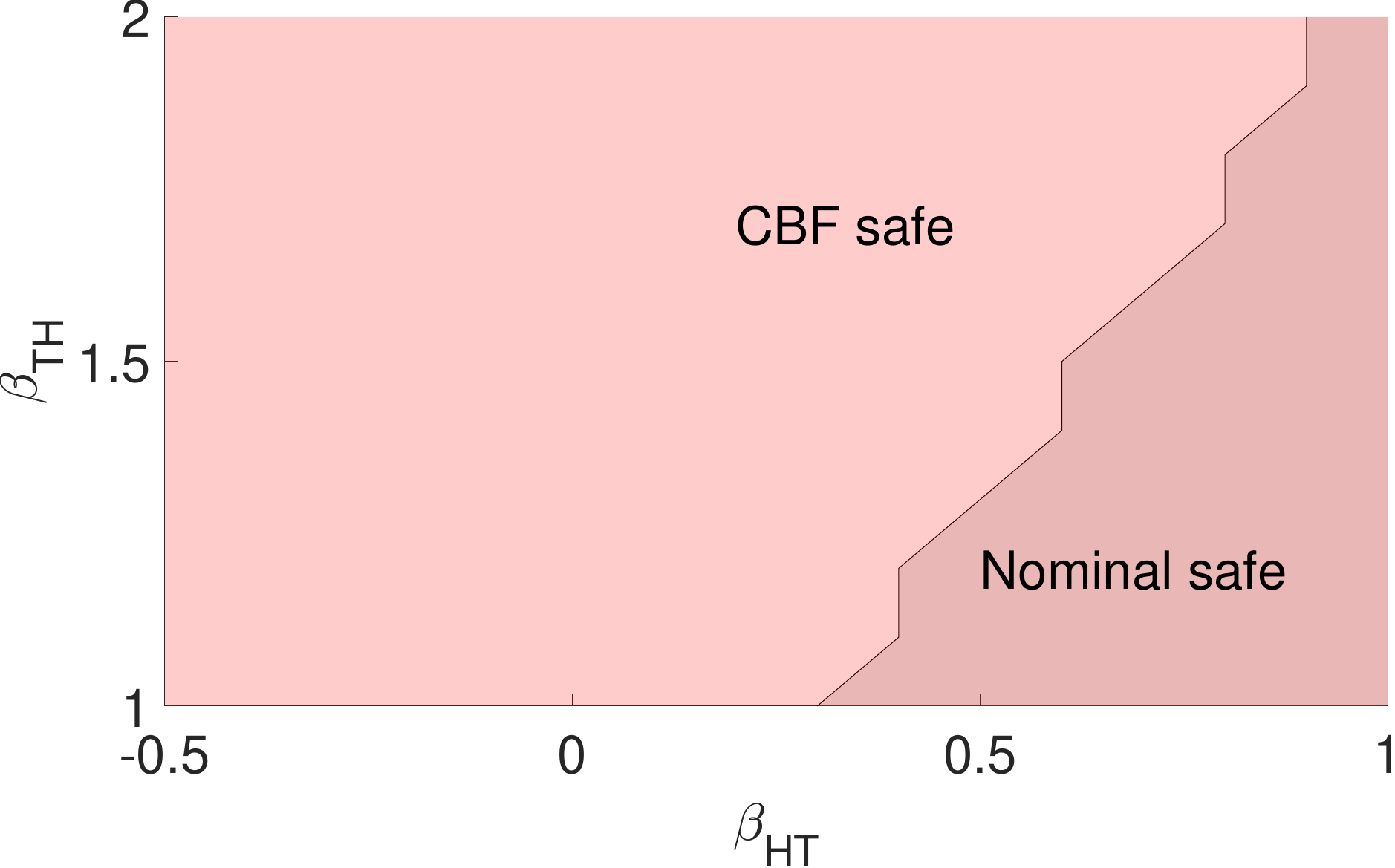}}
    \hspace{0.3em}
    \subfloat[Compare at $\Delta v_{\hhv} = 20$ m/s]{\includegraphics[width=0.24\linewidth]{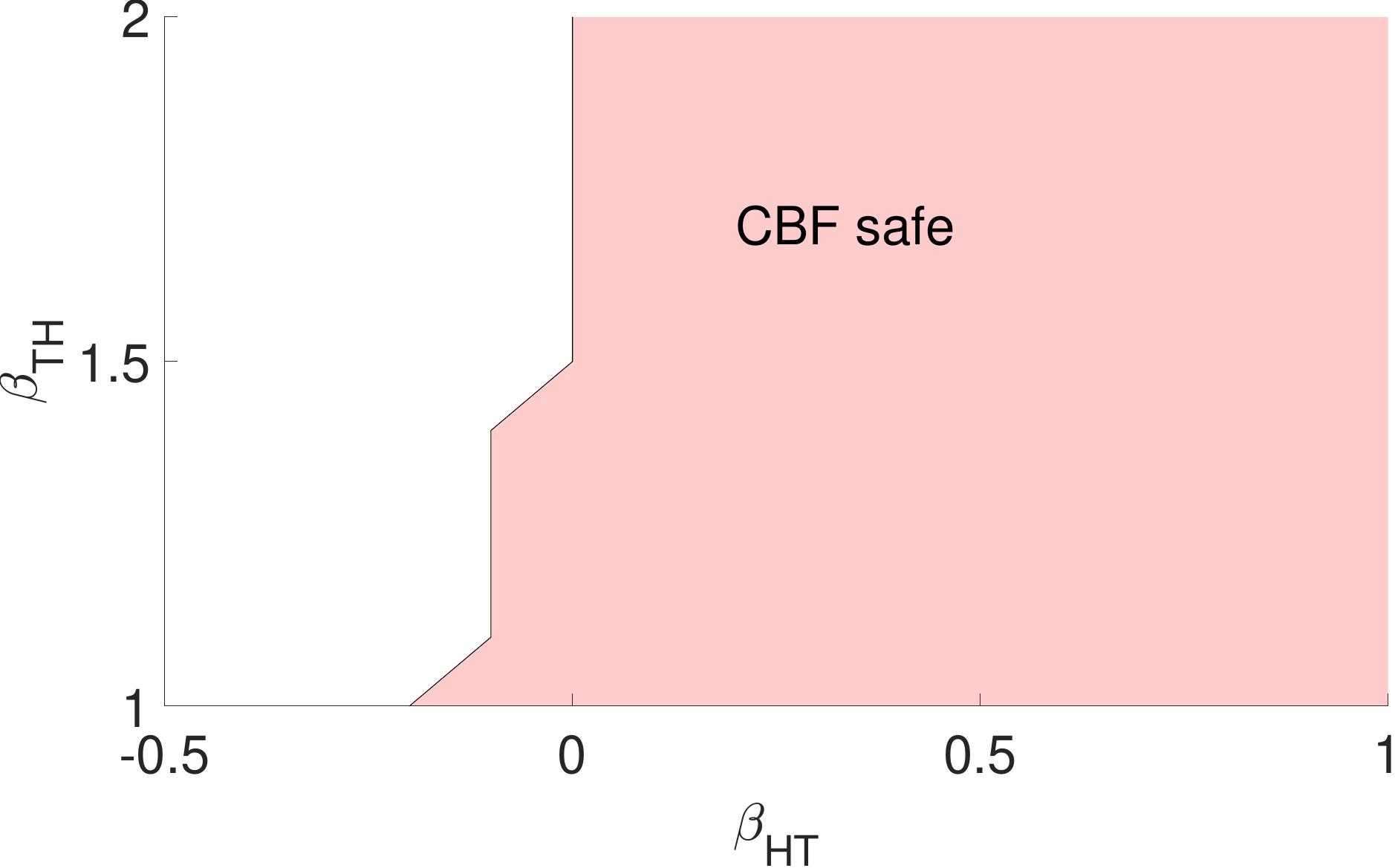}}
    \caption{Safety performance of the nominal controller and the safety-critical controller with CBF. The first and second row gives results for the head and tail CAV, respectively. The first and second columns show the maximum speed perturbation $\Delta v_\hhv$ for the head HV so that the CAV remains safe by using the nominal controller and the CBF, respectively, considering various controller gains $(\beta_{\headcav,\tailcav},\beta_{\tailcav,\headcav})$. The third and fourth columns provide the range of controller gains under which the CAV is safe for a fixed speed perturbation $\Delta v_\hhv$. The grey and red area gives the safe range for the nominal controller and CBF, respectively.}
    \label{fig:safe region} 
\end{figure}

\begin{figure}[t]
    \centering
    \subfloat[Nominal controller, $I$]{\includegraphics[width=0.24\linewidth]{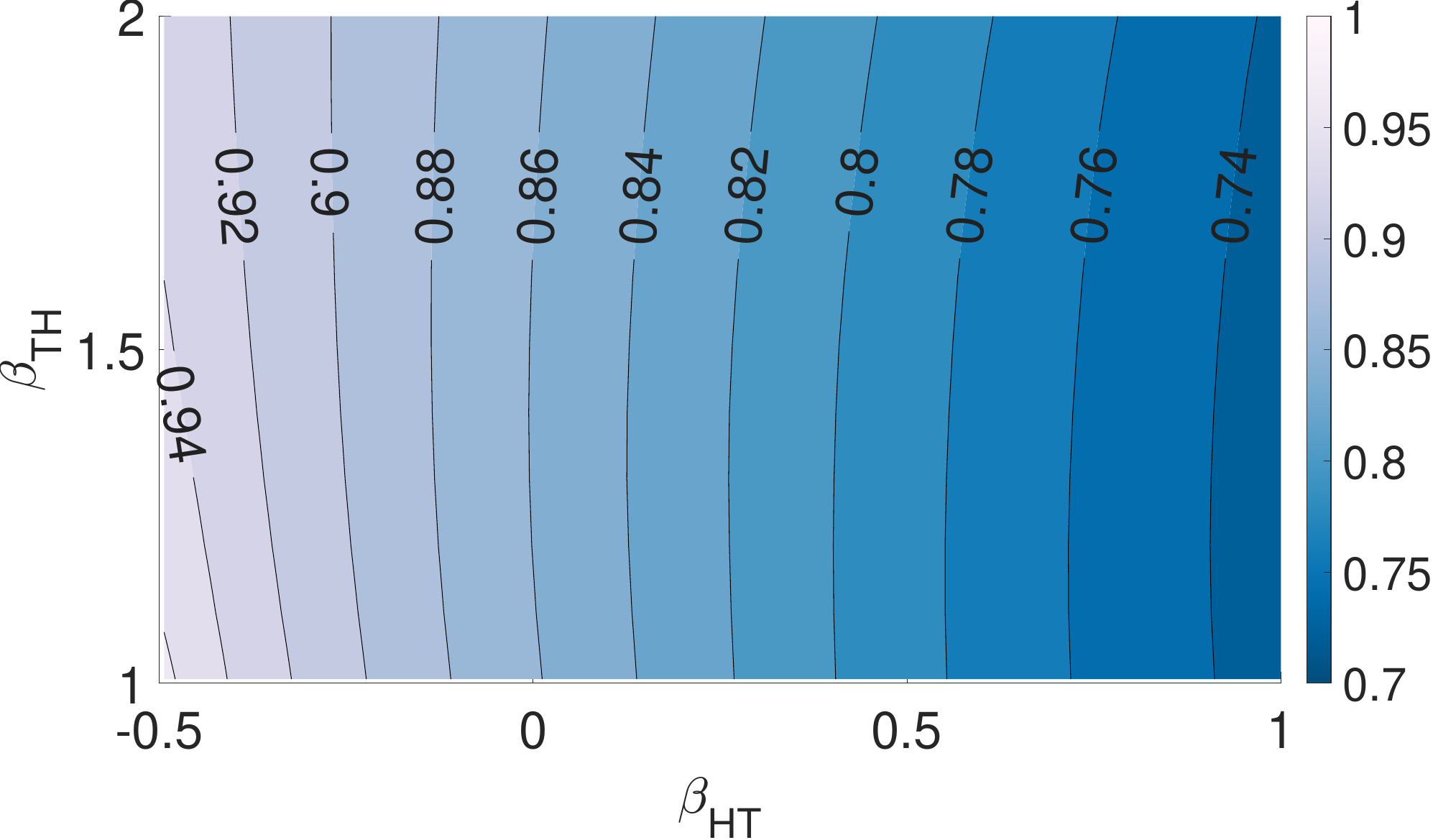}}
    \hspace{0.3em} \subfloat[CBF, $I$]{\includegraphics[width=0.24\linewidth]{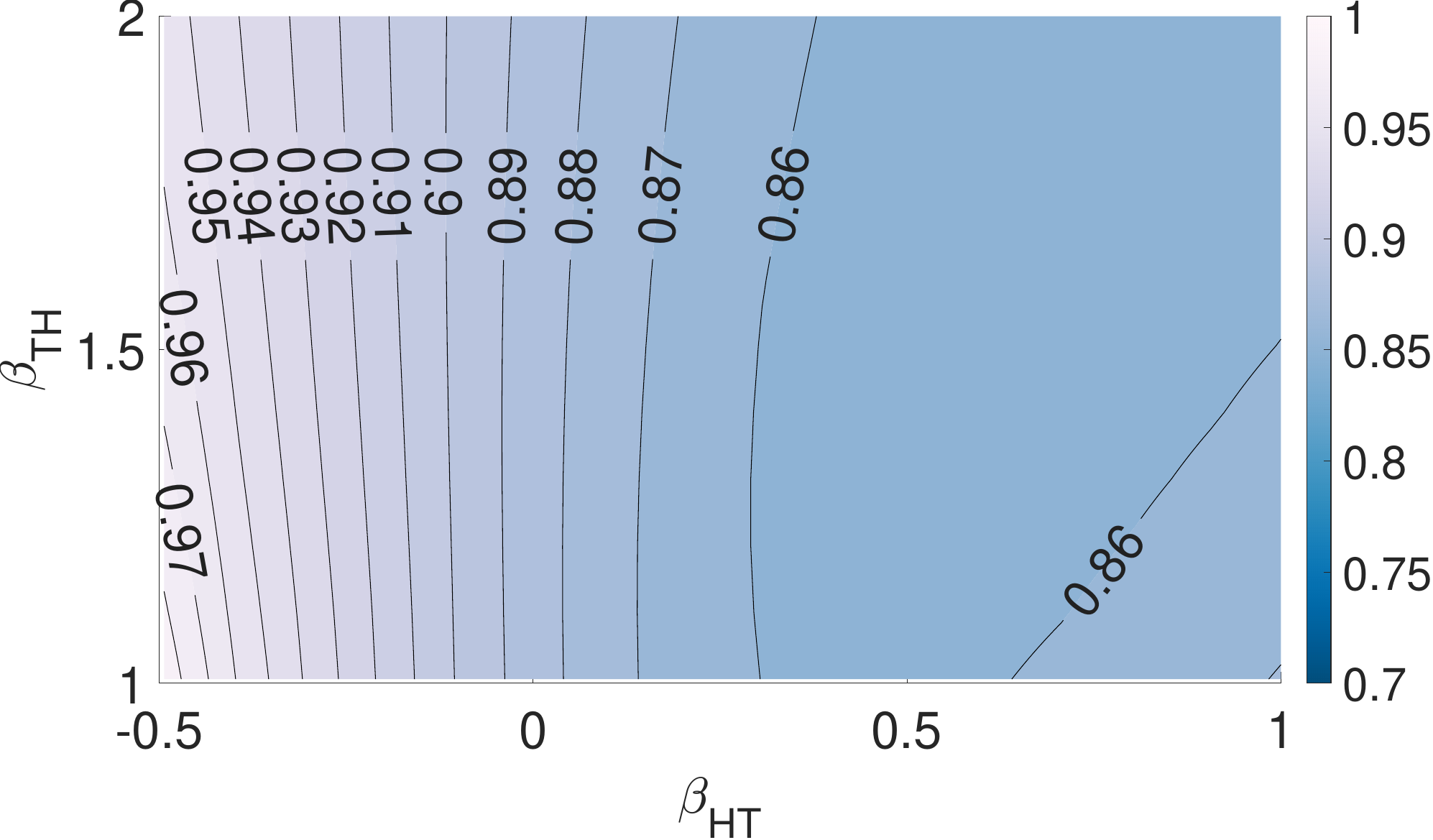}}
    \hspace{0.3em} \subfloat[Nominal controller, $\bar{I}$]{\includegraphics[width=0.24\linewidth]{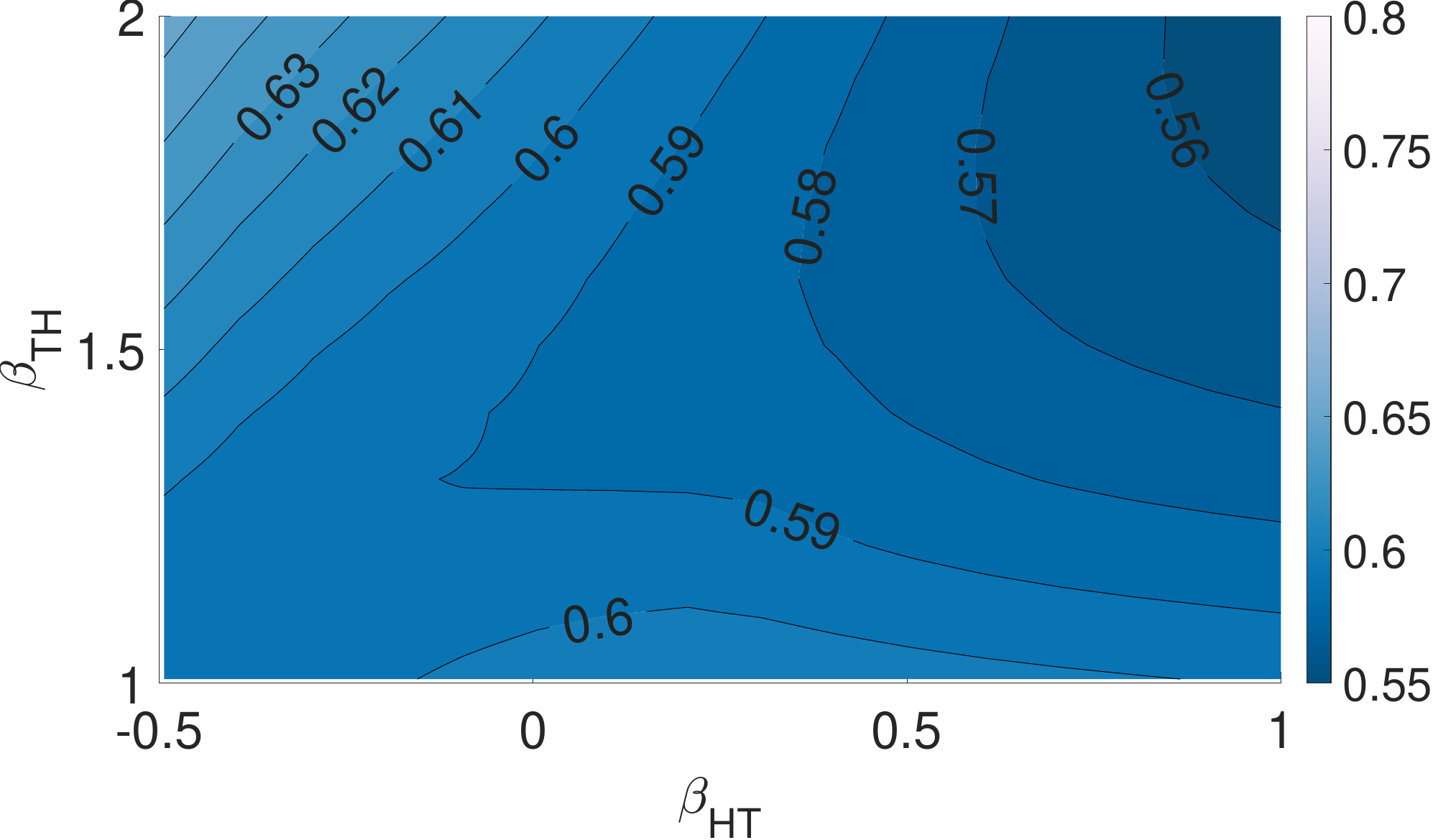}}
    \hspace{0.3em} \subfloat[CBF, $\bar{I}$]{\includegraphics[width=0.24\linewidth]{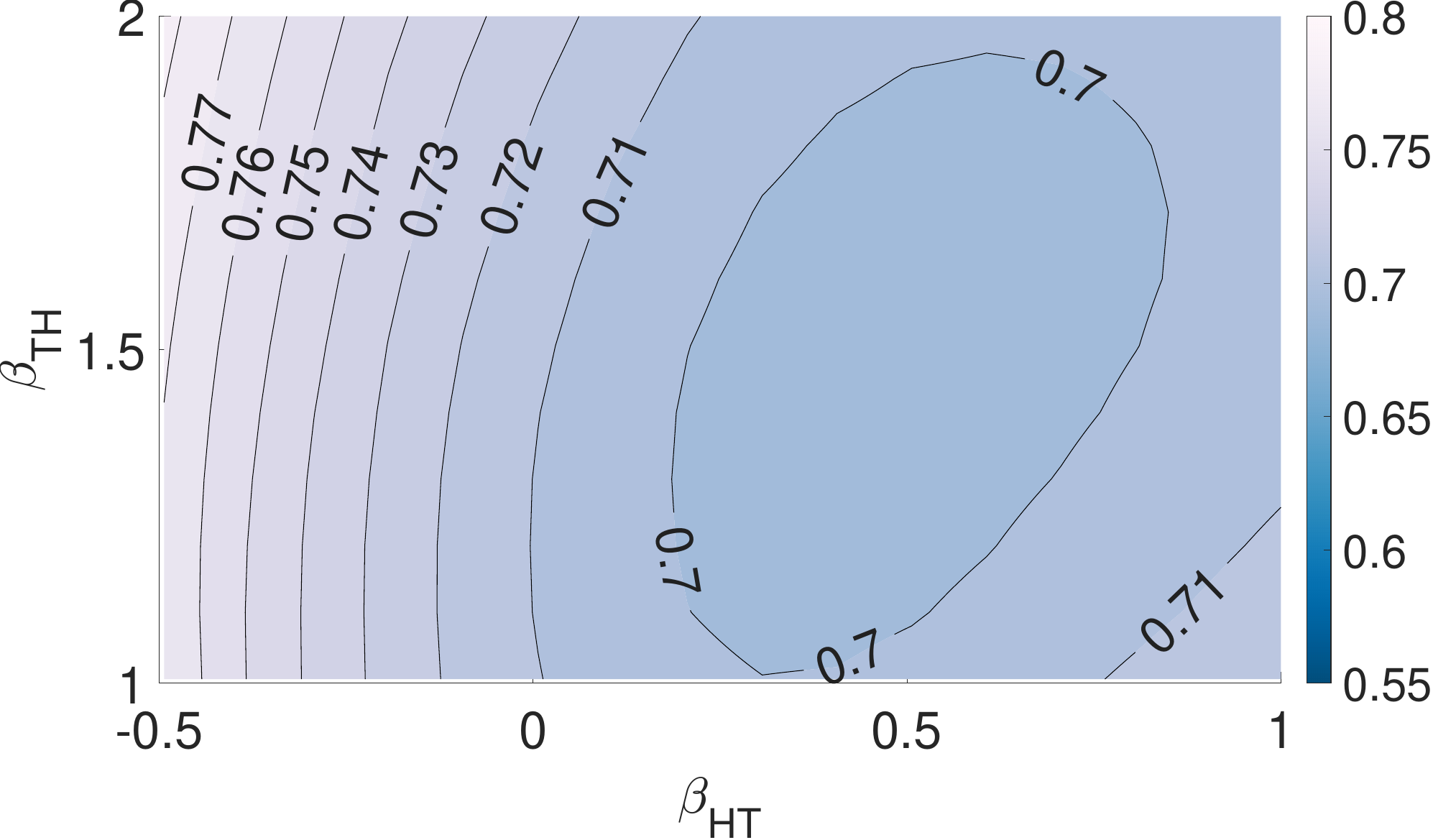}}
    \caption{Stability performance of the nominal controller and the safety-critical controller with CBF. The first and second columns show the head-to-tail string stability index $I$ defined in~\eqref{eq:stability index tail}. A smaller $I$ implies that the tail CAV has smaller speed perturbation, and thus the upstream traffic is smoother. With $I<1$, the mixed traffic system is considered head-to-tail string stable. The third and fourth columns depict the average string stability index $\bar{I}$ defined in~\eqref{eq:stability index platoon}.  A smaller $\bar{I}$ reflects that the entire mixed vehicle platoon drives smoother on average. A darker color implies a smaller $I$ or $\bar{I}$ and thus a smoother traffic.}
    \label{fig:analysis stability index}
\end{figure}

\begin{figure}[t]
    \centering
    Stability \vspace{0.5ex}
    \\
    \subfloat[Nominal controller $v_{\headcav}$]{\includegraphics[width=0.23\linewidth]{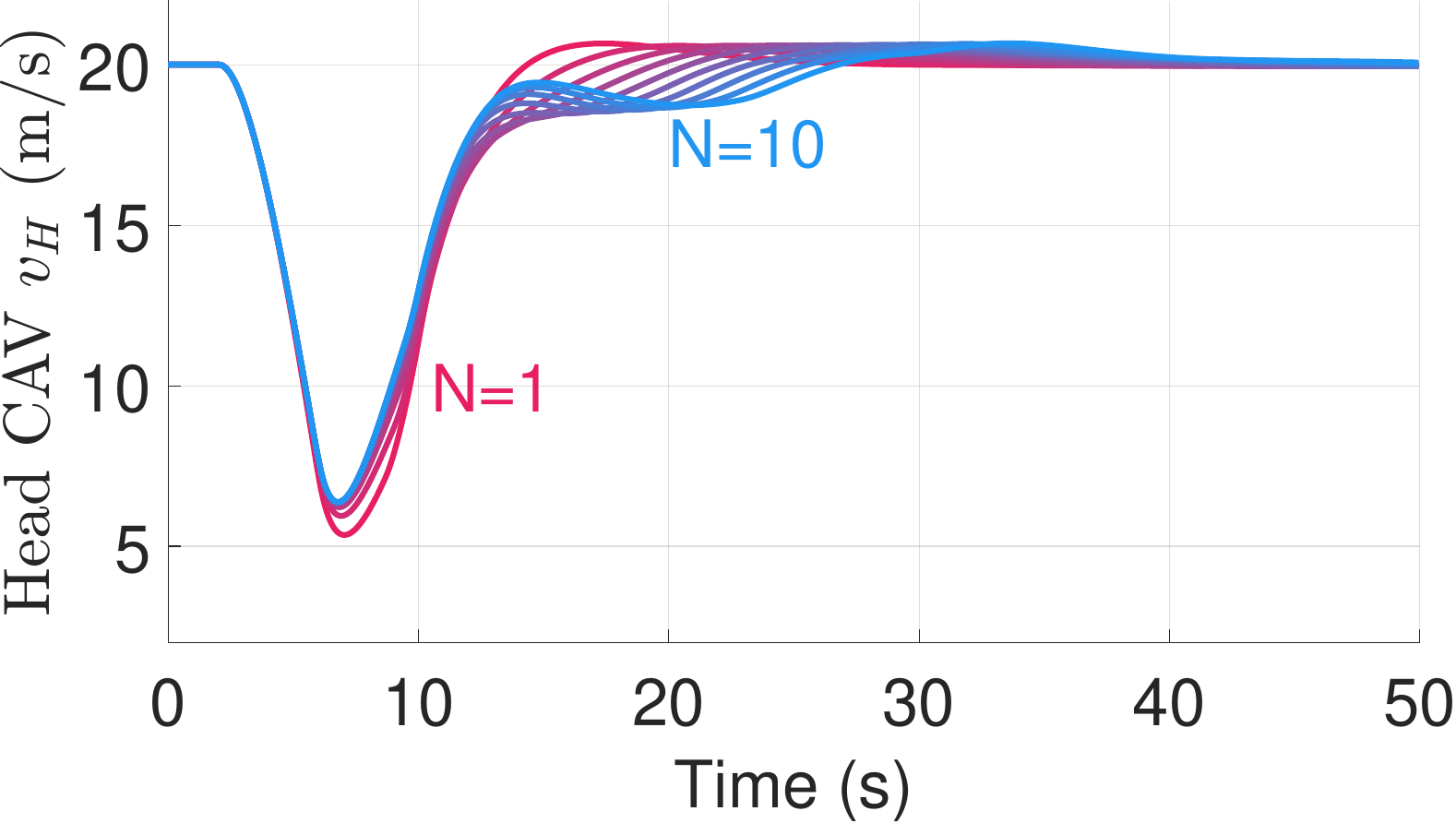}}
    \hspace{0.2em} \subfloat[CBF $v_{\headcav}$]{\includegraphics[width=0.23\linewidth]{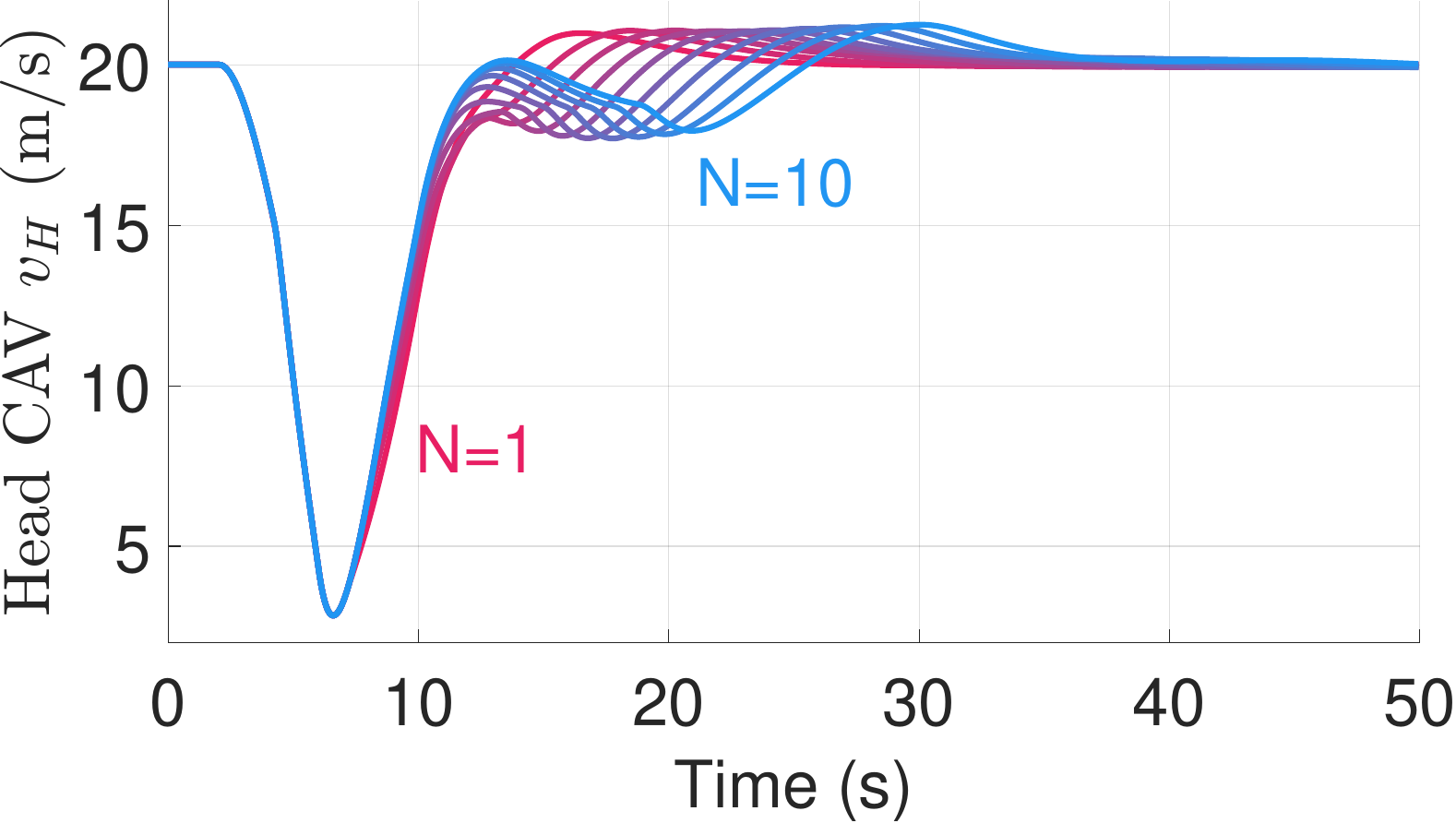}}
    \hspace{0.2em}\subfloat[Nominal Controller $v_{\tailcav}$]{\includegraphics[width=0.23\linewidth]{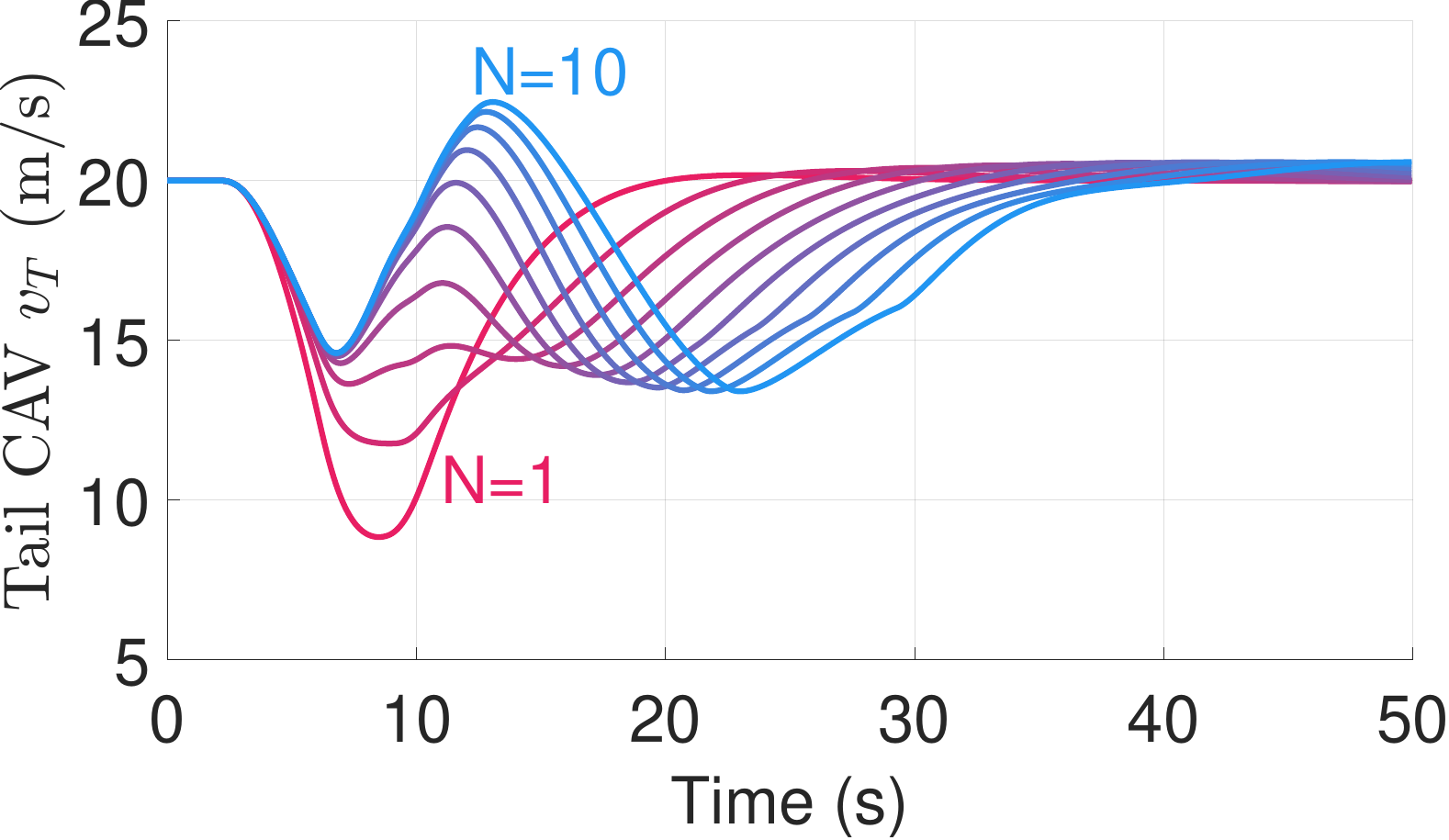}}
    \hspace{0.2em}\subfloat[CBF $v_{\tailcav}$]{\includegraphics[width=0.23\linewidth]{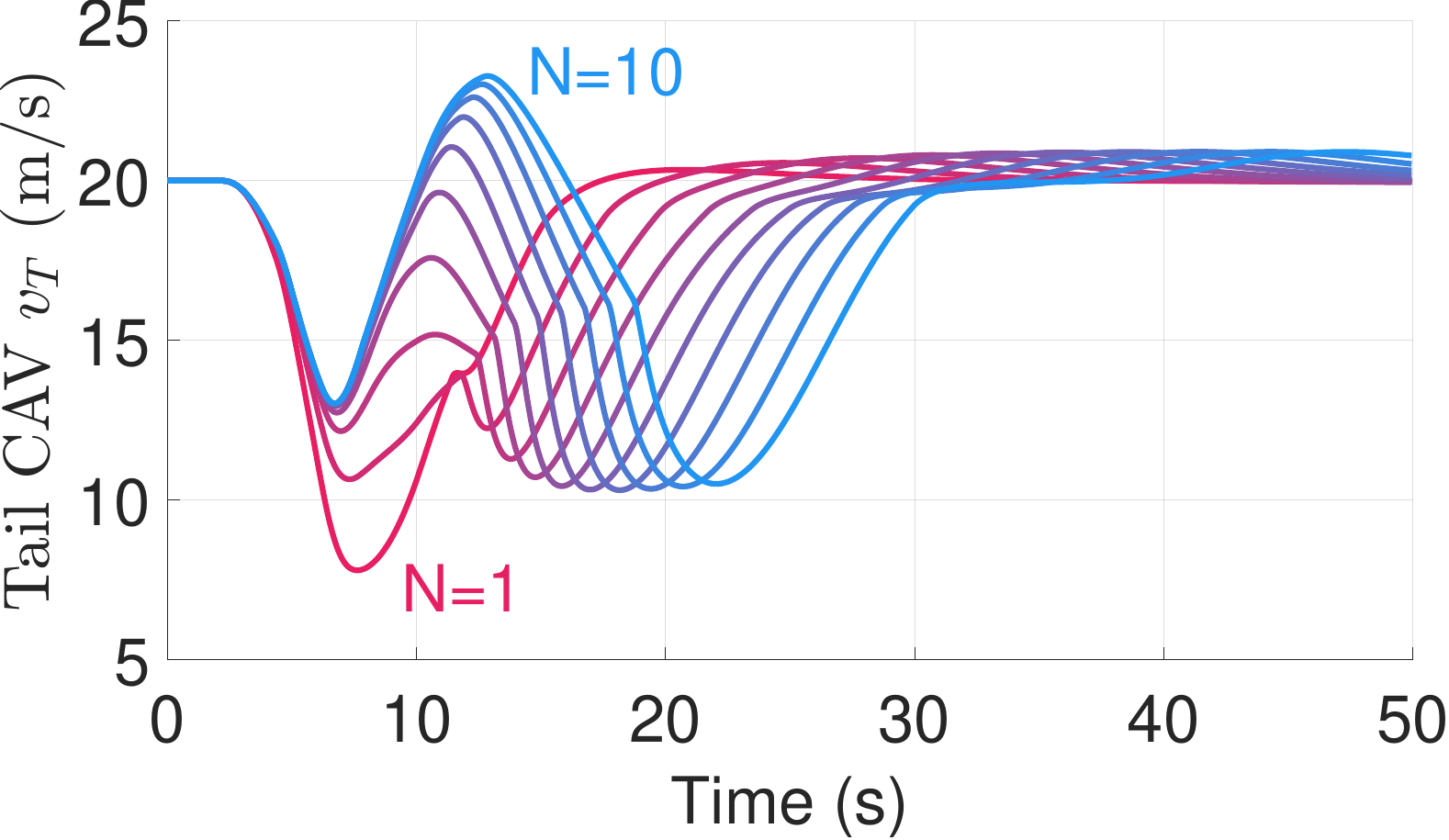}}
    \\
    \subfloat[Stability-safety trade-offs $I$]{\includegraphics[width=0.3\linewidth]{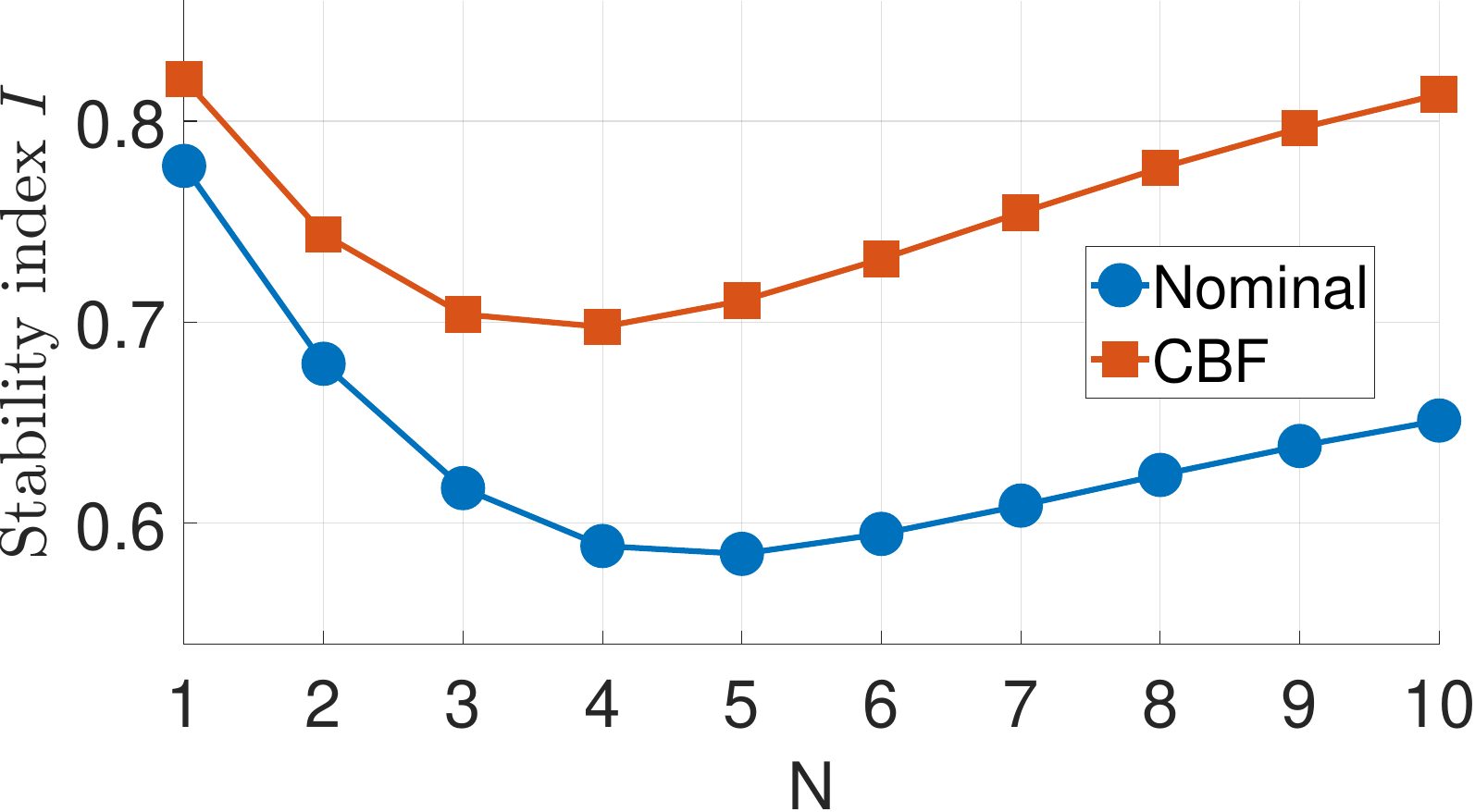}}
    \\
    \vspace{2ex} Safety \vspace{0.5ex} \\
    \subfloat[Nominal controller $h_{\headcav}$]{\includegraphics[width=0.23\linewidth]{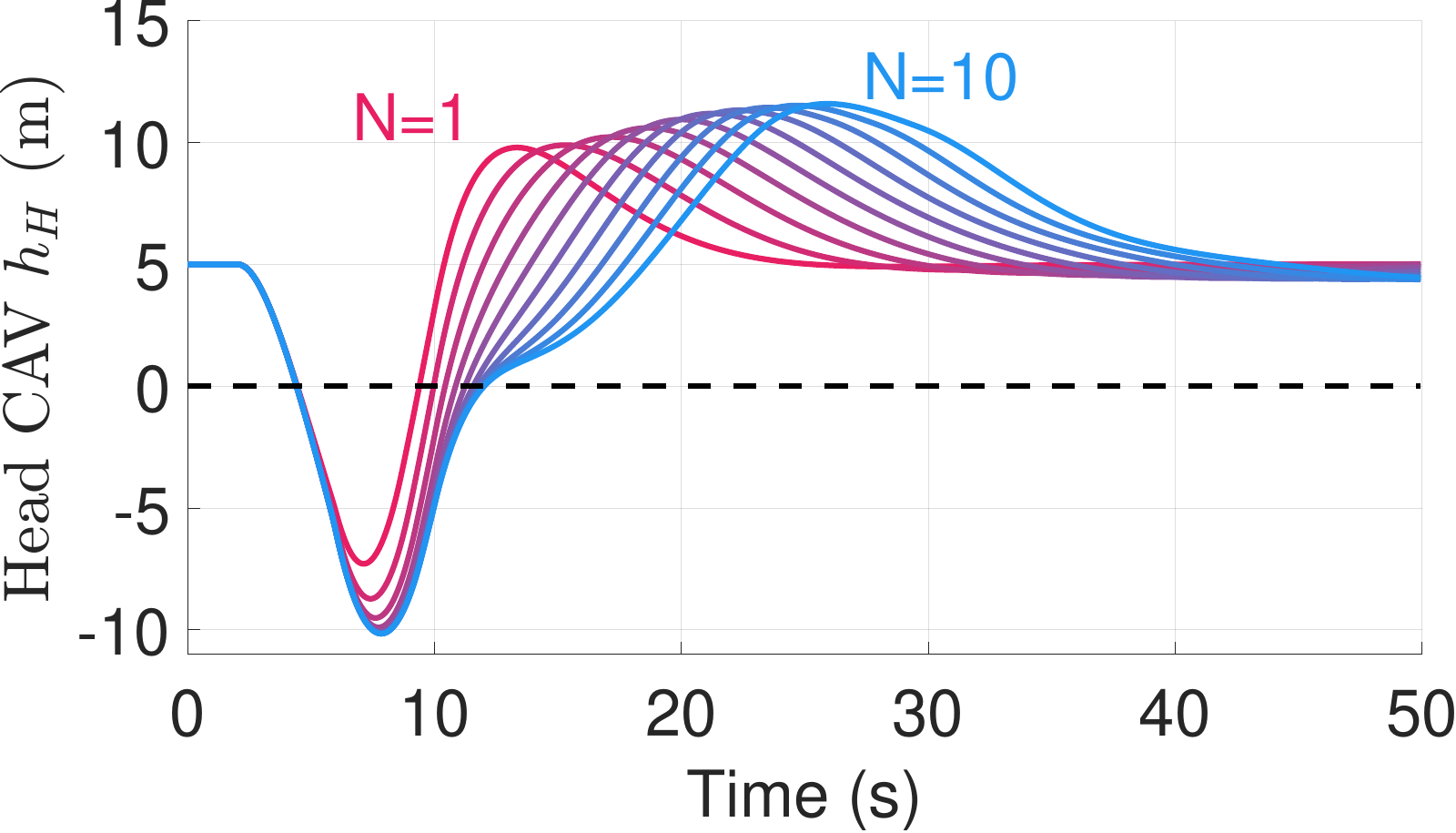}}
    \hspace{0.2em}\subfloat[CBF $h_{\headcav}$]{\includegraphics[width=0.23\linewidth]{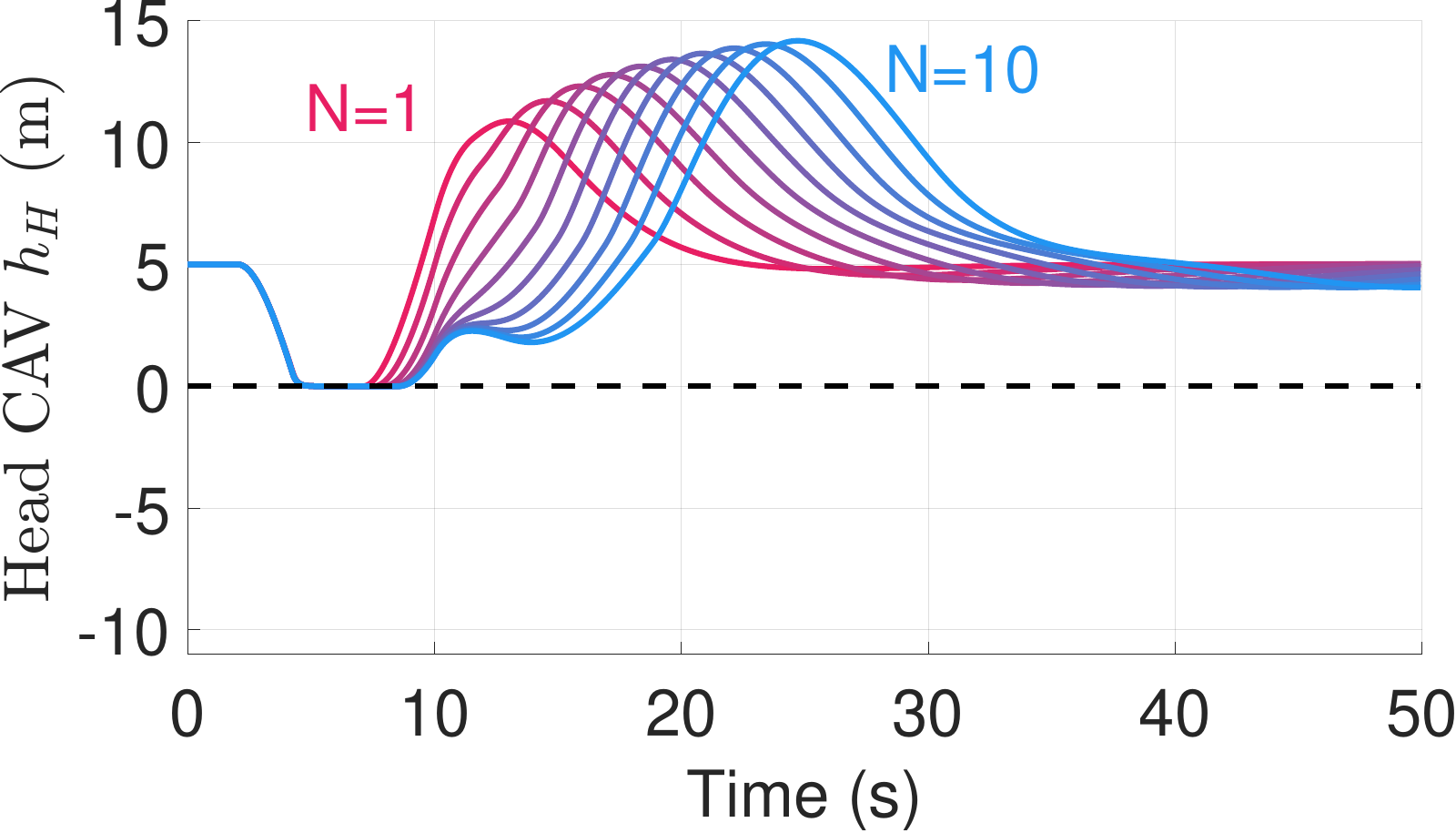}}
    \hspace{0.2em}\subfloat[Nominal controller $h_{\tailcav}$]{\includegraphics[width=0.23\linewidth]{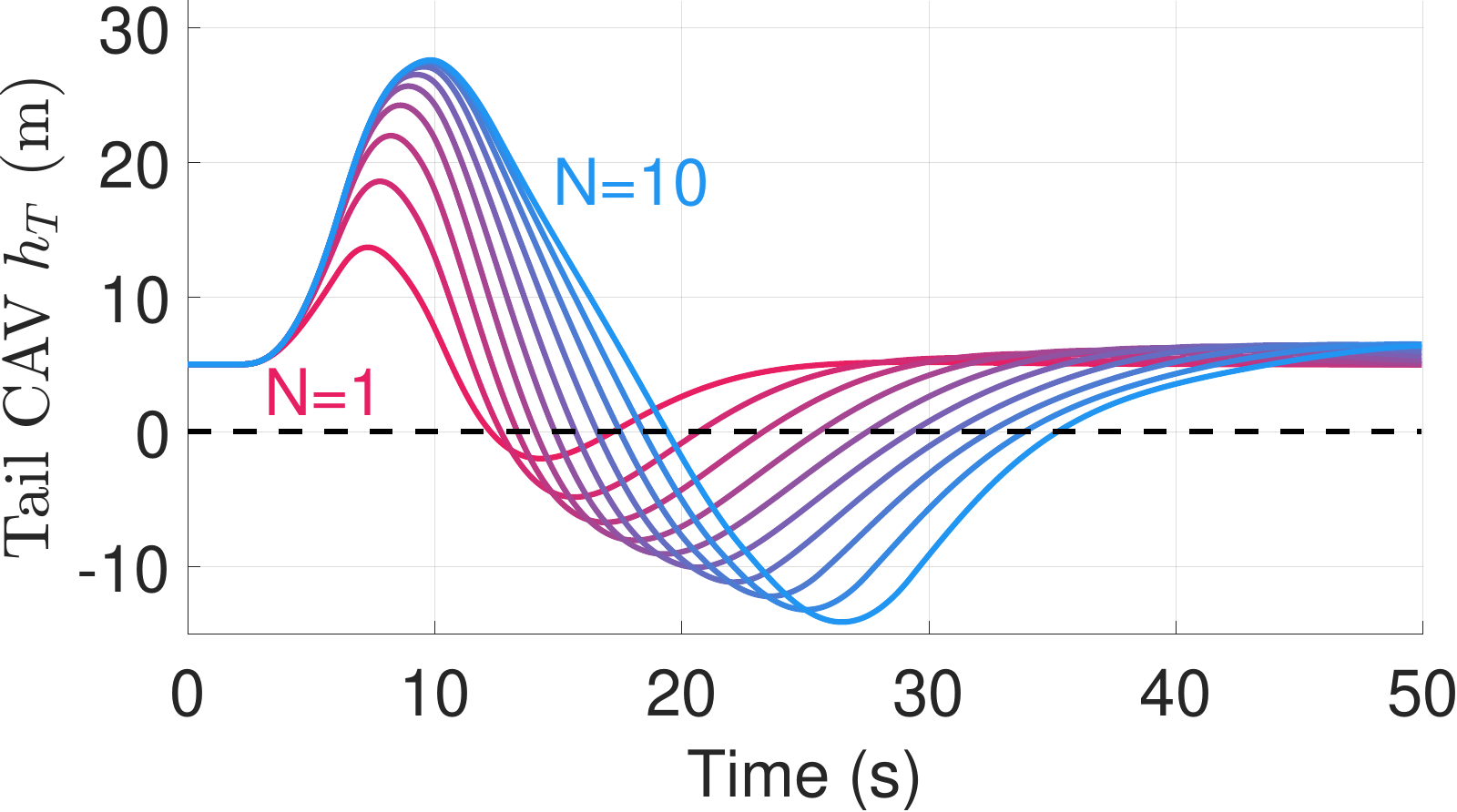}}
    \hspace{0.2em}\subfloat[CBF $h_{\tailcav}$]{\includegraphics[width=0.23\linewidth]{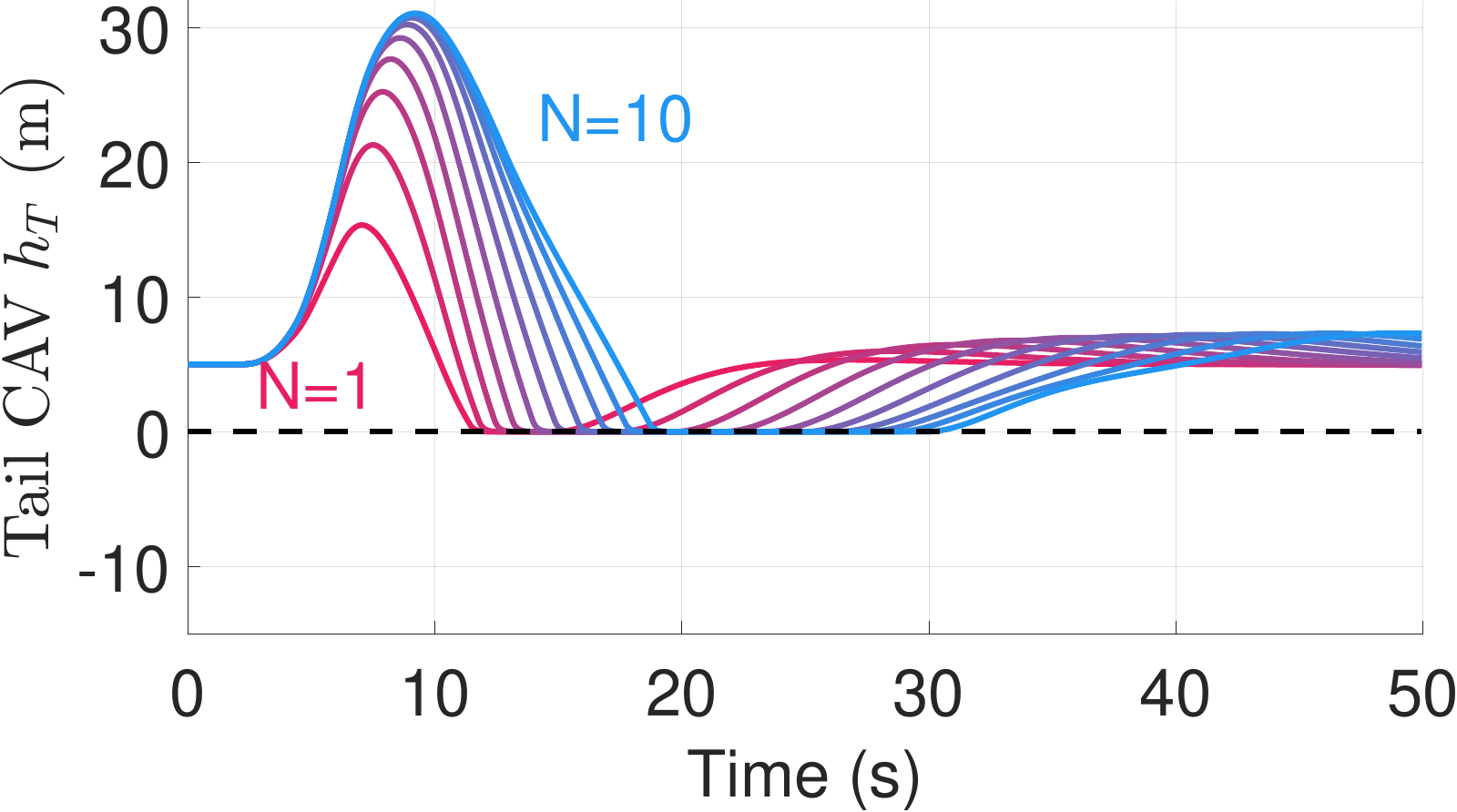}}
    \hspace{34mm}
    \caption{Simulated trajectories with different numbers of middle HVs, $N$. The proposed controller achieves both safe ($h\ge 0$) and string stable ($I<1$) driving for all considered $N$.
    }
    \label{fig:analysis HV number}
\end{figure}

\begin{figure}[t]
    \centering
    \subfloat[$a$]{\includegraphics[width=0.24    \linewidth]{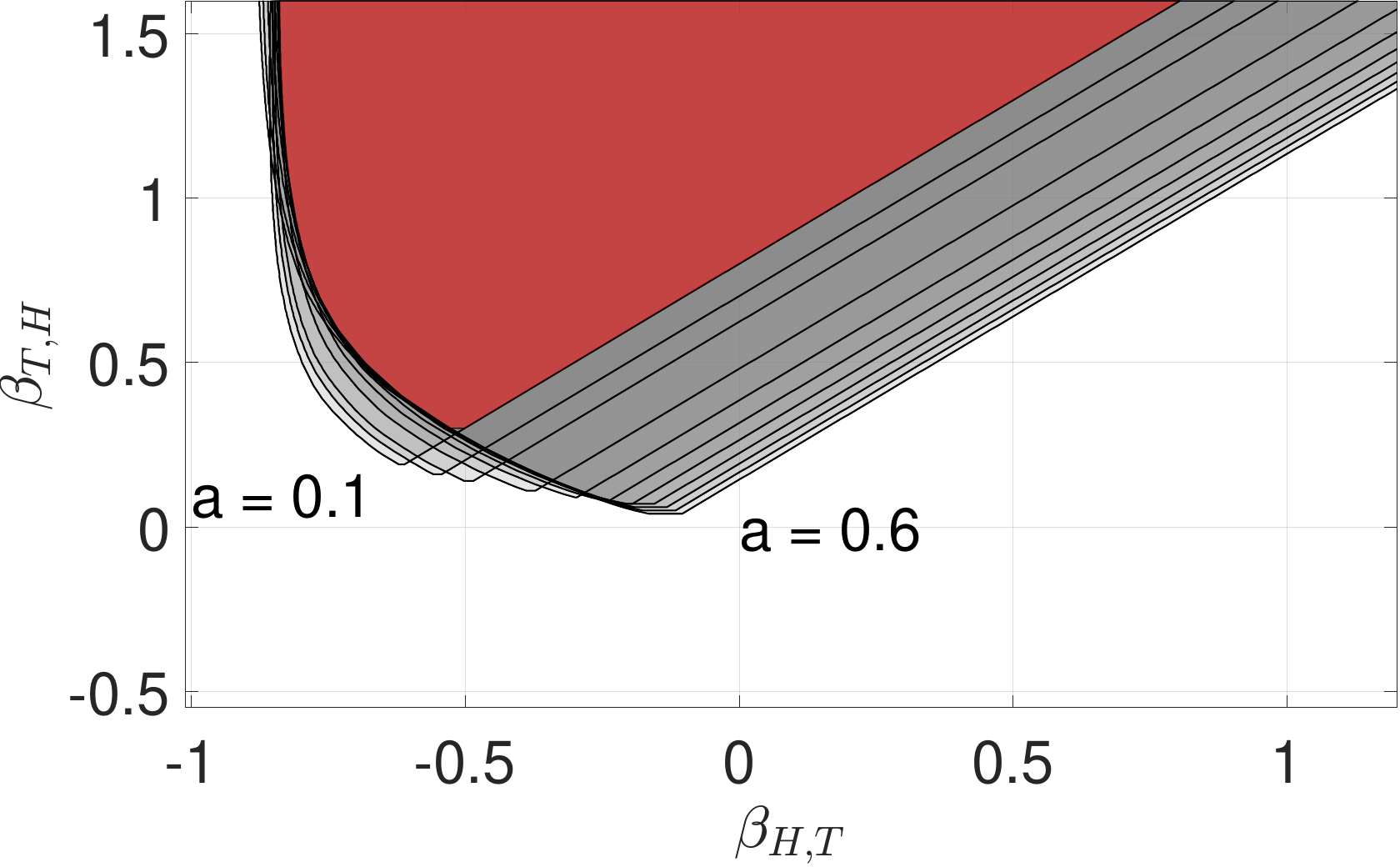}}
    \subfloat[$b$]{\includegraphics[width=0.24    \linewidth]{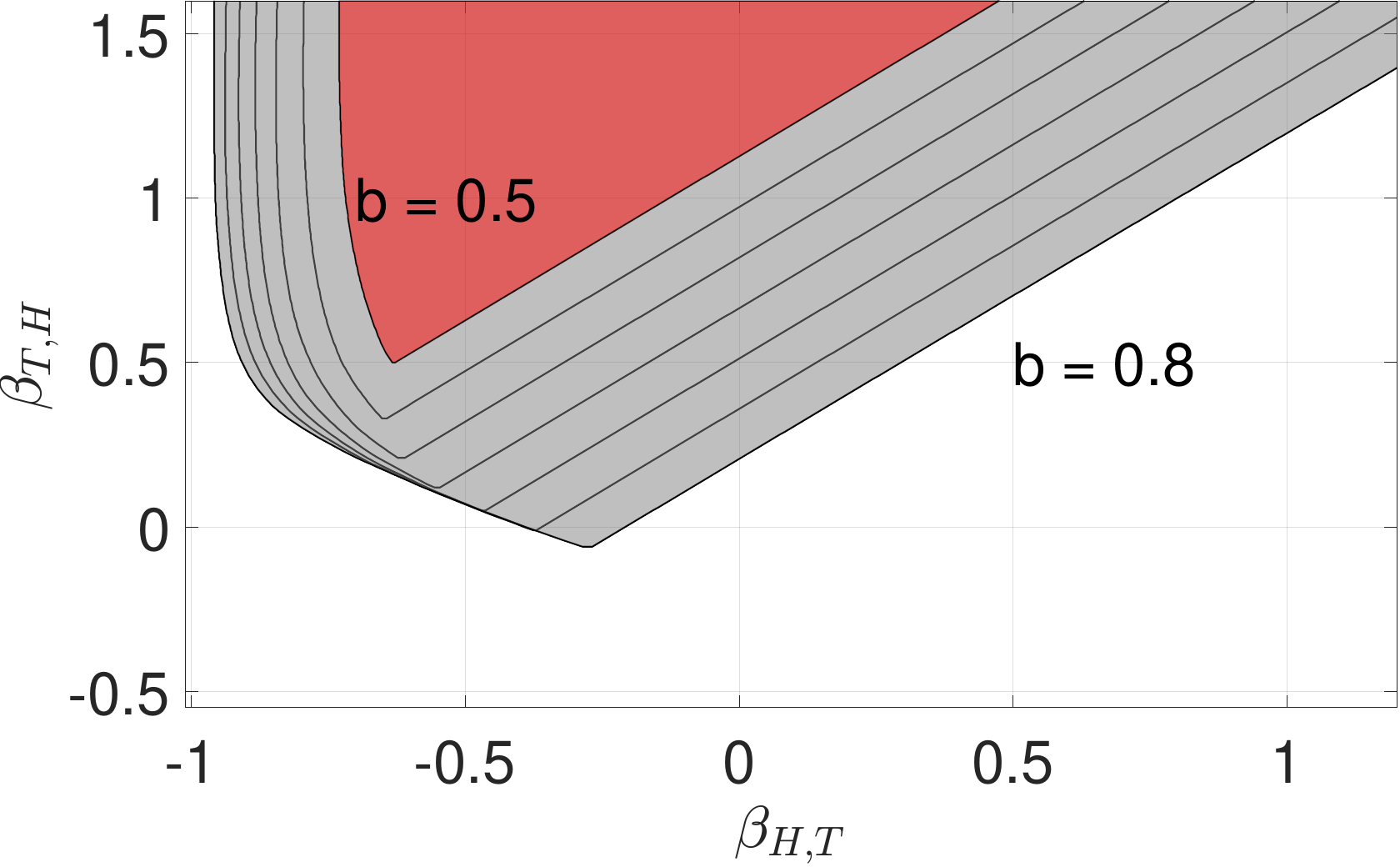}}
    \subfloat[$s_{\mathrm{st}}$]{\includegraphics[width=0.24    \linewidth]{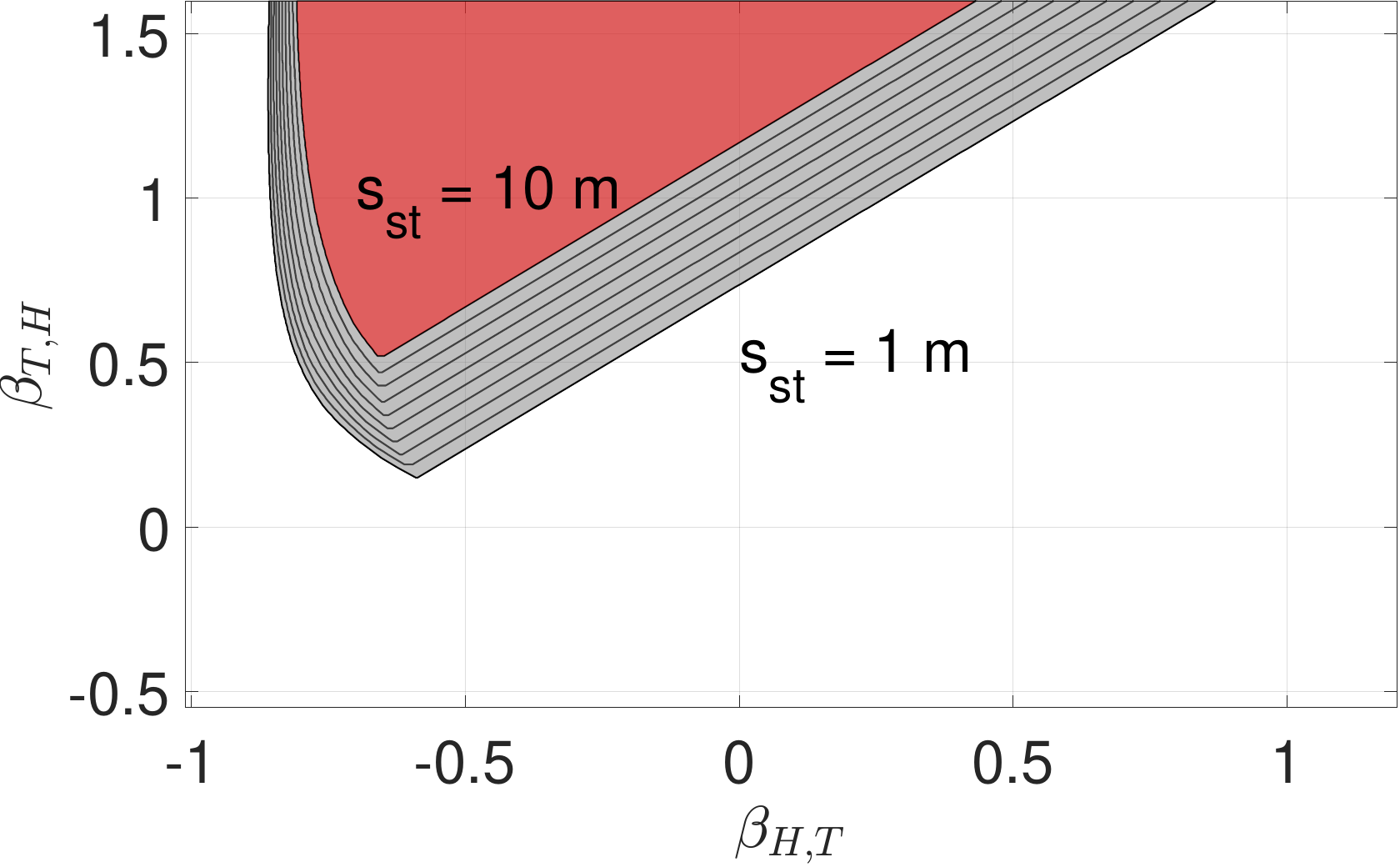}}
    \subfloat[$s_{\mathrm{go}}$]{\includegraphics[width=0.24    \linewidth]{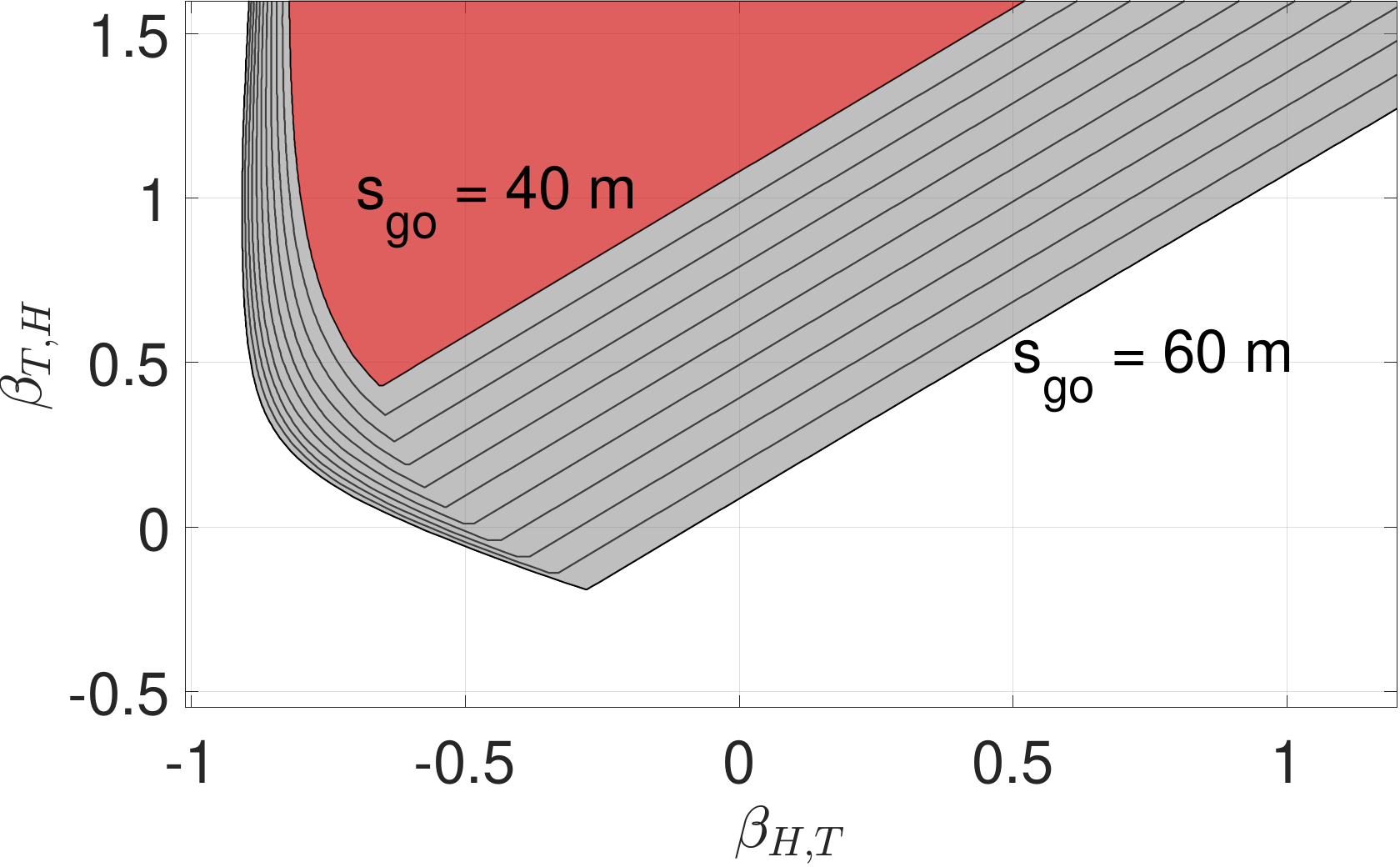}}
    \caption{Stability charts under different HV model parameters $a$, $b$, $s_{\mathrm{st}}$, and $s_{\mathrm{go}}$. In each subfigure, we plot the stability boundaries when one parameter changes with the other three parameters fixed. Each line gives a boundary of head-to-tail string stability for a given parameter value. Each grey area shades the range of $(\beta_{\headcav,\tailcav},\beta_{\tailcav,\headcav})$ gains that lead to string stability for the corresponding HV model parameter value. The red region is the overlap of string stability regions with varying parameters.}
    \label{fig:robust stability chart}
\end{figure}

\begin{figure}[t]
    \centering
    Nominal Controller \vspace{0.5ex}\\
    \includegraphics[width=0.24\linewidth]{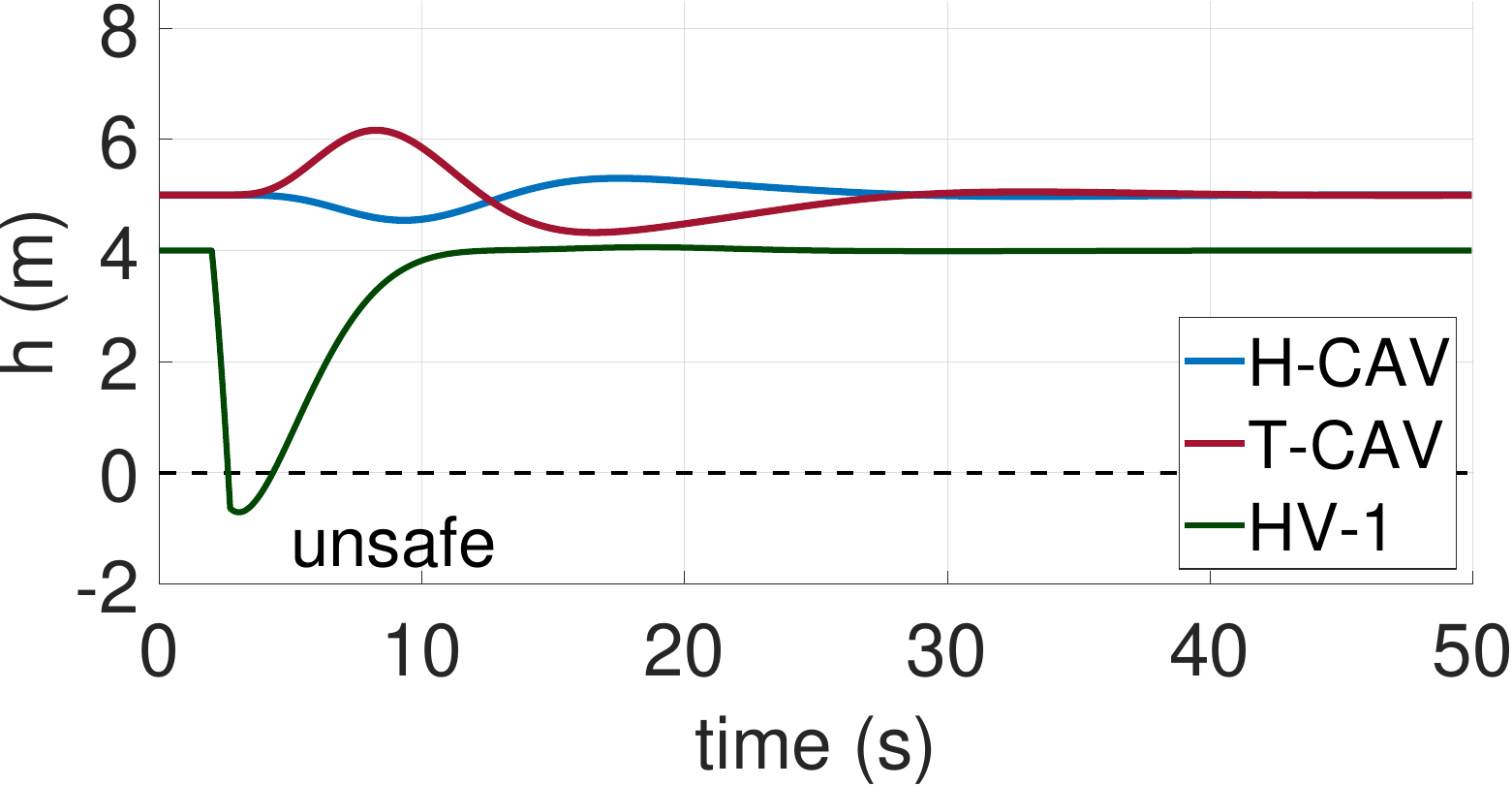}
    \includegraphics[width=0.24\linewidth]{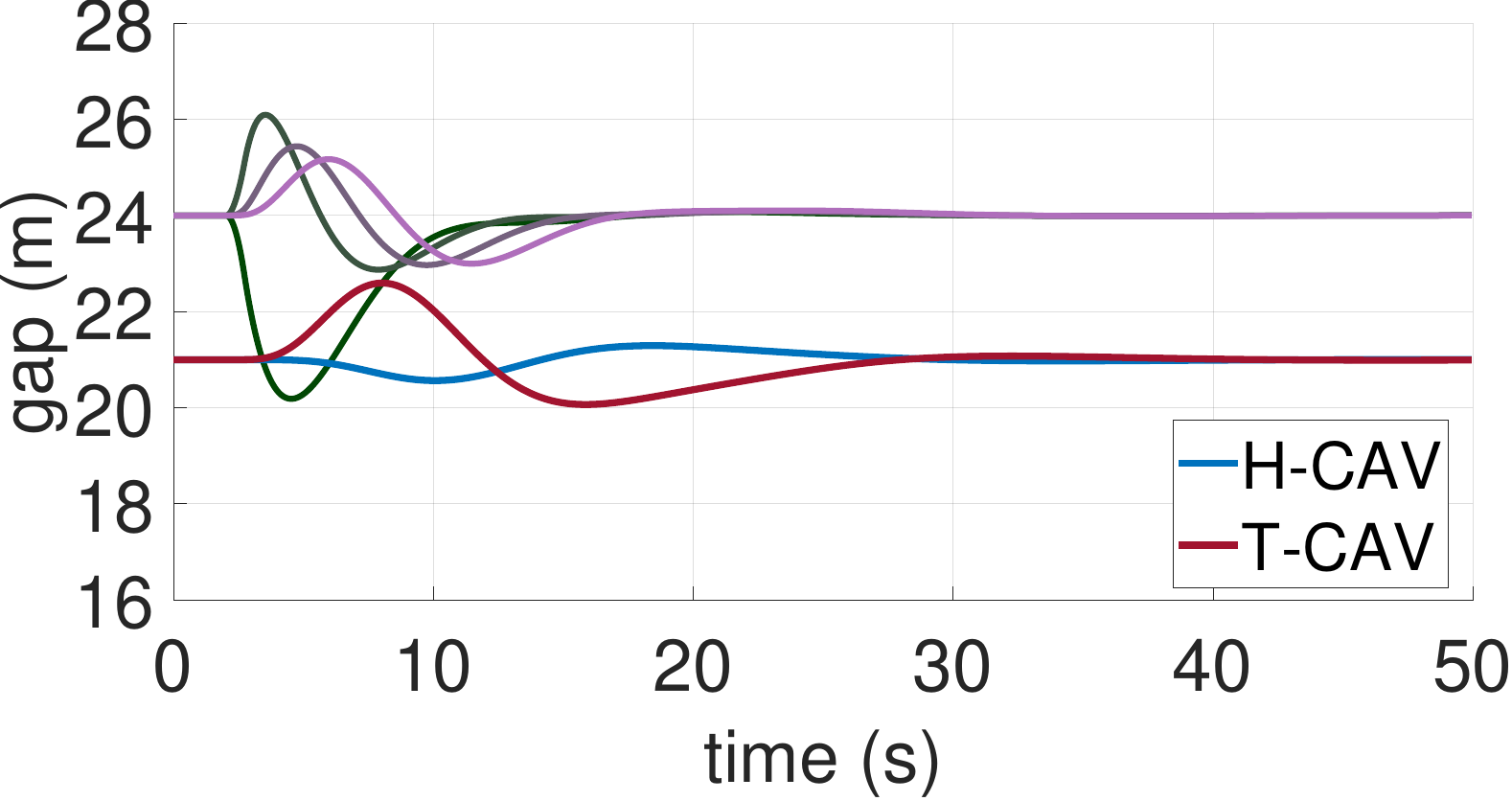}
    \includegraphics[width=0.24\linewidth]{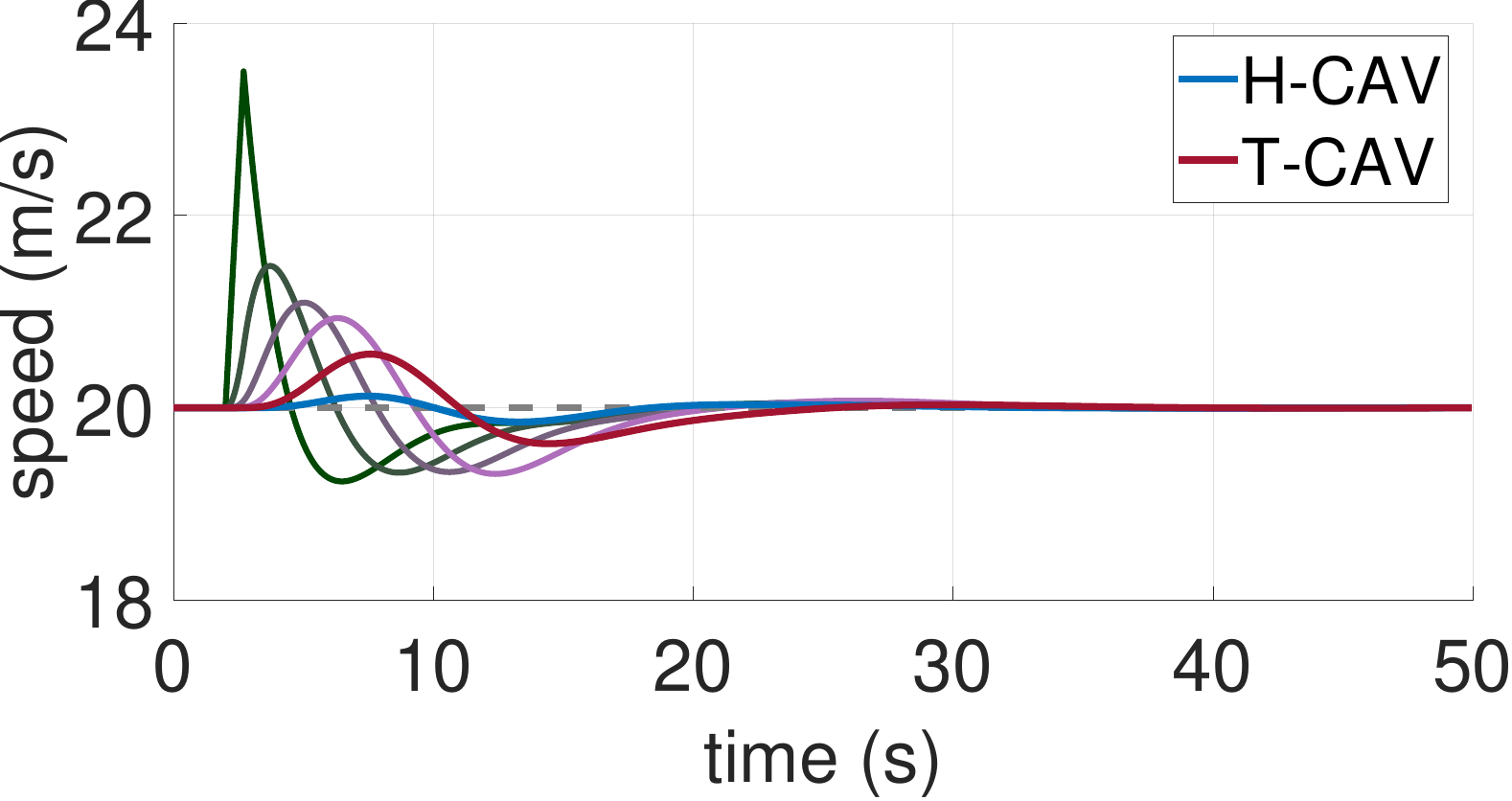}
    \includegraphics[width=0.24\linewidth]{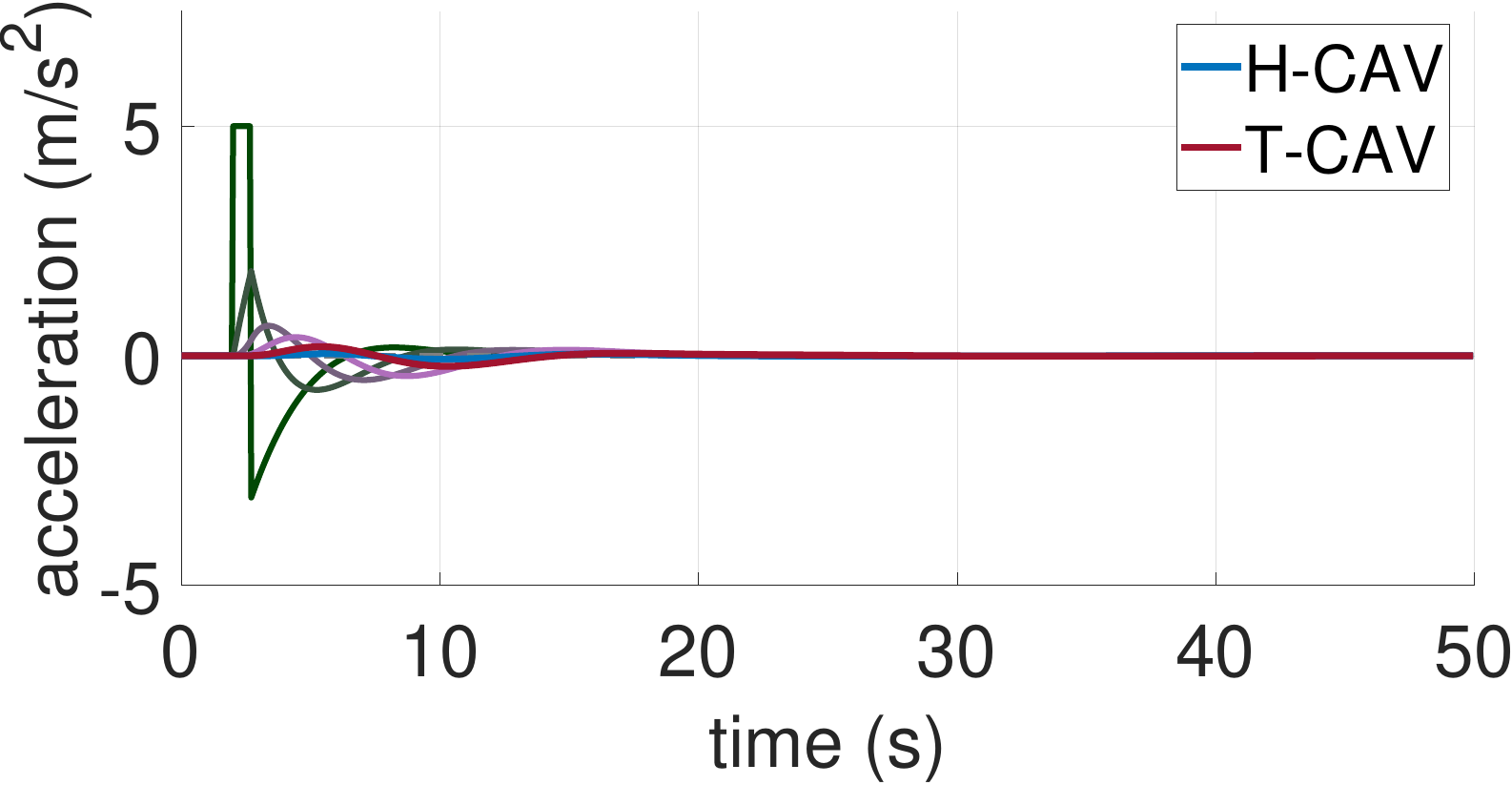}
    \\
    \vspace{2ex} CBF without robustness term \vspace{0.5ex} \\
    \includegraphics[width=0.24\linewidth]{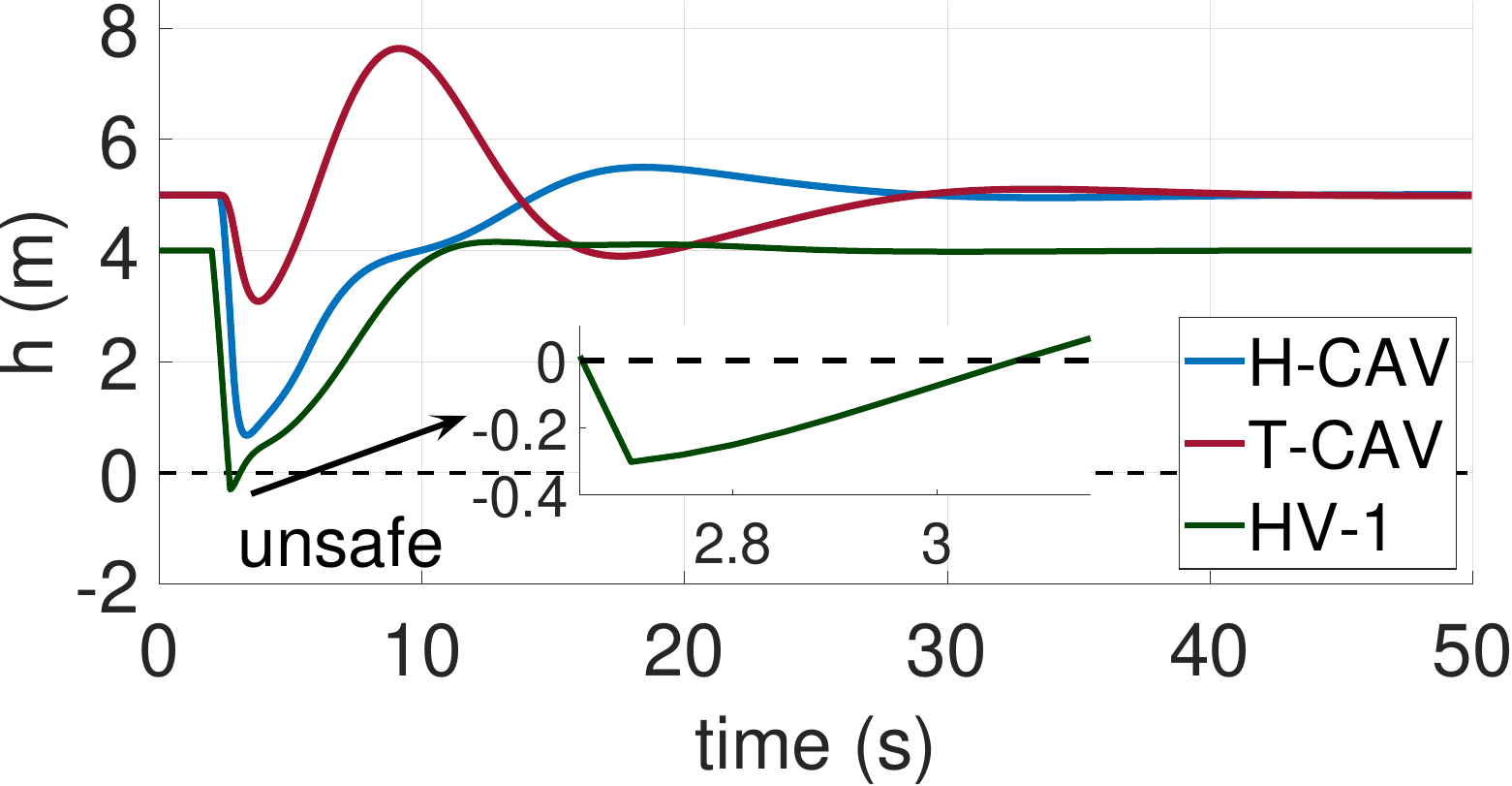}
    \includegraphics[width=0.24\linewidth]{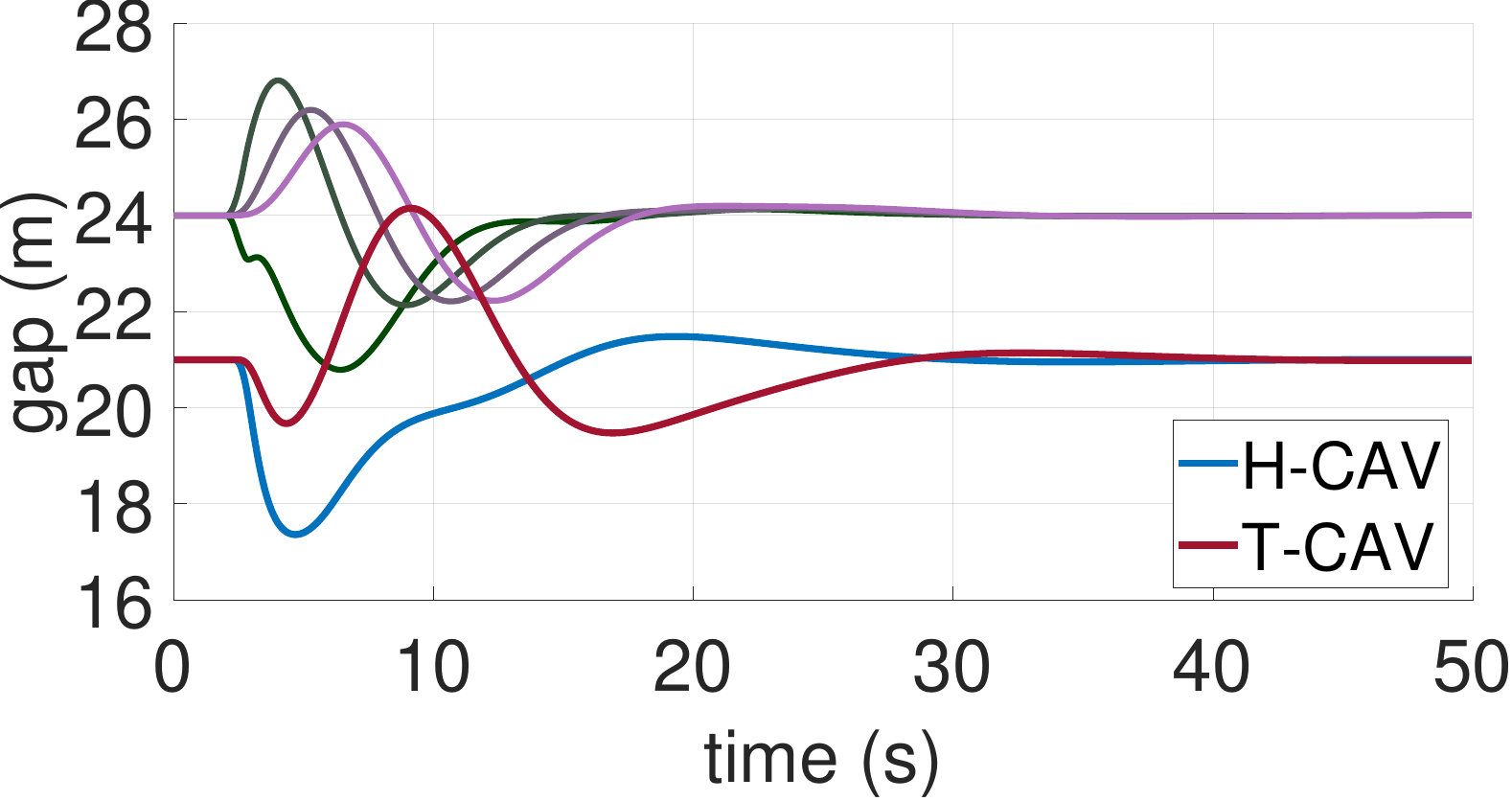}
    \includegraphics[width=0.24\linewidth]{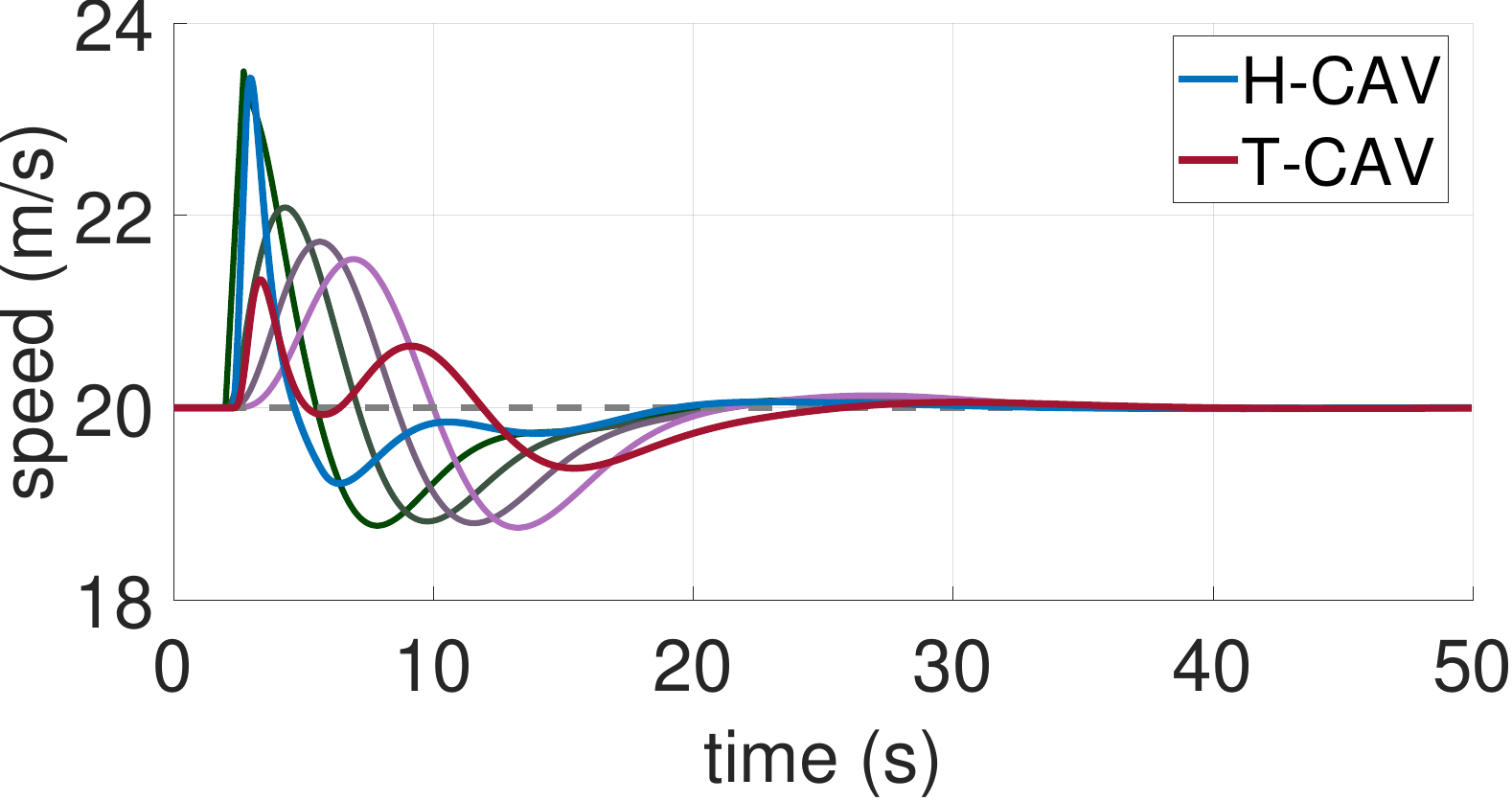}
    \includegraphics[width=0.24\linewidth]{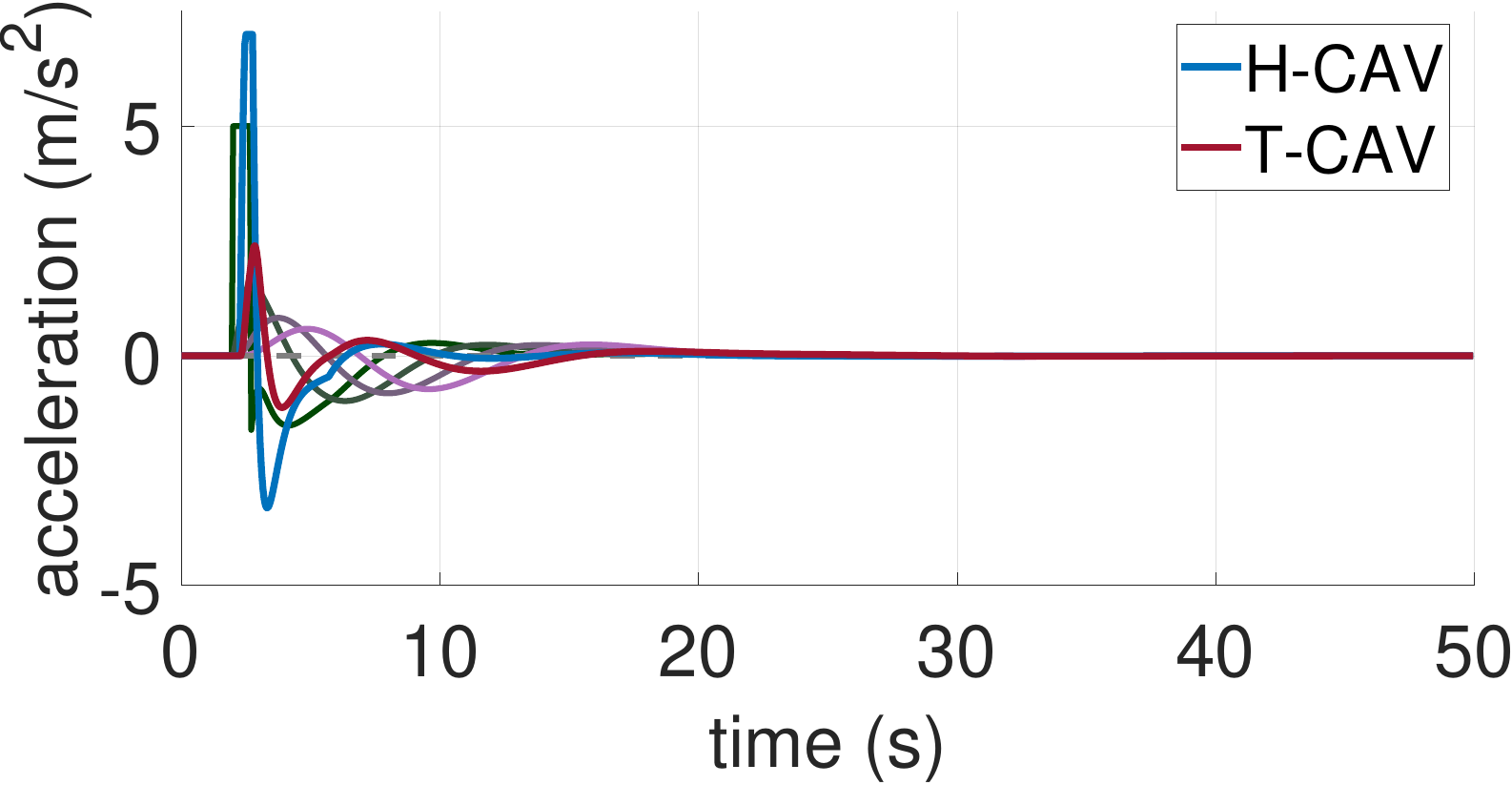}
    \\
    \vspace{2ex} Robust CBF \vspace{0.5ex}\\
    \includegraphics[width=0.24\linewidth]{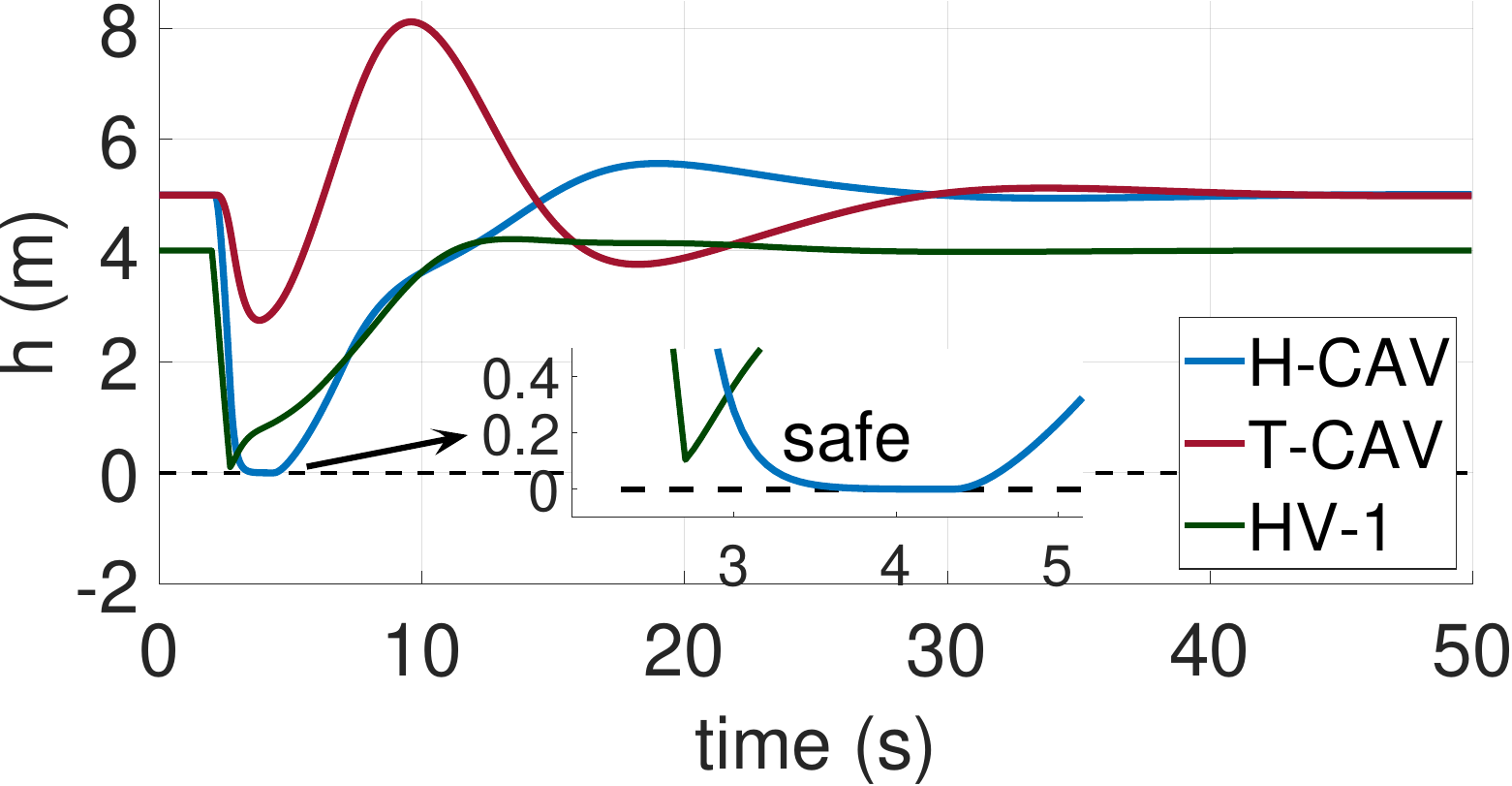}
    \includegraphics[width=0.24\linewidth]{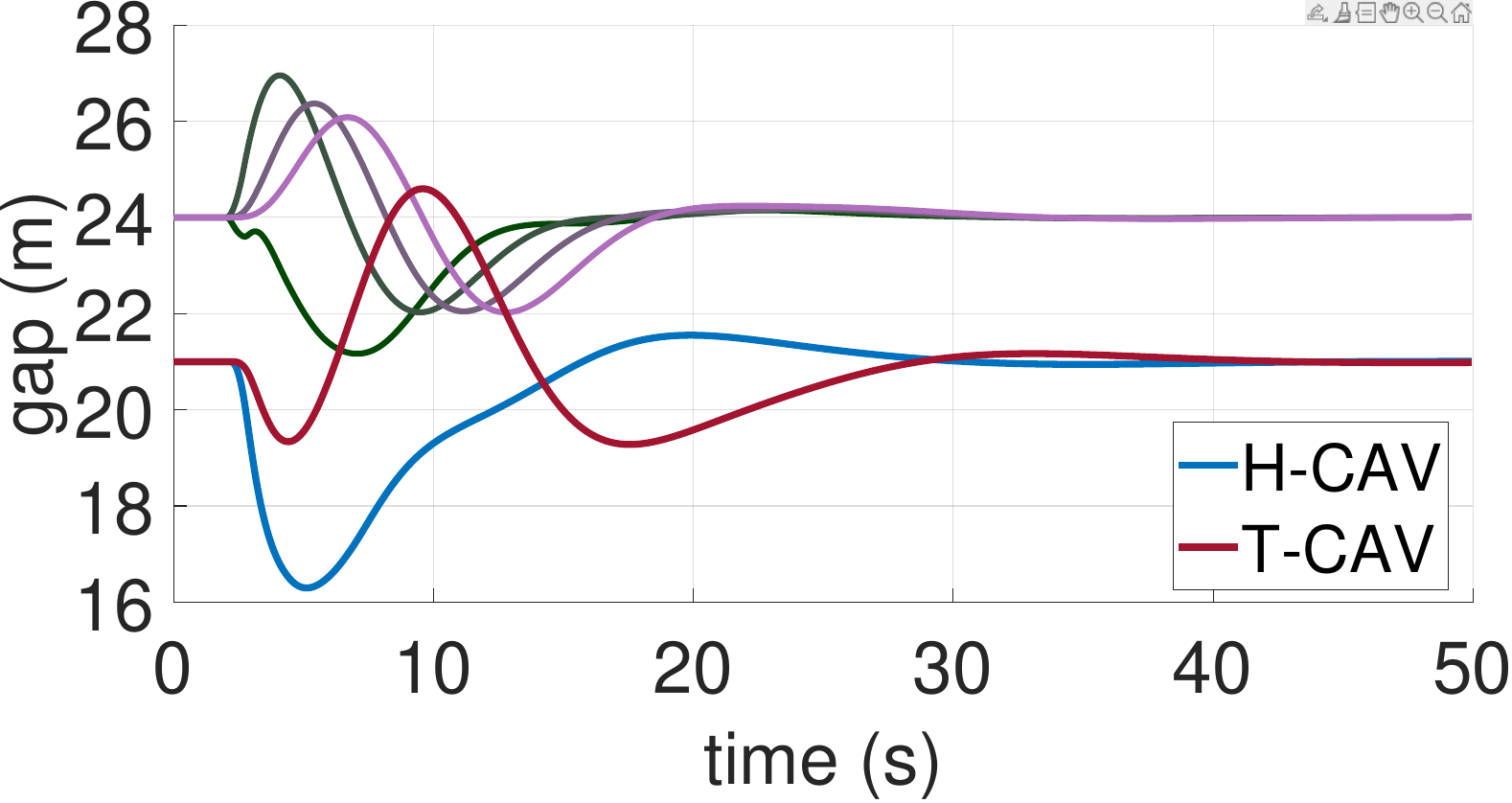}
    \includegraphics[width=0.24\linewidth]{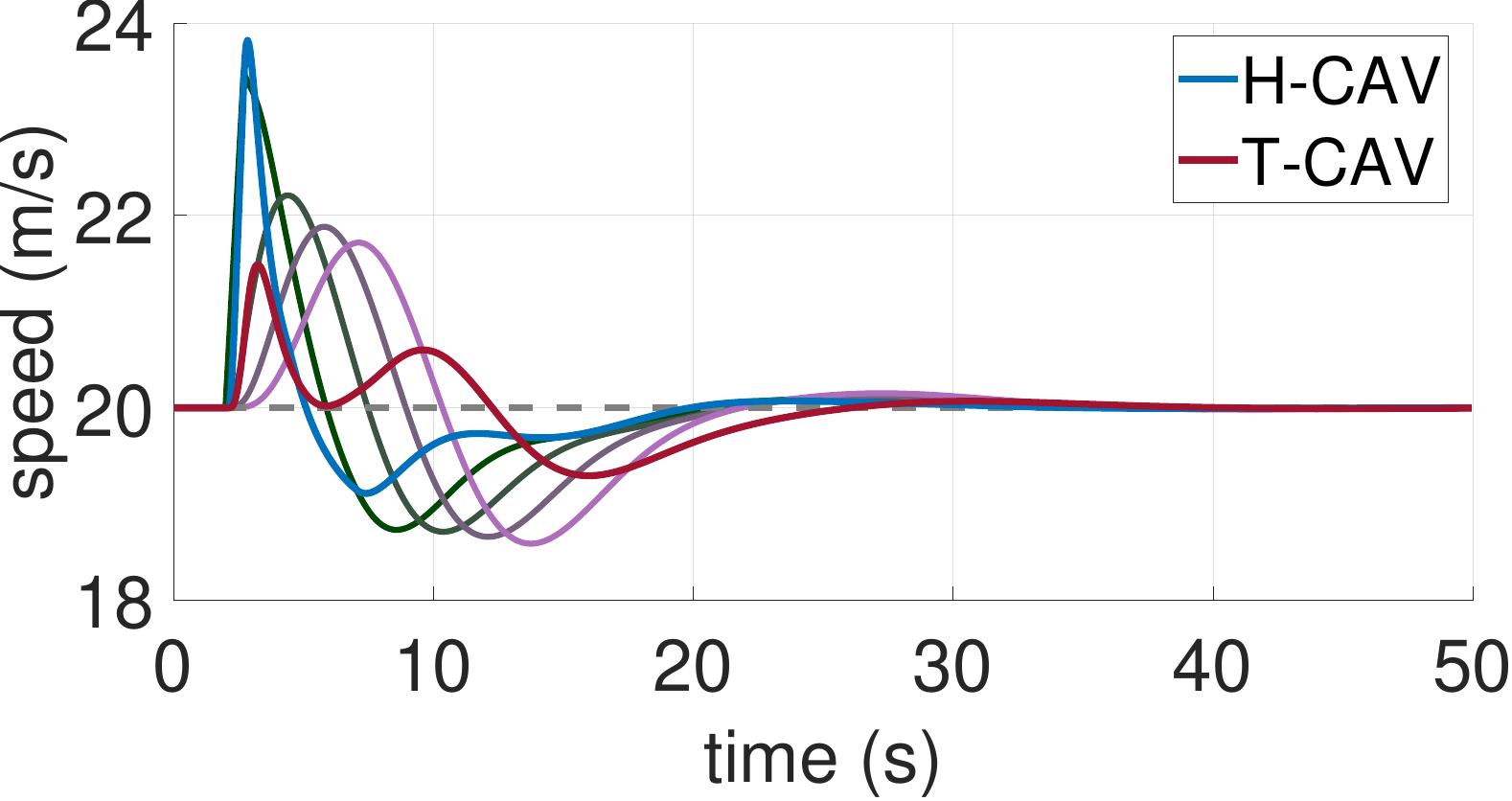}
    \includegraphics[width=0.24\linewidth]{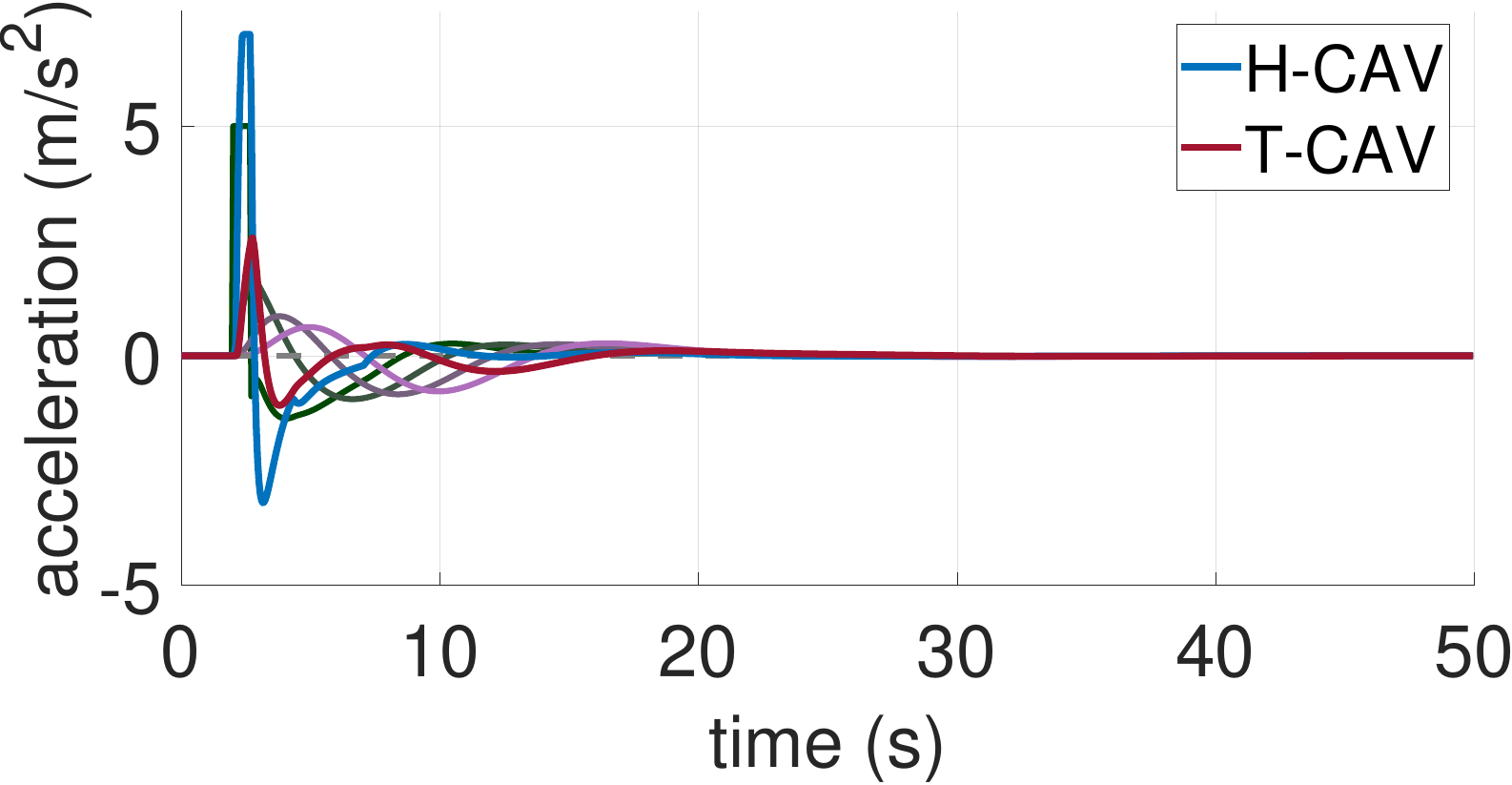}
    \caption{Simulated trajectories when one middle HV suddenly accelerates.
    When the human driver model used by the head CAV's control design is inaccurate, the CBF may fail to ensure HV safety.
    This can be remedied by using a robust CBF that accounts for the human driver uncertainty.}
    \label{fig:trajectory robust HV safety}
\end{figure}

\subsection{Trade-offs between stability and safety}

First, we consider the safety-critical controllers~\eqref{eq:min_controller_head},~\eqref{eq:min_controller_tail} that enforce CAV safety only.
We further investigate the effects of safety filters on safety and stability by conduction simulations with varying control gains $(\beta_{\headcav,\tailcav},\beta_{\tailcav,\headcav})$, considering the scenario where the head HV suddenly decelerates.
That is, the simulations in Fig.~\ref{fig:trajectory head HV dec} are repeated with various gains.

\subsubsection*{Effect of CBF on safety}
In theory, the safety-critical controllers guarantee safety with any nominal controller and in all scenarios. However, safety guarantees hold for bounded accelerations only. If the CBF requires too large acceleration or deceleration, the CAV's control input is saturated, and safety guarantees could be lost. We analyze how the nominal controller and CBF affect the two CAVs' safety with saturated accelerations.

Fig.~\ref{fig:safe region} presents simulation results with varying controller gains $(\beta_{\headcav,\tailcav},\beta_{\tailcav,\headcav})$. The first and second columns plot the maximum speed perturbation $\Delta v_{\hhv}$ of the head HV for which the head and tail CAV remains safe. A darker color indicates that the CAV is able to maintain safety at larger speed perturbations. Fig.~\ref{fig:safe region}(b) and (f) show that, by adding the CBF-based safety filters, both the head and tail CAV remains safe for almost all controller gains even when the head HV decelerates to a full stop, i.e., $\Delta v_{\hhv}  = 20$ m/s.
The third and fourth columns present the range of controller gains $(\beta_{\headcav,\tailcav},\beta_{\tailcav,\headcav})$ that ensure the safety of the head and tail CAVs for a fixed speed perturbation $\Delta v_{\hhv}$.  Fig.~\ref{fig:safe region}(c) and (g) give the safety region for $\Delta v_{\hhv} = 12$ m/s. As the grey area shows, the nominal controller can achieve safety for one of the CAVs with a limited choice of gains. However, by comparing panels (c) and (g), it can be concluded that no gains can enable the nominal controller to ensure safety for both CAVs. As opposed, by adding the CBF-based safety filters, both the head and tail CAVs remain safe for all considered $(\beta_{\headcav,\tailcav},\beta_{\tailcav,\headcav})$ gains, as the red region shows. 
Fig.~\ref{fig:safe region}(d) and (h) gives the safety region for $\Delta v_{\hhv} = 20$ m/s. We note from Fig.~\ref{fig:safe region}(h) that all nominal controllers fail to ensure the safety of the tail CAV (i.e., there is no grey region), while the CBF still guarantees its safety for most $(\beta_{\headcav,\tailcav},\beta_{\tailcav,\headcav})$ pairs (see the red region).
The unsafe domain (white region) is caused by the saturation of accelerations.

\subsubsection*{Effect of CBF on stability} 
As trajectories in Figs.~\ref{fig:trajectory head HV dec},~\ref{fig:trajectory HV dec}, and~\ref{fig:trajectory HV acc} show, the CBF may increase the acceleration or deceleration of CAVs to avoid collisions, i.e., there is a trade-off between safety and stability. To evaluate the CBF's effect on string stability (traffic smoothness), we run simulations when the head HV decelerates (i.e., for the scenario of Fig.~\ref{fig:trajectory head HV dec}), and we calculate the head-to-tail string stability index $I$ defined in~\eqref{eq:stability index tail}. 
Fig.~\ref{fig:analysis stability index}(a) and (b) depict the stability index obtained for the nominal controller and CBF using various $(\beta_{\headcav,\tailcav},\beta_{\tailcav,\headcav})$ controller gains. We see that after adding the CBF, we still have $I<1$ for all gains. Thus, combining the safe region in Fig.~\ref{fig:safe region} and the stability index in Fig.~\ref{fig:analysis stability index}, the safety-critical controller achieves both safety and string stability.

Besides evaluating head-to-tail string stability, that characterizes the smoothness of the tail CAV's motion, we also investigate the smoothness of the overall mixed traffic considering all vehicles in the platoon. We define the average string stability index:
\begin{align}\label{eq:stability index platoon}
    \bar{I} = \frac{1}{N+2} \left( \frac{ \sqrt{\int_0^T (v_{\headcav} - v^*)^2 \diff t}} {\sqrt{\int_0^T (v_{\hhv} - v^*)^2 \diff t}} + \frac{\sqrt{\int_0^T (v_{\tailcav} - v^*)^2 \diff t}}{\sqrt{\int_0^T (v_{\hhv} - v^*)^2 \diff t}} + \sum_{i=1}^N \frac{\sqrt{\int_0^T (v_i - v^*)^2 \diff t}}{\sqrt{\int_0^T (v_{\hhv} - v^*)^2 \diff t}} \right).
\end{align} 
Fig.~\ref{fig:analysis stability index}(c) and (d) plot $\bar{I}$ for the nominal controller and CBF. Similar to the head-to-tail string stability index $I$, the average string stability index $\bar{I}$ becomes higher when implementing CBF, since the acceleration becomes larger to maintain safety. Nevertheless, we still have $\bar{I}<1$ for the CBF, which implies that perturbations from the downstream traffic are attenuated by the platoon of two CAVs and $N$ middle HVs.

\subsection{Sensitivity to the CAV penetration rate}

We analyze how the number of middle HVs, $N$, affects the performance of the CAV controllers. Namely, how the CAV penetration rate affects the stability and safety of mixed traffic.    
We repeat the simulation considering that the head HV decelerates as Fig.~\ref{fig:trajectory head HV dec} for various $N$ values ranging from 1 to 10, with the penetration rate $p=2/(N+2)$ ranging from $67\%$ to $17 \%$. 
Note that the vehicle-to-vehicle communication connecting the two CAVs has a limited range, and 10 middle HVs is approximately the maximum $N$ that can be covered by the communication range.
Since the middle HVs are assumed to be non-connected,  the two CAVs do not respond to them, and the nominal controller remains the same for different numbers of middle HVs. We set all parameters except $N$ to be
the same as in Fig.~\ref{fig:trajectory head HV dec}.

Fig.~\ref{fig:analysis HV number} shows the simulated the speed $v_{\headcav}$, $v_{\tailcav}$ and safety function $h_{\headcav}$, $h_{\tailcav}$  for various $N$ values.
The first row of Fig.~\ref{fig:analysis HV number} depicts the profiles of speed $v_{\headcav}$ and $v_{\tailcav}$ under the nominal controller and CBF. We see that the nominal controller stabilizes the system for each $N$, i.e., the system state converges to the equilibrium value once the head HV drives at the equilibrium speed $v^*$. With more middle HVs, however, it takes more time to converge: the system gets close to the equilibrium around 20 sec for $N=1$, while it takes 40 sec for $N=10$. As for string stability, Fig.~\ref{fig:analysis HV number}(e) plots the string stability index $I$ as a function of $N$. We observe that $I<1$ for all $N$, i.e., the nominal controller achieves string stability. With the increase of $N$, $I$ first decreases and then increases.
This means that the tail CAV performs best if the head CAV is a few vehicles ahead ($N \approx 4$), which is consistent with the findings in connected cruise control~\cite{orosz2016connected}. When the safety filter is implemented, the effect of the HV number on stability is similar to the nominal controller. From Fig.~\ref{fig:analysis HV number}(b) and (d), the safety-critical controller still ensures plant stability, i.e., $v$ converges to $v^*$.  For the string stability, as Fig.~\ref{fig:analysis HV number}(e) shows, we still have $I<1$ for all $N$, which means the system is still string stable with the CBF. 

As for safety, we plot $h_{\headcav}$ and $h_{\tailcav}$ in the second row of Fig.~\ref{fig:analysis HV number}. For the nominal controller, both the head and tail CAV become more unsafe when there are more middle HVs (i.e., the minimum of $h$ becomes smaller as $N$ increases and $h$ goes below zero in each case in Fig.~\ref{fig:analysis HV number}(f) and (h)).  For the safety-critical controller, as Fig.~\ref{fig:analysis HV number}(g) and (i) show, both the head CAV and tail CAV maintain safety in the sense that $h_{\headcav}\ge 0$ and $h_{\tailcav}\ge 0$. To summarize, by implementing CBF constraints, the mixed vehicle platoon has safe and string stable driving. 

\subsection{Robustness to uncertainty of human driver behaviors}

Next, we study how the proposed controllers are affected by the uncertainty induced by human driver behaviors.

\subsubsection*{Robust stability with parameter uncertainty of the human driver models}

We analyze how the stability performance of the nominal controller is affected by the parameters of the human driver model. We consider the case where the two CAVs are connected only to each other, and repeat the stability calculations in Fig.~\ref{fig:stability chart}(a) with various HV parameters. 
There are four parameters in the human driver model~\eqref{eq:OVM}:
sensitivity to desired speed $a$, sensitivity to leader's speed $b$, 
stopping gap $s_{\mathrm{st}}$, and free-driving gap $s_{\mathrm{go}}$. We have found that the plant stability boundaries only have small changes with different HV model parameters. We depict the head-to-tail string stability region with different HV parameters in Fig.~\ref{fig:robust stability chart}. In each subfigure, we plot the string stability boundaries when one parameter changes while keeping the other three parameters the same as calibrated in Section~\ref{sec:subsec:stability analysis}. The overlap between the string stable regions associated with different HV parameters is shaded in red.
The red domain shows that there exists a large region of nominal controller gains that renders the system string stable even when the human driver behavior is uncertain.
Choosing gains from this region, therefore, provides robustness against human driver uncertainty. 

\subsubsection*{Robust safety with inaccurate human driver model}

Finally, we study how safety can be guaranteed when the behavior of human drivers is uncertain.
We note that the safety constraints of the head CAV~\eqref{eq:CBF head CAV}, tail CAV~\eqref{eq:CBF tail CAV}, and platoon~\eqref{eq:CBF platoon} do not depend on the human driver model. Therefore, CAV safety and platoon safety can be enforced even if the human-driver model is unknown or inaccurate. The human driver uncertainty only affects the HV safety~\eqref{eq:CBF HV}.
Below we outline a method to provide robust safety guarantees for HVs even with uncertain human driver model.

Recall that the uncertainty of the human driver model can be captured by a disturbance $d_{i}$ in~\eqref{eq:system HV v} that represents the error between the actual acceleration $\dot{v}_i$ and its model $F_i$. If the disturbance has a known upper bound $\bar{d}_i$, i.e., $|d_i| = |\dot{v}_i - F_i(s_i,v_i,\dot{s}_i)|\le \bar{d}_i$, then the HV safety constraint~\eqref{eq:CBF HV} can be modified based on robust CBF theory~\cite{jankovic2018robust} to the robust constraint:
\begin{align} \label{eq:CBF HV robust}
    L_{f} \bar{h}_{i}(x,v_{\hhv}) + L_{g_{\headcav}} \bar{h}_{i}(x) u_{\headcav} - \tau_i \bar{d}_i \ge -\gamma_{i} \bar{h}_{i}(x).
\end{align}
This leads to safety guarantees for HVs even with model uncertainty, which is established in  Theorem~\ref{theorem:safety HV robust} in Appendix~\ref{appdx:safety}.

We validate this robust CBF by simulations for the scenario of one middle HV suddenly accelerating.
We compare simulations with the nominal controller~\eqref{eq:nominal controller head CAV},~\eqref{eq:nominal controller tail CAV};
the safety-critical controller~\eqref{eq:QP_head},~\eqref{eq:min_controller_tail};
and the robust safety-critical controller where the constraints in~\eqref{eq:QP_head} are replaced by~\eqref{eq:CBF HV robust}. We take the calibrated HV model $F_i$ from the previous simulations in Fig.~\ref{fig:trajectory HV acc}, and we use this model to calculate the left-hand side of the safety constraint~\eqref{eq:CBF HV} and the robust safety constraint~\eqref{eq:CBF HV robust}.
Then, we simulate the vehicle platoon considering HVs with different parameters: $a=0.2$, $b=0.6$, $s_{\mathrm{st}} = 8$ m, and $s_{\mathrm{go}} = 40$ m.  The remaining simulation parameters, including the nominal controller gains, CBF parameters, and sudden HV acceleration settings, are the same as in Fig.~\ref{fig:trajectory HV acc}. 
Fig.~\ref{fig:trajectory robust HV safety} shows the simulated trajectories by the nominal controller, CBF~\eqref{eq:CBF HV} and robust CBF~\eqref{eq:CBF HV robust} with $\bar{d}_i = 5 \; \mathrm{m/s^2}$. The robust CBF maintains HV safety, while the other two controllers fail to do so.

\section{Conclusion}

In this paper, we coordinate a pair of CAVs traveling amongst HVs in mixed traffic. Feedback controllers are designed for the two CAVs to utilize CAV cooperation and, possibly, connected HV feedback. Stability and safety conditions are derived for the controller gains, and the effect of CAV coordination and HV connection on stability and safety is analyzed. We find that both CAV coordination and HV connection have opposite effects on stability and safety: including CAV cooperation or HV connection makes it easier to stabilize traffic but harder to maintain safety. To overcome this trade-off, safety filters are designed using control barrier functions considering CAV safety, HV safety, and platoon safety.
The controller performance is analyzed via numerical simulations. With the proposed controller, the mixed vehicle platoon travels safely and also mitigates perturbations from downstream traffic. Future extensions of this research include designing robust controllers under V2V communication failure, considering lateral movement, and evaluating the controller from more perspectives such as comfort and fuel consumption.

\appendices
\setcounter{equation}{0} 
\renewcommand{\theequation}{A.\arabic{equation}}

\section{Stability Analysis}
\label{appdx:stability}

This Appendix provides the mathematical background required for the stability analysis in Section~\ref{sec:stability}, including the derivation of the linearized dynamics in~\eqref{eq:linearized_system} and the corresponding head-to-tail transfer function in Lemma~\ref{theorem:G}.

\subsection{Linearized dynamics}

We first derive the linearized dynamics~\eqref{eq:linearized_system}.
Considering the perturbations around the equilibrium as in \eqref{eq:perturbation}, the dynamics are as follows.
For middle HV-$i$, by linearizing~\eqref{eq:system HV s}-\eqref{eq:system HV v}, we obtain:
\begin{align}
    \dot{\tilde{s}}_i &= \tilde{v}_{i-1} - \tilde{v}_i, \label{eq:linear system HV s}\\
    \dot{\tilde{v}}_i &= a_{i1} \tilde{s}_i - a_{i2} \tilde{v}_{i} + a_{i3} \tilde{v}_{i-1}, \label{eq:linear system HV v}
\end{align}
where the coefficients are
${a_{i1}=\frac{\partial F_i}{\partial s_i}(s_i^*,v^*,0)}$, ${a_{i2}=\frac{\partial F_i}{\partial \dot{s}_i}(s_i^*,v^*,0)-\frac{\partial F_i}{\partial v_i}(s_i^*,v^*,0)}$, and ${a_{i3}=\frac{\partial F_i}{\partial \dot{s}_i}(s_i^*,v^*,0)}$.
For the head CAV, the linearization of~\eqref{eq:system head CAV s}-\eqref{eq:system head CAV v} leads to:
\begin{align}
    \dot{\tilde{s}}_{\headcav} &= \tilde v_{\hhv} - \tilde{v}_{\headcav}, 	\label{eq:linear system head s}\\
    \dot{\tilde{v}}_{\headcav} & = \xi_{\headcav} \tilde{s}_{\headcav} - \eta_{\headcav} \tilde{v}_\headcav + \beta_{\headcav,\hhv} \tilde{v}_{\hhv} + \sum_{i\in \mathcal{N}_{\headcav}} 
     \beta_{\headcav,i} \tilde{v}_{i}   + \beta_{\headcav,\tailcav},  \label{eq:linear system head v}
\end{align}
with
${\xi_{\headcav} = \alpha_{\headcav} V'_{\headcav}(s_{\headcav}^*)}$
and
${\eta_{\headcav} = \alpha_{\headcav} + \beta_{\headcav,{\hhv}} + \sum_{i\in \mathcal{N}_{\headcav}} \beta_{\headcav,i}  + \beta_{\headcav,\tailcav}}$.
For the tail CAV, linearizing~\eqref{eq:system tail CAV s}-\eqref{eq:system tail CAV v} yields:
\begin{align}
    \dot{\tilde{v}}_{\tailcav} & = \tilde{s}_{\HVn} - \tilde{s}_{\tailcav}, \label{eq:linear system tail s} \\ 
    \dot{\tilde{v}}_{\tailcav} & = \xi_{\tailcav} \tilde{s}_{\tailcav} - \eta_{\tailcav} \tilde{v}_{\tailcav} + \beta_{\tailcav,\HVn} \tilde{v}_{\HVn} + \sum_{i\in \mathcal{N}_{\tailcav}} 
	 \beta_{\tailcav,i} \tilde{v}_{i}   + \beta_{\tailcav,\headcav} \tilde{v}_{\headcav}, \label{eq:linear system tail v}
\end{align}
with
${\xi_\tailcav = \alpha_\tailcav V_{\tailcav}'(s_{\tailcav}^*)}$
and
${\eta_\tailcav = \alpha_\tailcav + \beta_{\tailcav,\HVn} + \sum_{i\in \mathcal{N}_{\tailcav}} \beta_{\tailcav,i} + \beta_{\tailcav,\headcav}}$.

The linearized dynamics can be written compactly as system~\eqref{eq:linearized_system}, where the matrices $A$ and $B$ are:
\begin{equation}\label{eq:linear AB}
    A = \left[ \begin{array}{cc:cc:cc:c:cc:cc}
        0 & -1 & 0 & 0 & 0 & 0 & \cdots & 0 & 0 & 0& 0 \\
        \xi_{\headcav} & -\eta_{\headcav} & 0 &\beta_{\headcav,1} & 0 &\beta_{\headcav,2} & \cdots & 0 & \beta_{\headcav,\HVn} & 0 & \beta_{\headcav,\tailcav} \\ \hdashline
        0& 1 &0 &-1  & 0 & 0 & \cdots &0&0&0&0 \\
        0 & a_{13} & a_{11} & -a_{12} &0&0&  \cdots &0&0&0&0 \\ \hdashline
        0&0& 0& 1 &0 &-1  & \cdots &0&0&0&0 \\
        0&0& 0 & a_{23} & a_{21} & -a_{22} & \cdots &0&0&0&0 \\ \hdashline
        &&&&&& \ddots &  &   &  &   \\ \hdashline
        0&0&0&0&0&0&  \cdots & 0& 1 &0 &-1   \\
        0 & \beta_{\tailcav,\headcav}& 0 & \beta_{\tailcav,1} & 0 & \beta_{\tailcav,2} & \cdots & 0 & \beta_{\tailcav,\HVn}& \xi_{\tailcav} & -\eta_{\tailcav} \\
    \end{array}\right], \quad
    B = \left[ \begin{array}{c}
        1  \\
        \beta_{\headcav,\hhv} \\  
        0  \\  
        0  \\
        0 \\  
       \vdots   \\  
       0   \\
      0    \\
    \end{array}  \right].
\end{equation}

\subsection{Proof of Lemma~\ref{theorem:G}}

We derive the head-to-tail transfer function $G(s)$ as follows.
For each middle HV-$i$, we take the Laplace transform of the linearized dynamics~\eqref{eq:linear system HV s}-\eqref{eq:linear system HV v} considering zero initial conditions:
\begin{align}
    & s\widetilde{S}_i = \widetilde{V}_{i-1} - \widetilde{V}_i, \\
    & s \widetilde{V}_i = a_{i1} \widetilde{S}_{i} -a_{i2} \widetilde{V}_{i} +a_{i3} \widetilde{V}_{i-1},
\end{align}
where $\widetilde{S}$ and $\widetilde{V}$ denote the Laplace transforms of $\tilde{s}$ and $\tilde{v}$ (with the corresponding subscript).
This gives the relationship between $\widetilde{V}_{i-1}$ and $\widetilde{V}_i$ as:
\begin{align}
    \widetilde{V}_{i} = \frac{a_{i3} s +a_{i1}}{s^2 + a_{i2} s + a_{i1} } \widetilde{V}_{i-1}.
\end{align}
Considering that $\widetilde{V}_0 = \widetilde{V}_{\headcav}$, we express each $\widetilde{V}_i$ from $\widetilde{V}_{\headcav}$ as:
\begin{align} \label{eq:proof stability Vi VH}
    \widetilde{V}_i = \Omega_i \widetilde{V}_{\headcav},
\end{align}
with $
    \Omega_i = \prod_{j=1}^i \frac{a_{j3} s +a_{j1}}{s^2 + a_{j2} s + a_{j1}}$.
Note that $\Omega_i= P_i/P_0$ based on~\eqref{eq:G Pi}.

For the head CAV, we have the Laplace transform of  the linearized dynamics~\eqref{eq:linear system head s}-\eqref{eq:linear system head v} as:
\begin{align}
    & s\widetilde{S}_{\headcav} = \widetilde{V}_{\hhv} - \widetilde{V}_{\headcav}, \\
    & s \widetilde{V}_{\headcav} = \xi_{\headcav} \widetilde{S}_{\headcav} - \eta_{\headcav} \widetilde{V}_\headcav + \beta_{\headcav,\hhv} \widetilde{V}_{\hhv} + \sum_{i\in \mathcal{N}_{\headcav}} 
     \beta_{\headcav,i} \widetilde{V}_{i}   + \beta_{\headcav,\tailcav} \widetilde{V}_{\tailcav},
\end{align}
which gives:
\begin{align}
    \widetilde{V}_{\headcav} = \frac{(\beta_{\headcav,\hhv} s+ \xi_{\headcav} ) \widetilde{V}_{\hhv} + \beta_{\headcav,\tailcav} s \widetilde{V}_{\tailcav} + \sum_{i\in \mathcal{N}_{\headcav}}\beta_{\headcav,i} s \widetilde{V}_i
 }{s^2 + \eta_{\headcav} s + \xi_{\headcav}}.
\end{align}
By substituting $\widetilde{V}_i$ from~\eqref{eq:proof stability Vi VH}, we obtain:
\begin{align} \label{eq:proof stability VH Vd}
    \widetilde{V}_{\headcav} = \frac{(\beta_{\headcav,\hhv} s+ \xi_{\headcav} ) \widetilde{V}_{\hhv} + \beta_{\headcav,\tailcav} s \widetilde{V}_{\tailcav}}{s^2 + \eta_{\headcav} s + \xi_{\headcav} - \sum_{i\in \mathcal{N}_{\headcav}}\beta_{\headcav,i} s \Omega_i}.
\end{align}

For the tail CAV, taking the Laplace transform of the linearized dynamics~\eqref{eq:linear system tail s}-\eqref{eq:linear system tail v} leads to:
\begin{align}
    s\widetilde{S}_{\tailcav} &= \widetilde{V}_{\HVn} - \widetilde{V}_{\tailcav}, \\
    s \widetilde{V}_{\tailcav} &= \xi_{\tailcav} \widetilde{S}_{\tailcav} - \eta_{\tailcav} \widetilde{V}_\tailcav + \beta_{\tailcav,\HVn} \widetilde{V}_{\HVn} + \sum_{i\in \mathcal{N}_{\tailcav}} 
     \beta_{\tailcav,i} \widetilde{V}_{i}   + \beta_{\tailcav,\headcav} \widetilde{V}_{\headcav},
\end{align}
which gives:
\begin{align}
     \widetilde{V}_{\tailcav} = \frac{\beta_{\tailcav,\headcav}s \widetilde{V}_{\headcav} + (\beta_{\tailcav,\HVn}s + \xi_{\tailcav}) \widetilde{V}_{\HVn} + \sum_{i\in \mathcal{N}_{\tailcav}} \beta_{\tailcav,i} s \widetilde{V}_i }{s^2 + \eta_{\tailcav} s + \xi_{\tailcav}}.
\end{align}
By substituting $\widetilde{V}_i$ from~\eqref{eq:proof stability Vi VH}, we get:
\begin{align}\label{eq:proof stability VH VT}
     \widetilde{V}_{\headcav}  = \frac{(s^2 + \eta_{\tailcav} s + \xi_{\tailcav}) \widetilde{V}_{\tailcav}} {\beta_{\tailcav,\headcav}s   + (\beta_{\tailcav,\HVn}s + \xi_{\tailcav}) \Omega_{\HVn} + \sum_{i\in \mathcal{N}_{\tailcav}} \beta_{\tailcav,i} s \Omega_i}.
\end{align}

Equating~\eqref{eq:proof stability VH Vd} and~\eqref{eq:proof stability VH VT} leads to:
\begin{align}
    \frac{(\beta_{\headcav,\hhv} s+ \xi_{\headcav} ) \widetilde{V}_{\hhv} + \beta_{\headcav,\tailcav} s \widetilde{V}_{\tailcav}}{s^2 + \eta_{\headcav} s + \xi_{\headcav} - \sum_{i\in \mathcal{N}_{\headcav}}\beta_{\headcav,i} s \Omega_i}
    =
    \frac{(s^2 + \eta_{\tailcav} s + \xi_{\tailcav}) \widetilde{V}_{\tailcav}} {\beta_{\tailcav,\headcav}s   + (\beta_{\tailcav,\HVn}s + \xi_{\tailcav}) \Omega_{\HVn} + \sum_{i\in \mathcal{N}_{\tailcav}} \beta_{\tailcav,i} s \Omega_i}.
\end{align}
This can be rearranged to obtain the head-to-tail transfer function defined in~\eqref{eq:G_def}:
\begin{align}
    G(s) = \frac{(\beta_{\headcav,\hhv} s+ \xi_{\headcav} )(\beta_{\tailcav,\headcav}s   + (\beta_{\tailcav,\HVn}s + \xi_{\tailcav}) \Omega_{\HVn} + \sum_{i\in \mathcal{N}_{\tailcav}} \beta_{\tailcav,i} s \Omega_i)  }{\left( s^2 + \eta_{\headcav} s + \xi_{\headcav} - \sum_{i\in \mathcal{N}_{\headcav}}\beta_{\headcav,i} s \Omega_i \right) (s^2 + \eta_{\tailcav} s + \xi_{\tailcav})  -\beta_{\headcav,\tailcav}s (\beta_{\tailcav,\headcav}s   + (\beta_{\tailcav,\HVn}s + \xi_{\tailcav}) \Omega_{\HVn} + \sum_{i\in \mathcal{N}_{\tailcav}} \beta_{\tailcav,i} s \Omega_i)}.
\end{align}
Since $\Omega_i= P_i/P_0$ based on~\eqref{eq:G Pi}, we multiply both the numerator and denominator by $P_0$, and we get~\eqref{eq:G}-\eqref{eq:G D}.

\setcounter{equation}{0} 
\renewcommand{\theequation}{B.\arabic{equation}}

\section{Safety Analysis}
\label{appdx:safety}
In this Appendix, we first prove Theorem~\ref{theorem:safety nominal head} that provides safe controller gains for the nominal controller in Section~\ref{sec:safety}.
Then, we establish robust CBF constraints for HV safety, which are utilized in~\eqref{eq:CBF HV robust} in Section~\ref{sec:performance analysis}.

\subsection{Proof of Theorem~\ref{theorem:safety nominal head}}
\label{appdx:nominal_controller_safety}

\begin{proof}
We prove safety based on Lemma~\ref{theorem:nagumo safety}, by showing that ${\dot{h}_{\headcav}(x,v_{\hhv}) \ge 0}$ holds if ${h_{\headcav}(x) = 0}$.
We express $\dot{h}_{\headcav}(x,v_{\hhv})$:
\begin{equation}
    \dot {h}_{\headcav}(x,v_{\hhv}) = v_{\hhv} - v_{\headcav} -\tau_{\headcav} u_{\headcav},
\end{equation}
where we substitute the nominal controller ${u_{\headcav}=k_{\headcav,{\rm n}}(x,v_{\hhv})}$ from~\eqref{eq:nominal controller head CAV}:
\begin{equation}
    \dot {h}_{\headcav}(x,v_{\hhv}) =
    v_{\hhv} - v_{\headcav}
    -\tau_{\headcav} \bigg(
    \alpha_{\headcav} (V_\headcav(s_\headcav) - v_{\headcav})
    + \beta_{\headcav,{\hhv}}(W(v_{\hhv}) - v_{\headcav})
    + \sum_{i \in \mathcal{N}_{\headcav}} \beta_{\headcav,i} (W(v_i) - v_{\headcav})
    + \beta_{\headcav,\tailcav} (W(v_{\tailcav}) - v_{\headcav}) \bigg).
\end{equation}
Then, we consider $v_{\hhv}, v_{\headcav}, v_{i}, v_{\tailcav} \in [0,v_{\max}]$, $s_{\headcav} \in [s_{\mathrm{st}},s_{\mathrm{go}}]$, and we substitute $W$ and $V_{\headcav}$ from~\eqref{eq:W} and~\eqref{eq:Vs}:
\begin{equation}
     \dot {h}_{\headcav}(x,v_{\hhv}) =
     v_{\hhv} - v_{\headcav}
     -\tau_{\headcav} \bigg(
     \alpha_{\headcav} \big( \kappa (s_{\headcav}  - s_{\mathrm{st}}) - v_{\headcav} \big)
     + \beta_{\headcav,{\hhv}}(v_{\hhv} - v_{\headcav})
     + \sum_{i \in \mathcal{N}_{\headcav}} \beta_{\headcav,i} (v_i - v_{\headcav})
     + \beta_{\headcav,\tailcav} (v_{\tailcav} - v_{\headcav}) \bigg).
\end{equation}
Next, we rearrange the terms and consider ${h_{\headcav}(x) = 0}$, that is, ${s_{\headcav} = \tau_{\headcav} v_{\headcav}}$:
\begin{equation}
     \dot {h}_{\headcav}(x,v_{\hhv}) =
     (1 - \tau_{\headcav} \beta_{\headcav,{\hhv}}) (v_{\hhv} - v_{\headcav})
     -\alpha_{\headcav} \big( \kappa \tau_{\headcav} (s_{\headcav}  - s_{\mathrm{st}}) - s_{\headcav} \big)
     -\tau_{\headcav} \bigg(
     \sum_{i \in \mathcal{N}_{\headcav}} \beta_{\headcav,i} (v_i - v_{\headcav})
     + \beta_{\headcav,\tailcav} (v_{\tailcav} - v_{\headcav}) \bigg).
\end{equation}
Since $v_{\hhv}, v_{\headcav}, v_{i}, v_{\tailcav} \in [0,v_{\max}]$, we have
${|v_{\hhv} - v_{\headcav}| \le v_{\max}}$,
${|v_{i} - v_{\headcav}| \le v_{\max}}$, and
${|v_{\tailcav} - v_{\headcav}| \le v_{\max}}$.
Furthermore, since ${s_{\headcav} \in [s_{\mathrm{st}},s_{\mathrm{go}}]}$ and ${\kappa \leq 1/\tau_{\headcav}}$, we have
${\kappa \tau_{\headcav} (s_{\headcav}  - s_{\mathrm{st}}) - s_{\headcav} \le -s_{\mathrm{st}}}$.
Substituting these and using ${\alpha_{\headcav} \ge 0}$ leads to:
\begin{equation}
     \dot {h}_{\headcav}(x,v_{\hhv}) \ge
     -  |  1-\tau_{\headcav} \beta_{\headcav,{\hhv}}| v_{\max}
     +\alpha_{\headcav} s_{\mathrm{st}}
     -\tau_{\headcav} \bigg(
     \sum_{i \in \mathcal{N}_{\headcav}} |\beta_{\headcav,i}|
     + |\beta_{\headcav,\tailcav}| \bigg) v_{\max}.
\end{equation}
Finally, substituting~\eqref{eq:head_CAV_safe_gains} leads to ${\dot{h}_{\headcav}(x,v_{\hhv}) \ge 0}$ which completes the proof.
\end{proof}

\subsection{Robust CBF constraints for HV safety}

Finally, we discuss how the head CAV's controller may ensure HV safety when the human driver model is inaccurate, which is captured by nonzero disturbance $d \ne 0$ in~\eqref{eq:system} and $d_i \ne 0$ in~\eqref{eq:system HV v}.
To this end, we briefly discuss robust safety-critical control for systems with disturbance, based on robust CBF theory~\cite{jankovic2018robust}.

Consider system~\eqref{eq:controlafine} with an additive disturbance $d \in \mathcal{D} \subset \mathbb{R}^n$:
\begin{equation}
    \dot{x} = f(x) + g(x)u + d,
\end{equation}
cf.~\eqref{eq:system}.
Analogously to~\eqref{eq:CBF_constraint} in Theorem~\ref{theorem:safety}, it can be stated that controllers $u = k(x)$ satisfying:
\begin{align} \label{eq:CBF_constraint_disturbance}
    L_f h(x) + L_gh(x) k(x) + \nabla h(x) \cdot d \ge -\gamma(h(x)),
    \quad \forall x \in \mathcal{C},
    \quad \forall d \in \mathcal{D},
\end{align}
render the set $\mathcal{C}$ forward invariant (safe) for the corresponding closed-loop system with disturbance. The difficulty of satisfying this constraint is that the disturbance $d$ may be unknown.
However, if the disturbance has a known bound $\bar{d}>0$, that is, if $\|d\|_{\infty} \le \bar{d}$ holds, then robust CBF theory provides the following sufficient condition for a safe controller~\cite{jankovic2018robust}:
\begin{align} \label{eq:CBF_constraint_robust}
    L_f h(x) + L_gh(x) k(x) - \| \nabla h(x) \| \bar{d} \ge -\gamma(h(x)),
    \quad \forall x \in \mathcal{C},
\end{align}
which implies that~\eqref{eq:CBF_constraint_disturbance} holds.

In the context of guaranteeing HV safety, the following modification of~\eqref{eq:CBF HV} can be used as robust CBF constraint.

\begin{theorem}[Robust HV safety with bounded HV model error]\label{theorem:safety HV robust}
    Consider system~\eqref{eq:system} with disturbance given in~\eqref{eq:system_expressions}.
    Assume that there is a known bound $\bar{d}_i \in \mathbb{R}$ for the disturbance $d_i$, i.e., $|d_i| \le \bar{d}_i$.
    If the controller of the head CAV satisfies \eqref{eq:CBF head CAV} and:
    \begin{align} \label{eq:CBF HV robust appdx}
    L_{f} \bar{h}_{i}(x,v_{\hhv}) + L_{g_{\headcav}} \bar{h}_{i}(x) u_{\headcav} - \tau_i \bar{d}_i \ge -\gamma_{i} \bar{h}_{i}(x),
    \end{align}
    with $\gamma_i>0$,
    then HV-$i$ is safe w.r.t. the CTH policy defined in~\eqref{eq:safety_function_hv}.
\end{theorem}
\begin{proof}
Considering the system~\eqref{eq:system} and~\eqref{eq:system_expressions} with the safety function $\bar{h}_i$ in~\eqref{eq:safety_function_hv_valid}, the CBF constraint corresponding to~\eqref{eq:CBF_constraint_disturbance} is:
\begin{align} \label{eq:CBF HV disturbance}
    L_{f} \bar{h}_{i}(x,v_{\hhv}) + L_{g_{\headcav}} \bar{h}_{i}(x) u_{\headcav} + \frac{\partial \bar{h}_i}{\partial v_i} d_i \ge -\gamma_{i} \bar{h}_{i}(x),
\end{align}
where $\frac{\partial \bar{h}_i}{\partial v_i} = -\tau_i$.
Since the disturbance $d_i$ is upper bounded by $\bar{d}_i$, we have $\frac{\partial \bar{h}_i}{\partial v_i} d_i \ge - \tau_i \bar{d}_i$.
Therefore,~\eqref{eq:CBF HV robust appdx} implies~\eqref{eq:CBF HV disturbance} and leads to a sufficient condition for guaranteeing safety w.r.t.~$\bar{h}_i$, analogously to~\eqref{eq:CBF_constraint_robust}.
Furthermore, \eqref{eq:CBF head CAV} ensures safety w.r.t.~$h_{\headcav}$ defined in~\eqref{eq:safety_function_headcav}.
Based on~\eqref{eq:safety_function_hv_valid}, ensuring both $h_{\headcav}(x)\ge 0 $ and $\bar{h}_i(x)\ge 0$ implies $h_i(x) \ge 0$, which means safety w.r.t. the CTH policy defined in~\eqref{eq:safety_function_hv}.
\end{proof}

\begin{remark}[Effect of HV model error on safety constraints]
Compared with the HV safety constraint~\eqref{eq:CBF HV}, the robust HV safety constraint~\eqref{eq:CBF HV robust appdx} has an extra safety margin term $\tau_i \bar{d}_i$, which accounts for the human model error.   The safety margin is proportional to the model error $\bar{d}_i$, which means a less accurate human driver model will cause a more aggressive control strategy. Specifically, when $\bar{d}_i = 0$, the robust safety constraint~\eqref{eq:CBF HV robust appdx} becomes the same as the original safety constraint~\eqref{eq:CBF HV}. 
\end{remark}

\bibliographystyle{plain}
\bibliography{ref}

\begin{thebibliography}{10}

\bibitem{alan2024integrating}
Anil Alan, Chaozhe~R He, Tamas~G Molnar, Johaan~C Mathew, A~Harvey Bell, and
  G{\'a}bor Orosz.
\newblock Integrating safety with performance in connected automated truck
  control: {Experimental} validation.
\newblock {\em IEEE Transactions on Intelligent Vehicles}, 9(1):3075--3088,
  2024.

\bibitem{althoff2014online}
Matthias Althoff and John~M Dolan.
\newblock Online verification of automated road vehicles using reachability
  analysis.
\newblock {\em IEEE Transactions on Robotics}, 30(4):903--918, 2014.

\bibitem{ames2019control}
Aaron~D Ames, Samuel Coogan, Magnus Egerstedt, Gennaro Notomista, Koushil
  Sreenath, and Paulo Tabuada.
\newblock Control barrier functions: {Theory} and applications.
\newblock In {\em 18th European Control Conference}, pages 3420--3431. IEEE,
  2019.

\bibitem{ames2014control}
Aaron~D Ames, Jessy~W Grizzle, and Paulo Tabuada.
\newblock Control barrier function based quadratic programs with application to
  adaptive cruise control.
\newblock In {\em 53rd IEEE Conference on Decision and Control}, pages
  6271--6278. IEEE, 2014.

\bibitem{bai2022robust}
Weiqi Bai, Bin Xu, Hui Liu, Yechen Qin, and Changle Xiang.
\newblock Robust longitudinal distributed model predictive control of connected
  and automated vehicles with coupled safety constraints.
\newblock {\em IEEE Transactions on Vehicular Technology}, 72(3):2960--2973,
  2022.

\bibitem{bai2022hybrid}
Zhengwei Bai, Peng Hao, Wei Shangguan, Baigen Cai, and Matthew~J Barth.
\newblock Hybrid reinforcement learning-based eco-driving strategy for
  connected and automated vehicles at signalized intersections.
\newblock {\em IEEE Transactions on Intelligent Transportation Systems},
  23(9):15850--15863, 2022.

\bibitem{Bando1998}
Masako Bando, Katsuya Hasebe, Ken Nakanishi, and Akihiro Nakayama.
\newblock Analysis of optimal velocity model with explicit delay.
\newblock {\em Physical Review E}, 58(5):5429--5435, 1998.

\bibitem{bekiaris2023robust}
Nikolaos Bekiaris-Liberis.
\newblock Robust string stability and safety of {CTH} predictor-feedback
  {CACC}.
\newblock {\em IEEE Transactions on Intelligent Transportation Systems},
  24(8):8209--8221, 2023.

\bibitem{brunner2022comparing}
Johannes~S Brunner, Michail~A Makridis, and Anastasios Kouvelas.
\newblock Comparing the observable response times of {ACC} and {CACC} systems.
\newblock {\em IEEE Transactions on Intelligent Transportation Systems},
  23(10):19299--19308, 2022.

\bibitem{chen2023nearly}
Xiangdong Chen, Xi~Lin, Meng Li, Fang He, and Qiang Meng.
\newblock A nearly throughput-maximum knotted intersection design and control
  for connected and automated vehicles.
\newblock {\em Transportation Research Part B: Methodological}, 171:44--79,
  2023.

\bibitem{dai2022exploring}
Yulu Dai, Yuwei Yang, Zhiyuan Wang, and YinJie Luo.
\newblock Exploring the impact of damping on connected and autonomous vehicle
  platoon safety with {CACC}.
\newblock {\em Physica A: Statistical Mechanics and its Applications},
  607:128181, 2022.

\bibitem{dey2015review}
Kakan~C Dey, Li~Yan, Xujie Wang, Yue Wang, Haiying Shen, Mashrur Chowdhury, Lei
  Yu, Chenxi Qiu, and Vivekgautham Soundararaj.
\newblock A review of communication, driver characteristics, and controls
  aspects of cooperative adaptive cruise control ({CACC}).
\newblock {\em IEEE Transactions on Intelligent Transportation Systems},
  17(2):491--509, 2015.

\bibitem{di2019cooperative}
Marco Di~Vaio, Giovanni Fiengo, Alberto Petrillo, Alessandro Salvi, Stefania
  Santini, and Manuela Tufo.
\newblock Cooperative shock waves mitigation in mixed traffic flow environment.
\newblock {\em IEEE Transactions on Intelligent Transportation Systems},
  20(12):4339--4353, 2019.

\bibitem{ding2020penetration}
Jishiyu Ding, Huei Peng, Yi~Zhang, and Li~Li.
\newblock Penetration effect of connected and automated vehicles on cooperative
  on-ramp merging.
\newblock {\em IET Intelligent Transport Systems}, 14(1):56--64, 2020.

\bibitem{do2019simulation}
Wooseok Do, Omid~M Rouhani, Luis Miranda-Moreno, et~al.
\newblock Simulation-based connected and automated vehicle models on highway
  sections: {A} literature review.
\newblock {\em Journal of Advanced Transportation}, 2019:9343705, 2019.

\bibitem{NGSIM}
FHWA.
\newblock {Next Generation Simulation (NGSIM)}, 2007.

\bibitem{garg2023can}
Mohit Garg and M{\'e}lanie Bouroche.
\newblock Can connected autonomous vehicles improve mixed traffic safety
  without compromising efficiency in realistic scenarios?
\newblock {\em IEEE Transactions on Intelligent Transportation Systems},
  24(6):6674--6689, 2023.

\bibitem{giammarino2020traffic}
Vittorio Giammarino, Simone Baldi, Paolo Frasca, and Maria~Laura Delle~Monache.
\newblock Traffic flow on a ring with a single autonomous vehicle: {An}
  interconnected stability perspective.
\newblock {\em IEEE Transactions on Intelligent Transportation Systems},
  22(8):4998--5008, 2020.

\bibitem{gong2018cooperative}
Siyuan Gong and Lili Du.
\newblock Cooperative platoon control for a mixed traffic flow including human
  drive vehicles and connected and autonomous vehicles.
\newblock {\em Transportation Research Part B: Methodological}, 116:25--61,
  2018.

\bibitem{gunter2020commercially}
George Gunter, Derek Gloudemans, Raphael~E Stern, Sean McQuade, Rahul Bhadani,
  Matt Bunting, Maria~Laura Delle~Monache, Roman Lysecky, Benjamin Seibold,
  Jonathan Sprinkle, et~al.
\newblock Are commercially implemented adaptive cruise control systems string
  stable?
\newblock {\em IEEE Transactions on Intelligent Transportation Systems},
  22(11):6992--7003, 2020.

\bibitem{guo2023connected}
Sicong Guo, G{\'a}bor Orosz, and Tamas~G Molnar.
\newblock Connected cruise and traffic control for pairs of connected automated
  vehicles.
\newblock {\em IEEE Transactions on Intelligent Transportation Systems},
  24(11):12648--12658, 2023.

\bibitem{han2020energy}
Xiao Han, Rui Ma, and H~Michael Zhang.
\newblock Energy-aware trajectory optimization of {CAV} platoons through a
  signalized intersection.
\newblock {\em Transportation Research Part C: Emerging Technologies},
  118:102652, 2020.

\bibitem{jankovic2018robust}
Mrdjan Jankovic.
\newblock Robust control barrier functions for constrained stabilization of
  nonlinear systems.
\newblock {\em Automatica}, 96:359--367, 2018.

\bibitem{jin2018experimental}
I~Ge Jin, Sergei~S Avedisov, Chaozhe~R He, Wubing~B Qin, Mehdi Sadeghpour, and
  G{\'a}bor Orosz.
\newblock Experimental validation of connected automated vehicle design among
  human-driven vehicles.
\newblock {\em Transportation Research Part C: Emerging Technologies},
  91:335--352, 2018.

\bibitem{jin2020analysis}
Li~Jin, Mladen {\v{C}}i{\v{c}}i{\'c}, Karl~H Johansson, and Saurabh Amin.
\newblock Analysis and design of vehicle platooning operations on mixed-traffic
  highways.
\newblock {\em IEEE Transactions on Automatic Control}, 66(10):4715--4730,
  2020.

\bibitem{johansson2023platoon}
Alexander Johansson, Ting Bai, Karl~Henrik Johansson, and Jonas M{\aa}rtensson.
\newblock Platoon cooperation across carriers: From system architecture to
  coordination.
\newblock {\em IEEE Intelligent Transportation Systems Magazine},
  15(3):132--144, 2023.

\bibitem{kim2021compact}
Yeojun Kim, Jacopo Guanetti, and Francesco Borrelli.
\newblock Compact cooperative adaptive cruise control for energy saving: Air
  drag modelling and simulation.
\newblock {\em IEEE Transactions on Vehicular Technology}, 70(10):9838--9848,
  2021.

\bibitem{li2022trade}
Xiaopeng Li.
\newblock Trade-off between safety, mobility and stability in automated vehicle
  following control: {An} analytical method.
\newblock {\em Transportation Research Part B: Methodological}, 166:1--18,
  2022.

\bibitem{liu2020mobility}
Hao Liu, Xiao-Yun Lu, and Steven~E Shladover.
\newblock Mobility and energy consumption impacts of cooperative adaptive
  cruise control vehicle strings on freeway corridors.
\newblock {\em Transportation Research Record}, 2674(9):111--123, 2020.

\bibitem{molnar2023safetyCCC}
Tamas~G Molnar, G{\'a}bor Orosz, and Aaron~D Ames.
\newblock On the safety of connected cruise control: {Analysis} and synthesis
  with control barrier functions.
\newblock In {\em 62nd IEEE Conference on Decision and Control}, pages
  1106--1111. IEEE, 2023.

\bibitem{mousavi2021investigating}
Seyedeh~Maryam Mousavi, Osama~A Osman, Dominique Lord, Karen~K Dixon, and Bahar
  Dadashova.
\newblock Investigating the safety and operational benefits of mixed traffic
  environments with different automated vehicle market penetration rates in the
  proximity of a driveway on an urban arterial.
\newblock {\em Accident Analysis \& Prevention}, 152:105982, 2021.

\bibitem{nagumo1942lage}
Mitio Nagumo.
\newblock {\"U}ber die lage der integralkurven gew{\"o}hnlicher
  differentialgleichungen.
\newblock {\em Proceedings of the Physico-Mathematical Society of Japan. 3rd
  Series}, 24:551--559, 1942.

\bibitem{orosz2016connected}
G{\'a}bor Orosz.
\newblock Connected cruise control: {Modelling}, delay effects, and nonlinear
  behaviour.
\newblock {\em Vehicle System Dynamics}, 54(8):1147--1176, 2016.

\bibitem{papadoulis2019evaluating}
Alkis Papadoulis, Mohammed Quddus, and Marianna Imprialou.
\newblock Evaluating the safety impact of connected and autonomous vehicles on
  motorways.
\newblock {\em Accident Analysis \& Prevention}, 124:12--22, 2019.

\bibitem{qiu2023cooperative}
Jiahua Qiu and Lili Du.
\newblock Cooperative trajectory control for synchronizing the movement of two
  connected and autonomous vehicles separated in a mixed traffic flow.
\newblock {\em Transportation Research Part B: Methodological}, 174:102769,
  2023.

\bibitem{rios2016automated}
Jackeline Rios-Torres and Andreas~A Malikopoulos.
\newblock Automated and cooperative vehicle merging at highway on-ramps.
\newblock {\em IEEE Transactions on Intelligent Transportation Systems},
  18(4):780--789, 2016.

\bibitem{shi2021connected}
Haotian Shi, Yang Zhou, Keshu Wu, Xin Wang, Yangxin Lin, and Bin Ran.
\newblock Connected automated vehicle cooperative control with a deep
  reinforcement learning approach in a mixed traffic environment.
\newblock {\em Transportation Research Part C: Emerging Technologies},
  133:103421, 2021.

\bibitem{shladover2012impacts}
Steven~E Shladover, Dongyan Su, and Xiao-Yun Lu.
\newblock Impacts of cooperative adaptive cruise control on freeway traffic
  flow.
\newblock {\em Transportation Research Record}, 2324(1):63--70, 2012.

\bibitem{stern2018dissipation}
Raphael~E Stern, Shumo Cui, Maria~Laura Delle~Monache, Rahul Bhadani, Matt
  Bunting, Miles Churchill, Nathaniel Hamilton, Hannah Pohlmann, Fangyu Wu,
  Benedetto Piccoli, et~al.
\newblock Dissipation of stop-and-go waves via control of autonomous vehicles:
  {Field} experiments.
\newblock {\em Transportation Research Part C: Emerging Technologies},
  89:205--221, 2018.

\bibitem{talebpour2016influence}
Alireza Talebpour and Hani~S Mahmassani.
\newblock Influence of connected and autonomous vehicles on traffic flow
  stability and throughput.
\newblock {\em Transportation Research Part C: Emerging Technologies},
  71:143--163, 2016.

\bibitem{treiber2000congested}
Martin Treiber, Ansgar Hennecke, and Dirk Helbing.
\newblock Congested traffic states in empirical observations and microscopic
  simulations.
\newblock {\em Physical Review E}, 62:1805--1824, 2000.

\bibitem{wang2021review}
Chen Wang, Yuanchang Xie, Helai Huang, and Pan Liu.
\newblock A review of surrogate safety measures and their applications in
  connected and automated vehicles safety modeling.
\newblock {\em Accident Analysis \& Prevention}, 157:106157, 2021.

\bibitem{wang2024robust}
Jian Wang, Anye Zhou, Zhiyuan Liu, and Srinivas Peeta.
\newblock Robust cooperative control strategy for a platoon of connected and
  autonomous vehicles against sensor errors and control errors simultaneously
  in a real-world driving environment.
\newblock {\em Transportation Research Part B: Methodological}, 184:102946,
  2024.

\bibitem{wang2021leading}
Jiawei Wang, Yang Zheng, Chaoyi Chen, Qing Xu, and Keqiang Li.
\newblock Leading cruise control in mixed traffic flow: System modeling,
  controllability, and string stability.
\newblock {\em IEEE Transactions on Intelligent Transportation Systems},
  23(8):12861--12876, 2022.

\bibitem{wang2018distributed}
Ziran Wang, Guoyuan Wu, and Matthew Barth.
\newblock Distributed consensus-based cooperative highway on-ramp merging using
  {V2X} communications.
\newblock Technical report, SAE Technical Paper, 2018.

\bibitem{xiao2021bridging}
Wei Xiao, Christos~G Cassandras, and Calin~A Belta.
\newblock Bridging the gap between optimal trajectory planning and
  safety-critical control with applications to autonomous vehicles.
\newblock {\em Automatica}, 129:109592, 2021.

\bibitem{ye2019evaluating}
Lanhang Ye and Toshiyuki Yamamoto.
\newblock Evaluating the impact of connected and autonomous vehicles on traffic
  safety.
\newblock {\em Physica A: Statistical Mechanics and its Applications},
  526:121009, 2019.

\bibitem{yi2020using}
Ziwei Yi, Linheng Li, Xu~Qu, Yang Hong, Peipei Mao, and Bin Ran.
\newblock Using artificial potential field theory for a cooperative control
  model in a connected and automated vehicles environment.
\newblock {\em Transportation Research Record}, 2674(9):1005--1018, 2020.

\bibitem{yu2018stabilization}
Huan Yu, Shumon Koga, and Miroslav Krstic.
\newblock Stabilization of traffic flow with a leading autonomous vehicle.
\newblock In {\em Dynamic Systems and Control Conference}, volume 51906, page
  V002T22A006. American Society of Mechanical Engineers, 2018.

\bibitem{zhang2022hybrid}
Hanyu Zhang, Lili Du, and Jinglai Shen.
\newblock Hybrid {MPC} system for platoon based cooperative lane change control
  using machine learning aided distributed optimization.
\newblock {\em Transportation Research Part B: Methodological}, 159:104--142,
  2022.

\bibitem{zhao2024safetyACC}
Chenguang Zhao, Tamas~G Molnar, and Huan Yu.
\newblock Safety-critical stabilization of mixed traffic by pairs of {CAV}s.
\newblock In {\em 2024 American Control Conference}. IEEE, 2024.
\newblock Accepted.

\bibitem{zhao2023safety}
Chenguang Zhao, Huan Yu, and Tamas~G Molnar.
\newblock Safety-critical traffic control by connected automated vehicles.
\newblock {\em Transportation Research Part C: Emerging Technologies},
  154:104230, 2023.

\bibitem{zhao2022formal}
Tong Zhao, Ekim Yurtsever, Joel~A Paulson, and Giorgio Rizzoni.
\newblock Formal certification methods for automated vehicle safety assessment.
\newblock {\em IEEE Transactions on Intelligent Vehicles}, 8(1):232--249, 2022.

\bibitem{zheng2019cooperative}
Yuan Zheng, Bin Ran, Xu~Qu, Jian Zhang, and Yi~Lin.
\newblock Cooperative lane changing strategies to improve traffic operation and
  safety nearby freeway off-ramps in a connected and automated vehicles
  environment.
\newblock {\em IEEE Transactions on Intelligent Transportation Systems},
  21(11):4605--4614, 2019.

\bibitem{zhou2021analytical}
Jiazu Zhou and Feng Zhu.
\newblock Analytical analysis of the effect of maximum platoon size of
  connected and automated vehicles.
\newblock {\em Transportation Research Part C: Emerging Technologies},
  122:102882, 2021.

\bibitem{zhou2023data}
Yang Zhou, Xinzhi Zhong, Qian Chen, Soyoung Ahn, Jiwan Jiang, and Ghazaleh
  Jafarsalehi.
\newblock Data-driven analysis for disturbance amplification in car-following
  behavior of automated vehicles.
\newblock {\em Transportation Research Part B: Methodological}, 174:102768,
  2023.

\bibitem{zhou2023autonomous}
Zhi Zhou, Linheng Li, Xu~Qu, and Bin Ran.
\newblock An autonomous platoon formation strategy to optimize cav
  car-following stability under periodic disturbance.
\newblock {\em Physica A: Statistical Mechanics and its Applications},
  626:129096, 2023.

\end{thebibliography}

\end{document}